\documentclass[acmsmall,screen,nonacm]{acmart}
\newif\ifarxivPreprint
\arxivPreprinttrue

\setcopyright{none}
\AtBeginDocument{%
}

\usepackage{mathtools}
\usepackage[capitalise,noabbrev,nameinlink]{cleveref}
\usepackage{xspace}
\usepackage{tikz}
\usetikzlibrary{arrows.meta,positioning,fit,backgrounds,tikzmark,decorations.pathreplacing,automata,bending,tikzmark}
\usepackage{yfonts}
\usepackage[colorinlistoftodos,prependcaption,textsize=tiny]{todonotes}
\usepackage{colortbl}
\usepackage{xcolor}
\usepackage{multirow}
\usepackage{makecell}
\usepackage{array}
\usepackage{stmaryrd}
\usepackage{ebproof}
\usepackage{subcaption}
\usepackage{tabularx}
\usepackage{enumitem}
\usepackage{eqparbox}
\usepackage{adjustbox}
\usepackage{pdflscape}
\usepackage{thm-restate}
\usepackage{nicefrac}
\usepackage{wrapfig}
\usepackage{transparent}
\usepackage{hhline}

\crefname{theorem}{theorem}{theorems}
\Crefname{theorem}{Theorem}{Theorems}
\crefname{definition}{definition}{definitions}
\Crefname{definition}{Definition}{Definitions}
\crefname{lemma}{lemma}{lemmas}
\Crefname{lemma}{Lemma}{Lemmas}
\crefname{prop}{proposition}{propositions}
\Crefname{prop}{Proposition}{Propositions}
\crefname{corollary}{corollary}{corollaries}
\Crefname{corollary}{Corollary}{Corollaries}
\crefname{statement}{statement}{statements}
\crefname{example}{example}{examples}
\Crefname{example}{Example}{Examples}
\Crefname{statement}{Statement}{Statements}
\crefname{implication}{implication}{implications}
\Crefname{implication}{Implication}{Implications}
\crefname{equation}{equation}{equations}
\Crefname{equation}{Equation}{Equations}
\crefname{figure}{figure}{figures}
\Crefname{figure}{Figure}{Figures}
\crefname{table}{table}{tables}
\Crefname{table}{Table}{Tables}
\crefname{example}{example}{examples}
\Crefname{example}{Example}{Examples}

\newcommand{\withabbrev}[2]{%
  \begingroup
  \crefname{theorem}{thm.}{thms.}%
  \Crefname{theorem}{Thm.}{Thms.}%
  \crefname{definition}{def.}{defs.}%
  \Crefname{definition}{Def.}{Defs.}%
  \crefname{lemma}{lem.}{lems.}%
  \Crefname{lemma}{Lem.}{Lems.}%
  \crefname{prop}{prop.}{props.}%
  \Crefname{prop}{Prop.}{Props.}%
  \crefname{corollary}{cor.}{cors.}%
  \Crefname{corollary}{Cor.}{Cors.}%
  \crefname{statement}{stmt.}{stmts.}%
  \Crefname{statement}{Stmt.}{Stmts.}%
  \crefname{implication}{impl.}{impls.}%
  \Crefname{implication}{Impl.}{Impls.}%
  \crefname{equation}{eq.}{eqs.}%
  \Crefname{equation}{Eq.}{Eqs.}%
  \crefname{figure}{fig.}{figs.}%
  \Crefname{figure}{Fig.}{Figs.}%
  \crefname{table}{tab.}{tabs.}%
  \Crefname{table}{Tab.}{Tabs.}%
  \crefname{example}{ex.}{exs.}%
  \Crefname{example}{Ex.}{Exs.}%
  #1{#2}%
  \endgroup
}

\newcommand{\abCref}[1]{\withabbrev{\Cref}{#1}}

\newcommand{\fullref}[1]{\Cref{#1}: \nameref{#1}}
\newcommand{\titleFullref}[2]{\texorpdfstring{#1\NoCaseChange{\fullref{#2}}}{#1Section \ref*{#2}: \nameref*{#2}}}

\newcommand{\semiring}{S}
\newcommand{\Semiring}{\mathcal{S}}

\newcommand{\sbigadd}{\bigoplus}
\newcommand{\splus}{\mathbin{\oplus}}
\newcommand{\ssplus}{\morespace{\splus}}
\newcommand{\stimes}{\mathbin{\otimes}}
\newcommand{\sdot}{\mathbin{\odot}}
\newcommand{\szero}{\textswab{0}}
\newcommand{\snull}{\textswab{0}}
\newcommand{\sone}{\raisebox{-.75pt}{\textswab{1}}\hspace*{.25pt}}
\newcommand{\sonetiny}{\scalebox{.6}{\textswab{1}}}
\newcommand{\sgeq}{\hgeq}
\newcommand{\sleq}{\hleq}
\newcommand{\seq}{=}

\newcommand{\Semiringlong}{\Semiring = (\semiring, \oplus, \odot, \snull, \sone)}

\newcommand{\monoid}{M}
\newcommand{\Monoid}{\mathcal{M}}
\newcommand{\monoiddef}{\Monoid = (\monoid,\, \monop,\, \monone)}
\newcommand{\monop}{\sdot}
\newcommand{\monone}{\sone}

\newcommand{\module}{W}
\newcommand{\Module}{\mathcal{W}}
\newcommand{\moduledef}{\Module = (\module, \oplus, \snull, \otimes)}

\newcommand{\trop}{\texttt{trop}}

\newcommand{\comb}{\texttt{nats}}
\newcommand{\prob}{\texttt{prob}}

\newcommand{\why}{\texttt{why}}
\newcommand{\bool}{\texttt{bool}}

\newcommand{\lang}{\texttt{lang}}
\newcommand{\clear}{\texttt{clear}}

\newcommand{\modbool}{\Module_{\bool}}
\newcommand{\monbool}{\Monoid_{\bool}}
\newcommand{\monboollong}{(\Bools, \land, \true)}
\newcommand{\modboollong}{(\Bools, \lor, \false, \land)}

\newcommand{\modlang}{\Module_{\lang}}

\newcommand{\monlang}{\Monoid_{\lang}}

\newcommand{\modnats}{\Module_{\comb}}

\newcommand{\modprob}{\Module_{\prob}}
\newcommand{\modproblong}{(\PosRealsInf, +, 0, \cdot)}
\newcommand{\monprob}{\Monoid_{\prob}}
\newcommand{\monproblong}{(\PosReals, \cdot, 1)}

\newcommand{\montrop}{\Monoid_{\trop}}
\newcommand{\modtrop}{\Module_{\trop}}

\newcommand{\modwhy}{\Module_{\why}}
\newcommand{\pDNF}{\mathsf{pDNF}}
\newcommand{\modwhylong}{(\pDNF, \vee, 0, \wedge)}
\newcommand{\monwhy}{\Monoid_{\why}}
\newcommand{\monwhylong}{(\pDNF, \wedge, 1)}

\newcommand{\modclear}{\Module_{\clear}}
\newcommand{\clearUniverse}{\{\mathsf{P, C, S, TS}\}}
\newcommand{\modclearlong}{(\clearUniverse, \min, \mathsf{TS}, \max)}
\newcommand{\monclear}{\Monoid_{\clear}}
\newcommand{\monclearlong}{(\clearUniverse, \max, \mathsf{P})}

\newcommand{\sem}[1]{\llbracket #1 \rrbracket}
\newcommand{\nosem}[1]{#1}

\newcommand{\BRANCH}[2]{\{\, #1\, \} \BranchSymbol \{\, #2\, \}}
\newcommand{\BIGBRANCH}[3]{\sbigadd_{#1}^{#2} \, \{\, #3\, \}}
\newcommand{\WeighSymbol}{\stimes}
\newcommand{\WEIGH}[1]{{\WeighSymbol}\,#1\,}

\newcommand{\AssignSymbol}{\coloneqq}
\newcommand{\ASSIGN}[2]{\ensuremath{#1 \AssignSymbol #2}}

\newcommand{\IF}[1]{\ensuremath{\symIf\,\left(\, {#1} \,\right)\,\{}}
\newcommand{\ELSE}{\ensuremath{\}\,\symElse\,\{}}
\newcommand{\ITE}[3]{\ensuremath{\symIf\,\left(\, {#1} \,\right)\,\left\{\, {#2} \,\right\}\,\symElse\,\left\{\, {#3} \,\right\}}}
\newcommand{\WHILE}[1]{\ensuremath{\symWhile \left(\, {#1} \,\right)\left\{\right.}}
\newcommand{\WHILEDO}[2]{\ensuremath{\symWhile \left(\, {#1} \,\right)\left\{\, {#2} \,\right\}}}

\newcommand{\semcol}{\,\fatsemi\;}

\DeclareMathOperator{\supp}{supp}
\newcommand{\config}[1]{\langle {#1} \rangle}

\newcommand{\morespace}[1]{~{}#1{}~}
\newcommand{\qmorespace}[1]{\quad{}#1{}\quad}
\newcommand{\qqmorespace}[1]{\qquad{}#1{}\qquad}

\newcommand{\hheyloleq}{\morespace{\heyloleq}}
\newcommand{\hheylogeq}{\morespace{\heylogeq}}

\newcommand{\qiff}{\qmorespace{\textnormal{iff}}}
\newcommand{\qqiff}{\qqmorespace{\textnormal{iff}}}

\newcommand{\qimplies}{\qmorespace{\textnormal{implies}}}

\definecolor{Platinum}{HTML}{f0f0f0}
\definecolor{Charcoal}{HTML}{5C5C5C}
\definecolor{Black}{HTML}{000000}

\definecolor{CornflowerOcean}{HTML}{0072b2}
\definecolor{EmeraldDepths}{HTML}{095541}
\definecolor{Teal}{HTML}{008080}
\definecolor{CoralGlow}{HTML}{FF7F50}
\definecolor{ForestMoss}{HTML}{63922F}
\definecolor{Olive}{HTML}{87912F}
\definecolor{RebeccaPurple}{HTML}{682F91}
\definecolor{BrightLemon}{HTML}{FFED27}

\definecolor{TurfGreen}{HTML}{077547}
\definecolor{StormyTeal}{HTML}{076C75}
\definecolor{YaleBlue}{HTML}{074775}
\definecolor{RegalNavy}{HTML}{073176}
\definecolor{YaleBlue2}{HTML}{074475}
\definecolor{YaleBlue}{HTML}{075675}
\definecolor{DeepOcean}{HTML}{077570}

\definecolor{BrickRed}{HTML}{B11F1F}
\definecolor{BrickAmber}{HTML}{B03820}
\definecolor{RustySpice}{HTML}{B05020}
\definecolor{DarkGoldenrod}{HTML}{b08020}

\definecolor{DeepSpaceBlue}{HTML}{203b57}
\definecolor{CharcoalBlue}{HTML}{204457}
\definecolor{DarkTeal}{HTML}{204d57}
\definecolor{DarkGarnet}{HTML}{6e0e0a}
\definecolor{MoltenLava}{HTML}{6e1e0a}
\definecolor{DarkWalnut}{HTML}{6e2e0a}

\colorlet{basicBackground}{Charcoal}
\colorlet{primalColor}{Teal} %
\colorlet{dualColor}{CoralGlow} %

\colorlet{moduleColor}{RustySpice} %
\colorlet{weightedColor}{BrickAmber} %
\colorlet{wGCLColor}{BrickRed} %

\colorlet{heytingColor}{TurfGreen} %
\colorlet{prepostColor}{StormyTeal} %
\colorlet{heyvlColor}{YaleBlue} %

\colorlet{stmtColor}{EmeraldDepths} %
\colorlet{hintcolor}{Charcoal}
\colorlet{neutralColor}{Olive}

\newcommand{\colweighted}[1]{{\color{weightedColor}#1}}
\newcommand{\colwGCL}[1]{{\color{wGCLColor}#1}}
\newcommand{\colheylo}[1]{{\color{prepostColor}#1}}
\newcommand{\colheyvl}[1]{{\color{heyvlColor}#1}}
\newcommand{\colproc}[1]{{\color{hintcolor}#1}}

\newcommand{\hint}[1]{\color{hintcolor}{#1}}
\newcommand{\explain}[1]{\text{\small\color{hintcolor}{(#1)}}\quad}
\newcommand{\inlineExplain}[1]{\tag*{\small\color{hintcolor}{(#1)}}}

\newcommand{\lsupquant}[2]{\ensuremath{\bigsqcup_{#1} #2}}
\newcommand{\linfquant}[2]{\ensuremath{\bigsqcap_{#1} #2}}

\newcommand{\ifThenElse}[3]{\begin{cases}{#2}, & \text{if } {#1}\\{#3}, & \text{otherwise} \end{cases}}
\newcommand{\ifThenElseDot}[3]{\begin{cases}{#2}, & \text{if } {#1}\\{#3}, & \text{otherwise}~. \end{cases}}

\newcommand{\true}{\mathsf{true}}
\newcommand{\false}{\mathsf{false}}

\newcommand{\impl}{\rightarrow}
\newcommand{\coimpl}{\leftarrow}
\newcommand{\coneg}{{\sim}}
\newcommand{\heylovalidate}[1]{\ensuremath{\triangle\!\left( #1 \right)}}
\newcommand{\heylocovalidate}[1]{\ensuremath{\triangledown\!\left( #1 \right)}}

\newcommand{\heyloleq}{\sqsubseteq}
\newcommand{\heylogeq}{\sqsupseteq}
\newcommand{\hleq}{\heyloleq}
\newcommand{\hhleq}{\morespace{\hleq}}
\newcommand{\bhleq}{\hleq_{\statesclosed}}
\newcommand{\bhhleq}{\morespace{\hleq_\beta}}

\newcommand{\hgeq}{\heylogeq}

\newcommand{\hand}{\sqcap}
\newcommand{\hbigand}{\bigsqcap}

\newcommand{\hor}{\sqcup}
\newcommand{\hbigor}{\bigsqcup}

\newcommand{\hsym}{\mathpzc{H}}
\newcommand{\reflectH}[2]{\reflectbox{$#1\mathpzc{H}$}}
\newcommand{\cohsym}{\mathpalette\reflectH\relax}
\newcommand{\bihsym}{\mathcal{H}}

\newcommand{\hdom}{H}

\newcommand{\hdef}{\hsym = (\hdom, \hleq, \hand, \hor, \impl)}
\newcommand{\hcodef}{\cohsym = (\hdom, \hleq, \hand, \hor, \coimpl)}
\newcommand{\hbidef}{\bihsym = (\hdom, \hleq, \hand, \hor, \impl, \coimpl)}

\newcommand{\hla}{\ensuremath{\varphi}}
\newcommand{\hlb}{\ensuremath{\psi}}
\newcommand{\hlc}{\ensuremath{\rho}}

\newcommand{\inv}{\ensuremath{I}}

\newcommand{\wla}{\colweighted{f}}
\newcommand{\wlb}{\colweighted{g}}
\newcommand{\wlc}{\colweighted{h}}
\newcommand{\wX}{\colweighted{X}}

\newcommand{\hhla}{\colheylo{\hla}}
\newcommand{\hhlb}{\colheylo{\hlb}}
\newcommand{\hhlc}{\colheylo{\hlc}}

\newcommand{\iinv}{\colweighted{\ensuremath{I}}}
\newcommand{\hinv}{\colheylo{\inv}}
\newcommand{\hinvseq}[1]{\hinv_{\colheylo{#1}}}

\newcommand{\substBy}[2]{[{#1} \mapsto {#2}]}

\newcommand{\iverson}[1]{\left[{#1}\right]}
\newcommand{\booth}[1]{\left<{#1}\right>}
\newcommand{\interpretsimple}[1]{\llbracket{#1}\rrbracket}

\newcommand{\eval}[1]{\llbracket{#1}\rrbracket}

\newcommand{\isValid}[1]{#1 \text{ is valid}}
\newcommand{\isCovalid}[1]{#1 \text{ is covalid}}

\newcommand{\qfFrag}{\wHeyLo_{\mathsf{qf}}}
\newcommand{\downQFrag}{\wHeyLo_\downarrow}
\newcommand{\upQFrag}{\wHeyLo_\uparrow}
\newcommand{\downPFrag}{\mathcal P_\downarrow}
\newcommand{\upPFrag}{\mathcal P_\uparrow}

\newcommand{\downQelim}[1]{\mathsf{qelim}^{\downarrow}(#1)}
\newcommand{\upQelim}[1]{\mathsf{qelim}^{\uparrow}(#1)}
\newcommand{\downPrenex}[1]{\mathsf{prenex}^{\downarrow}(#1)}
\newcommand{\upPrenex}[1]{\mathsf{prenex}^{\uparrow}(#1)}
\newcommand{\downStrip}[2]{\mathsf{strip}^{\downarrow}(#1,#2)}
\newcommand{\upStrip}[2]{\mathsf{strip}^{\uparrow}(#1,#2)}

\newcommand{\midequiv}{&\quad\equiv\quad}%
\newcommand{\midequivL}{\midequiv}

\newcommand{\lfp}[2]{\mathrm{lfp}\,{#1}.~{#2}}
\newcommand{\gfp}[2]{\mathrm{gfp}\,{#1}.~{#2}}

\newcommand{\Bools}{\mathbb{B}}
\newcommand{\Nats}{\mathbb{N}}

\newcommand{\NatsX}{\mathbb{N}^{\infty}}
\newcommand{\PosReals}{\mathbb{R}_{\geq 0}}
\newcommand{\PosRealsInf}{\mathbb{R}_{\geq 0}^{\infty}}
\newcommand{\eureal}{\PosRealsInf}

\newcommand{\Weighting}{\mathbb{W}}
\newcommand{\States}{\Sigma}
\newcommand{\State}{\sigma}
\newcommand{\Vals}{\mathbb{N}}

\newcommand{\sfsymbol}[1]{\textsf{\upshape {#1}}}
\newcommand{\Vars}{\sfsymbol{Vars}\xspace}
\newcommand{\varset}{\ensuremath{\beta}}
\newcommand{\GCL}{\sfsymbol{GCL}\xspace}
\newcommand{\pGCL}{\sfsymbol{pGCL}\xspace}
\newcommand{\wGCL}{\textsf{\upshape wGCL}\xspace}
\newcommand{\HeyVL}{\sfsymbol{HeyVL}\xspace}
\newcommand{\HeyLo}{\sfsymbol{HeyLo}\xspace}
\newcommand{\wHeyVL}{\sfsymbol{wHeyVL}\xspace}
\newcommand{\wHeyLo}{\sfsymbol{wHeyLo}\xspace}
\newcommand{\Caesar}{\textsc{Caesar}}
\newcommand{\sflangsub}[2]{\ensuremath{\textsf{\upshape #1}_{#2}}\xspace}
\newcommand{\wHeyVLEUReal}{\sflangsub{wHeyVL}{\mathtt{\prob}}}
\newcommand{\wHeyVLBool}{\sflangsub{wHeyVL}{\mathtt{\bool}}}
\newcommand{\wHeyVLTrop}{\sflangsub{wHeyVL}{\mathtt{\trop}}}
\newcommand{\wHeyVLClearance}{\sflangsub{wHeyVL}{\mathtt{\clear}}}
\newcommand{\wHeyVLWhy}{\sflangsub{wHeyVL}{\mathtt{\why}}}
\newcommand{\wHeyVLLang}{\sflangsub{wHeyVL}{\mathtt{\lang}}}
\newcommand{\semiringPrefix}[1]{\symStmt{\texttt{#1}}\xspace}
\newcommand{\rigEUReal}{\semiringPrefix{\prob}}
\newcommand{\rigBool}{\semiringPrefix{\bool}}
\newcommand{\rigTrop}{\semiringPrefix{\trop}}
\newcommand{\rigClear}{\semiringPrefix{\clear}}
\newcommand{\rigWhy}{\semiringPrefix{\why}}
\newcommand{\rigLang}{\semiringPrefix{\lang}}

\newcommand{\MonExp}{\sfsymbol{MExp}\xspace}
\newcommand{\ArithExp}{\sfsymbol{AExp}\xspace}
\newcommand{\BExp}{\sfsymbol{BExp}\xspace}

\newcommand{\typeof}[2]{\ensuremath{#1 \colon {#2}}}
\newcommand{\typefont}[1]{\mathsf{#1}}
\newcommand{\listtype}{\typefont{Lists}}

\newcommand{\typevar}{\ensuremath{\tau}}

\newcommand{\aexpr}{\ensuremath{a}}
\newcommand{\mexpr}{\ensuremath{m}}

\newcommand{\bexpr}{\ensuremath{b}}
\newcommand{\procname}{\colproc{P}}

\newcommand{\symStmt}[1]{\ensuremath{{\color{stmtColor}{#1}}}}
\newcommand{\symAssign}{\morespace{\coloneqq}}
\newcommand{\symRasgn}{\coloneqq}
\newcommand{\symSemi}{\symStmt{\texttt{;}~}}
\newcommand{\symIf}{\symStmt{\texttt{if}}}
\newcommand{\symElse}{\symStmt{\texttt{else}}}
\newcommand{\symAssert}{\symStmt{\texttt{assert}}}
\newcommand{\symAssume}{\symStmt{\texttt{assume}}}
\newcommand{\symDemonic}{\symStmt{\symIf~(\sqcap)}}
\newcommand{\symAngelic}{\symStmt{\symIf~(\sqcup)}}
\newcommand{\symHavoc}{\symStmt{\texttt{havoc}}}
\newcommand{\symWhile}{\symStmt{\texttt{while}}}
\newcommand{\symTick}{\symStmt{\texttt{reward}}}
\newcommand{\symProc}{\symStmt{\texttt{proc}}}
\newcommand{\symcoProc}{\symStmt{\texttt{coproc}}}
\newcommand{\symReturns}{\texttt{->}}
\newcommand{\symRequires}{\symStmt{\texttt{pre}}}
\newcommand{\symEnsures}{\symStmt{\texttt{post}}}
\newcommand{\symValidate}{\symStmt{\texttt{validate}}}
\newcommand{\symVar}{\symStmt{\texttt{var}}}
\newcommand{\symWeightedBranching}{\symStmt{\symIf~(+)}}
\newcommand{\symWeigh}{\symStmt{\texttt{weigh}}}

\newcommand{\blockStart}{\ensuremath{\{}}
\newcommand{\blockEnd}{\ensuremath{\}}}

\newcommand{\weightedProc}[4]{{#1}~\symProc~{\mathit{#2}}\,\texttt{(}{#3}\texttt{)}~\symReturns~\texttt{(}{#4}\texttt{)}}
\newcommand{\weightedCoproc}[4]{{#1}~\symcoProc~{\mathit{#2}}\,\texttt{(}{#3}\texttt{)}~\symReturns~\texttt{(}{#4}\texttt{)}}

\newcommand{\Ensures}[1]{\symEnsures~{\color{prepostColor}#1}}
\newcommand{\Requires}[1]{\symRequires~{\color{prepostColor}#1}}

\newcommand{\varin}{\ensuremath{\mathit{in}}}
\newcommand{\varout}{\ensuremath{\mathit{out}}}
\newcommand{\multi}[1]{\ensuremath{\overline{#1}}}

\newcommand{\stmt}{C}

\newcommand{\sstmt}{{\color{heyvlColor}\stmt}}

\newcommand{\sstmts}[1]{\sstmt_{\colheyvl{#1}}}
\newcommand{\wcom}{\ensuremath{C}}
\newcommand{\wwcom}{\colwGCL{\wcom}}
\newcommand{\wwcoma}{\colwGCL{\wcom'}}
\newcommand{\wwcomi}{\colwGCL{\wcom_1}}
\newcommand{\wwcomii}{\colwGCL{\wcom_2}}
\newcommand{\wwcomtimes}[1]{\colwGCL{\wcom^{\colwGCL{#1}}}}
\newcommand{\cc}{\wcom}
\newcommand{\ccs}[1]{\cc_{#1}}

\newcommand{\cctick}{\wwcom\colweighted{'}}
\newcommand{\stmtAsgn}[2]{\ensuremath{{#1} \symAssign {#2}}}
\newcommand{\stmtDeclInit}[3]{\symVar~\ensuremath{{#1}\colon {#2} \symRasgn {#3}}}
\newcommand{\stmtRasgn}[2]{\ensuremath{{#1} \symRasgn {#2}}}
\newcommand{\stmtSeq}[2]{\ensuremath{{#1}\symSemi {#2}}}
\newcommand{\stmtIfStart}[1]{\ensuremath{\symIf~({#1})~\{ }}

\newcommand{\stmtHavoc}[1]{\ensuremath{\symHavoc~{#1}}}
\newcommand{\stmtAssert}[1]{\ensuremath{\symAssert~{#1}}}
\newcommand{\stmtAssume}[1]{\ensuremath{\symAssume~{#1}}}

\newcommand{\stmtDemonicStart}{\ensuremath{\symDemonic~\{ }}
\newcommand{\stmtElseStart}{\ensuremath{\symElse~\{ }}
\newcommand{\stmtAngelicStart}{\ensuremath{\symAngelic~\{ }}

\newcommand{\stmtDemonic}[2]{\ensuremath{\stmtDemonicStart{#1} \}~\stmtElseStart {#2} \}}}
\newcommand{\stmtWeighted}[2]{\ensuremath{\symWeightedBranching~\{ {#1} \}~\symElse~\{ {#2} \}}}
\newcommand{\BranchSymbol}{\splus}

\newcommand{\stmtAngelic}[2]{\ensuremath{\stmtAngelicStart{#1} \}~\stmtElseStart {#2} \}}}
\newcommand{\stmtTick}[1]{\ensuremath{\symTick~{#1}}}
\newcommand{\stmtValidate}{\ensuremath{\symValidate}}
\newcommand{\stmtWeigh}[1]{\ensuremath{\symWeigh~{#1}}}

\newcommand{\symDown}{}

\newcommand{\Havoc}[1]{\ensuremath{\symDown\symHavoc~{#1}}}
\newcommand{\Assert}[1]{\ensuremath{\symDown\symAssert~{#1}}}
\newcommand{\Assume}[1]{\ensuremath{\symDown\symAssume~{#1}}}
\newcommand{\Validate}{\ensuremath{\symDown\stmtValidate}}

\newcommand{\symUp}{\symStmt{\texttt{co}}}

\newcommand{\coHavoc}[1]{\ensuremath{\symUp\symHavoc~{#1}}}
\newcommand{\coAssert}[1]{\ensuremath{\symUp\symAssert~{#1}}}
\newcommand{\coAssume}[1]{\ensuremath{\symUp\symAssume~{#1}}}
\newcommand{\coValidate}{\ensuremath{\symUp\stmtValidate}}

\newcommand{\monus}{\mathbin{\dot-}}

\newcommand{\interpretsimpleState}[1]{\interpretsimple{#1}(\State)}
\newcommand{\interpretsimpleStateSubstBy}[3]{\interpretsimple{#1}(\State\substBy{#2}{#3})}

\newcommand{\evalState}[1]{\eval{#1}(\State)}
\newcommand{\evalStateSubstBy}[3]{\eval{#1}(\State\substBy{#2}{#3})}

\newcommand{\symWp}{\sfsymbol{wp}}
\newcommand{\symVc}{\sfsymbol{vp}}
\newcommand{\symVcB}{\sfsymbol{vc}}

\renewcommand{\wp}[1]{\symWp\llbracket{#1}\rrbracket}

\newcommand{\vc}[1]{\symVc\llbracket{#1}\rrbracket}
\newcommand{\vcB}[1]{\symVcB\llbracket{#1}\rrbracket}

\newcommand{\symWlp}{\sfsymbol{wlp}}
\newcommand{\wlp}[1]{\symWlp\llbracket{#1}\rrbracket}

\newcommand{\embed}[1]{\scalebox{0.85}{\textnormal{\textsf{?}}}({#1})}
\newcommand{\coEmbed}[1]{\scalebox{0.85}{$\neg$\textnormal{\textsf{?}}}({#1})}

\newcommand{\wpPhi}{\prescript{\symWp}{}{\Phi}}
\newcommand{\wlpPhi}{\prescript{\symWlp}{}{\Phi}}

\newcommand{\definedAs}{\triangleq}

\newcommand{\Spec}[2]{[#1, #2]}
\newcommand{\coSpec}[3]{{#1} : \Spec{#2}{#3}_{\scriptstyle\heylogeq}}
\newcommand{\stmtSpec}[3]{{#1} : \Spec{#2}{#3}}

\newcommand{\pChoice}[3]{\{\,#1\,\}~[#2]~\{\,#3\,\}}

\newcommand{\wpenc}[1]{\texttt{enc}_{\symWp}({#1})}

\newcommand{\procenc}[1]{\texttt{enc}({#1})}
\newcommand{\coprocenc}[1]{\texttt{enc}^{\texttt{co}}({#1})}

\newcommand{\specEnc}[1]{\texttt{specEnc}({#1})}
\newcommand{\specEncShort}{\texttt{specEnc}}
\newcommand{\coSpecEnc}[1]{\texttt{specEnc}^{\texttt{co}}({#1})}

\newcommand{\charEnc}[1]{\texttt{charEnc}({#1})}
\newcommand{\charEncShort}{\texttt{charEnc}}
\newcommand{\coCharEnc}[1]{\texttt{charEnc}^{\texttt{co}}({#1})}

\newcommand{\parkEnc}[1]{\texttt{parkEnc}({#1})}
\newcommand{\parkEncShort}{\texttt{parkEnc}}
\newcommand{\coParkEnc}[1]{\texttt{parkEnc}^{\texttt{co}}({#1})}

\newcommand{\kappaUnrollEnc}[1]{\kappa\texttt{-unrollEnc}({#1})}
\newcommand{\kappaUnrollEncShort}{\kappa\texttt{-unrollEnc}}
\newcommand{\unrollEnc}[2]{\texttt{{#1}-unrollEnc}({#2})}

\newcommand{\kappaEnc}[1]{\kappa\texttt{-Enc}({#1})}
\newcommand{\kappaEncShort}{\kappa\texttt{-Enc}}

\newcommand{\omegaEnc}[1]{\omega\texttt{-Enc}({#1})}
\newcommand{\omegaEncShort}{\omega\texttt{-Enc}}

\newcommand{\FLMonVar}{\ell}
\newcommand{\FLAllWords}{\Gamma^\infty}
\newcommand{\FLMonitor}[1]{\ensuremath{\mathsf{monitor}\!\left(#1\right)}}
\newcommand{\FLMonitorPrime}[1]{\ensuremath{\mathsf{monitor}'\!\left(#1\right)}}
\newcommand{\FLUpMonitor}[1]{\ensuremath{\mathsf{monitor}^{\uparrow}\!\left(#1\right)}}
\newcommand{\FLUpMonitorPrime}[1]{\ensuremath{\mathsf{monitor}^{\uparrow\prime}\!\left(#1\right)}}
\newcommand{\FLIn}[1]{\ensuremath{(\FLMonVar \subseteq #1)}}
\newcommand{\FLOut}[1]{\ensuremath{(#1 \subseteq \FLMonVar)}}
\newcommand{\TropicalTrans}[1]{\ensuremath{\mathtt{\trop}\!\left(#1\right)}}
\newcommand{\AccessControlTrans}[1]{\ensuremath{\mathtt{\clear}\!\left(#1\right)}}
\newcommand{\FLType}{\mathtt{\lang}}
\newcommand{\symWhyBool}{\ensuremath{\mathtt{\bool}}}
\newcommand{\WhyBool}[1]{\ensuremath{\symWhyBool\!\left(#1\right)}}
\newcommand{\symWhyTrans}{\ensuremath{\mathsf{tr}}}
\newcommand{\WhyTrans}[1]{\ensuremath{\symWhyTrans\!\left(#1\right)}}
\newcommand{\symVcBool}{\ensuremath{\mathsf{vc}_{\mathtt{\bool}}}}
\newcommand{\symVcTropical}{\ensuremath{\mathsf{vc}_{\mathtt{\trop}}}}
\newcommand{\symVcAccess}{\ensuremath{\mathsf{vc}_{\mathtt{\clear}}}}
\newcommand{\symVcWhy}{\ensuremath{\mathsf{vc}_{\mathtt{\why}}}}
\newcommand{\vcTropical}[1]{\ensuremath{\symVcTropical\!\left(#1\right)}}
\newcommand{\vcAccess}[1]{\ensuremath{\symVcAccess\!\left(#1\right)}}
\newcommand{\vcWhy}[1]{\ensuremath{\symVcWhy\!\left(#1\right)}}
\newcommand{\weightedProcHead}[3]{{#1}~\symProc~{\mathit{#2}}\,\texttt{(}{#3}\texttt{)}}
\newcommand{\weightedCoprocHead}[3]{{#1}~\symcoProc~{\mathit{#2}}\,\texttt{(}{#3}\texttt{)}}
\newcommand{\weightedProcOut}[1]{\hspace{4em}\symReturns~\texttt{(}{#1}\texttt{)}}
\newcommand{\WhyEval}[2]{\ensuremath{\mathsf{eval}_{#1}\!\left(#2\right)}}

\newcommand{\proofheading}[1]{\par\textit{#1.}\ }
\newenvironment{proofcases}{\begin{itemize}}{\end{itemize}}
\newcommand{\proofcase}[1]{\item[-]\textsc{Case} #1.}

\newcommand{\Set}[1]{\left\{#1\right\}}

\newcommand{\dotsqcap}{\sqcap}
\newcommand{\dotsqcup}{\sqcup}

\newcommand{\dotdualimpl}{\coimpl}
\newcommand{\dotcoimpl}{\dotdualimpl}

\newcommand{\myrel}[2]{%
  \mathrel{\overset{\raisebox{.6ex}{\makebox[0pt]{\mbox{\normalfont\tiny\sffamily #1}}}}{#2}}%
}

\newcommand{\itemExplain}[1]{%
  \hfill\text{\small\itshape\color{hintcolor}(#1)}%
}

\newcommand{\weightings}{\mathbb{W}}

\newcommand{\wlpchar}{\Psi_{\wla}}
\newcommand{\wpchar}{\Phi_{\wla}}

\newcommand{\wlpcharhhla}{\Psi_{\hhla}}
\newcommand{\wpcharhhla}{\Phi_{\hhla}}

\newcommand{\betawlpchar}{\wlpchar}

\newcommand{\varsetreach}[1]{\mathrm{Reach}_\beta(#1)}
\newcommand{\statesclosed}{\Gamma}

\newcommand{\leqannotate}[1]{\colheylo{{}^{\heyloleq}\!\!\!\! \fatslash \!\!\! \fatslash \,\,\, {#1}}}
\newcommand{\geqannotate}[1]{\colheylo{{}^{\heylogeq}\!\!\!\! \fatslash \!\!\! \fatslash \,\,\, {#1}}}
\newcommand{\wpannotate}[1]{\, \colheylo{{}^{\symWp}\!\!\!\! \fatslash \!\!\! \fatslash \,\,\, {#1}}}
\newcommand{\wlpannotate}[1]{\, \colheylo{{}^{\symWlp}\!\!\!\! \fatslash \!\!\! \fatslash \,\,\, {#1}}}

\newcommand{\tablespace}{4pt}

\newcommand{\totimpl}{\rightarrow_{\text{tot}}}
\newcommand{\totcoimpl}{\leftarrow_{\text{tot}}}
\newcommand{\matimpl}{\rightarrow_{\text{bool}}}
\newcommand{\matcoimpl}{\leftarrow_{\text{bool}}}

\definecolor{mathcolor}{RGB}{70, 130, 180}
\definecolor{progcolor}{RGB}{34, 139, 34}
\definecolor{logiccolor}{RGB}{220, 20, 60}
\definecolor{ivlcolor}{RGB}{255, 140, 0}
\definecolor{encodingcolor}{RGB}{138, 43, 226}
\definecolor{highlightcolor}{RGB}{178, 34, 34}

\newcommand{\freevars}[1]{\texttt{free}({#1})}

\newenvironment{halfboxl}{%
\noindent
\begin{minipage}[t]{0.4875\linewidth}
}{%
\end{minipage}\hspace{0.025\linewidth}}
\newenvironment{halfboxr}{%
\begin{minipage}[t]{0.4875\linewidth}
\vspace{0pt}
}{%
\end{minipage}\bigbreak}

\DeclareMathAlphabet{\mathpzc}{OT1}{pzc}{m}{it}
\DeclareMathOperator{\Ima}{Im}

\newcommand{\hvcpreweight}{\hhlc}
\newcommand{\hvcpostweight}{\hhlb}
\newcommand{\hpreweight}{\hhlc}

\newcommand{\kk}{\kappa}
\newcommand{\kindsym}{\Lambda}
\newcommand{\kindop}[1]{\kindsym^{#1}}
\newcommand{\kindopk}{\kindsym^{ \kk }}

\newcommand{\ww}{\omega}

\newcommand{\clf}{\wcom_{LF}}

\newcommand{\graphreach}{\wcom_{GR}}

\begin{document}

\title[Towards a Verification Infrastructure for \wGCL]{\texorpdfstring{Towards a Deductive Verification Infrastructure \\ for Weighted Programming}{Towards a Deductive Verification Infrastructure for Weighted Programming}}

\author{Emma Ahrens}
\authornote{Equal contribution.}
\email{ahrens@cs.rwth-aachen.de}
\orcid{0000-0002-6394-3351}
\affiliation{%
  \institution{RWTH Aachen University}
  \city{Aachen}
  \country{Germany}
}
\author{Samuel Rode}
\authornotemark[1]
\email{samuel.rode@rwth-aachen.de}
\orcid{0009-0002-3854-9309}
\affiliation{%
  \institution{RWTH Aachen University}
  \city{Aachen}
  \country{Germany}
}
\author{Philipp Schröer}
\authornotemark[1]
\email{phisch@cs.rwth-aachen.de}
\orcid{0000-0002-4329-530X}
\affiliation{%
  \institution{RWTH Aachen University}
  \city{Aachen}
  \country{Germany}
}
\author{Joost-Pieter Katoen}
\email{katoen@cs.rwth-aachen.de}
\orcid{0000-0002-6143-1926}
\affiliation{%
  \institution{RWTH Aachen University}
  \city{Aachen}
  \country{Germany}
}
\renewcommand{\shortauthors}{Ahrens, Rode, Schröer, and Katoen}

\begin{abstract}
    \emph{Weighted programs} extend guarded commands with trace weights drawn from a \emph{semiring}, or more generally a \emph{monoid-module}.
    Varying this algebra gives one programmatic syntax for a variety of quantitative and symbolic models.
    \emph{Weakest-preweighting semantics} provides a compositional basis for reasoning about those programs.

    We present a deductive verification framework based on a \emph{weighted assertion language} and an \emph{intermediate verification language}.
    Its weight domains are ordered structures with \emph{implication} and \emph{coimplication}, which let verification conditions express lower- and upper-bound obligations internally.
    We prove sound translations of core commands and reusable encodings for various proof rules applying to procedure calls and loops.
    To facilitate automation, we prove soundness of a quantifier elimination procedure for our assertion language.
    A prototype in the \Caesar{} verifier checks case studies for probabilistic queueing costs, recursive database provenance with cyclic dependencies, clearance bounds for networks of arbitrary size, and formal-language reasoning about lock-freedom of a compare-and-swap counter.
\end{abstract}


\maketitle

\section{Introduction}

Quantitative reasoning is becoming increasingly relevant across many different branches of computer science. Recent work ranges from weighted transition systems and automata \cite{HansenLMP18, DrosteK21} to weighted logics \cite{BolligG09,cmp-lg-9705006} and declarative languages \cite{BatzGKKW22, BelleR20,Zilberstein25}. In database theory, provenance analysis uses
semiring semantics to investigate why queries
return the answers they do
\cite{provenance_foundational_paper,gradel2025provenance}.
Weighted semantics have
also been studied for Datalog, a programming language for data reasoning
\cite{BistarelliMS08}, while \cite{GuZ21} considers logic programming with
semiring semantics. Further developments include weighted rewriting
\cite{AhrensKGK25,AvanziniY25}, weighted pushdown automata \cite{Kuich97},
Guarded Kleene Algebra with Tests with semiring semantics \cite{KoeveringR025},
and Weighted NetKAT \cite{weightedNetKAT26}.
Overall, this amount of research raises a natural question: To what extent can we automate quantitative reasoning?

\paragraph{Weighted Programming.}
This work takes a step towards semi-automatic verification of weighted programs, as introduced by \citeauthor{BatzGKKW22}, which is a generalization of classical programs where executions accumulate weights from an algebraic domain.
This domain is usually a semiring, and in the general form a monoid-module with additive aggregation $\oplus$ and scalar multiplication $\otimes$.
The weights can express quantitative properties, such as costs, probabilities, and security clearances, as well as symbolic properties, such as formal trace languages.
\emph{Weigh statements} multiply weight along execution traces, while non-deterministic choice combines traces via the additive operation.

For example, the Boolean semiring, with addition $\oplus = \lor$ and scalar multiplication $\otimes = \land$, yields programs with Boolean assertions and angelic non-determinism.
The \emph{tropical monoid-module} $\modtrop$ over $\NatsX$ yields programs that optimize for minimal costs ($\oplus = \min$), where costs are additive ($\otimes = +$).
The \emph{Ski-rental problem} (\Cref{fig:intro-wgcl-ski-rental}) is a classical example of an optimization problem that can be modeled as a weighted program over the tropical monoid-module.
Given either the choice of renting skis for a day at a cost of 1, or buying skis for a cost of $y$, the goal is to minimize the total cost over $n$ days.
We can also represent a compare-and-swap (CAS) concurrent counter as a weighted program over the monoid-modules of formal languages, shown in \Cref{fig:intro-wgcl-lock-freedom}.
We want to show \emph{lock-freedom}, i.e. that there always exists a thread that makes progress.

\begin{figure*}[t]
    \centering
    \begin{subfigure}[c]{0.48\textwidth}
        \centering
        \[
        \begin{aligned}
            & \ASSIGN{c}{1} \fatsemi \\
            & \WHILE{n > 0} \\
            & \qquad \ASSIGN{n}{n - 1} \fatsemi \\
            & \qquad \{ \\
            & \qquad\qquad \WEIGH{c} \\
            & \qquad \} \mathrel{\BranchSymbol} \{ \\
            & \qquad\qquad \WEIGH{y} \fatsemi \ASSIGN{n}{0}
            & \} ~ \}%
        \end{aligned}
        \]
        \caption{Ski rental problem (tropical monoid-module). For $n$ days of skiing, the program models the choice between renting skis for a day at a cost of 1, or buying skis for a cost of $y$.}
        \label{fig:intro-wgcl-ski-rental}
    \end{subfigure}
    \hfill
    \begin{subfigure}[c]{0.48\textwidth}
        \centering
        \[
        \begin{aligned}
            & \WHILE{\true} \\
            & \qquad \BIGBRANCH{j=1}{N}{\ASSIGN{i}{j}} \fatsemi \\
            & \qquad \IF{\ell[i] = v} \\
            & \qquad\qquad \WEIGH{S_i} \fatsemi \ASSIGN{v}{\ell[i] + 1} \fatsemi \ASSIGN{\ell[i]}{v} \\
            & \qquad \ELSE \\
            & \qquad\qquad \WEIGH{F_i} \fatsemi \ASSIGN{\ell[i]}{v} & \} ~ \} \\
        \end{aligned}
        \]
        \caption{Compare-and-swap counter with $N$ threads (formal languages monoid-module). Each thread $i$ attempts to increment a shared counter $\ell[i]$ from a value $v$ to $v+1$. Failure weighs by $F_i$, success by $S_i$.}
        \label{fig:intro-wgcl-lock-freedom}
    \end{subfigure}

    \caption{Weighted programs for the ski rental problem and a compare-and-swap counter.}
    \label{fig:intro-wgcl-programs}
    \Description{Source wGCL programs for ski rental and lock-freedom of a compare-and-swap counter.}
\end{figure*}

Weakest preweightings generalize Dijkstra's weakest preconditions~\cite{Dijkstra75} and probabilistic weakest preexpectations~\cite{McIverM05} to weighted programs~\cite{BatzGKKW22}.
The transformer $\symWp$ proceeds backwards through the program and computes a \emph{weighting} that maps each program state to the aggregate weight of all executions starting in that state.
For the Boolean semiring, this states whether there exists an execution satisfying all assertions.
For the tropical monoid-module, this is the minimal cost among all non-deterministic execution traces.

\paragraph{Assertions, Assumptions, Proof Rules.}
To automate such reasoning, we extend weighted programs with \emph{verification statements} such as \emph{assert} and \emph{assume}.
These statements allow us to reify \emph{proof rules as programs}: reasoning principles for loops and procedure calls become program encodings that manipulate states and weights.
In the weighted setting, verification statements must be quantitative, i.e., they must support lower- and upper-bound reasoning.
For example, procedure calls can be represented by a program that asserts a preweighting $\hvcpreweight$, forgets (``havocs'') the values of variables $\varset$ modifiable by the callee, and assumes a postweighting $\hvcpostweight$.
Similarly, loop rules such as Park induction and $k$-induction~\cite{k_induction} can be encoded as programs that check quantitative invariants.

By reifying proof rules as programs, one can use a small, reusable backend for many verification techniques.
This is the idea of an \emph{intermediate verification language} (IVL), a simple programming language with verification statements and verification-condition generation.
Successful classical IVL-based infrastructures include Boogie~\cite{BarnettCDJL05} and Why3~\cite{FilliatreP13}.
A program in the source language is translated to the IVL, with the proof rules encoded in the IVL.
Here, generic transformations and optimizations can be applied to the IVL program, and then verification conditions are generated and discharged by an SMT solver.

\paragraph{Goals of this Paper.}
We want to build a verification infrastructure for weighted programming, supporting a wide variety of algebraic domains and verification techniques.
Weakest preweighting semantics exists for weighted programs, but we lack an assertion language, an IVL, and encodings of proof rules for weighted programming.
We therefore generalize the \Caesar{} verification infrastructure~\cite{SchroerBKKM23}, from probabilistic to weighted programming.

Because weighted programming is parametric over the concrete monoid-module, we need to instantiate the assertion language and IVL for different algebraic domains.
In particular, we need notions of quantitative conjunctions and implications, with support for both lower- and upper-bound reasoning.
We address this by parameterizing \wHeyLo{} and \wHeyVL{} over \emph{Heyting algebras}, which provide the conjunction and implication operations,\footnote{Although \Caesar{}'s \HeyLo{} and \HeyVL{} are named after Heyting algebras, prior work only supported the setting specific to expected-value reasoning, not general Heyting algebras~\cite{SchroerBKKM23}.} and further generalizing encodings of proof rules.
Another challenge is automation.
The \wHeyVL{} encodings of proof rules introduce quantitative infimum and supremum quantifiers, which are difficult for SMT solvers to handle directly.
While complete quantifier elimination exists for restricted quantitative assertion fragments~\cite{DBLP:conf/fossacs/BatzKO25}, we also have structural specifications and domain-specific expressions outside of such a fragment. 
We therefore provide a best-effort quantifier-elimination pass for \wHeyLo{} that eliminates supported quantifiers and leaves the rest unchanged.

\begin{figure}[t]
    \centering
    \begin{tikzpicture}[
        box/.style={rectangle, rounded corners=3pt, minimum width=2.5cm, minimum height=0.7cm, text centered, font=\scriptsize\sffamily, text width=2.5cm, inner sep=3pt, very thick},
        arrow/.style={-{Stealth[length=2.5mm]}, thick, draw=Charcoal!50},
        labelarrow/.style={arrow, font=\scriptsize\sffamily, text=Charcoal!50},
        pipelinearrow/.style={ %
            draw=Platinum, font=\scriptsize\sffamily, text=Platinum, line width=20pt},
        highlightarrow/.style={-{Stealth[length=3mm]}, line width=2pt, draw=Charcoal},
        backgroundbox/.style={rectangle, rounded corners=6pt, draw=ivlcolor!30, fill=ivlcolor!5, thick, dashed, inner sep=9pt}
    ]

        \def\height{1}
        \def\width{3.5}

        \node[box, fill=moduleColor!75, opacity=0.75] (semirings) at (1*\width, 3.5*\height)    {$\omega$-Continuous \\Monoid-Modules};
        \node[box, fill=heytingColor!50, draw=heytingColor!75, ]  (heyting)   at (1*\width, 2*\height)  {Heyting Algebras};

        \node[box, fill=weightedColor!75, opacity=0.75] (wp) at (2*\width, 3.5*\height) {Weakest Preweightings};
        \node[box, fill=prepostColor!50, draw=prepostColor!75]  (vp)     at (2*\width, 2*\height)   {\wHeyLo: Verification\\Preweightings};

        \node[box, fill=neutralColor!50, draw=neutralColor!75]  (qe)     at (2.5*\width, 0.75*\height)   {Quantifier\\Elimination};
        \node[font=\scriptsize\sffamily] (smt)    at (3.5*\width, 0.75*\height)  {SMT Solver};

        \node[box, fill=wGCLColor!75, opacity=0.75] (wgcl)   at (3*\width, 3.5*\height)     {\wGCL: Weighted Programming Language};
        \node[box, fill=heyvlColor!50, draw=heyvlColor!75]   (wheyvl) at (3*\width, 2*\height)   {\wHeyVL: Verification Language};

        \node[box, fill=neutralColor!50, draw=neutralColor!75] (loops) at (4*\width, 2.75*\height) {Proof Rules \\[2pt]
            {\tiny (Specification Statements, local Park Induction, local $\kappa$-Induction, $\omega$-Invariants)}
        };

        \draw[labelarrow] (semirings) -- (wp)
            node[midway, above, yshift=0.3cm] {semantics of};
        \draw[labelarrow] (wp) -- (wgcl)
            node[midway, above, yshift=0.3cm] {semantics of};
        \draw[labelarrow] (loops) -- (wgcl.east)
            node[midway, above, xshift=0.7cm] {reasons about};

        \draw[labelarrow] (semirings) -- (heyting)
            node[midway, right, align=left] {adds \\ implications};
        \draw[labelarrow] (wp) -- (vp)
            node[midway, right, align=left] {formalisation};
        \draw[labelarrow] (wgcl) -- (wheyvl)
            node[midway, right, align=left] {encoded in};
        \draw[labelarrow] (loops) -- (wheyvl.east)
            node[midway, below, xshift=0.7cm] {encoded in};

        \draw[labelarrow] (heyting) -- (vp)
            node[midway, above] {};
        \draw[labelarrow] (vp) -- (wheyvl)
            node[midway, above] {};

        \begin{scope}[on background layer]
            \draw[pipelinearrow] plot [smooth, tension=0.5]
            coordinates {(wp.center) (wgcl.center) ([yshift=0pt]loops.north) (wheyvl.north)
                (vp.south) (qe.north west) ([xshift=10pt]smt.east) };
        \end{scope}
    \end{tikzpicture}

    \caption{Overview of our verification infrastructure for weighted programming. The left side shows the algebraic structures and assertion language, while the right side shows the IVL and verification pipeline.}
    \label{fig:overview}
\end{figure}
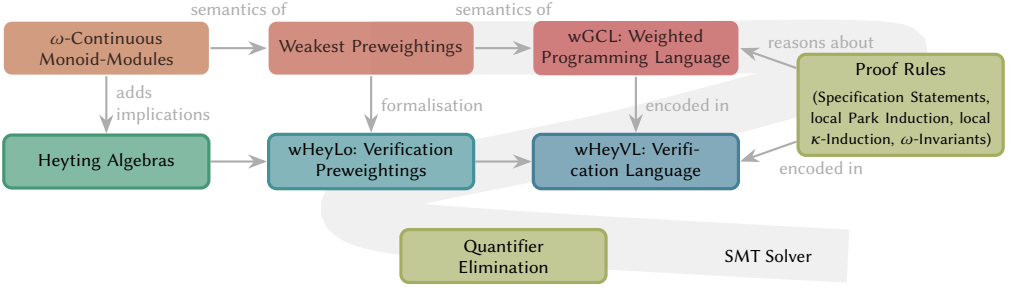

\paragraph{Overview of Our Approach.}
The first layer in \Cref{fig:overview} is the source language: monoid-modules are the basis for the weakest preweighting semantics for the $\wGCL$ language of weighted programs.
The second layer is our verification infrastructure: based on Heyting algebras for semantics, we define the \wHeyLo{} assertion language.
$\wGCL{}$ programs are translated to $\wHeyVL{}$ programs, using encodings of proof rules for procedure calls and loops.
The semantics of $\wHeyVL{}$ is defined via verification preweightings in $\wHeyLo{}$, and approximates the original $\wGCL{}$ semantics (soundness).
We provide a quantifier elimination procedure for \wHeyLo{} to enable SMT-based reasoning.

\paragraph{Contributions.}
We develop a deductive verification infrastructure for weighted programming.
\begin{itemize}
  \item We introduce \wHeyLo, a weighted assertion language parameterized by bi-Heyting algebras, and instantiate it for totally ordered and Boolean algebras, and the \emph{Why} monoid-module.
  \item We define \wHeyVL, an intermediate verification language for lower- and upper-bound reasoning.
  \item For proof rules, we encode specification statements, Park induction, $\kappa$-induction, and $\omega$-invariants as \wHeyVL programs and address local variables in our soundness proofs.
  \item We formalize a best-effort quantifier-elimination pass for \wHeyLo that preserves validity/covalidity on arbitrary verification conditions and is complete on identified fragments.
  \item We prototype the approach in \Caesar{} and evaluate it on case studies spanning expected costs, provenance, security clearances, and formal-language reasoning about lock-freedom.
\end{itemize}

\section{Weighted Programming and Reasoning via Weakest Preweightings}

\label{sec:weighted_programming}

We define the $\wGCL$ language of weighted programs, named for a \emph{weighted} version of the probabilistic $\pGCL$ language~\cite{McIverM05}, in turn named after Dijkstra's guarded command language $\GCL$~\cite{Dijkstra75}.
$\wGCL$ is based on monoid-modules, providing a general algebraic structure for trace weights.

\subsection{Monoid-Modules}

Monoid-modules consist of a monoid $\Monoid$ and a module $\Module$ over that monoid, and they generalize semirings, which are more commonly used in the literature on weighted automata and optimization.

\begin{definition}[Monoid-Module] \label{def:monoid-module}
    A \emph{monoid} $\Monoid = (\monoid, \odot, \sone)$ is a set $\monoid$ together with an associative \emph{multiplication} $\odot: \monoid\times \monoid\to \monoid$ and the corresponding neutral element $\sone \in \monoid$.

    A \emph{module $\Module = (\module, \oplus, \snull, \otimes)$ over the monoid} $\Monoid$ consists of a commutative monoid $(\module, \oplus, \snull)$ together with a \emph{scalar multiplication} $\otimes: \monoid\times \module \to \module$ such that
    \begin{itemize}
        \item multiplication is associative and distributive, hence for all $v, w \in \monoid$ and $x, y \in \module$ it is
        $
            (v \odot w) \otimes x = v \otimes (w \otimes x) \quad \text{and} \quad v \otimes (x \oplus y) = (v \otimes x) \oplus (v \otimes y),
        $

        \item the neutral element $\sone$ is also neutral w.r.t. the scalar multiplication and the neutral element $\snull$ annihilates, hence for all $v \in \monoid$ and $x \in \module$ it is
        $ \sone \otimes x = x \quad \text{and} \quad v \otimes \snull = \snull. $
    \end{itemize}
\end{definition}

With different choices of monoid-modules, we recover classical and probabilistic programming.
For instance, the \emph{angelic Boolean} monoid-module embeds Boolean best-case reasoning:
\[
    \textsc{Angelic Boolean:}\qquad \monbool = \monboollong \quad \text{and} \quad \modbool = \modboollong~,
\]
where at least one trace ($\oplus = \lor$) is true ($\odot = \otimes = \land$).
We embed probabilistic reasoning via:
\[
    \textsc{Probability:}\qquad \monprob = \monproblong \quad \text{and} \quad \modprob = \modproblong~,
\]
i.e. the monoid multiplies (finite) non-negative reals via the standard multiplication and the module adds them via the standard addition, where $\sone = 1$ and $\snull = 0$\footnote{Throughout this paper, we use the convention $0 \cdot \infty = 0$ and $x \cdot \infty = \infty$ for $x>0$.}.
Each step of the execution trace may be weighted by a finite value from $\PosReals$, and the total weight may even be infinite (from $\PosRealsInf$).

The \emph{tropical} monoid-module is often used to reason about optimization problems, where different actions are possible, choosing between actions corresponds to taking the action with the \emph{minimal} costs, and taking several actions corresponds to taking the \emph{sum} of their costs.
\[
    \textsc{Tropical:}\qquad \montrop = (\NatsX, +, 0) \quad \text{and} \quad \modtrop = (\NatsX, \min, \infty, +)~,
\]

One can also capture traces of program executions as words over an alphabet $\Gamma$.
The \emph{formal languages} monoid-module allows scalar multiplication with (finite) words from $\Gamma^*$, with the module operating on possibly infinite words from $2^{\Gamma^{\infty}}$.
With $\cdot$ denoting concatenation of words, and $\epsilon$ being the empty word:
\[
    \textsc{Formal languages:}\qquad \monlang = (\Gamma^*, \cdot, \epsilon) \quad \text{and} \quad \modlang = (2^{\Gamma^{\infty}}, \cup, \emptyset, \cdot)~.
\]
The additive operation $\oplus$ induces the following binary relation.

\begin{definition}[Natural Order]
    The \emph{natural order} on a monoid-module $\moduledef$ is a binary relation $\sleq\; \subseteq \module\times \module$, where
    $x \sleq y$ if and only if there exists an element $z\in \module$ with $x \oplus z = y. $
\end{definition}

The natural order on the probability module $\modprob$ is the standard order on the reals, since $a \sleq b$ \emph{if and only if} $a + c = b$ for some $c \in \PosRealsInf$ \emph{if and only if} $a \leq b$ for all $a, b \in \PosRealsInf$.
On the tropical module $\modtrop$, the natural order $\sleq$ is the reversed standard order $\geq$ on the extended natural numbers $\NatsX$.
For $x, y\in \NatsX$, $x\sleq y$ \emph{if and only if} there exists an element $z\in \NatsX$ such that $x \oplus z = x \min z = y$ \emph{if and only if} $x \geq y$.
It follows, e.g., that $\infty \sleq 20 \sleq 0$, $\infty$ is the smallest element, and $0$ the largest.

\begin{definition}[$\omega$-Continuous Monoid-Module]
    A module $\moduledef$ over the monoid $\monoiddef$ is $\omega$-continuous, if
    \begin{itemize}
        \item the natural order is a partial order (reflexive, transitive, and antisymmetric), a least element $\bot \in \module$ exists, and every $\omega$-ascending chain $(x_i)_{i\in\Nats} \subseteq \module$ has a supremum $\bigsqcup_{i\in\Nats} x_i \in \module$,\footnote{This defines exactly an $\omega$-complete partial order.}

        \item scalar multiplication and addition are $\omega$-continuous; for every $\omega$-ascending chain $(x_i)_{i\in\Nats} \subseteq \module$ and element $y\in \monoid$, it is
        \begin{gather*}
            y \otimes \bigsqcup_{i\in\Nats} x_i = \bigsqcup_{i\in\Nats} (y \otimes x_i) \qquad \text{and} \qquad
            y \splus \bigsqcup_{i\in\Nats} x_i = \bigsqcup_{i\in\Nats} (y \splus x_i).
        \end{gather*}
    \end{itemize}
\end{definition}

We want to support both lower- and upper-bound reasoning, so we require \emph{$\omega$-bicontinuous monoid modules}, which satisfy both $\omega$-continuity and $\omega$-cocontinuity.
The latter is continuity with respect to the reversed natural order $\sgeq$\footnote{The reversed natural order $\sgeq$ is defined as $x \sgeq y$ if and only if $y \sleq x$ for all $x, y \in \Module$.}.
All our monoid-modules are $\omega$-bicontinuous (\Cref{app:case-studies}).

\subsection{Syntax and Semantics of \texorpdfstring{$\wGCL$}{wGCL}}

The syntax of $\wGCL$ is standard imperative programming with a \emph{weigh} statement $\WEIGH{\mexpr}$, which weighs branches with the monoid multiplication $\stimes$, and a \emph{branching} statement $\BRANCH{\wwcomi}{\wwcomii}$.

\begin{definition}[Syntax of $\wGCL$]
    \label{def:wGCL_syntax}
    For an $\omega$-continuous module $\Module$ over the monoid $\Monoid$ and a countable set of variables $\Vars$ ranging over the natural numbers $\Nats$, the set $\wGCL$ of weighted programs is constructed by the grammar in \Cref{fig:wGCL_syntax}.
\end{definition}

\begin{figure*}[t]
    \centering
    {
    \[
    \renewcommand{\arraystretch}{1.25}%
    \setlength{\arraycolsep}{3pt}
        \begin{array}{@{}cc@{}}
            \begin{array}[t]{@{}rcl@{}}
                \aexpr & \rightarrow
                    & n
                    \mid x
                    \mid \aexpr + \aexpr
                    \mid \aexpr \cdot \aexpr
                    \mid \aexpr \monus \aexpr \\
                \bexpr & \rightarrow
                    & \aexpr < \aexpr
                    \mid \neg \bexpr
                    \mid \bexpr \wedge \bexpr \\
                \mexpr & \rightarrow
                    & v
                    \mid \mexpr \sdot \mexpr
                    \mid \booth{\bexpr}\cdot \mexpr
                    \mid \mexpr^\aexpr
            \end{array}
            \hspace{2em}&\hspace{2em}
            \begin{array}[t]{@{}rcl@{}}
                \wwcom & \rightarrow
                    & \ASSIGN{x}{\aexpr}
                    \mid \wwcomi\semcol \wwcomii
                    \mid \WEIGH{\mexpr} \\
                  & \mid
                    & \ITE{\bexpr}{\wwcomi}{\wwcomii} \\
                  & \mid
                    & \WHILEDO{\bexpr}{\wwcomi}
                    \mid \BRANCH{\wwcomi}{\wwcomii}
            \end{array}
        \end{array}
    \]
    }

    \caption{Syntax of expressions and \wGCL programs, with $x \in \Vars$, $n \in \Nats$, $v \in \monoid$, $\wwcomi,\wwcomii \in \wGCL$, $\aexpr \in \ArithExp$, $\bexpr \in \BExp$, and $\mexpr \in \MonExp$.}
    \label{fig:wGCL_syntax}
    \Description{Grammar for wGCL commands and arithmetic, Boolean, and monoid expressions.}
\end{figure*}

Programs in $\wGCL$ operate on \emph{program states} $\Sigma = \{ \sigma: \Vars \to \Nats \mid \supp(\sigma) \text{ is finite} \}$, i.e. functions mapping variables to natural numbers with finite support.
The semantics $\sem{\aexpr}: \Sigma \to \Nats$ for arithmetic expressions and $\sem{\bexpr}: \Sigma \to \{0,1\}$ for Boolean expressions are standard (c.f. \Cref{fig:semantics_arith_bool}).
In particular, $\monus$ denotes the truncated subtraction, which is defined as $\interpretsimple{\aexpr \monus \aexpr'}(\sigma) = \max\{0, \interpretsimple{\aexpr}(\sigma) - \interpretsimple{\aexpr'}(\sigma)\}$ for all $\sigma\in\Sigma$.
Monoid expressions $\mexpr$ describe elements of monoid $\Monoid$ and have semantics $\sem{\mexpr}: \Sigma \to \Monoid$ with $\sem{v}(\sigma) = v$, $\sem{\mexpr \sdot \mexpr'}(\sigma) = \sem{\mexpr}(\sigma) \sdot \sem{\mexpr'}(\sigma)$.
$\booth{\bexpr}\cdot \mexpr$ is a conditional monoid expression, which evaluates to $\sem{\mexpr}(\sigma)$ if $\sem{\bexpr}(\sigma)$ is true and to $\sone$ otherwise.
The expression $\mexpr^\aexpr$ denotes the $\sem{\aexpr}(\sigma)$-fold monoid multiplication of $\sem{\mexpr}(\sigma)$.

A program execution starts in some state $\sigma \in \Sigma$ with the neutral weight $\sone \in \Monoid$.
With the assignment $\ASSIGN{x}{\aexpr}$, we \emph{update} a variable $x\in\Vars$ to an arithmetic expression $\aexpr$ in a given program state $\sigma\in\Sigma$ by
\(\sigma[x \mapsto \aexpr] = \lambda y.\;\text{if } y = x \text{ then } \interpretsimple{\aexpr}(\sigma) \text{ else } \sigma(y).\)
When weighing with $\WEIGH{\mexpr}$, the previous weight is multiplied by $\sem{\mexpr}(\sigma)$ for the program state $\sigma \in \Sigma$.
The branching statement $\BRANCH{\wwcomi}{\wwcomii}$ combines the weights of the branches $\wwcomi$ and $\wwcomii$ by the module addition $\oplus$.

\begin{example}
    Recall the ski rental example (\Cref{fig:intro-wgcl-ski-rental}), defined on the tropical module $\modtrop$ over the monoid $\montrop$.
    Each iteration of the loop in state $\sigma$ and with current weight $m \in \Monoid$ chooses the minimum cost ($\oplus = \min$) of either renting skis for a day at a cost of 1  (left branch) or buying skis for a cost of $c$ (right branch).
    When renting, the new weight is $1 \otimes w = 1 + w$ and the new state is $\sigma\substBy{x}{x + 1}$.
    When buying, the new weight is $c \otimes w = c + w$ with new state $\sigma\substBy{x}{0}$.

    For the CAS counter example (\Cref{fig:intro-wgcl-lock-freedom}), we use the monoid-module of formal languages $\modlang$ over the alphabet $\Set{S_i, F_i \mid i \in \Nats}$.
    Each of the two if-branches corresponds to either successful or unsuccessful CAS operations, and we weigh the branches with either the success word $S_i$ or failure word $F_i$ for thread $i$.
    Therefore, the weight of a program execution is an infinite word representing the sequence of successful and unsuccessful CAS operations.
    Full operational semantics can be found in \Cref{app:operational-semantics}.
\end{example}

\subsection{Weakest (Liberal) Preweightings}

We reason about weights of program executions via \emph{weightings}.

\begin{definition}[Weighting] \label{def:weighting}
    For a  monoid-module $\Module$ and a set of variables $\Vars$, a \emph{weighting} $\wla \in \Weighting = \left\{\; \Sigma\to\module \;\right\}$ maps each program state to an element, called \emph{weight}, in the module.
\end{definition}
For a fixed $\omega$-bicontinuous module $\Module$ over the monoid $\Monoid$, the set $\Weighting$ of weightings is also an $\omega$-bicontinuous module over $\Monoid$, since all properties are lifted pointwise.
The \emph{Iverson bracket} $[\bexpr]$ filters weightings to states satisfying the Boolean expression $\bexpr$.
Let $\wla\in\Weighting$ and $\sigma\in\Sigma$.
Then:
\[
    ([\bexpr]\cdot \wla)(\sigma) = \begin{cases}
        \wla(\sigma), & \text{ if } \sem{\bexpr}(\sigma), \\
        \snull, & \text{ otherwise.}
    \end{cases}
\]
Substitution of a variable $x\in\Vars$ by an arithmetic expression $\aexpr$ in a weighting $\wla\in\Weighting$ is given by $\wla\substBy{x}{\aexpr} = \lambda \sigma. \wla(\sigma\substBy{x}{\sigma(\aexpr)}).$
We reason about weighted programs using \emph{weakest preweightings}~\cite{BatzGKKW22}, which are the generalization of classical weakest preconditions~\cite{Dijkstra75} and probabilistic weakest preexpectations~\cite{McIverM05} to the setting of weighted programs.
They are defined recursively over the program structure, with both a non-liberal and liberal variant.

\begin{definition}[Weakest Preweightings]
    \label[definition]{def:wp_wlp}
    For an $\omega$-continuous module $\Module$ over monoid $\Monoid$ and program $\wwcom \in \wGCL$, \emph{the weakest (liberal) preweightings} $\wp{\wwcom},\wlp{\wwcom} \colon (\Weighting \to \Weighting)$ are defined in \Cref{fig:weakest_pre}.
\end{definition}

\begin{figure}
    \centering
    \small
    \renewcommand{\arraystretch}{1.15}
    \setlength{\tabcolsep}{4pt}
    \begin{tabular}
        {
            >{\raggedright\arraybackslash}m{0.24\linewidth}
            !{\color{Platinum}\vrule width 1.5pt}
            >{\raggedright\arraybackslash}m{0.335\linewidth}
            !{\color{Platinum}\vrule width 2.5pt}
            >{\raggedright\arraybackslash}m{0.335\linewidth}
        }
        \toprule
        $\wwcom$
            & $\wp{\wwcom}(\wla)$
            & $\wlp{\wwcom}(\wla)$
        \\
        \midrule
        \cellcolor{basicBackground!2}$\ASSIGN{x}{\aexpr}$
            & \cellcolor{dualColor!5}$\wla\substBy{x}{\aexpr}$
            & \cellcolor{primalColor!5}$\wla\substBy{x}{\aexpr}$
        \\
        \cellcolor{basicBackground!2}$\wwcomi \semcol \wwcomii$
            & \cellcolor{dualColor!5}$\wp{\wwcomi}(\wp{\wwcomii}(\wla))$
            & \cellcolor{primalColor!5}$\wlp{\wwcomi}(\wlp{\wwcomii}(\wla))$
        \\
        \cellcolor{basicBackground!2}$\ITE{\bexpr}{\wwcomi}{\wwcomii}$
            & \cellcolor{dualColor!5}
              $[\bexpr]\cdot \wp{\wwcomi}(\wla) \oplus [\neg\bexpr]\cdot \wp{\wwcomii}(\wla)$
            & \cellcolor{primalColor!5}
              $[\bexpr]\cdot \wlp{\wwcomi}(\wla) \oplus [\neg\bexpr]\cdot \wlp{\wwcomii}(\wla)$
        \\
        \cellcolor{basicBackground!2}$\BRANCH{\wwcomi}{\wwcomii}$
            & \cellcolor{dualColor!5}
              $\wp{\wwcomi}(\wla) \oplus \wp{\wwcomii}(\wla)$
            & \cellcolor{primalColor!5}
              $\wlp{\wwcomi}(\wla) \oplus \wlp{\wwcomii}(\wla)$
        \\
        \cellcolor{basicBackground!2}$\WEIGH{\mexpr}$
            & \cellcolor{dualColor!5}$\mexpr \otimes \wla$
            & \cellcolor{primalColor!5}$\mexpr \otimes \wla$
        \\
        \cellcolor{basicBackground!2}$\WHILEDO{\bexpr}{\wwcomi}$
            & \cellcolor{dualColor!5}
              $\lfp{\wX}{\underbrace{[\neg\bexpr]\cdot \wla \oplus [\bexpr]\cdot \wp{\wwcomi}(\wX)}_{=\; \wpchar(\wX)}}$
            & \cellcolor{primalColor!5}
              $\gfp{\wX}{\underbrace{[\neg\bexpr]\cdot \wla \oplus [\bexpr]\cdot \wlp{\wwcomi}(\wX)}_{=\; \wlpchar(\wX)}}$
        \\
        \bottomrule
    \end{tabular}

    \caption{The iterative definition of the weakest (liberal) preweightings $\sfsymbol{w(l)p}\llbracket \wwcom\rrbracket (\wla)$  for $\wwcom \in \wGCL$ and $\wla\in\weightings$.}
    \label{fig:weakest_pre}
    \Description{The rules defining $\symWp$ ($\symWlp$).}
\end{figure}

The rules in \Cref{fig:weakest_pre} define $\wp{\wwcom}(\wla)$ for a program $\wwcom \in \wGCL$ and weighting $\wla \in \Weighting$ by propagating $\wla$ backwards through the program.
Because $\wla$ is evaluated at the end of the program, it is called a \emph{postweighting}, and $\wp{\wwcom}(\wla)$ is called a \emph{preweighting}.

Assignments substitute the assigned expression into the postweighting, while sequential composition composes transformers in reverse program order, yielding $\wp{\wwcomi}(\wp{\wwcomii}(\wla))$ for $\wwcomi \semcol \wwcomii$.
Conditionals select the corresponding branch pointwise: states satisfying $\bexpr$ use $\wp{\wwcomi}(\wla)$, and states satisfying $\neg\bexpr$ use $\wp{\wwcomii}(\wla)$.
Branching combines the two branch preweightings by the module addition $\oplus$, and a weighing statement $\WEIGH{\mexpr}$ scalar-multiplies the postweighting by the monoid expression $\mexpr$.
Loops $\WHILEDO{\bexpr}{\wwcomi}$ are interpreted as the least fixed point of the characteristic function $\wpchar(\wX) = [\neg\bexpr]\cdot \wla \oplus [\bexpr]\cdot \wp{\wwcomi}(\wX)$, which exists by $\omega$-continuity and Kleene's fixed-point theorem~\cite{BatzGKKW22}.
The liberal transformer $\wlp{\wwcom}(\wla)$ is analogous, except that loops use the greatest fixed point of the function $\wlpchar(\wX) = [\neg\bexpr]\cdot \wla \oplus [\bexpr]\cdot \wlp{\wwcomi}(\wX)$.

The resulting difference in both transformers is that $\wp{\wwcom}(\wla)$ collects the weights of all terminating executions with $\wla$ evaluated the final state, while $\wlp{\wwcom}(\wla)$ additionally collects the weights of all non-terminating executions.

\begin{example}
    Recall the Ski rental problem (\Cref{fig:intro-wgcl-ski-rental}) over the tropical monoid-module $\modtrop$.
    For the loop body $\wwcoma$, we obtain the weakest preweighting $\wp{\wwcoma}(\wla) = 1 \otimes \wla\substBy{n}{n-1} \oplus c \otimes \wla\substBy{n}{0}$.
    The characteristic function of the loop is $\wpchar(\wX) = [n = 0] \cdot \wla \oplus [n > 0] \cdot (1 \otimes \wX\substBy{n}{n-1} \oplus c \otimes \wX\substBy{n}{0})$.
    One can prove that the least fixed point evaluates to $n \oplus y$.
\end{example}

\begin{theorem}[\cite{BatzGKKW22}] \label{thm:weakest_pre}
    The weighting transformer $\wp{\wwcom}(\cdot)$ is a well-defined, $\omega$-continuous function. In particular, the least fixed point of the function $\Phi$ exists for all $\wla\in \Weighting$ and $\wwcom \in \wGCL$.
\end{theorem}

Analogously, we define weakest \emph{liberal} preweightings for $\omega$-cocontinuous modules $\Module$, see \Cref{def:wlp} and \Cref{thm:weakest_liberal_pre} on \cpageref{def:wlp,thm:weakest_liberal_pre}.

\section{Heyting (Co-)Implications for Totally Ordered and Boolean Algebras}
\label{sec:heyting}

Central to the specifications and proofs in the weakest preweighting calculi are \emph{comparisons} between weightings.
In addition, we need to express assumptions during proofs, e.g. assuming that a certain precondition holds, or that some value is bounded by a certain threshold.
We use \emph{Heyting algebras} to internalize the comparison and assumption operations in the logic,  which give us quantitative \emph{implications} and dual \emph{co-implications}.
Although \HeyLo{} and \HeyVL{} of prior work~\cite{SchroerBKKM23} define quantitative (co-)implications, they only do so for the specific setting of probabilistic programs (corresponding to $\modprob$).
Here, we generalize to more general $\omega$-bicontinuous monoid-modules.

In the following, we define instances of Heyting algebras for (1) totally ordered lattices, (2) Boolean algebras, and (3) the specific setting of the why semiring $\modwhy$.
\Cref{fig:modules} gives an overview of the different monoid-modules and their corresponding (co-)implications.

\begin{figure}
    \centering
    \begin{adjustbox}{max width=\linewidth,max totalheight=0.76\textheight,center}
    \begin{minipage}{\linewidth}
    \centering

    \begin{tikzpicture}[overlay, remember picture]
        \node[xshift=-0.5em, yshift=-0.2*\baselineskip] (adown) at (pic cs:a) {};
        \node[xshift=-0.5em, yshift=0.2*\baselineskip] (bup) at (pic cs:b) {};
        \node[xshift=-0.5em, yshift=-0.2*\baselineskip] (bdown) at (pic cs:b) {};
        \node[xshift=-0.5em, yshift=0.2*\baselineskip] (cup) at (pic cs:c) {};
        \node[xshift=-0.5em, yshift=-0.2*\baselineskip] (cdown) at (pic cs:c) {};
        \node[xshift=-1em, yshift=0.2*\baselineskip] (dup) at (pic cs:d) {};
        \node[xshift=-1em, yshift=-0.2*\baselineskip] (ddown) at (pic cs:d) {};
        \node[xshift=-0.3em, yshift=-1.2*\baselineskip] (eup) at (pic cs:e) {};

        \draw [decorate, decoration = {brace, mirror, amplitude=\tablespace},
            transform canvas={yshift=\baselineskip, xshift=-3*\tablespace}, color=hintcolor, thick]
            (adown) -- (bup)
            node[midway, above, xshift=-\tablespace, yshift=\baselineskip, rotate=90]
            {$\substack{\text{\textbf{weightings} } \weightings:\; \Sigma \to \Module \\ \text{map into module } \Module}$};

        \draw [decorate, decoration = {brace, mirror, amplitude=\tablespace},
            transform canvas={yshift=\baselineskip, xshift=-3*\tablespace}, color=hintcolor, thick]
            (bdown) -- (cup)
            node[midway, above, xshift=-\tablespace, rotate=90]
            {$\substack{\text{\textbf{weights} from} \\ \text{monoid } \Monoid}$};

        \draw [decorate, decoration = {brace, mirror, amplitude=\tablespace},
            transform canvas={yshift=\baselineskip, xshift=-3*\tablespace}, color=hintcolor, thick]
            (cdown) -- (dup)
            node[midway, above, xshift=-\tablespace, rotate=90]
            {$\substack{\text{\textbf{bounds} via } \\ \omega-\text{bicomplete} \\ \text{partial order}}$};

        \draw [decorate, decoration = {brace, mirror, amplitude=\tablespace},
            transform canvas={yshift=\baselineskip, xshift=-3*\tablespace}, color=hintcolor, thick]
            (ddown) -- (eup)
            node[midway, above, xshift=-\tablespace, rotate=90]
            {$\substack{\text{\textbf{comparison} via} \\ \text{Heyting algebra}}$};
    \end{tikzpicture}

    \begin{align*}
        \renewcommand{\arraystretch}{1.35}%
        \setlength{\tabcolsep}{4pt}
        \begin{array}{!{\color{black}\vrule width \heavyrulewidth}c!{\color{black}\vrule width \lightrulewidth}}
            \\
            \arrayrulecolor{Platinum}\cmidrule[2.5pt]{1-1} \tikzmark{a}
            \module
            \\
            \oplus
            \\
            \snull
            \\
            \otimes
            \\ \arrayrulecolor{Platinum}\cmidrule[1.5pt]{1-1} \tikzmark{b}
            \monoid
            \\
            \odot
            \\
            \sone
            \\ \arrayrulecolor{Platinum}\cmidrule[2.5pt]{1-1} \tikzmark{c}
            \hleq
            \\
            \bot
            \\
            \top
            \\
            \hand
            \\
            \hor
            \\ \arrayrulecolor{Platinum}\cmidrule[1.5pt]{1-1} \tikzmark{d}
            \null
            \\
            \impl
            \\
            \tikzmark{e}
            \coimpl
        \end{array}
        \begin{array}{c}
            \cellcolor{basicBackground!2}{\mathtt{\trop}}
            \\ \arrayrulecolor{Platinum}\cmidrule[2.5pt]{1-1}
            \cellcolor{basicBackground!2}{\NatsX}
            \\
            \cellcolor{basicBackground!2}{\min}
            \\
            \cellcolor{basicBackground!2}{\infty}
            \\
            \cellcolor{basicBackground!2}{+}
            \\ \arrayrulecolor{Platinum}\cmidrule[1.5pt]{1-1}
            \cellcolor{basicBackground!2}{\NatsX}
            \\
            \cellcolor{basicBackground!2}{+}
            \\
            \cellcolor{basicBackground!2}{0}
            \\ \arrayrulecolor{Platinum}\cmidrule[2.5pt]{1-1}
            \cellcolor{basicBackground!2}{\geq}
            \\
            \cellcolor{basicBackground!2}{\infty}
            \\
            \cellcolor{basicBackground!2}{0}
            \\
            \cellcolor{basicBackground!2}{\max}
            \\
            \cellcolor{basicBackground!2}{\min}
            \\ \arrayrulecolor{Platinum}\cmidrule[1.5pt]{1-1}
            \cellcolor{basicBackground!2}{\substack{\text{totally} \\ \text{ordered}}}
            \\
            \cellcolor{basicBackground!2}{\totimpl}
            \\
            \cellcolor{basicBackground!2}{\totcoimpl}
        \end{array}
        \begin{array}{c}
            \cellcolor{basicBackground!2}{\mathtt{\comb}}
            \\ \arrayrulecolor{Platinum}\cmidrule[2.5pt]{1-1}
            \cellcolor{basicBackground!2}{\NatsX}
            \\
            \cellcolor{basicBackground!2}{+}
            \\
            \cellcolor{basicBackground!2}{0}
            \\
            \cellcolor{basicBackground!2}{\cdot}
            \\ \arrayrulecolor{Platinum}\cmidrule[1.5pt]{1-1}
            \cellcolor{basicBackground!2}{\Nats}
            \\
            \cellcolor{basicBackground!2}{\cdot}
            \\
            \cellcolor{basicBackground!2}{1}
            \\ \arrayrulecolor{Platinum}\cmidrule[2.5pt]{1-1}
            \cellcolor{basicBackground!2}{\leq}
            \\
            \cellcolor{basicBackground!2}{0}
            \\
            \cellcolor{basicBackground!2}{\infty}
            \\
            \cellcolor{basicBackground!2}{\min}
            \\
            \cellcolor{basicBackground!2}{\max}
            \\ \arrayrulecolor{Platinum}\cmidrule[1.5pt]{1-1}
            \cellcolor{basicBackground!2}{\substack{\text{totally} \\ \text{ordered}}}
            \\
            \cellcolor{basicBackground!2}{\totimpl}
            \\
            \cellcolor{basicBackground!2}{\totcoimpl}
        \end{array}
        \begin{array}{c}
            \cellcolor{basicBackground!2}{\mathtt{\prob}}
            \\ \arrayrulecolor{Platinum}\cmidrule[2.5pt]{1-1}
            \cellcolor{basicBackground!2}{\PosRealsInf}
            \\
            \cellcolor{basicBackground!2}{+}
            \\
            \cellcolor{basicBackground!2}{0}
            \\
            \cellcolor{basicBackground!2}{\cdot}
            \\ \arrayrulecolor{Platinum}\cmidrule[1.5pt]{1-1}
            \cellcolor{basicBackground!2}{\PosReals}
            \\
            \cellcolor{basicBackground!2}{\cdot}
            \\
            \cellcolor{basicBackground!2}{1}
            \\ \arrayrulecolor{Platinum}\cmidrule[2.5pt]{1-1}
            \cellcolor{basicBackground!2}{\leq}
            \\
            \cellcolor{basicBackground!2}{0}
            \\
            \cellcolor{basicBackground!2}{\infty}
            \\
            \cellcolor{basicBackground!2}{\min}
            \\
            \cellcolor{basicBackground!2}{\max}
            \\ \arrayrulecolor{Platinum}\cmidrule[1.5pt]{1-1}
            \cellcolor{basicBackground!2}{\substack{\text{totally} \\ \text{ordered}}}
            \\
            \cellcolor{basicBackground!2}{\totimpl}
            \\
            \cellcolor{basicBackground!2}{\totcoimpl}
        \end{array}
        \begin{array}{c}
            \cellcolor{basicBackground!2}{\mathtt{\clear}}
            \\ \arrayrulecolor{Platinum}\cmidrule[2.5pt]{1-1}
            \cellcolor{basicBackground!2}{\clearUniverse}
            \\
            \cellcolor{basicBackground!2}{\min}
            \\
            \cellcolor{basicBackground!2}{\mathsf{TS}}
            \\
            \cellcolor{basicBackground!2}{\max}
            \\ \arrayrulecolor{Platinum}\cmidrule[1.5pt]{1-1}
            \cellcolor{basicBackground!2}{\clearUniverse}
            \\
            \cellcolor{basicBackground!2}{\max}
            \\
            \cellcolor{basicBackground!2}{\mathsf{P}}
            \\ \arrayrulecolor{Platinum}\cmidrule[2.5pt]{1-1}
            \cellcolor{basicBackground!2}{\geq}
            \\
            \cellcolor{basicBackground!2}{\mathsf{TS}}
            \\
            \cellcolor{basicBackground!2}{\mathsf{P}}
            \\
            \cellcolor{basicBackground!2}{\max}
            \\
            \cellcolor{basicBackground!2}{\min}
            \\ \arrayrulecolor{Platinum}\cmidrule[1.5pt]{1-1}
            \cellcolor{basicBackground!2}{\substack{\text{totally} \\ \text{ordered}}}
            \\
            \cellcolor{basicBackground!2}{\totimpl}
            \\
            \cellcolor{basicBackground!2}{\totcoimpl}
            \\
        \end{array}
        \begin{array}{c}
            \cellcolor{basicBackground!2}{\mathtt{\bool}}
            \\ \arrayrulecolor{Platinum}\cmidrule[2.5pt]{1-1}
            \cellcolor{basicBackground!2}{\{0,1\}}
            \\
            \cellcolor{basicBackground!2}{\vee}
            \\
            \cellcolor{basicBackground!2}{0}
            \\
            \cellcolor{basicBackground!2}{\wedge}
            \\ \arrayrulecolor{Platinum}\cmidrule[1.5pt]{1-1}
            \cellcolor{basicBackground!2}{\{0,1\}}
            \\
            \cellcolor{basicBackground!2}{\wedge}
            \\
            \cellcolor{basicBackground!2}{1}
            \\ \arrayrulecolor{Platinum}\cmidrule[2.5pt]{1-1}
            \cellcolor{basicBackground!2}{\leq}
            \\
            \cellcolor{basicBackground!2}{0}
            \\
            \cellcolor{basicBackground!2}{1}
            \\
            \cellcolor{basicBackground!2}{\wedge}
            \\
            \cellcolor{basicBackground!2}{\vee}
            \\ \arrayrulecolor{Platinum}\cmidrule[1.5pt]{1-1}
            \cellcolor{basicBackground!2}{\substack{\text{Boolean} \\ \text{algebra}}}
            \\
            \cellcolor{basicBackground!2}{\matimpl}
            \\
            \cellcolor{basicBackground!2}{\matcoimpl}
        \end{array}
        \begin{array}{c}
            \cellcolor{basicBackground!2}{\mathtt{\lang}}
            \\ \arrayrulecolor{Platinum}\cmidrule[2.5pt]{1-1}
            \cellcolor{basicBackground!2}{2^{\Gamma^{\infty}}}
            \\
            \cellcolor{basicBackground!2}{\cup}
            \\
            \cellcolor{basicBackground!2}{\emptyset}
            \\
            \cellcolor{basicBackground!2}{\cdot}
            \\ \arrayrulecolor{Platinum}\cmidrule[1.5pt]{1-1}
            \cellcolor{basicBackground!2}{\Gamma^*}
            \\
            \cellcolor{basicBackground!2}{\cdot}
            \\
            \cellcolor{basicBackground!2}{\varepsilon}
            \\ \arrayrulecolor{Platinum}\cmidrule[2.5pt]{1-1}
            \cellcolor{basicBackground!2}{\subseteq}
            \\
            \cellcolor{basicBackground!2}{\emptyset}
            \\
            \cellcolor{basicBackground!2}{\Gamma^{\infty}}
            \\
            \cellcolor{basicBackground!2}{\cap}
            \\
            \cellcolor{basicBackground!2}{\cup}
            \\ \arrayrulecolor{Platinum}\cmidrule[1.5pt]{1-1}
            \cellcolor{basicBackground!2}{\substack{\text{Boolean} \\ \text{algebra}}}
            \\
            \cellcolor{basicBackground!2}{\matimpl}
            \\
            \cellcolor{basicBackground!2}{\matcoimpl}
        \end{array}
        \begin{array}{c!{\color{black}\vrule width \heavyrulewidth}}
            \cellcolor{basicBackground!2}{\mathtt{\why}}
            \\ \arrayrulecolor{Platinum}\cmidrule[2.5pt]{1-1}
            \cellcolor{basicBackground!2}{\pDNF}
            \\
            \cellcolor{basicBackground!2}{\vee}
            \\
            \cellcolor{basicBackground!2}{0}
            \\
            \cellcolor{basicBackground!2}{\wedge}
            \\ \arrayrulecolor{Platinum}\cmidrule[1.5pt]{1-1}
            \cellcolor{basicBackground!2}{\pDNF}
            \\
            \cellcolor{basicBackground!2}{\wedge}
            \\
            \cellcolor{basicBackground!2}{1}
            \\ \arrayrulecolor{Platinum}\cmidrule[2.5pt]{1-1}
            \cellcolor{basicBackground!2}{\leq}
            \\
            \cellcolor{basicBackground!2}{0}
            \\
            \cellcolor{basicBackground!2}{1}
            \\
            \cellcolor{basicBackground!2}{\wedge}
            \\
            \cellcolor{basicBackground!2}{\vee}
            \\ \arrayrulecolor{Platinum}\cmidrule[1.5pt]{1-1}
            \cellcolor{basicBackground!2}{\substack{\text{finite} \\ \text{distributive}}}
            \\
            \cellcolor{basicBackground!2}{\impl_{\pDNF}}
            \\
            \cellcolor{basicBackground!2}{\coimpl_{\pDNF}}
        \end{array}
    \end{align*}

    \begin{align*}
        \begin{array}{ll}
            a \totimpl b = \left. \begin{cases}
                \top, &\text{ if } a \hleq b, \\ b, &\text{ otherwise}
            \end{cases} \right\}
            \qquad &
            a \totcoimpl b = \left. \begin{cases}
                \bot, &\text{ if } b \hleq a, \\ b, &\text{ otherwise}
            \end{cases} \right\}
            \\[15pt]
            a \matimpl b = \neg a \hor b
            &
            a \matcoimpl b = \neg a \hand b
        \end{array}
    \end{align*}

    \end{minipage}
    \end{adjustbox}

    \caption{An overview over $\omega$-bicontinuous monoid-modules that are bi-Heyting algebras. All instances from Table~1 in \cite{BatzGKKW22} for weighted programming reappear, though extended by the Heyting (co-)implications. \\
    For $\clear$, $P<C<S<TS$ and the displayed order $\geq$ is the natural order induced by $\min$.
    For $\lang$, $\Gamma^\infty = \Gamma^* \cup \Gamma^\omega$, $L \in 2^{\Gamma^{\infty}}$ is a language of finite and infinite words, $g \in \Gamma^*$ is a finite word, and $g \cdot L = \{ g\cdot l \mid l \in L \}$.
    For $\why$, $\pDNF$ denotes positive DNF formulae over a finite set of variables, ordered by semantic implication; $\impl_{\pDNF}$ and $\coimpl_{\pDNF}$ are the residuals of this finite distributive lattice.}
    \label{fig:modules}
\end{figure}

\subsection{Heyting Algebras}

In our setting, we consider monoid-modules that are bounded lattices\footnote{In a bounded lattice $(\hdom, \hleq)$, for all $a \in \hdom$, we have $\bot \hleq a \hleq \top$.} and equip them with an implication. This yields a Heyting algebra, where the implication $a \impl ( \cdot ): \hdom \impl \hdom$ is right adjoint to the meet $a \hand ( \cdot ): \hdom \impl \hdom$.

\begin{definition}[Bi-Heyting Algebra] \label{def:heyting_algebra}
    A \emph{Heyting algebra} $\hdef$ is a bounded lattice $(\hdom, \hleq)$ together with an \emph{implication} $\impl: \hdom \times \hdom \to \hdom$ such that for all $a, b, x \in \hdom$,
    \begin{gather*}
        (a \hand x) \hleq b \quad\text{ iff }\quad x \hleq (a \impl b).
    \end{gather*}
    For every $a \in \hdom$, let $\neg a := a \impl \bot$ be the \emph{Heyting negation}.
    The dual structure $\hcodef$ is a \emph{co-Heyting algebra} with the \emph{co-implication} $\coimpl: \hdom \times \hdom \to \hdom$ such that for all $a, b, x \in \hdom$,
    \begin{gather*}
        b \hleq (a \hor x) \quad\text{ iff }\quad (a \coimpl b) \hleq x.
    \end{gather*}
    For every $a \in \hdom$, let $\coneg a := a \coimpl \top$ be the \emph{co-negation}.
    A Heyting algebra which is also a co-Heyting algebra is called a \emph{bi-Heyting algebra}.
\end{definition}

The following algebraic form of the \emph{deduction theorem}~\cite{Kleene02} allows us to internalize comparisons in the logic via the implication $\impl$ and co-implication $\coimpl$.

\begin{restatable}[Deduction Theorem]{theorem}{thmDeductionTheorem}
    \label[theorem]{thm:deduction_theorem}
    Let $\hsym$ be a bi-Heyting algebra and $a, b \in \hdom$.
    Then, we have
	\begin{align*}
			a \hheyloleq b \qiff a \impl b = \top
			\qquad\textnormal{and}\qquad
			a \hheylogeq b \qiff a \coimpl b = \bot.
	\end{align*}%
\end{restatable}
With the \emph{Boolean embedding} $\embed{\cdot} \colon \Bools \to \hdom$ that maps $\true$ to $\top$ and $\false$ to $\bot$, we can internalize the Boolean comparison $\impl$ and $\coimpl$ in the logic as well and even encode if-else statements.

\begin{restatable}{lemma}{lemBoolEmbeddingImplCoimpl}
    \label[lemma]{lem:bool_embedding_impl_coimpl}
    Let $\hsym$ be a bi-Heyting algebra.
    For all $b \in \Bools$ and $x \in \hdom$, we have
    \begin{align*}
        \embed{b} \impl x = \ifThenElse{b}{x}{\top}
        \qquad\textnormal{and}\qquad
        \embed{b} \coimpl x = \ifThenElse{\neg b}{x}{\bot}
    \end{align*}
\end{restatable}

\begin{example}[Ternary Conditionals]
    Let $\hdef$ be a Heyting algebra and $c \in \Bools$.
    Then, for all $a, b \in \hdom$, we have $(\embed{c} \impl a) \hand (\embed{\neg c} \impl b) = \text{ if } c \text{ then } a \text{ else } b$.
    Dually, for co-Heyting algebra $\hcodef$, we have $(\embed{c} \coimpl a) \hor (\embed{\neg c} \coimpl b) = \text{ if } c \text{ then } b \text{ else } a$.
\end{example}

Instead of using a Boolean embedding, we want to express \emph{qualitative} properties directly via our bi-Heyting algebra. To use proof rules such as Park induction, we need to check whether a certain entailment holds.
To connect the quantitative domain $\hsym$ with qualitative reasoning, we define \emph{validation} and \emph{covalidation} operators that "booleanize" statements by mapping them to the extreme values $\bot$ and $\top$ of the bi-Heyting algebra.

\begin{definition}[Validation and Covalidation]
    \label[definition]{def:heyting-validation}
    Let $\hdef$ be a Heyting algebra with $x \in \hdom$, then we define the {validation} $\heylovalidate{x}$ and \emph{covalidation} $\heylocovalidate{x}$ as follows:
    \begin{gather*}
        \heylovalidate{x}
        = \ifThenElse{x = \top}{\top}{\bot}
        \quad \text{and} \quad
        \heylocovalidate{x}
        = \ifThenElseDot{x = \bot}{\bot}{\top}
    \end{gather*}
\end{definition}

To obtain direct definitions of Heyting (co-)implications for totally ordered lattices, Boolean algebras, and the why semiring $\modwhy$, the following characterization helps.

\begin{restatable}[Residual Characterization of (Co-)Implication]{lemma}{lemImplAsResidual}
    \label{lem:impl_as_residual}
Let $(\hdom,\hleq)$ be a bounded lattice. Then, the following statements are equivalent:
\begin{enumerate}
    \item There exists an implication $\impl$ such that $(\hdom,\hleq,\hand,\hor,\impl)$ is a Heyting algebra.
    \item For every $a,b\in \hdom$, the %
    supremum $\hbigor\{x\in \hdom \mid a\hand x \hleq b\}$ is well-defined.
\end{enumerate}
Then, we derive that the Heyting implication is $a \impl b = \hbigor\{x\in \hdom \mid a\hand x \hleq b\}$ for all $a, b \in \hdom$.

\end{restatable}
Dually, the Heyting co-implication can be characterized as the least element satisfying $b \hleq a \hor x$. These extrema exist since the greatest (least) element of a subset of a partially ordered set is always its supremum (infimum).

\subsection{Totally Ordered Bounded Lattices}

\HeyLo from~\cite{SchroerBKKM23} extends the probability module $\modprob$ by the \emph{Gödel implications} to obtain a bi-Heyting algebra. We generalize those implications for any totally ordered bounded lattice $(\hdom, \hleq)$, which yields a bi-Heyting algebra with the (co-)implications $\totimpl$ and $\totcoimpl$.

\Cref{lem:impl_as_residual} provides a constructive characterization of the implications and, since $\hdom$ is totally ordered and bounded, we derive that the greatest element $x \in \hdom$ such that $a \hand x \hleq b$ holds is either $\top$, if $a \hleq b$, or $b$, otherwise.
A dual construction works for the co-implication.

\begin{restatable}{theorem}{thmTotimplTotcoimplBoundedLattice}
    \label{thm:totimpl_totcoimpl_bounded_lattice}
    For any totally ordered bounded lattice $(\hdom, \hleq)$, we obtain a bi-Heyting algebra $\hsym = (\hdom, \hleq, \hand, \hor, \totimpl, \totcoimpl)$, with the (co-)implications
    \begin{align*}
        a \totimpl b = \left. \begin{cases} \top, &\text{ if } a \hleq b, \\
            b, &\text{ otherwise} \end{cases} \right\}
        \quad \text{and} \quad
        a \totcoimpl b = \left. \begin{cases} \bot, &\text{ if } b \hleq a, \\
            b, &\text{ otherwise} \end{cases} \right\}.
    \end{align*}
\end{restatable}

The extended naturals module $\modnats$, the tropical semiring $\modtrop$, and the probability module $\modprob$ are totally ordered bounded lattices, hence we use the previous implications in this work.

\begin{example}
    For the probability module $\modprob$, this yields the Gödel implication $a \totimpl b = \infty$, if $a \leq b$, and $a \totimpl b = b$, otherwise.
    This yields, e.g., $4 \totimpl 5 = \infty$ and $5 \totimpl 4 = 4$.
    The natural order of the tropical semiring $\modtrop$ is reversed, hence we obtain $a \totimpl b = 0$ if $a \geq b$ and $a \totimpl b = b$, otherwise.
    Then, we have $4 \totimpl 5 = 5$ and $5 \totimpl 4 = 0$.
\end{example}

\subsection{Boolean Algebras}

A second source of Heyting algebras are Boolean algebras. Unlike general Heyting algebras, where negation is only a pseudo-complement, Boolean algebras require that for every element $a$ an additive inverse $\neg a$ exists, satisfying $a \hand \neg a = \bot$ and $a \hor \neg a = \top$.

By the characterization in \Cref{lem:impl_as_residual}, $a \matimpl b$ is the greatest element $x \in \hdom$ such that $a \hand x \hleq b$.
The candidate $\neg a \hor b$ satisfies $a \hand (\neg a \hor b) = (a \hand \neg a) \hor (a \hand b) = a \hand b \hleq b$.
This motivates the material implication $a \matimpl b = \neg a \hor b$ from classical logic.

\begin{restatable}{theorem}{thmMatimplMatcoimplBoundedLattice}
    \label{thm:matimpl_matcoimpl_bounded_lattice}
    For any Boolean algebra $(\hdom, \hleq, \neg)$ with $\hand$ and $\hor$ defined, we obtain a bi-Heyting algebra $\hsym = (\hdom, \hleq, \hand, \hor, \matimpl, \matcoimpl)$ with
    $ a \matimpl b = \neg a \hor b
        \text{ and }
        a \matcoimpl b = \neg a \hand b. $
\end{restatable}

\begin{example}
    Instances of Boolean algebras are the Boolean module $\modbool$\footnote{The Boolean module $\modbool$ is both a totally ordered lattice and a Boolean algebra, and both \Cref{thm:totimpl_totcoimpl_bounded_lattice,thm:matimpl_matcoimpl_bounded_lattice} yield the same (co-)implications.} and the formal languages module $\modlang$ for any alphabet $\Gamma$.
    For the Boolean module, the material implication is the standard Boolean implication, e.g., $(\true \matimpl \false) = \false$ and $(\false \matimpl \true) = \true$.
    For the formal languages module, the material implication is the set-theoretic implication, e.g., $(\Set{a, b} \matimpl \Set{b, c}) = \Set{b, c}$ and $(\Set{a, b} \matimpl \Set{a, b, c}) = \Set{a, b, c}$.
\end{example}

\subsection{The Why Module}

Not all monoid-modules fall into the previous two categories.
One example is the \emph{Why module} $\modwhy$, which tracks \emph{provenance} of computations, i.e., why a certain result was obtained.

Let $A$ be a finite set of variables and $\pDNF = \Set{\bigvee_{j=1}^{n}\bigwedge_{X\in I_j} a
\mid n \ge 0,\; I_j \subseteq A}_{/{\equiv}}$ the set of positive formulae in DNF module logical equivalence.
We order $\pDNF$ by semantic entailment, i.e., $P \hleq Q$ if and only if for all $\rho: A \to \{0,1\}$ it is ($\rho \models P$ implies $\rho \models Q$).
\[
    \textsc{Why:}\qquad
    \monwhy = (\pDNF, \wedge, 1)
    \quad\text{and}\quad
    \modwhy = (\pDNF, \vee, 0, \wedge).
\]
Since $\pDNF$ is defined modulo logical equivalence over the finite set $A$, it is itself finite and forms a bounded distributive lattice with meet $\wedge$, join $\vee$, bottom $0$, and top $1$.
It is neither totally ordered, nor a Boolean algebra, since positive formulae are not closed under negation.
For $P,Q \in \pDNF$,
\[
    P \impl_{\pDNF} Q
    =
    \bigvee \Set{R \in \pDNF \mid P \wedge R \hleq Q}.
\]
Since $\pDNF$ is finite, this join is well-defined and computable; the order check is entailment over finitely many valuations.
The coimplication $\coimpl_{\pDNF}$ is defined dually as the residual of $\vee$.

\begin{restatable}{theorem}{thmWhySemiringBiHeyting}
    \label{thm:why_semiring_bi_heyting}
    For every finite set of variables $A$, the carrier $\pDNF$ of the Why monoid-module $\modwhy$ forms a bi-Heyting algebra
    \(
        (\pDNF, \hleq, \wedge, \vee, \impl_{\pDNF}, \coimpl_{\pDNF}).
    \)
\end{restatable}

Indeed, $P \impl_{\pDNF} Q$ is by construction the greatest $R$ with $P \wedge R \hleq Q$; hence it is the Heyting implication by \Cref{lem:impl_as_residual}, and the co-Heyting implication follows dually.

\section{Verification Infrastructure for Weighted Programming}

In this section, we specify the \emph{quantitative assertion language} \wHeyLo that is parameterized over $\omega$-bicontinuous monoid-modules that are bounded lattices and admit Heyting (co-)implications. From now on, we consider only modules $\Module$ over monoids $\Monoid$ that meet these prerequisites.

Furthermore, we define the \emph{intermediate verification language} \wHeyVL which allows us to encode $\wGCL$ programs, assumptions, and proof rules. Its semantics are \wHeyLo formulae which are computable via the verification preweighting transformer $\symVc$.

\begin{definition}[Quantitative Assertions]
    The set of \wHeyLo formulae is defined by the grammar in \Cref{fig:wheylo_semantics}, where $\aexpr$ is an arithmetic expression \ArithExp, $\bexpr$ is a Boolean expression \BExp, and $\mexpr$ is a monoid expression \MonExp. The semantics $\sem{\hhla}: \Sigma \to \Module$ is defined in \Cref{fig:wheylo_semantics}.
    The notation via Iverson brackets $\iverson{\bexpr}\cdot \hhla$ may be used as syntactic sugar for the \wHeyLo formula $\embed{\neg\bexpr} \coimpl \hhla$.
\end{definition}

\begin{table}
    \centering
    \renewcommand{\arraystretch}{1.35}
    \setlength{\tabcolsep}{4pt}
    \begin{tabular}
        {
            >{\raggedright\arraybackslash}p{0.09\linewidth}
            !{\color{Platinum}\vrule width 1.5pt}
            >{\raggedright\arraybackslash}p{0.36\linewidth}
            !{\color{Platinum}\vrule width 2.5pt}
            >{\raggedright\arraybackslash}p{0.09\linewidth}
            !{\color{Platinum}\vrule width 1.5pt}
            >{\raggedright\arraybackslash}p{0.36\linewidth}
        }
            \toprule
            $\hhla$
                & $\sem{\hhla}(\sigma) \in \Module$
                & $\hhla$
                & $\sem{\hhla}(\sigma) \in \Module$
            \\
            \midrule
            \cellcolor{basicBackground!2}$w \in \Module$
                & \cellcolor{basicBackground!2}$w$
                & \cellcolor{basicBackground!2}
                & \cellcolor{basicBackground!2}
            \\
            \cellcolor{basicBackground!2}$\aexpr \cdot \hhla$
                & \cellcolor{basicBackground!2}
                  {$\bigoplus_{i = 1}^{\sem{\aexpr}(\sigma)} \sem{\hhla}(\sigma)$}
                & \cellcolor{basicBackground!2}
                  \multirow{-2}{=}{$\embed{\bexpr}$}
                & \cellcolor{basicBackground!2}
                  \multirow{-2}{=}
                  {$\ifThenElse{\sem{\varphi}(\sigma)}{\top}{\bot}$}
            \\
            \cellcolor{basicBackground!2}$\mexpr \stimes \hhla$
                & \cellcolor{basicBackground!2}
                  {$\sem{\mexpr}(\sigma) \stimes \sem{\hhla}(\sigma)$}
                & \cellcolor{basicBackground!2}$\hhla \splus \hhlb$
                & \cellcolor{basicBackground!2}
                  {$\sem{\hhla}(\sigma) \splus \sem{\hhlb}(\sigma)$}
            \\
            \cellcolor{primalColor!5}$\hhla \hand \hhlb$
                & \cellcolor{primalColor!5}
                  {$\sem{\hhla}(\sigma) \hand \sem{\hhlb}(\sigma)$}
                & \cellcolor{dualColor!5}$\hhla \hor \hhlb$
                & \cellcolor{dualColor!5}
                  {$\sem{\hhla}(\sigma) \hor \sem{\hhlb}(\sigma)$}
            \\
            \cellcolor{primalColor!5}$\hhla \impl \hhlb$
                & \cellcolor{primalColor!5}
                  {$\sem{\hhla}(\sigma) \impl \sem{\hhlb}(\sigma)$}
                & \cellcolor{dualColor!5}$\hhla \coimpl \hhlb$
                & \cellcolor{dualColor!5}
                  {$\sem{\hhla}(\sigma) \coimpl \sem{\hhlb}(\sigma)$}
            \\
            \cellcolor{primalColor!5}$\hbigand_{x}{\hhla}$
                & \cellcolor{primalColor!5}
                  {$\hbigand \left\{ \sem{\hhla}(\sigma[x \mapsto n]) \mid n \in \Nats \right\}$}
                & \cellcolor{dualColor!5}$\hbigor_{x}{\hhla}$
                & \cellcolor{dualColor!5}
                  {$\hbigor \left\{ \sem{\hhla}(\sigma[x \mapsto n]) \mid n \in \Nats \right\}$}
            \\
            \cellcolor{primalColor!5}$\heylovalidate{\hhla}$
                & \cellcolor{primalColor!5}
                  {$\heylovalidate{w} \text{ with } \sem{\hhla}(\sigma) = w$}
                & \cellcolor{dualColor!5}$\heylocovalidate{\hhla}$
                & \cellcolor{dualColor!5}
                  {$\heylocovalidate{w} \text{ with } \sem{\hhla}(\sigma) = w$}
            \\
            \bottomrule
    \end{tabular}
    \caption{Syntax and semantics of \wHeyLo formulae. Cells colored in \textcolor{primalColor!15}{\rule{0.7em}{0.7em}} and \textcolor{dualColor!15}{\rule{0.7em}{0.7em}} mark order-dual formula constructors.}
    \label[table]{fig:wheylo_semantics}
\end{table}

The following theorem ensures that the semantics of \wHeyLo formulae are well-defined. The main challenge here are the supremum and infimum of - seemingly - arbitrary sets, even though the monoid-modules are only $\omega$-bicontinuous. However, the definition of \wHeyLo ensures that $\sem{\hbigor_x \hhla} = \hbigor_{i \in \Nats} w_i \in \Module$ exists for an $\omega$-ascending chain $(w_i)_{i \in \Nats} \subseteq \Module$ (and dually for the infimum).

\begin{restatable}{theorem}{thmWHeyLoWellDefined}
    \label{thm:wHeyLo_well_defined}
    If the module $\moduledef$ over the monoid $\monoiddef$ is $\omega$-bicontinuous and a bi-Heyting algebra, then the semantics of the \wHeyLo formulae are well-defined.
\end{restatable}

See \Cref{app:heyting} for a proof.
Two \wHeyLo formulae $\hhla$ and $\hhlb$ are \emph{equivalent}, denoted $\hhla \equiv \hhlb$, iff $\interpretsimple{\hhla} = \interpretsimple{\hhlb}$. By lifting the natural order of the module pointwise, we obtain $\hhla \hleq \hhlb$ iff $\hhla \impl \hhlb$ is valid and $\hhla \hgeq \hhlb$ iff $\hhla \coimpl \hhlb$ is covalid. The formula $\hhlc \in \wHeyLo$ is valid iff $\hhlc \equiv \top$ and covalid iff $\hhlc \equiv \bot$.

\begin{table}
    \centering
    \renewcommand{\arraystretch}{1.35}
    \setlength{\tabcolsep}{4pt}
    \begin{tabular}
        {
            >{\raggedright\arraybackslash}m{0.17\linewidth}
            !{\color{Platinum}\vrule width 1.5pt}
            >{\raggedright\arraybackslash}m{0.28\linewidth}
            !{\color{Platinum}\vrule width 2.5pt}
            >{\raggedright\arraybackslash}m{0.17\linewidth}
            !{\color{Platinum}\vrule width 1.5pt}
            >{\raggedright\arraybackslash}m{0.28\linewidth}
        }
            \toprule
            $\sstmt$
                & $\vc{\sstmt}(\hhla)$
                & $\sstmt$
                & $\vc{\sstmt}(\hhla)$
            \\
            \midrule
            \cellcolor{basicBackground!2}
              \raisebox{-0.5\baselineskip}[0pt][0pt]{\stmtRasgn{x}{\aexpr}}
                & \cellcolor{basicBackground!2}
                  \raisebox{-0.5\baselineskip}[0pt][0pt]{$\hhla\substBy{x}{\aexpr}$}
                & \cellcolor{basicBackground!2}
                  \begin{tabular}[c]{@{}l@{}l@{}}
                      \(\symWeightedBranching\) & ~\(\{\) \(S_1 ~\}\) \\
                      \(\symElse\)              & ~\(\{\) \(S_2 ~\}\)
                  \end{tabular}
                & \cellcolor{basicBackground!2}
                  \raisebox{-0.5\baselineskip}[0pt][0pt]
                  {$\vc{\sstmts{1}}(\hhla) + \vc{\sstmts{2}}(\hhla)$}
            \\
            \cellcolor{basicBackground!2}\stmtSeq{\sstmts{1}}{\sstmts{2}}
                & \cellcolor{basicBackground!2}$\vc{\sstmts{1}}(\vc{\sstmts{2}}(\hhla))$
                & \cellcolor{basicBackground!2}\stmtWeigh{\mexpr}
                & \cellcolor{basicBackground!2}$\mexpr \otimes \hhla$
            \\
            \cellcolor{primalColor!5}
                \begin{tabular}[c]{@{}l@{}l@{}}
                    \(\symDemonic\) & ~\(\{\) \(S_1 ~\}\) \\
                    \(\symElse\)    & ~\(\{\) \(S_2 ~\}\)
                \end{tabular}
                & \cellcolor{primalColor!5}
                  \raisebox{-0.5\baselineskip}[0pt][0pt]
                  {$\vc{\sstmts{1}}(\hhla)\dotsqcap \vc{\sstmts{2}}(\hhla)$}
                & \cellcolor{dualColor!5}
                  \begin{tabular}[c]{@{}l@{}l@{}}
                      \(\symAngelic\) & ~\(\{\) \(S_1 ~\}\) \\
                      \(\symElse\)    & ~\(\{\) \(S_2 ~\}\)
                  \end{tabular}
                & \cellcolor{dualColor!5}
                  \raisebox{-0.5\baselineskip}[0pt][0pt]
                  {$\vc{\sstmts{1}}(\hhla)\dotsqcup \vc{\sstmts{2}}(\hhla)$}
            \\
            \cellcolor{primalColor!5}\Assert{\hhlb}
                & \cellcolor{primalColor!5}$\hhlb \sqcap \hhla$
                & \cellcolor{dualColor!5}\coAssert{\hhlb}
                & \cellcolor{dualColor!5}$\hhlb \sqcup \hhla$
            \\
            \cellcolor{primalColor!5}\Assume{\hhlb}
                & \cellcolor{primalColor!5}$\hhlb \impl \hhla$
                & \cellcolor{dualColor!5}\coAssume{\hhlb}
                & \cellcolor{dualColor!5}$\hhlb \coimpl \hhla$
            \\
            \cellcolor{primalColor!5}\Havoc{x}
                & \cellcolor{primalColor!5} $\hbigand_{x}{\hhla}$
                & \cellcolor{dualColor!5}\coHavoc{x}
                & \cellcolor{dualColor!5} $\hbigor_{x}{\hhla}$
            \\
            \cellcolor{primalColor!5}\Validate
                & \cellcolor{primalColor!5}$\heylovalidate{\hhla}$
                & \cellcolor{dualColor!5}\coValidate
                & \cellcolor{dualColor!5}$\heylocovalidate{\hhla}$
            \\
            \bottomrule
    \end{tabular}
    \caption{Syntax and semantics of \wHeyVL. Cells colored in \textcolor{primalColor!15}{\rule{0.7em}{0.7em}} represent primal operators needed for underapproximations, while cells in the complementary color \textcolor{dualColor!15}{\rule{0.7em}{0.7em}} represent their dual counterparts for overapproximations.}
    \label[table]{tab:heyvl-semantics-modern}
\end{table}

The \emph{intermediate verification language} \wHeyVL generalizes the \HeyVL{} language presented by \citet{SchroerBKKM23} for probabilistic programs over the probabilistic module $\modprob$.
Intermediate verification languages are not designed to be executed directly, but allow intermediate representations that encode proof obligations for programs in an imperative style.
$\wHeyVL{}$ is a \emph{weighted} language that can be instantiated for various weighted settings.

\emph{\wHeyVL statements} are closely related to \wGCL commands, but extend them with statements that allow the encoding of proof rules (there are several examples in \Cref{sec:proof_rules}). In \Cref{tab:heyvl-semantics-modern}, we list the syntax of \wHeyVL, where the first statements have a one-on-one correspondence to statements in \wGCL. The remaining statements are the extensions that allow intermediate reasoning.

The \emph{assignment} $\stmtRasgn{x}{\aexpr}$ assigns the value of an arithmetic expression $\aexpr\in\ArithExp$ to a variable $x \in \Vars$. \emph{Sequential composition} $\stmtSeq{\sstmts{1}}{\sstmts{2}}$ executes $\sstmts{1}$ followed by $\sstmts{2}$ and \emph{weighted branching} $\stmtWeighted{\sstmts{1}}{\sstmts{2}}$ executed both substatements and combines them via the additive operation $\oplus$ of the module. The \emph{weighting} $\stmtWeigh{\mexpr}$ weighs $\mexpr \in \MonExp$ via the multiplicative operation $\otimes$ of the module.
\emph{Demonic branching} $\stmtDemonic{\sstmts{1}}{\sstmts{2}}$ executes either $\sstmts{1}$ or $\sstmts{2}$, where the choice is made to minimize the verification preweighting.
The \emph{assertion} $\stmtAssert{\hhlb}$ and \emph{assumption} $\stmtAssume{\hhlb}$ generalize classical assertions and assumptions, where $\hhlb$ is a \wHeyLo formula.
The \emph{havoc} statement $\stmtHavoc{x}$ assigns an arbitrary value to variable $x$ and the \emph{validation} statement $\stmtValidate$ turn quantitative weightings into qualitative ones (c.f. \Cref{tab:heyvl-semantics-modern}).

The previous statements are generally associated with reasoning about lower bounds of the verification preweighting. Dually, there are the \emph{angelic branching} $\stmtAngelic{\sstmts{1}}{\sstmts{2}}$, \emph{coassertion} $\coAssert{\hhlb}$, \emph{coassumption} $\coAssume{\hhlb}$, \emph{cohavoc} $\coHavoc{x}$, and \emph{covalidation} $\coValidate$ statements, which are used for over-approximations.

We specify the semantics via the recursively defined \emph{verification preweighting transformer}
$$\vc{\sstmt}: \wHeyLo \to \wHeyLo \text{ for all } C \in \wHeyVL$$
that is again related to the weakest (liberal) preweighting transformers $\symWp$ and $\symWlp$, see \Cref{def:wp_wlp}.
Let $\hhla \in \wHeyLo$ be a post. For the assignment statement $\stmtAsgn{x}{\aexpr}$, the sequential composition $\stmtSeq{\sstmts{1}}{\sstmts{2}}$, the weighted branching $\stmtWeighted{\sstmts{1}}{\sstmts{2}}$, and the weighting $\stmtWeigh{\aexpr}$, the semantics is defined just like the $\symWp$/$\symWlp$ semantics in \wGCL (see \Cref{sec:weighted_programming}).

The non-deterministic branchings $\stmtDemonic{\sstmts{1}}{\sstmts{2}}$ and $\stmtAngelic{\sstmts{1}}{\sstmts{2}}$ take the infimum $\sqcap$ and supremum $\sqcup$ of the semantics of the branches, respectively. The assertion $\stmtAssert{\hhlb}$ takes the infimum $\sqcap$ of the assertion $\hhlb$ and the post $\hhla$, while the assumption statement $\stmtAssume{\hhlb}$ is defined by the implication $\impl$. Havoc $\stmtHavoc{x}$ takes the infimum $\sqcap$ over all possible values of $x$ and the validation statement $\stmtValidate$ applies the $\heylovalidate{\cdot}$ operator to the post $\hhla$.
Dually, the semantics of the angelic branching $\stmtAngelic{\sstmts{1}}{\sstmts{2}}$, $\coAssert{\hhlb}$, $\coAssume{\hhlb}$, $\coHavoc{x}$, and $\coValidate$ statements are defined by taking the supremum $\sqcup$, the coimplication $\coimpl$, the supremum $\sqcup$ over all possible values of $x$, and the $\heylocovalidate{\cdot}$ operator, respectively. Note that the classical if-else branching corresponds to $ \stmtWeighted{\stmtAssume{\embed{b}}\symSemi \sstmts{1}}{\stmtAssume{\embed{\neg b}}\symSemi \sstmts{2}}$.

\emph{$\wHeyVL{}$ programs} are collections of \emph{procedures} and \emph{coprocedures}, which are again dual constructs. Its syntax is schematically shown in \Cref{fig:wheyvl-proc} and consists of
\begin{itemize}
    \item its \emph{signature} specifying the weighted instance $\semiringPrefix{rig}$ (e.g., \rigTrop, \rigBool, or \rigLang), the keyword $\semiringPrefix{proc}$ or $\semiringPrefix{coproc}$, the name $P$, and typed input and output parameters $\multi{\varin: \typevar}$ and $\multi{\varout: \typevar}$,
    \item its \emph{specification} via a preweighting $\hhla \in \wHeyLo$ and a postweighting $\hhlb \in \wHeyLo$, and
    \item its body $C \in \wHeyVL$ that is a \wHeyVL {statement}.
\end{itemize}
Input parameters are read-only and can only be used in the pre $\hhla$, post $\hhlb$, and body $\sstmt$, while output parameters may be written to, but cannot be used in the pre $\hhla$.

\begin{figure}[t]
    \begin{minipage}{0.45\textwidth}
        \centering
        \colorbox{primalColor!15}{\textbf{Procedure} ($\hhla \hleq \vc{\sstmt}(\hhlb)$)}
        \begin{flalign*}
             & \weightedProc{\semiringPrefix{rig}}{\procname}{\multi{\varin: \typevar}}{\multi{\varout: \typevar}} \\
             & \qquad\Requires{\hhla}                                                             \\
             & \qquad\Ensures{\hhlb}                                                              \\
             & \blockStart~\sstmt~\blockEnd
        \end{flalign*}
    \end{minipage}%
    \hfill%
    \begin{minipage}{0.45\textwidth}
        \centering
        \colorbox{dualColor!15}{\textbf{Coprocedure} ($\hhla \hgeq \vc{\sstmt}(\hhlb)$)}
        \begin{flalign*}
             & \weightedCoproc{\semiringPrefix{rig}}{\procname}{\multi{\varin: \typevar}}{\multi{\varout: \typevar}} \\
             & \qquad\Requires{\hhla}                                                               \\
             & \qquad\Ensures{\hhlb}                                                                \\
             & \blockStart~\sstmt~\blockEnd
        \end{flalign*}
    \end{minipage}

    \caption{\wHeyVL procedures (left) verify lower bounds and coprocedures (right) verify upper bounds.}
    \label{fig:wheyvl-proc}
\end{figure}

A procedure verifies if the pre $\hhla$ is a \emph{lower bound} of the verification preweighting of $\sstmt$ with respect to the post $\hhlb$. Dually, a coprocedure verifies upper bounds.
\begin{definition}
    Consider the (co)procedures from \Cref{fig:wheyvl-proc}. A procedure $P$ (with keyword $\symProc$) \emph{verifies} iff $\hhla \hleq \vc{C}(\hhlb)$.
    A coprocedure $P$ (with keyword $\symcoProc$) verifies iff $\hhla \hgeq \vc{C}(\hhlb)$.
\end{definition}
By the deduction theorem (\Cref{thm:deduction_theorem}), the verification conditions for procedures and coprocedures can be internalized in \wHeyLo as $\top \equiv \hhla \impl \vc{\sstmt}(\hhlb)$ and $\bot \equiv \hhla \dotcoimpl \vc{\sstmt}(\hhlb)$.

\section{Proof Rules \& Encodings in \wHeyVL}
\label{sec:proof_rules}

In this section, we bridge the gap between the theoretical concept \emph{weighted programming} and our verification infrastructure from the previous section. The translation of a loop-free program $\wwcom\in \wGCL$ into the verification language $\wHeyVL$ is straight-forward and depicted in \Cref{tab:encoding_loop_free_spec}. For this translation, the following theorem shows that the correspondence between the weakest (liberal) preweightings and the verification preweightings is exact.

\begin{restatable}[Loop-Free Encodings are Exact]{theorem}{thmWpEncodingExactness}
    \label{thm:wp-encoding-exactness}
    For a loop-free program $\wwcom \in \wGCL$ and postweighting \(\hhla \in \wHeyLo\), we have, for all $\sigma \in \Sigma$, that
    $
        \nosem{\vc{\procenc{\wwcom}}(\hhla)}(\State) = \wp{\wwcom}(\nosem{\hhla})(\State).
    $
\end{restatable}

\begin{figure}
    \centering

    \begin{minipage}[t]{0.65\textwidth}
        \centering
        \renewcommand{\arraystretch}{0.95}
        \setlength{\tabcolsep}{4pt}
        \begin{align*}
        &\vcenter{\hbox{%
        \begin{tabular}
            {
                >{\raggedright\arraybackslash}m{0.39\linewidth}
                !{\color{Platinum}\vrule width 1.5pt}
                >{\raggedright\arraybackslash}m{0.53\linewidth}
            }
            \toprule
            $\wwcom$
                & \(\procenc{\wwcom}\) \hfill \(\hint{= \coprocenc{\wwcom}}\)
            \\
            \midrule
            \cellcolor{basicBackground!2}$\ASSIGN{x}{E}$
                & \cellcolor{basicBackground!2}$\stmtRasgn{x}{E}$
            \\
            \cellcolor{basicBackground!2}$\WEIGH{a}$
                & \cellcolor{basicBackground!2}$\stmtWeigh{a}$
            \\
            \cellcolor{basicBackground!2}$\wwcomi \fatsemi~ \wwcomii$
                & \cellcolor{basicBackground!2}
                  $\stmtSeq{\procenc{\wwcomi}}{\procenc{\wwcomii}}$
            \\
            \cellcolor{basicBackground!2}$\BRANCH{\wwcomi}{\wwcomii}$
                & \cellcolor{basicBackground!2}
                  $\stmtWeighted{\procenc{\wwcomi}}{\procenc{\wwcomii}}$
            \\
            \cellcolor{basicBackground!2}
              \begin{tabular}[c]{@{}l@{}}
                  \(\stmtIfStart{b}\wwcomi \}\) \\
                  \(\quad \stmtElseStart\wwcomii \}\)
              \end{tabular}
                & \cellcolor{basicBackground!2}
                  \begin{tabular}[c]{@{}l@{}}
                      \(\stmtDemonicStart
                        {\stmtSeq{\stmtAssume{\embed{b}}}{\procenc{\wwcomi}}} \}\) \\
                      \(\quad \stmtElseStart
                        {\stmtSeq{\stmtAssume{\embed{\neg b}}}{\procenc{\wwcomii}}} \}\)
                  \end{tabular}
            \\
            \cellcolor{basicBackground!2}$\WHILEDO{b}{\wwcomi}$
                & \cellcolor{basicBackground!2}$\explain{without translation}$
            \\
            \bottomrule
        \end{tabular}
        }}
        \end{align*}
    \end{minipage}
    \hfill
    \begin{minipage}[t]{0.3\textwidth}
        {
        \begin{align*}
            & \specEnc{\stmtSpec{\beta}{\hhlc}{\hhlb}} \\
            \midrule
            & \stmtAssert{\hhlc}\symSemi \inlineExplain{filter for $\hhlc$} \\
            & \stmtHavoc{\varset}\symSemi \inlineExplain{quantify modified states} \\
            & \stmtValidate\symSemi \inlineExplain{booleanize the check} \\
            & \stmtAssume{\hhlb} \inlineExplain{entailment check}
        \end{align*}}
    \end{minipage}

    \caption{Translation of loop-free \wGCL statements into \wHeyVL and encoding of the specification statement for states $\varset \subseteq \Vars$ and weightings $\hhlc, \hhlb \in \wHeyLo$.}
    \label{tab:encoding_loop_free_spec}
\end{figure}

For loop-free \wGCL constructs, $\symWp$ and $\symWlp$ coincide. However, there exists no direct translation for programs with loops.
We consider common proof rules, like e.g. Park induction \cite{BatzGKKW22} and $\kappa$-induction \cite{k_induction}.
Furthermore, we modify Park induction and $\kappa$-induction to corresponding \emph{local} proof rules that disregard variables that are not altered during the execution of the loop.
With those modified proof rules, we can use more invariants to show lower (resp. upper) bounds.
Then, we give the encoding of the premises into \wHeyVL-programs, which allows the calculation of lower (resp. upper) bounds of the weakest (liberal) preweightings.
Every proof rule in this section has a version that verifies a lower bound and uses procedures.
Every rule also has a dual version that verifies an upper bound, uses coprocedures, and dual statements, which we omit for brevity.
We write $\hhlb$ instead of $\sem{\hhlb}$ and $\vc{\cdot}(\cdot)(\sigma)$ instead of $\sem{\vc{\cdot}(\cdot)}(\sigma)$ for all $\hhlb \in \wHeyLo$ and $\sigma \in \Sigma$.

\subsection{Specification Statements}
\label{sec:spec_statement}

Before considering proof rules for loops, we look at the encoding of \emph{specification statements} \cite{specificationstatement} in \wHeyVL. Intuitively, such a statement represents a block of code by specifying a preweighting, a postweighting, and which variables are changed during the execution.
Therefore, they enable \emph{implementation abstraction}; instead of a concrete implementation, specification statements describe a program by their effects. This is key to \emph{modular verification}, where we verify components of a program separately and then analyse the derived specification statements to verify the whole program.
Furthermore, they give rise to \emph{program refinement}, where we start with specification statements and incrementally refine them into concrete implementations.

\begin{definition}
    Let $\varset \subseteq \Vars$, $\hhlc, \hhlb \in \wHeyLo$.
    The (under-approximating) specification statement $\stmtSpec{\varset}{\hhlc}{\hhlb}$ is a placeholder for the set of all \wHeyVL programs $\sstmt$ that modify only variables in $\beta$ and satisfy $\hhlc \hleq \vc{\sstmt}(\hhlb)$.
\end{definition}

To achieve a desired post $\hhla\in\wHeyLo$ using the specification $\stmtSpec{\varset}{\hhlc}{\hhlb}$, two conditions must be met. First, the pre $\hhlc$ must hold in the initial state. Second, the specification's post $\hhlb$ must be strong enough to guarantee $\hhla$ after arbitrarily changing the variables in $\varset$ while leaving all other variables untouched. The set $\varset$ contains all \emph{local variables} and, starting in an initial state $\sigma\in \Sigma$, the following definition gives us an over-approximation of all states reachable by modifying the local variables.

\begin{definition} [$\varset$-Reachable States]
    \label[definition]{def:heyvl-varset-reachable-states}
    For any state $\sigma\in\States$ and any $\varset\subseteq\Vars$, we define
    \begin{align*}
        \varsetreach{\sigma}=\{\sigma'\in\Sigma \mid \forall v\in\Vars \setminus \varset:~ \sigma(v)=\sigma'(v) \}.
    \end{align*}
\end{definition}

Now, the previous intuition translates directly into a \wHeyVL encoding. We first $\symAssert$ that the pre $\hhlc$ is a lower bound. Then, arbitrary changes to the local variables $\varset$ are mimicked by $\symHavoc$ and $\symValidate$ booleanizes the result. Finally, we can $\symAssume$ that the postcondition $\hhlb$ holds.

\begin{figure}
    \centering

    \begin{minipage}[t]{0.4\textwidth}
        \begin{align*}
            & \coSpecEnc{\coSpec{\beta}{\hhlc}{\hhlb}} \\
            \midrule
            & \coAssert{\hhlc}\symSemi
            \coHavoc{\varset}\symSemi
            \coValidate\symSemi
            \coAssume{\hhlb} \\[1em]
            & \charEnc{\wwcom, \hinv} \\
            \midrule
            & \stmtIfStart{b} \\
            & \quad \procenc{\wwcom} \symSemi \inlineExplain{encoding of loop body} \\
            & \quad \stmtAssert{\hinv} \symSemi \inlineExplain{filter for invariant} \\
            & \quad \stmtAssume{\embed{\false}} \symSemi \inlineExplain{establish truth} \\
            & \}~\stmtElseStart{} \}
        \end{align*}
    \end{minipage}
    \hfill
    \begin{minipage}[t]{0.3\textwidth}
        \begin{align*}
            & \parkEnc{\wwcom, \beta, \hinv} = \\
            \midrule
            & \tikzmark{spec1} \stmtAssert{\hinv}\symSemi
            \stmtHavoc{\beta}\symSemi \\
            & \tikzmark{spec2} \stmtValidate\symSemi
            \stmtAssume{\hinv}\symSemi \\[0.5em]
            & \tikzmark{char1} \stmtIfStart{b} \\
            & \quad \procenc{\wwcom}\symSemi \\
            & \quad \stmtAssert{\hinv}\symSemi \\
            & \quad \stmtAssume{\embed{\false}}\symSemi \\
            & \tikzmark{char2} \}~\stmtElseStart \}
        \end{align*}
    \end{minipage}

    \begin{tikzpicture}[overlay, remember picture]
        \draw[thick, decorate, decoration={brace, amplitude=5pt}, color=hintcolor] ([xshift=-1.5em]pic cs:spec2) -- node[midway, rotate=90, yshift=2em, xshift=0pt] () {\small \shortstack{$\specEncShort$ \\ $(\stmtSpec{\beta}{\hinv}{\hinv})$}} ([xshift=-1.5em, yshift=5pt]pic cs:spec1);

        \draw[thick, decorate, decoration={brace, amplitude=5pt}, color=hintcolor] ([shift={(-1.5em,0)}]pic cs:char2) -- node[midway, rotate=90, yshift=1.5em, xshift=0cm] () {\small $\charEnc{\wwcom, \hinv}$} ([shift={(-1.5em,5pt)}]pic cs:char1);
    \end{tikzpicture}

    \caption{Encoding of the specification statement, the characteristic function, and Park induction for variables $\varset \subseteq \Vars$, weightings $\hhlc, \hhlb \in \wHeyLo$, loop $\wwcom = \WHILEDO{b}{\wwcoma} \in \wGCL$, and invariant $\hinv \in \wHeyLo$.}
    \label{tab:encoding_spec_char_park}
\end{figure}

\begin{restatable}[Semantics]{lemma}{lemHeyvlSpecificationEncoding}
    \label[lemma]{lem:heyvl-specification-encoding}
    For $\varset \subseteq \Vars$ and $\hhlc, \hhlb, \hhla \in \wHeyLo$, and $\sigma \in \Sigma$, the \emph{specification statement encoding} $\specEnc{\stmtSpec{\varset}{\hhlc}{\hhlb}}$ given in \Cref{tab:encoding_loop_free_spec} satisfies
    \begin{align*}
        \nosem{\vc{\specEnc{\stmtSpec{\varset}{\hhlc}{\hhlb}}}(\hhla)}(\sigma)=
        \begin{cases}
            \nosem{\hhlc}(\sigma), & \text{if } ~\nosem{\hhlb}(\sigma') \hleq \nosem{\hhla}(\sigma')~\forall \sigma' \in \varsetreach{\sigma} \\
            \bot, & \text{otherwise}
        \end{cases}.
    \end{align*}
\end{restatable}

By applying $\symVc$ sequentially, we derive
\[
    \vc{\stmtAssert{\hhlc}}(\vc{\stmtHavoc{\beta}}(\vc{\stmtValidate}(\vc{\stmtAssume{\hhlb}}(\hhla))))
    = \hhlc \hand \hbigand_{\beta} \heylovalidate{\hhlb \impl \hhla},
\]
where $\symAssume$ evaluates to $\top$, if $\nosem{\hhlb} \hleq \nosem{\hhla}$, and $\nosem{\hhlb}$, otherwise. Due to $\symValidate$, this yields $\top$, if $\nosem{\hhlb} \hleq \nosem{\hhla}$, and $\bot$, otherwise. The statement $\symHavoc$ checks this condition for all states in $\varsetreach{\State}$ and $\symAssert$ filters for $\hhlc$ (see \Cref{sec:appendix-specification-statements}).
The soundness theorem below states that the encoding $\specEnc{\stmtSpec{\varset}{\hhlc}{\hhlb}}$ under-approximates the semantics of any statement that satisfies the respective specification.

\begin{restatable}[Soundness]{theorem}{thmHeyvlSpecificationSoundness}
    \label[theorem]{thm:heyvl-specification-soundness}
    Let $\varset \subseteq \Vars$, $\hhlc, \hhlb \in \wHeyLo$, and $\sstmt \in \wHeyVL$ be a statement that only modifies variables in $\varset$ and satisfies $\hhlc \hleq \vc{\sstmt}(\hhlb)$.
    Then, for all $\hhla \in \wHeyLo$, we have $\vc{\specEnc{\stmtSpec{\varset}{\hhlc}{\hhlb}}}(\hhla) \hleq \vc{\sstmt}(\hhla)$.
\end{restatable}

The proof follows directly from \Cref{lem:heyvl-specification-encoding} and the monotonicity of the $\wHeyVL$ semantics.

We specify \emph{over-approximating} specification statements $\coSpec{\beta}{\hhlc}{\hhlb}$, where we consider upper bounds and $\vc{\sstmt}(\hhlb) \hleq \hhlc$ holds. The corresponding encoding $\coSpecEnc{\coSpec{\beta}{\hhlc}{\hhlb}}$ is obtained by substituting all statements in $\specEnc{\cdot}$ by their dual \texttt{co}-statements, see \Cref{tab:encoding_spec_char_park}. The semantics, the soundness theorem, and the corresponding proofs are exactly dual.

\subsection{Finite Loop Unrolling}
\label{sec:loop_unrolling}

While the encoding of loop-free programs is straightforward, we do not have an exact encoding of programs with loops. We can, however, approximate loops by proof rules.
This leads us to the first proof rule \emph{finite loop unrolling}. We can literally unroll a loop $\wwcom = \WHILEDO{b}{\wwcoma} \in \wGCL$ once, twice, or more times and approximate it via
$$ \IF{b}{~ \wwcoma} ~\}, \quad \IF{b}{~ \wwcoma \semcol \IF{b}{~ \wwcoma}} ~\} ~\}, \quad \ldots $$
For an $n$-unrolling $\wwcomtimes{n}$ of the loop $\wwcom$ and postweighting $\hhla \in \wHeyLo$, the weakest preweighting is $\wp{\wwcomtimes{n}}(\hhla) = \wpcharhhla^n(\bot)$. We know from \Cref{sec:weighted_programming} and Kleene's fixed point theorem \cite{BatzGKKW22} that the weakest preweightings is
$
    \wp{\wwcom}(\hhla) = \lfp X \wpcharhhla(X) = \hbigor_{n \in \Nats} \wpcharhhla^n (\bot)
$
and hence $\wp{\wwcomtimes{n}}(\hhla) \hleq \wp{\wwcom}(\hhla)$ for all $n \in \Nats$. The more times we unroll the loop, the better the approximation becomes.

\begin{restatable}[Finite Loop Unrolling]{lemma}{lemFiniteLoopUnrolling}
    \label[lemma]{lem:loop_unrolling}
    For a loop $\wwcom = \WHILEDO{b}{\wwcoma} \in \wGCL$, weightings $\wla, \wlc \in \weightings$, and the corresponding characteristic function $\wpchar$, we have, for any $n \in \Nats$, that
    $$
        \nosem{\wlc} \hleq \wpchar^n(\bot)
        \quad\text{implies}\quad
        \nosem{\wlc} \hleq \wp{\wwcom}(\nosem{\wla}).
    $$
\end{restatable}

The proof of the previous lemma applies Kleene's fixed point theorem and uses the monotonicity of the join $\hor$. For the encoding and a postweighting $\hhla \in \wHeyLo$, we translate the characteristic function $\wpcharhhla(\hinv) = \iverson{b}\cdot \wp{\wwcoma}(\hinv) \oplus \iverson{\neg b}\cdot \hhla$ into the \wHeyVL program $\charEnc{\wwcom, \hinv}$ in \Cref{tab:encoding_spec_char_park}, where $\hinv \in \wHeyLo$. If $b$ holds, this program encodes the loop body $\procenc{\wwcoma}$, $\symAssert$s that $\hinv$ holds, and uses $\stmtAssume{\embed{\false}}$ afterwards to discard the current assertion. If $b$ does not hold, then the postweighting $\hhla$ holds.

\begin{restatable}[Semantics]{lemma}{lemHeyvlCharacteristicEncodingSemantics}
    \label[lemma]{lem:heyvl-characteristic-encoding-semantics}
    For a loop $\wwcom = \WHILEDO{b}{\wwcoma} \in \wGCL$, weightings $\hhla, \hinv \in \wHeyLo$, an encoding $\procenc{\wwcoma}$ with $\vc{\procenc{\wwcoma}}(\hhlb) \hleq \wp{\wwcoma}(\hhlb)$ for all $\hhlb \in \wHeyLo$, and $\sigma\in\Sigma$, the encoding of the characteristic function $\wpcharhhla(\hinv)$ given in \Cref{tab:encoding_spec_char_park} satisfies
    $$
        \nosem{\vc{\charEnc{\wwcom, \hinv}}(\hhla)}(\sigma) =
        \begin{cases}
            \nosem{\vc{\procenc{\wwcoma}}(\hinv)}(\sigma), & \text{if } \sem{\bexpr}(\sigma) \\
            \nosem{\hhla}(\sigma), & \text{otherwise}
        \end{cases}.
    $$
\end{restatable}

Hence, the translation $\charEncShort$ encodes exactly the characteristic function $\Phi$ and, via applying the program $n \in \Nats$ times, we use \Cref{lem:loop_unrolling} to derive the following property.

\begin{restatable}[Soundness]{lemma}{lemCharacteristicEncodingSoundness}
    \label[lemma]{lem:characteristic_function_soundness}
    For a loop $\wwcom = \WHILEDO{b}{\wwcoma} \in \wGCL$, weightings $\hhla, \hinv \in \wHeyLo$, $n \in \Nats$, and an encoding $\procenc{\wwcoma}$ with $\vc{\procenc{\wwcoma}}(\hhlb) \hleq \wp{\wwcoma}(\hhlb)$ for all $\hhlb \in \wHeyLo$, the encoding $\charEncShort$ is a sound under-approximation, thus
    $
        \nosem{\vc{\charEnc{\wwcom, \bot}^n}(\hhla)} \hhleq \wp{\wwcom}(\nosem{\hhla}).
    $
\end{restatable}

We use this proof rule by checking $\hhlc \hleq \vc{\charEnc{\wwcom, \bot}^n}(\hhla)$ for some $n \in \Nats$ and $\hhlc, \hhla \in \wHeyLo$. If this holds, we can conclude that $\hhlc \hleq \wp{\wwcom}(\hhla)$ is a lower bound.

The dual proof rule exists for the weakest \emph{liberal} preweightings of loops, where we obtain an upper bound by applying the characteristic function repeatedly. In this case, the encoding $\coCharEnc{\wwcom, \hinv}$ contains the dual \texttt{co}-statements (note that the dual of $\stmtAssume{\embed{\false}}$ is $\coAssume{\embed{\true}}$) and everything follows similarly.

\subsection{Local Park Induction}
\label{sec:local_park_induction}

Park induction allows reasoning about an unbounded number of steps.
Also, loop unrolling gives us a lower bound for the weakest preweighting of a loop, while Park induction yields an upper bound (and the dual holds for weakest \emph{liberal} preweightings).
Instead of applying the characteristic function to the least or greatest weighting, we need to \emph{guess} a suitable loop \emph{invariant} to under-approximate the weakest liberal preweighting of the loop (and over-approximate weakest preweightings).

More specifically, we call $\hinv \in \wHeyLo$ a \emph{subinvariant}, if $\nosem{\hinv} \hleq \wlpchar(\nosem{\hinv})$, and a \emph{superinvariant}, if $\wpchar(\nosem{\hinv}) \hleq \hinv$.
\emph{Park induction} states that the invariant is a lower bound on the weakest liberal preweighting of the loop $\wwcom\in \wGCL$ for postweighting $\hhla \in \wHeyLo$, hence $\nosem{\hinv} \hleq \wlp{\wwcoma}(\hhla)$, if $\hinv$ is a subinvariant \cite{BatzGKKW22}. For $\wp{\wwcom}(\hhla)$, the weighting $\hinv$ is an upper bound if it is a superinvariant.

Here, we define \emph{local Park induction} on a restricted state space that depends variables that are modified during the execution of the loop, so-called \emph{local} variables (see \Cref{sec:spec_statement}). This alteration allows us to prove more (correct) invariants than the original Park induction rule and makes the rule more complete.
In the classical setting of Boolean verification, Boogie/BoogiePL-based encodings implement this idea as part of their encodings \cite{leino2007a} and related loop-cutting transformations have been validated operationally \cite{parthasarathy2021formallyvalidatingpracticalverification}.
In the quantitative setting, the tool \Caesar{} also employs this enhanced rule, although its proof was missing.

Suppose that $\statesclosed \subseteq \Sigma$ is a set of states that is \emph{closed under the $\beta$-reachable states}, more formally, $\varsetreach{\sigma} \subseteq \statesclosed$ for all $\sigma \in \statesclosed$. Then, we alter the original proof rule to \emph{local Park induction} by considering $\statesclosed$-restricted weightings $\hhla_{|\statesclosed}$ and restricted inequalities $\bhleq$.

\begin{definition}
    For weightings $\wla, \wlb \in \weightings$ and a set of states $\statesclosed \subseteq \Sigma$ that is closed under $\beta$-reachability, we consider $\statesclosed$-restricted weightings $\wla_{|\statesclosed}, \wlb_{|\statesclosed}: \statesclosed \to \Module$ with the \emph{$\statesclosed$-restricted order}
    $
        \wla_{|\statesclosed} \bhleq \wlb_{|\statesclosed} \text{ if and only if }
        \wla(\sigma) \hleq \wlb(\sigma) ~\forall \sigma \in \statesclosed.
    $
\end{definition}

This restriction keeps all the previous properties, since the set of restricted weightings is again an $\omega$-bicontinuous monoid-module and in particular also $\omega$-bicomplete partially ordered.

\begin{restatable}[Closedness under Reachability]{lemma}{lemClosedness}
    \label[lemma]{lem:closedness_under_reachability}
    For a program $\wwcom \in \wGCL$ with local variables $\varset \subseteq \Vars$ and a set $\statesclosed \subseteq \Sigma$ of states that is closed under $\varset$-reachability, $\wlp{\wwcom}: \weightings \to \weightings$ is closed under $\statesclosed$-restricted weightings: for all $\wla \in \weightings$, there exists $\wlb \in \weightings$ such that
    $\wlp{\wwcom}(\wla_{|\statesclosed}) = \wlb_{|\statesclosed}$.
\end{restatable}

Since the weighting transformer $\betawlpchar$ is $\omega$-cocontinuous, we derive from the previous lemma that its restriction to an set $\statesclosed$ (closed under $\varset$-reachable states) is again $\omega$-bicontinuous. Finally, we can state and prove the local variant of Park induction. Here, we use global weightings $\wla, \iinv \in \weightings$ but consider only the restricted state space $\statesclosed$.

\begin{restatable}[Local Park Induction]{theorem}{thmLocalParkInduction}
    \label[theorem]{thm:local-park-induction}
    For a loop $\wwcom = \WHILEDO{b}{\wwcoma} \in \wGCL$ with local variables $\varset \subseteq \Vars$, a set $\statesclosed \subseteq \Sigma$ of states that is closed under $\varset$-reachability, and weightings $\wla, \iinv \in \weightings$, it is
    $$
        \nosem{\iinv} \bhleq \wlpchar(\nosem{\iinv}) \quad\text{implies}\quad \nosem{\iinv} \bhleq \wlp{\WHILEDO{b}{\wwcoma}}(\nosem{\wla}).
    $$
\end{restatable}

The Ski-rental problem as introduced in \Cref{fig:intro-wgcl-ski-rental} is a simple example where the correctness of the invariant depends on constant variables that are not modified during the execution of the loop. The optimal costs of the program are $\min\{n, y\}$, where $n$ is the number of days and $y$ is the cost of buying skis. However, the original Park induction rule does not suffice to prove the correctness for the invariant $I = \min\{n, y\}$ for postweighting $0 \in \weightings$, since it considers the global state space $\Sigma$: A counterexample is
$ \sigma \in \Sigma \text{ with } \sigma(n) = 4, \sigma(y) = 6, \text{ and } \sigma(c) = 2$, where the optimal costs are $\min\{2\cdot n, y\} = 6$ and the inequality on the left-hand side in \Cref{thm:local-park-induction} does not hold, since
$$ \sigma(I) = 4 \not\hleq \wlpchar(I)(\sigma) = 5 \qquad \hint{\iff  \sigma(I) = 4 \not\geq \wlpchar(I)(\sigma) = 5}. $$
The local proof rule considers only the restricted state space $\Gamma = \{ \sigma \in \Sigma \mid \sigma(c) = 1 \}$. Here, we can show that $\hinv$ is indeed a subinvariant and thus verify the correctness of the program.

The encoding of local Park induction is $\parkEncShort$ in \Cref{tab:encoding_spec_char_park} and \Cref{lem:semantics-local-park-induction-encoding} shows that its semantics satisfies
$$
    \nosem{\vc{\parkEnc{\wwcom, \hinv}}(\hhla)}(\State) \hleq
        \ifThenElseDot{\hinv \hleq_{\varsetreach{\sigma}} \wlpcharhhla(\hinv)}{\nosem{\hinv}(\sigma)}{\bot}
$$

\begin{restatable}[Soundness]{theorem}{thmHeyvlSoundnessLocalParkInductionEncoding}
    \label[theorem]{thm:heyvl-soundness-local-park-induction-encoding}
    For a loop $\wwcom = \WHILEDO{b}{\wwcoma} \in \wGCL$ with local variables $\varset \subseteq \Vars$ and weightings $\hhla, \hinv \in \wHeyLo$, the encoding $\parkEncShort$ is a sound under-approximation, thus
    $$
        \nosem{\vc{\parkEnc{\wwcom, \varset, \hinv}}(\hhla)} \hhleq \wlp{\wwcom}(\nosem{\hhla}).
    $$
\end{restatable}

The dual proof rule exists for weakest preweightings of loops, where we check whether $\hinv \in \wHeyLo$ is a superinvariant and obtain an upper bound by computing the encoding $\coParkEnc{\wwcom, \beta, \hinv}$ that contains the corresponding \texttt{co}-statements and everything follows analogously.
This allows us to prove programs where the correctness of the invariant depends on constant variables that are not modified during the execution.

\subsection{Local \texorpdfstring{$\kappa$}{kappa}-Induction}
\label{sec:k-induction}

\begin{figure}
    \hspace{4em}
    \begin{minipage}[t]{0.38\textwidth}
        \begin{align*}
            & \kappaEnc{\wwcom, \beta, I} \\
            \midrule
            & \tikzmark{spec21} \stmtAssert{\hinv}\symSemi
            \stmtHavoc{\beta}\symSemi \\
            & \tikzmark{spec22} \stmtValidate\symSemi
            \stmtAssume{\hinv}\symSemi \\[1em]
            & \tikzmark{unroll21} \stmtIfStart{b} \\
            & \quad \procenc{\wwcoma} \symSemi \coAssert{\hinv} \symSemi \\
            & \tikzmark{unroll1} \quad \stmtIfStart{b} \\
            & \quad\quad \procenc{\wwcoma} \symSemi \coAssert{\hinv} \symSemi \\
            & \tikzmark{char31} \quad\quad \stmtIfStart{b} \\
            & \quad\quad\quad \procenc{\wwcoma} \symSemi \stmtAssert{\hinv} \symSemi \\
            & \quad\quad\quad \stmtAssume{\embed{\false}} \symSemi \\
            & \tikzmark{char32} \quad\quad\quad \}~\stmtElseStart{} \} \\
            & \tikzmark{unroll2} \quad\quad \}~\stmtElseStart{} \} \\
            & \tikzmark{unroll22} \quad \}~\stmtElseStart{} \}
        \end{align*}

        \begin{tikzpicture}[overlay, remember picture]
           \draw[thick, decorate, decoration={brace, amplitude=5pt}, color=hintcolor] ([shift={(-1em,0)}]pic cs:spec22) -- node[midway, rotate=90, yshift=2em, xshift=0pt] () {\small \shortstack{$\specEncShort$ \\ $(\stmtSpec{\beta}{\hinv}{\hinv})$}}([shift={(-1em,5pt)}]pic cs:spec21);

           \draw[thick, decorate, decoration={brace, amplitude=5pt}, color=hintcolor] ([shift={(-3.5em,0)}]pic cs:unroll22) -- node[midway, rotate=90, yshift=1.5em, xshift=0cm] () {\small $\kappaUnrollEnc{\wwcom, \varset, \hinv}$} ([shift={(-3.5em,5pt)}]pic cs:unroll21);

           \draw[thick, decorate, decoration={brace, amplitude=5pt}, color=hintcolor] ([shift={(-1em,0)}]pic cs:unroll2) -- node[midway, rotate=90, yshift=1.5em, xshift=0cm] () {\small $\unrollEnc{($\kappa$-1)}{\wwcom, \varset, \hinv}$} ([shift={(-1em,5pt)}]pic cs:unroll1);

           \draw[thick, decorate, decoration={brace, amplitude=5pt}, color=hintcolor] ([shift={(1.5em,0)}]pic cs:char32) -- node[midway, rotate=90, yshift=1.5em, xshift=0cm] () {\small $\charEnc{\wwcom, \varset, \hinv}$} ([shift={(1.5em,5pt)}]pic cs:char31);
        \end{tikzpicture}
    \end{minipage}
    \hfill
    \begin{minipage}[t]{0.45\textwidth}
        \begin{align*}
            & \unrollEnc{0}{\wwcom, \varset, \hinv} = \\
            & \charEnc{\wwcom, \varset, \hinv} \\[2em]
            & \text{for } \kk > 0 : \\
            & \unrollEnc{$\kappa$}{\wwcom, \varset, \hinv} = \\
            & \stmtIfStart{b} \\
            & \quad \procenc{\wwcoma} \symSemi \inlineExplain{encoding of loop body} \\
            & \quad \coAssert{\hinv} \symSemi \inlineExplain{admit invariant} \\
            & \quad \unrollEnc{($\kappa$-1)}{\wwcom, \varset, \hinv} \symSemi \\
            & \}~\stmtElseStart{} \}
        \end{align*}
    \end{minipage}
    \hfill

    \caption{Encoding of $\kk$-induction for the loop $\wwcom = \WHILEDO{b}{\wwcoma} \in \wGCL$, variables $\varset \subseteq \Vars$, weightings $\hhlc, \hhlb \in \wHeyLo$, an invariant $\hinv \in \wHeyLo$, and $\kappa = 2$.}
    \label{fig:encoding_kappa_induction}
\end{figure}

The $\kk$-induction proof rule, formalized by \citeauthor{k_induction}~\cite{k_induction} for complete lattices, extends Park induction by relaxing its strict one-step inductivity requirement. We adapt this technique to our setting of $\omega$-bicomplete partial orders and extend it to \emph{local} $\kk$-induction.

Recall that, for a loop $\wwcom$ and postweighting $\hhla \in \wHeyLo$, the weakest liberal preweighing is the greatest fixed point of the characteristic function and due to Kleene's fixed point theorem \cite{BatzGKKW22} it is $\wlp{\wwcom}(\hhla) = \gfp{X}{\wlpcharhhla(X)} = \hbigand_{n\in\Nats} \wlpcharhhla^n(\top)$. Via local Park induction, we can prove a lower bound $\hinv \in \wHeyLo$ and refine it by iteratively applying the characteristic function, hence
$$ \hinv \bhleq \wlpcharhhla(\hinv) \bhleq \wlpcharhhla^2(\hinv) \bhleq \cdots \bhleq \wlp{\wwcom}(\hhla). $$
However, the reverse direction of Park induction does not hold in general.
Even if $\hinv$ is a valid upper bound with $\hinv \bhleq \wlp{\wwcom}(\hhla)$, the weighting $\hinv$ might \emph{not} be \emph{inductive} in the sense that $\hinv \not\bhleq \wlpcharhhla(\hinv)$.
To tackle this, we create an increasing chain with the following operator.

\begin{definition}
    \label[definition]{def:kappa_operator}
    For a loop $\wwcom = \WHILEDO{b}{\wwcoma} \in \wGCL$ and weightings $\wla, \iinv \in \weightings$, the \emph{$\kk$-induction operator} is defined as $\kindsym: \weightings \to \weightings,~ X \mapsto \iinv \hor \wlpchar(X)$.
\end{definition}

Since we extend $\kk$-induction to \emph{local} $\kk$-induction, we consider a set of states $\statesclosed \subseteq \Sigma$ that is closed under $\varset$-reachability, where $\varset$ are the local variables of the program. Then, $\statesclosed$-restricted weightings $\hhla_{|\statesclosed}, \hhlb_{|\statesclosed}: \statesclosed \to \Module$ consider a smaller state space and allow for a $\statesclosed$-restricted order $\bhleq$ (compare with \Cref{sec:local_park_induction}). We have shown in \Cref{lem:closedness_under_reachability} that the characteristic function $\wlpcharhhla$ is closed under $\statesclosed$-restricted weightings and we show the same for the $\kk$-induction operator in \Cref{lem:closedness_under_reachability_k-induction}.

Now, even if the weighting $\hinv \in \wHeyLo$ is not inductive ($\hinv \not\bhleq \wlpcharhhla(\hinv)$), the operator $\kindsym$ enforces the increasing chain
$ \hinv \bhhleq \kindsym(\hinv) = \hinv \hor \wlpcharhhla(\hinv) \bhhleq \kindsym^2(\hinv) = \hinv \hor \wlpcharhhla(\hinv \hor \wlpcharhhla(\hinv)) \bhhleq \cdots $ \\
Applying $\kindsym$ iteratively can be interpreted as searching for a inductive weighting; there exists a $\kk \in \Nats$ such that ($\kindopk(\hinv) \bhleq \wlpcharhhla(\kindopk(\hinv))$ is a local subinvariant) if and only if ($\nosem{\hinv} \bhleq \wlpcharhhla(\kindopk(\hinv))$). Then we use local Park induction to derive that $\hinv \bhleq \kindopk(\hinv) \bhleq\wlp{\WHILEDO{b}{\wwcoma}}(\nosem{\hhla})$.

\begin{restatable}[$\kappa$-Induction]{theorem}{thmKappaInduction}
    \label[theorem]{thm:kappa_induction}
    For a loop $\wwcom = \WHILEDO{b}{\wwcoma} \in \wGCL$ with local variables $\varset \subseteq \Vars$, a set $\statesclosed \subseteq \Sigma$ of states that is closed under $\varset$-reachability, weightings $\wla, \iinv \in \weightings$, and $\kk \in \Nats$, it is
    $$
        \nosem{\iinv} \bhleq \wlpchar(\kindopk(\nosem{\iinv})) \quad\text{implies}\quad \nosem{\iinv} \bhleq \wlp{\WHILEDO{b}{\wwcoma}}(\nosem{\wla}).
    $$
\end{restatable}

We encode local $\kk$-induction via the encoding $\kappaEnc{\wwcom, \varset, \hinv} = \specEnc{\stmtSpec{\varset}{\hinv}{\hinv}}\symSemi$ \\ $\kappaUnrollEnc{\wwcom, \varset, \hinv}$ that modifies the variables $\varset$, applies the $\kk$-induction operator $\kk$ times, and then computes the characteristic function $\wlpcharhhla$, see \Cref{fig:encoding_kappa_induction}. The \Cref{lem:kappa_induction_semantics} shows that the semantics of the encoding $\kappaEncShort$ satisfies
$$
    \nosem{\vc{\kappaEnc{\wwcom, \varset, \hinv}}(\hhla)}(\sigma)
    \hleq \ifThenElseDot{\nosem{\hinv} \bhleq \wlpcharhhla(\kindopk(\nosem{\hinv}))}{\nosem{\hinv}(\sigma)}{\bot}
$$
The soundness of the encoding follows from the $\kk$-induction proof rule and the semantics of the encoding. This is formalized in \Cref{thm:kappa_induction_soundness}, which states that $\kappaEncShort$ is a sound under-approximation with
$$
    \nosem{\vc{\kappaEnc{\wwcom, \varset, \hinv}}(\hhla)} \hhleq \wlp{\wwcom}(\nosem{\hhla}).
$$

Again, the dual proof rule exists for proving upper bounds by replacing the encodings in \Cref{fig:encoding_kappa_induction} by their dual \text{co}-statements. The corresponding proofs are completely dual.

\subsection{\texorpdfstring{$\omega$}{omega}-Invariants}
\label{sec:w-invariants}

\begin{figure}
    \begin{minipage}[t]{0.45\textwidth}
        \begin{align*}
            & \omegaEnc{\wwcom, \hinvseq{n}} \\
            \midrule
            & \coHavoc{n} \symSemi \inlineExplain{supremum over $n$} \\
            & \stmtAssert{\hinv(n)} \symSemi \inlineExplain{filter for $\hinvseq{n}$} \\
            & \stmtAssume{\embed{\false}} \inlineExplain{establish truth} \\[1em]
            & \weightedProc{\semiringPrefix{rig}}{base}{\overline{x}_1:\tau_1, \ldots, \overline{x}_k:\tau_k}{} \\
            & ~\Requires{\hinv(0)} \\
            & ~\Ensures{\embed{\false}} \\
            & \blockStart~\blockEnd
        \end{align*}
    \end{minipage}
    \hfill
    \begin{minipage}[t]{0.45\textwidth}
        \begin{align*}
            & \weightedProcHead{\semiringPrefix{rig}}{step}{\overline{x}_1:\tau_1, \ldots, \overline{x}_k:\tau_k} \\
            & \weightedProcOut{x_1:\tau_1, \ldots, x_k:\tau_k} \\
            & ~\Requires{\hinv(n+1)} \\
            & ~\Ensures{\hhla} \\
            & \blockStart~\stmtDeclInit{x_1}{\typevar_1}{\overline{x}_1}\symSemi \ldots\symSemi \stmtDeclInit{x_k}{\typevar_k}{\overline{x}_k} \symSemi \\
            & \quad\stmtIfStart{b} \\
            & \quad\quad \procenc{\wwcoma}\symSemi \stmtAssert{\hinv(n)}\symSemi \stmtAssume{\embed{\false}}\symSemi \\
            & \quad \}~\stmtElseStart \}~ \blockEnd
        \end{align*}
    \end{minipage}

    \caption{Encoding of the $\omega$-rule for the loop $\wwcom = \WHILEDO{b}{\wwcoma} \in \wGCL$, weighting $\hhla \in \wHeyLo$, and an $\omega$-invariant $\hinv(n) \in \wHeyLo$ with free variable $n \in \Vars$.}
    \label[figure]{fig:encoding_omega}
\end{figure}

Until now, loop unrolling is the only proof rule that computes lower bounds for weakest preweightings (and upper bounds for weakest \emph{liberal} preweightings). However, this rule considers exactly $n$ unrollings for a fixed $n \in \Nats$. The $\omega$-rule allows us to prove lower bounds for all $n \in \Nats$ by constructing a chain of weightings $(\hinvseq{n})_{n \in \Nats} \in\wHeyLo^\Nats$, where $\hinvseq{n}$ lower-bounds the $n$-th unrolling, hence $\hinvseq{n} \hleq \wpcharhhla^n(\bot)$ for every $n \in \Nats$. We show that its limit is a sound lower bound of the weakest preweighting of the loop. The dual proof rule exists for upper bounds of weakest liberal preweightings.

We call an $\omega$-ascending chain of weightings $(\hinvseq{n})_{n \in \Nats} \in\wHeyLo^\Nats$ an \emph{$\omega$-subinvariant}, if $\hinvseq{0} = \bot$ and $\hinvseq{n+1} \hleq \wpcharhhla(\hinvseq{n})$ for all $n \in \Nats$ holds. An $\omega$-descending chain $(\hinvseq{n})_{n \in \Nats} \in\wHeyLo^\Nats$ is an \emph{$\omega$-superinvariant}, if $\hinvseq{0} = \top$ and $\wlpcharhhla(\hinvseq{n}) \hleq \hinvseq{n+1}$ for all $n \in \Nats$.
The $\omega$-rule states that, given an $\omega$-subinvariant $(\hinvseq{n})_{n\in\Nats}$, its supremum is a lower bound of the weakest preweighting of the loop. We can show by induction that every term $\hinvseq{n}$ bounds the subsequent Kleene term $\wpcharhhla^{n+1}(\bot)$ and prove that the limit of the sequence $(\hinvseq{n})_{n \in \Nats}$ soundly approximates the fixed point.

\begin{restatable}[$\omega$-Rule]{theorem}{thmOmegaInvariants}
    \label[theorem]{thm:omega-invariants}
    For a loop $\wwcom = \WHILEDO{b}{\wwcoma} \in \wGCL$, a postweighting $\wla \in \weightings$, and an $\omega$-subinvariant $(\iinv_n)_{n\in\Nats}\in\weightings^\Nats$, it is
    $
        \sup_{n \in \Nats} \nosem{\iinv_n} \hleq \wp{\WHILEDO{b}{\wwcoma}}(\nosem{\wla}).
    $
\end{restatable}

For the encoding of the $\omega$-rule, we represent the sequence $(\hinvseq{n})_{n \in \Nats}$ by the formula $\hinv(n) \in \wHeyLo$ with free variable $n \in \Vars$.
Then, we check that $(\hinv(n))_{n \in \Nats}$ is an $\omega$-subinvariant via the two auxiliary programs in \Cref{fig:encoding_omega} that encode the base case and the induction step of a natural induction on $n$.

\begin{restatable}[Checking $\omega$-Invariant]{lemma}{lemCheckInvariant}
    \label[lemma]{lem:semantics-check-invariant}
    For a weighting $\hinv(n) \in \wHeyLo$ with free variable $n \in \Vars$, the programs $base$ and $step$ in \Cref{fig:encoding_omega} encode
    \begin{align*}
        \begin{array}{rlcl}
            &\hinv(0) \hleq \vc{base}(\hhla)
            \quad&\text{if and only if}\quad
            &\hinv(0) = \bot \\
            \text{and}\quad
            &\hinv(n+1) \hleq \vc{step}(\hhla)
            \quad&\text{implies}\quad
            &\forall n \in \Nats: ~ \hinv(n+1) \hleq \wpcharhhla(\hinv(n)).
        \end{array}
    \end{align*}
\end{restatable}

If the two auxiliary programs are valid, \Cref{thm:omega-invariants} induces that $\sup_{n \in \Nats} \hinvseq{n}$ is a lower bound of the weakest preweighting of the loop. This is encoded in $\omegaEncShort$ in \Cref{fig:encoding_omega} and \Cref{lem:semantics-omega-encoding} shows that
$$
    \nosem{\vc{\omegaEnc{\wwcom, (\hinvseq{n})_{n \in \Nats}}}(\hhla)}(\State)
    = \hbigor_{n \in \Nats} \hinv(n).
$$

If one of the two procedures fail to verify, the proposed sequence is not a valid $\ww$-invariant and the whole verification fails. Otherwise, the encoding guarantees a sound lower bound (\Cref{thm:omega-encoding-soundness}) on the weakest preweighting with
$
    \nosem{\vc{\omegaEnc{\wwcom, (\hinvseq{n})_{n \in \Nats}}}(\hhla)} \hleq \wp{\wwcom}(\nosem{\hhla}).
$

\section{Case Studies}
\label{sec:case-studies}

In this section, we illustrate the application of our framework to verify quantitative properties of weighted programs. We demonstrate the use of the proof rules that we developed in \Cref{sec:proof_rules} such as Park induction and the $\omega$-rule. Moreover, we focus on the pipeline that our framework imposes.

Take a program $\wwcom \in \wGCL$ and a postweighting $\wla \in \weightings$. For each loop, pick a proof rule from \Cref{sec:proof_rules} and guess an invariant $\iinv \in \weightings$. Decide whether to under- or overapproximate.
Transform $\wwcom$ into a \wHeyVL program by encoding loop-free fragments directly and loops via the encoding of the proof rule. Encode the weightings as \wHeyLo formulae.
Apply the verification preweighting transformer $\symVc$ to obtain a preweighting $\hlc \in \wHeyLo$ and eliminate quantifier via the algorithm in \Cref{sec:quantifier-elimination}, before discharging to an SMT solver.

\subsection{Network Security Analysis}
\label{sec:network-security-clearance}

Let us consider a (simplified) example, where an attacker needs $n$ steps in a network to reach a target database with hidden information. At each step, the attacker has the nondeterministic choice between a standard attack (which requires clearance level $\mathsf{Clearance}$ to succeed) and a specific exploit of service that usually requires the high clearance level $\mathsf{TopSecret}$, but has a critical vulnerability at the last step (hence clearance level $\mathsf{Confidential}$ suffices).

Many automated approaches to network security modeling, such as the attack graphs generated by \citeauthor{attack_graph_model_checking} \cite{attack_graph_model_checking}, require a finite-state representation of the network. These methods evaluate fixed topologies and cannot verify properties across networks of unbounded size $n$.
Compositional reasoning via weakest preweightings often avoids these limitations for network topologies that are representable as weighted programs. Instead of generating a graph of size $n$, we compactly encode the rules of the topology using a $\symWhile$-loop.

Our example $C_{AC}$ is encoded in \Cref{fig:case_studies_wGCL_clear} and uses the monoid-module
\[
    \textsc{Clearance:}\qquad \monclear = \monclearlong \quad \text{and} \quad \modclear = \modclearlong~,
\]
with totally-ordered values $\mathsf{Public} < \mathsf{Confidential} < \mathsf{Secret} < \mathsf{Top Secret}$, the additive operation $\min$ (choosing the path with least security clearance) and the multiplicative operation $\max$ (the clearance of a path is the maximum clearance of its steps)

By calculating an upper bound on the weakest preweighting $\wp{C_{AC}}(\mathsf{P}) \sleq \hla$, corresponding to $\wp{C_{AC}}(\mathsf{P}) \geq \hla$, we determine the minimal clearance an attacker needs to successfully compromise the target for an arbitrary network size $n$. \Cref{tab:encoding_spec_char_park} and the invariant $I(n) = \min\{ \iverson{n = 0}\cdot \mathsf{P}, \iverson{n = 1}\cdot \mathsf{C}, \iverson{n \geq 2}\cdot \mathsf{S} \}$ gives us the encoding $\mathit{networkSecurity} = \parkEnc{C_{AC}, \varset, I(n)}$ in \Cref{fig:case_studies_clear_lang}, where $\varset = \{ n \}$. Using \Cref{thm:heyvl-soundness-local-park-induction-encoding}, we derive that
$$\wp{\mathit{networkSecurity}}(\mathsf{P}) \hleq \vc{\mathit{networkSecurity}} \hleq I(n_0).$$ The minimal security clearance is hence $\mathsf{P}$ if the attacker is already at the target ($n = 0$), $\mathsf{C}$ if the target is one step away ($n = 1$), and $\mathsf{S}$, otherwise ($n \geq 2$).

\subsection{Lock-Freedom of the Compare-And-Swap Technique}
\label{sec:cas-counter}

In concurrency theory, we need guarantees on how multiple processes interact when competing for shared resources. In \cite{BatzGKKW22}, the authors analyse starvation-freedom of a mutual exclusion protocol via weakest preweightings. A common non-blocking progress condition is \emph{lock-freedom} \cite[p. 60]{art_of_multiprocessor_programming}, which guarantees \emph{global} progress\footnote{At every step, some thread will make progress within a finite number of steps.}, but individual threads may starve \cite{nature_of_progress}.

Consider a system with $N$ threads trying to increment a shared counter ($v$) using the compare-and-swap technique, where each process $i$ saves a local copy of the counter in $\ell[i]$ and computes the incremented value $\ell[i] + 1$. In each iteration, a process is non-deterministically chosen and updates the variable $v$, if its local copy of $v$ is up-to-date, otherwise, it updates its local copy.

We write infinite words $w$ over the alphabet $\Gamma = \{ \mathsf{S}_1, \mathsf{F}_1, \ldots, \mathsf{S}_N, \mathsf{F}_N \}$ via the formal languages module $\modlang$ to analyse the progress of $N$ processes. E.g., the word $\mathsf{S}_1 \mathsf{F}_2 \mathsf{S}_2 \cdots$ depicts that process $1$ successfully updated the counter $v$, then process $2$ failed and updated its local copy, and afterwards process $2$ successfully updated the counter $v$. The program $C_{CAS} \in \wGCL$ in \Cref{fig:intro-wgcl-lock-freedom} encodes the compare-and-swap technique and the corresponding language of infinite words.

We define a descending $\omega$-chain of languages $(\mathcal{L}_n)_{n \in \Nats}$, where $\mathcal{L}_n$ contains all infinite words with at most $N$ consecutive failures in the first $n$ steps. Then, we show that $\wlp{C_{CAS}}(\bot) \hleq \inf_{n \in \Nats} \mathcal{L}_n$, which implies that every computed word contains at most $N$ consecutive failures and thus the program is lock-free.
We specify auxiliary languages $\mathcal{L}_n^k$, $n \in \Nats$ and $1 \leq k \leq N$,
$$
    \mathcal{L}_0^k = \Gamma^{\omega}, \qquad
    \mathcal{L}_{n+1}^0 = \bigcup_{1 \leq i \leq N} \mathsf{S}_i \cdot \mathcal{L}_n^N, \text{ and} \qquad
    \mathcal{L}_{n+1}^{k+1} = \bigg( \bigcup_{1 \leq i \leq N} \mathsf{F}_i \cdot \mathcal{L}_n^k \bigg) \cup \bigg( \bigcup_{1 \leq i \leq N} \mathsf{S}_i \cdot \mathcal{L}_n^N \bigg).
$$
Again, the subscript $n$ specifies that the prefix of length $n$ contains at most $N$ consecutive failures. The superscript $k$ specifies that the word starts with at most $k$ consecutive failures. As an example, we have $\mathsf{F}_3 \mathsf{F}_1 \mathsf{S}_3 \mathsf{F}_3 \mathsf{F}_2 \mathsf{F}_1 \mathsf{F}_3 \mathsf{S}_1 \in \mathcal{L}_6^2$ for $N = 3$, which starts with exactly $2$ consecutive failures and contains at most $3$ consecutive failures \emph{in the prefix} of length $6$. Now, we define $\mathcal{L}_n = \mathcal{L}_n^N$.

\Cref{thm:omega-invariants} states that if $\wlpPhi_{\bot}(\mathcal{L}_n) \hleq \mathcal{L}_{n+1}$ for all $n \in \Nats$, then $\inf_{n \in \Nats} \mathcal{L}_n$ is an upper bound on the weakest liberal preweighting. The encoding of this proof rule in $\wHeyVL$ is given by $\mathit{baseCAS}$, $\mathit{stepCAS}$, and $\mathit{CAS} = \omegaEnc{C_{CAS}, \mathcal{L}_n}$, see \Cref{fig:case_studies_clear_lang}. Via \Cref{thm:omega-encoding-soundness}, we conclude that $\wlp{C_{CAS}}(\bot) \hleq \vc{\mathit{CAS}} \hleq \inf_{n \in \Nats} \mathcal{L}_n$.

\subsection{Queueing Background Jobs}

We consider a queue of $k$ background jobs (e.g. image processing, mail delivery, report generation). Each job succeeds with probability 0.8 and fails with probability 0.2. Upon failure, the system either performs a simple retry with costs 4 or restarts the system and performs a full retry with costs 6. We want to obtain a bound on the expected total costs of processing the queue.

The probability module $\modprob$ allows us to model this problem as a program $C_Q \in \wGCL_{\prob}$ (see \Cref{fig:case_studies_wGCL_prob}). In $\wGCL_{\prob}$, a program $\BRANCH{\WEIGH{0.8}\fatsemi\, C_1}{\WEIGH{0.2}\fatsemi\, C_2}$ corresponds to the probabilistic choice $\pChoice{C_1}{0.8}{C_2}$ from $\pGCL$. Since this generalizes and all other statements of $\pGCL$ are native in $\wGCL_{\prob}$ (compare with \cite{McIverM05}), weighted programming is a conservative extension of probabilistic programming over non-negative integer variables.

The expected costs of processing the queue are tracked in the variable $e$. Hence, we compute an upper bound of $\wp{C_Q}(e)$ and illustrate this for $\kappa$-induction with $\kappa = 2$ by upper bounding $\wpPhi_e(\wpPhi_e(I) \hand I)$ for some invariant $I \in \wHeyLo$. The encoding is $\mathit{queueing} = \kappaEnc{C_Q, \varset, I} \in \wHeyVL$ for local variables $\varset = \{ k, e, b \}$ and invariant $I = e + 4 \cdot k + 16 \cdot \iverson{b \neq b_0}$ (see \Cref{fig:case_studies_prob_why}). We derive that $\wp{C_Q}(e) \sleq \vc{\mathit{queueing}} \sleq e_0 + 4\cdot k_0$, where $e_0$ and $k_0$ are the initial values of $e$ and $k$. Hence, the additional costs are bounded by $4\cdot k$ for a queue of length $k$.

\subsection{Database Provenance}
\label{sec:database-provenance}

\begin{wrapfigure}{r}{1.5cm}
    \vspace{-0em}
    \centering
    \begin{tabular}{cc}
        a & b \\ \hline
        2 & 3 \\
        2 & 1 \\
        \hint{3} & \hint{2} \\
        \hint{3} & \hint{1} \\
        4 & 2 \\
        \hint{5} & \hint{2} \\
        \hint{6} & \hint{3}
    \end{tabular}
    \vspace{-1em}
\end{wrapfigure}
In database theory, \emph{Why provenance} \cite{provenance_foundational_paper, gradel2025provenance,complexityofwhyprovenance, whysemiring} analyses which entries in a database determine the result of a query. As an example, we consider the organisational hierarchy of a company, where new ideas are discussed with the direct supervisor.
Employee $a$ may discuss an idea with its manager $b$ in the table. This is a simplified version of the \emph{path accessibility problem} \cite{Cook74,complexityofwhyprovenance}. \Cref{fig:case-studies-recursive-query} visualises the organiation. We consider the query:
    Can information travel from employee $4$ to manager $1$?

Why provenance determines \emph{why} that query is correct, hence which employees need to know about the new idea in order to pass it on to the top manager. Or, more formally, which entries in the table lead to the result of the query. Here, one possible explanation for the result are the rows 1, 2, and 5 in the table. However, there exist other explanations (e.g. rows 1, 3, and 5 or rows 1, 2, 3, and 5) and the new idea can be passed back and forth between employees 2 and 3 arbitrarily often, which leads to infinitely many explanations. The Why semiring \cite{whysemiring} provides finite representations of the infinite explanations due to the idempotent addition $\vee$ and we represent that employee $a$ knows about an idea by the atom $X_a$.
\Cref{fig:case_studies_wGCL_why} contains an implementation in $\wGCL_{\why}$ and computes the information flow from $4$ to $1$.

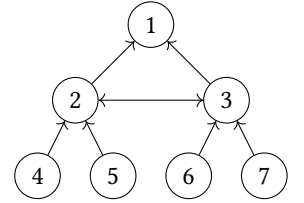
\begin{wrapfigure}{r}{3.5cm}
    \vspace{-1em}
    \begin{tikzpicture}
        \node[draw, circle] (1) at (0,0) {1};
        \node[draw, circle] (2) at (-1,-1) {2};
        \node[draw, circle] (3) at (1,-1) {3};
        \node[draw, circle] (4) at (-1.5,-2) {4};
        \node[draw, circle] (5) at (-0.5,-2) {5};
        \node[draw, circle] (6) at (0.5,-2) {6};
        \node[draw, circle] (7) at (1.5,-2) {7};

        \draw[->] (4) -- (2);
        \draw[->] (5) -- (2);
        \draw[->] (6) -- (3);
        \draw[->] (7) -- (3);
        \draw[<->] (2) -- (3);
        \draw[->] (2) -- (1);
        \draw[->] (3) -- (1);
    \end{tikzpicture}

    \caption{Organizational hierarchy as a graph.}
    \label{fig:case-studies-recursive-query}
\end{wrapfigure}
To obtain the provenance, we are interested in the weakest liberal preweighting $\wlp{C_{GR}}(\top)$, where the postweighting $\hla = \top$ means that we do not impose any additional constraints. We use Park induction to show that $I \hleq \wlpchar(I)$ for the invariant $I = \iverson{e = 1}\cdot X_1 \vee \iverson{e = 2}\cdot (X_1 \wedge X_2 \wedge X_3) \vee \iverson{e = 3}\cdot (X_1 \wedge X_2 \wedge X_3) \vee \iverson{e = 4}\cdot (X_1 \wedge X_2 \wedge X_3 \wedge X_4)$. \Cref{tab:encoding_spec_char_park} gives us the encoding $\mathit{graphReachability} = \parkEnc{C_{GR}, \{e\}, I}$ in \Cref{fig:why-lowering-example-employees}. Using \Cref{thm:heyvl-soundness-local-park-induction-encoding}, we derive that $X_1 \wedge X_2 \wedge X_3 \wedge X_4 \hleq \vc{\mathit{graphReachability}} \hleq \wlp{C_{GR}}(\top)$, which means that the information may travel via the employees 1, 2, 3, and 4.

Note that reasoning via weakest preweightings and the $\symVc$ calculus elegantly avoid looking at infinitely many derivation trees (representing infinitely many explanations), relying on complex, infinite formal power series, or implementing specialized algorithms to prevent infinite loops, since loops are evaluted as fixed points and the idempotent operation algebraically collapses the representation of infinite derivation trees into finite, closed-form expressions. This captures exact provenance without requiring infinite data structures or cycle detection.

\section{Automated Reasoning}
\label[section]{sec:quantifier-elimination}

\subsection{Verification via Caesar}

For the verification of the case studies in the previous section, we give a translation from certain instances of \wHeyVL into the existing tool \Caesar{}~\cite{SchroerBKKM23}.
This is a first step towards full automation of weighted programs.
This is simple for the probability module $\modprob$, since \wGCL is a conservative extension of \pGCL, \wHeyVL is also a conservative extension of \HeyVL.
For the Why semiring $\modwhy$, we translate statements $\WEIGH{m}$ into $\stmtAssume{\embed{m}}$ and statements $\symWeightedBranching$ into $\stmtAngelic{\cdot}{\cdot}$ (which corresponds to booleanizing the formulae).
The Clearance module is translated by replacing $\WEIGH{m}$ with $\coAssert{m}$ and $\symWeightedBranching$ with $\stmtDemonic{\cdot}{\cdot}$.
And the formal languages module considers a variable $\ell$ that contains the language that will be computed in the remainder of the execution. Then, $\WEIGH{m}$ translates to $\stmtAssert{\ell \subseteq m \cdot \Gamma^{\infty}}$, which ensures that all words in $\ell$ start with $m$.
\Cref{sec:appendix-weighted-to-ordinary-heyvl} contains the full translations for the instances that we consider in this work.

\subsection{Quantifier Elimination}

For automated verification, we want to discharge verification conditions to an SMT solver.
The quantitative quantifiers of \wHeyLo{} are an obstacle: as they correspond to infima and suprema, classical techniques for quantifier elimination do not apply. Hence, we develop a technique for quantifier elimination in our setting.
This also covers the previously not formalized implementation in use by \Caesar{}, so far limited to the probabilistic module $\modprob$.
Quantifier elimination is necessary in particular due to the use of \symHavoc/\symUp\symHavoc{}, which induces infima/suprema. For example, the encoding of Park induction (c.f. \Cref{sec:local_park_induction}) has a verification condition of the form $\top \equiv \hpreweight \impl (\hinv \sqcap \bigsqcap_x \heylovalidate{\Psi})$, where $\hinv \in \wHeyLo$ is the loop invariant and $\hpreweight \in \wHeyLo$ is a $\symRequires$, and $\Psi \in \wHeyLo$ is for the loop body.\footnote{Instead of $\hpreweight \heyloleq \hinv \sqcap \bigsqcap_x \heylovalidate{\Psi}$, we use the internalized version $\top \equiv \hpreweight \impl (\hinv \sqcap \bigsqcap_x \heylovalidate{\Psi})$ (recall \Cref{thm:deduction_theorem}).}
The key observation is that, one can first move the quantifier (\emph{prenex}) and then strip leading infima, preserving equivalence with respect to validity.
We assume $x$ is only bound in $\Psi$ and does not occur free in $\hpreweight$ or $\hinv$.\footnote{Variables bound by quantifiers can always be renamed to satisfy this condition.}
\[
  \begin{array}{r@{\quad}c@{\quad}c@{\quad}l}
    \textsc{1. Prenex:} &
    \hpreweight \impl \hinv \sqcap \bigsqcap_x \heylovalidate{\Psi}
    & \equiv &
    \bigsqcap_x (\hpreweight \impl (\hinv \sqcap \heylovalidate{\Psi}))
    \\[0.5em]
    \textsc{2. Strip:} &
    \isValid{\bigsqcap_x \left(\hpreweight \impl (\hinv \sqcap \heylovalidate{\Psi})\right)}
    & \text{iff} &
    \isValid{\hpreweight \impl (\hinv \sqcap \heylovalidate{\Psi})}~.
  \end{array}
\]

\newcommand{\qeTag}[2]{\ifmmode\text{\textcolor{#1}{\scriptsize\textsf{#2}}}\else\textcolor{#1}{\scriptsize\textsf{#2}}\fi}
\newcommand{\qeAlways}{\qeTag{neutralColor}{always}}
\newcommand{\qeWeightedAdd}{\qeTag{weightedColor}{W+}}
\newcommand{\qeWeightedScalar}{\qeTag{weightedColor}{Wx}}
\newcommand{\qeWf}{\qeTag{hintcolor}{wf.}}
\newcommand{\qeUwf}{\qeTag{hintcolor}{uwf.}}
\newcommand{\qeExtra}{\qeTag{hintcolor}{extra}}
\newcommand{\qeReq}[1]{\unskip\nobreak\hfill\mbox{#1}}

\begin{figure}[t]
    \centering

    \begingroup
    \footnotesize
    \renewcommand{\arraystretch}{1.18}
    \begin{adjustbox}{max width=\linewidth}
        \begin{tabularx}{\linewidth}
            {@{}>{\columncolor{basicBackground!2}\raggedright\arraybackslash}p{0.12\linewidth}
            !{\color{Platinum}\vrule width 1.5pt}
            >{\columncolor{basicBackground!2}\raggedright\arraybackslash}X
            !{\color{Platinum}\vrule width 2.5pt}
            >{\columncolor{basicBackground!2}\raggedright\arraybackslash}X@{}}
            \toprule
            \cellcolor{white}Context
             & \cellcolor{white}\makecell[c]{Infimum-prefix rules                                \\[-0.15em]$\bigsqcap_x(\cdot)$}
             & \cellcolor{white}\makecell[c]{Supremum-prefix rules                               \\[-0.15em]$\bigsqcup_x(\cdot)$} \\
            \midrule
            Lattice
             & $(\bigsqcap_x\hhla)\star\hhlb,\ \hhla\star(\bigsqcap_y\hhlb)
                \equiv \bigsqcap(\hhla\star\hhlb)$\qeReq{\qeAlways}
             & $(\bigsqcup_x\hhla)\star\hhlb,\ \hhla\star(\bigsqcup_y\hhlb)
                \equiv \bigsqcup(\hhla\star\hhlb)$\qeReq{\qeAlways}
            \\
            Validation
             & $\heylovalidate{\bigsqcap_x\hhla}
                \equiv \bigsqcap_x\heylovalidate{\hhla}$\qeReq{\qeAlways}
             & $\heylovalidate{\bigsqcup_x\hhla}
                \equiv \bigsqcup_x\heylovalidate{\hhla}$\qeReq{\qeUwf}
            \\
            Covalidation
             & $\heylocovalidate{\bigsqcap_x\hhla}
                \equiv \bigsqcap_x\heylocovalidate{\hhla}$\qeReq{\qeWf}
             & $\heylocovalidate{\bigsqcup_x\hhla}
                \equiv \bigsqcup_x\heylocovalidate{\hhla}$\qeReq{\qeAlways}
            \\
            Implication
             & $(\bigsqcup_x\hhla)\impl\hhlb,\ \hhla\impl(\bigsqcap_y\hhlb)
                \equiv \bigsqcap(\hhla\impl\hhlb)$\qeReq{\qeAlways}
             & $(\bigsqcap_x\hhla)\impl\hhlb,\ \hhla\impl(\bigsqcup_y\hhlb)
                \equiv \bigsqcup(\hhla\impl\hhlb)$\qeReq{\qeUwf}
            \\
            Coimplication
             & $(\bigsqcup_x\hhla)\coimpl\hhlb,\ \hhla\coimpl(\bigsqcap_y\hhlb)
                \equiv \bigsqcap(\hhla\coimpl\hhlb)$\qeReq{\qeWf}
             & $(\bigsqcap_x\hhla)\coimpl\hhlb,\ \hhla\coimpl(\bigsqcup_y\hhlb)
                \equiv \bigsqcup(\hhla\coimpl\hhlb)$\qeReq{\qeAlways}
            \\
            Addition
             & $(\bigsqcap_x\hhla)\splus\hhlb,\ \hhla\splus(\bigsqcap_y\hhlb)
                \equiv \bigsqcap(\hhla\splus\hhlb)$\qeReq{\qeWeightedAdd}
             & $(\bigsqcup_x\hhla)\splus\hhlb,\ \hhla\splus(\bigsqcup_y\hhlb)
                \equiv \bigsqcup(\hhla\splus\hhlb)$\qeReq{\qeWeightedAdd}
            \\
            Scalar
             & $a\stimes(\bigsqcap_x\hhla)
                \equiv \bigsqcap_x(a\stimes\hhla)$\qeReq{\qeExtra}
             & $a\stimes(\bigsqcup_x\hhla)
                \equiv \bigsqcup_x(a\stimes\hhla)$\qeReq{\qeWeightedScalar}
            \\
            \bottomrule
        \end{tabularx}
    \end{adjustbox}
    \endgroup

    \caption{
         Prenexing rules for $\hhla,\hhlb \in \wHeyLo$ and $a \in \MonExp$ where $x$ is not free in $\hhlb$ and $y$ is not free in $\hhla$.
         Comma-separated formulas are separate prenexing laws, one for each displayed position of the quantified subformula.
         On the right, \(\bigsqcap(\cdot)\) and \(\bigsqcup(\cdot)\) bind the same variable as the quantified subformula being moved.
         In the lattice row, $\star\in\{\sqcap,\sqcup\}$. \qeAlways{} rules hold always, \qeWeightedAdd{}/\qeWeightedScalar{}: under extra weighted assumptions, \qeUwf{}/\qeWf{} under (upwards) well-foundedness, and \qeExtra{} under scalar meet preservation.
      }
    \label{fig:qelim-prenex-landscape}
\end{figure}

\Cref{fig:qelim-prenex-landscape} shows the landscape of prenexing rules that hold under various assumptions.
Many hold always (\qeAlways).
For the example above where we eliminate infima, we make use of the rules in the second column and prenex over the validation $\heylovalidate{\cdot}$, binary meet $\sqcap$, and implication $\impl$.

The rules in the second column generate a prefix of infimum quantifiers, and the third column generate suprema.
Generally, we can strip leading infima with respect to validity (equivalence to $\top$), and suprema with respect to covalidity (equivalence to $\bot$):
\begin{restatable}[Quantifier Stripping]{theorem}{thmQelimStrip}
  \label{thm:qelim-strip}
  For every $\hhla \in \wHeyLo$:
  \[
    \linfquant{x}{\hhla} \text{ is valid } \qiff \hhla \text{ is valid}, \qquad
    \lsupquant{x}{\hhla} \text{ is covalid } \qiff \hhla \text{ is covalid}.
  \]
\end{restatable}

In contrast to classical logic, we cannot always prenex quantifiers~\cite{baaz2024goedellogicsprenexfragments}.
For example, $\heylocovalidate{\bigsqcap_x\hhla} \equiv \bigsqcap_x\heylocovalidate{\hhla}$ only holds when our monoid-module is \emph{well-founded}, i.e. for all $\State \in \States$, there exists a value $v \in \Vals$ such that $\sem{\linfquant{x}{\hhla}}(\State) = \sem{\hhla\substBy{x}{v}}(\State)$ for all $\hhla \in \wHeyLo$.
Over the probabilistic module $\modprob$, let \(\hhla = (1/2)^x \otimes 1\).
Then \(\bigsqcap_x \hhla = 0\), but no \(x \in \mathbb{N}\) realizes this infimum; hence \(\nabla(\bigsqcap_x \hhla)=0\), whereas \(\bigsqcap_x \nabla(\hhla)=\infty\).
By contrast, the tropical monoid-module $\modtrop$ is \emph{upwards} well-founded: its natural order is reversed on \(\NatsX\), so suprema are ordinary minima and are always attained.
Further, all of our monoid-modules also satisfy the \qeWeightedAdd{}, \qeWeightedScalar{}, \qeExtra{} assumptions, which allow prenexing over addition and scalar multiplication (suprema and infima), respectively.\footnote{For \qeWeightedAdd{}, $\splus$ must preserve binary meets and joins, and for \qeWeightedScalar{}, $a \stimes (\cdot)$ must preserve binary joins, for \qeExtra{}, binary meets.}
From now on, we focus on rules that always hold for the monoid-modules in our paper (\qeAlways, \qeWeightedAdd{}, \qeWeightedScalar{}, \qeExtra{} but not \qeWf{} nor \qeUwf{}).
The constructions can be easily extended for when other assumptions hold.

\begin{figure}[t]
  \centering
  \renewcommand{\arraystretch}{1.24}
  \begin{minipage}[t]{0.47\linewidth}
    \vspace{0pt}
    \centering
    \textbf{Elimination for Validity} \\
    $\mathsf{qelim}^{\downarrow} \colon \wHeyLo \to \wHeyLo$:
    \[
      \begin{array}{@{}r@{\,\mapsto\,}l@{}}
        \alpha & \alpha \quad \text{for atomic $\alpha$, $\top$, $\bot$}
        \\[0.35em]
        \heylovalidate{\hhla} & \heylovalidate{\mathsf{qelim}^{\downarrow}(\hhla)}
        \\[0.35em]
        \hhla \mathbin{\star} \hhlb
        & \mathsf{qelim}^{\downarrow}(\hhla) \mathbin{\star} \mathsf{qelim}^{\downarrow}(\hhlb)
        \\[0.05em]
        \multicolumn{2}{@{}l@{}}{\qquad \text{for $\star \in \{\sqcap,\sqcup,\splus\}$}}
        \\[0.35em]
        a \stimes \hhla & a \stimes \mathsf{qelim}^{\downarrow}(\hhla)
        \\[0.35em]
        \hhla \impl \hhlb
        & \mathsf{qelim}^{\uparrow}(\hhla) \impl \mathsf{qelim}^{\downarrow}(\hhlb)
        \\[0.35em]
        \bigsqcap_x \hhla & \mathsf{fresh}_x(\mathsf{qelim}^{\downarrow}(\hhla))
        \\[0.35em]
        \hhla & \hhla \quad \text{otherwise}
      \end{array}
    \]
  \end{minipage}
  \hfill
  \begin{minipage}[t]{0.47\linewidth}
    \vspace{0pt}
    \centering
    \textbf{Elimination for Covalidity} \\
    $\mathsf{qelim}^{\uparrow} \colon \wHeyLo \to \wHeyLo$:
    \[
      \begin{array}{@{}r@{\,\mapsto\,}l@{}}
        \alpha & \alpha \quad \text{for atomic $\alpha$, $\top$, $\bot$}
        \\[0.35em]
        \heylocovalidate{\hhla} & \heylocovalidate{\mathsf{qelim}^{\uparrow}(\hhla)}
        \\[0.35em]
        \hhla \mathbin{\star} \hhlb
        & \mathsf{qelim}^{\uparrow}(\hhla) \mathbin{\star} \mathsf{qelim}^{\uparrow}(\hhlb)
        \\[0.05em]
        \multicolumn{2}{@{}l@{}}{\qquad \text{for $\star \in \{\sqcap,\sqcup,\splus\}$}}
        \\[0.35em]
        a \stimes \hhla & a \stimes \mathsf{qelim}^{\uparrow}(\hhla)
        \\[0.35em]
        \hhla \coimpl \hhlb
        & \mathsf{qelim}^{\downarrow}(\hhla) \coimpl \mathsf{qelim}^{\uparrow}(\hhlb)
        \\[0.35em]
        \bigsqcup_x \hhla & \mathsf{fresh}_x(\mathsf{qelim}^{\uparrow}(\hhla))
        \\[0.35em]
        \hhla & \hhla \quad \text{otherwise}
      \end{array}
    \]
  \end{minipage}
  \caption{
    Quantifier elimination procedures $\mathsf{qelim}^{\downarrow}$ and $\mathsf{qelim}^{\uparrow}$ for validity and covalidity, respectively.
    The function $\mathsf{fresh}_x(\cdot) \colon \wHeyLo \to \wHeyLo$ replaces the eliminated bound variable $x$ by a fresh free variable.}
  \label{fig:qelim-algorithm}
\end{figure}

In the above, we used the rules in \Cref{fig:qelim-prenex-landscape} to shift quantifiers bottom-up and then strip a prefix.
However, when the shifted quantifier cannot be stripped (e.g. a supremum w.r.t. validity), we want to avoid quantifiers capturing larger formulas than before.
Therefore, we present a \emph{top-down} approach in \Cref{fig:qelim-algorithm}.
Starting from the root, we directly remove eliminatable quantifiers while possible and stop recursing at quantifiers that cannot be eliminated.
When a quantifier is removed, we replace bound variables $x$ by fresh ones to avoid wrong capture.
We have two entry points, $\mathsf{qelim}^{\downarrow}$ and $\mathsf{qelim}^{\uparrow}$, corresponding to elimination with respect to validity and covalidity, respectively.
When the pattern of the current formula matches, then we may stop, recurse, or switch to the dual procedure.
For example, $\mathsf{qelim}^{\downarrow}(\heylovalidate{\linfquant{x}{x}}) = \heylovalidate{\mathsf{qelim}^{\downarrow}(x)} = \heylovalidate{x}$ and $\mathsf{qelim}^{\downarrow}(\lsupquant{x}{x} \impl \linfquant{y}{y}) = x \impl y$.
For soundness, we prove the equivalence to prenexing and stripping leading quantifiers via \Cref{thm:qelim-strip}.

\noindent
\begin{restatable}{theorem}{thmQelimSoundness}
  \label[theorem]{thm:qelim-soundness}
  Assume \qeAlways, \qeWeightedAdd{}, \qeWeightedScalar{}, and \qeExtra{} hold.
  Then, for every $\hhla \in \wHeyLo$,
  \[
    \isValid{\hhla} \qiff \isValid{\downQelim{\hhla}}~,
    \qquad
    \isCovalid{\hhla} \qiff \isCovalid{\upQelim{\hhla}}~.
  \]
\end{restatable}

Because the procedures traverse every sub-formula at most once, they run in linear time.\footnote{$\downQelim{\cdot}$ and $\upQelim{\cdot}$ can be combined with the variable renaming of $\mathsf{fresh}$ to avoid redundant traversals.}
The resulting formulas are of the same sizes as before, only with quantifiers removed and fresh variables introduced for every bound variable.
For the following fragment, \emph{all} quantifiers can be eliminated.

\begin{restatable}{theorem}{thmQelimCompleteFragment}
  \label[theorem]{thm:qelim-complete-fragment}
  Assume \qeAlways, \qeWeightedAdd{}, \qeWeightedScalar{}, and \qeExtra{} hold.
  Let $\downQFrag, \upQFrag \subseteq \wHeyLo$ be given by the following mutually recursive grammars, where $\hhla_{\mathsf{qf}} \in \wHeyLo$ is quantifier-free.
  \begin{align*}
    \hhla_\downarrow ~::=~&
    \hhla_{\mathsf{qf}}
    \mid \heylovalidate{\hhla_\downarrow}
    \mid \bigsqcap_x \hhla_\downarrow
    \mid \hhla_\downarrow \sqcap \hhla_\downarrow
    \mid \hhla_\downarrow \sqcup \hhla_\downarrow
    \mid \hhla_\downarrow \splus \hhla_\downarrow
    \mid a \stimes \hhla_\downarrow
    \mid \hhla_\uparrow \impl \hhla_\downarrow,
    \\
    \hhla_\uparrow ~::=~&
    \hhla_{\mathsf{qf}}
    \mid \heylocovalidate{\hhla_\uparrow}
    \mid \bigsqcup_x \hhla_\uparrow
    \mid \hhla_\uparrow \sqcap \hhla_\uparrow
    \mid \hhla_\uparrow \sqcup \hhla_\uparrow
    \mid \hhla_\uparrow \splus \hhla_\uparrow
    \mid a \stimes \hhla_\uparrow
    \mid \hhla_\downarrow \coimpl \hhla_\uparrow.
  \end{align*}
  Let $\hhla \in \downQFrag$.
  Then, $\downQelim{\hhla} \in \qfFrag$.
  For $\hhla \in \upQFrag,$ we have $\upQelim{\hhla} \in \qfFrag$.
\end{restatable}

As a consequence, we can completely eliminate the quantifiers introduced by \symHavoc{} statements in procedures for encodings of the specification statement, Park and $k$-induction proof rules (c.f. \Cref{sec:proof_rules}).
Dually, the same holds for the dual encodings using \symUp\symHavoc{} statements in coprocedures.

\section{Related Work}
\label{sec:related-work}

Because we generalize using monoid-modules and not just semirings, our work subsumes probabilistic programming verification, in particular the \HeyVL{} language~\cite{SchroerBKKM23}.\footnote{This work only subsumes probabilistic programs without \symTick{} statements, as those are still missing in weighted programming as well.}
We address several missing details, including soundness of \emph{local} proof rules, and quantifier elimination.
Although both are internally used in the \Caesar{} implementation, they were previously not formally addressed.

We build a general deductive verification framework for weighted programs, extending weakest preweighting semantics~\cite{BatzGKKW22} with an assertion language and an intermediate verification language.
For proof rules, we newly generalize the proof rules for specification statements~\cite{specificationstatement}, $\kappa$-induction~\cite{k_induction}, and $\omega$-invariants, introducing \emph{local} variants.

Weighted NetKAT~\cite{weightedNetKAT26} provides a semiring-parametric extension of NetKAT and verifies quantitative network properties by compiling policies to weighted automata.
This is complementary to our work: Weighted NetKAT exploits the finite packet model of NetKAT to obtain decision procedures, while our framework provides a general deductive verification infrastructure for weighted programs.
However, Weighted NetKAT policies can be encoded as \wGCL{} programs, and the corresponding scalar properties can be stated as weakest-preweighting bounds and verified deductively.
Outcome Logic~\cite{Zilberstein25} is a Hoare-style logic with algebraic effects based on semirings.

\citeauthor{DBLP:conf/fossacs/BatzKO25} study quantifier elimination for piecewise-linear quantities over the extended rationals~\cite{DBLP:conf/fossacs/BatzKO25}.
In contrast, we also support e.g. formulas with exponentials (c.f. \Cref{sec:quantifier-elimination}) and other monoid-modules.
Whereas they target full logical equivalence, we only target validity/covalidity, which is sufficient for verification conditions.
For Gödel logics, \citeauthor{baaz2024goedellogicsprenexfragments} study prenex fragments~\cite{baaz2024goedellogicsprenexfragments} and show that prenex normal forms are not always available.
In contrast to both of the above, we also degrade gracefully when quantifiers cannot be eliminated.

\section{Conclusion}

We presented a deductive verification infrastructure for weighted programs, parameterized by $\omega$-bicontinuous monoid-modules whose natural orders form bi-Heyting algebras.
Its assertion language \wHeyLo{} and intermediate language \wHeyVL{} uniformly support lower- and upper-bound reasoning and encode specification statements, local Park and $\kappa$-induction, and $\omega$-invariants.
We proved these encodings sound and presented a linear-time, best-effort quantifier-elimination procedure that preserves validity or covalidity and is complete on identified fragments.
Sound lowerings to \Caesar{} enabled the verification of expected costs, security clearances, database provenance, and formal-language properties.
Future work includes a native implementation, automated invariant generation, and support for higher moments and multiple objectives.

\bibliographystyle{ACM-Reference-Format}
\bibliography{references}

@article{SchroerBKKM23,
  author       = {Philipp Schr{\"{o}}er and
                  Kevin Batz and
                  Benjamin Lucien Kaminski and
                  Joost{-}Pieter Katoen and
                  Christoph Matheja},
  title        = {A Deductive Verification Infrastructure for Probabilistic Programs},
  journal      = {Proc. {ACM} Program. Lang.},
  volume       = {7},
  number       = {{OOPSLA2}},
  pages        = {2052--2082},
  year         = {2023},
  url          = {https://doi.org/10.1145/3622870},
  doi          = {10.1145/3622870},
  bibsource    = {dblp computer science bibliography, https://dblp.org}
}

@book{McIverM05,
  author       = {Annabelle McIver and
                  Carroll Morgan},
  title        = {Abstraction, Refinement and Proof for Probabilistic Systems},
  series       = {Monographs in Computer Science},
  publisher    = {Springer},
  address      = {New York, NY},
  year         = {2005},
  url          = {https://doi.org/10.1007/b138392},
  doi          = {10.1007/B138392},
  isbn         = {978-0-387-40115-7},
  bibsource    = {dblp computer science bibliography, https://dblp.org}
}

@article{BatzGKKW22,
  author       = {Kevin Batz and
                  Adrian Gallus and
                  Benjamin Lucien Kaminski and
                  Joost{-}Pieter Katoen and
                  Tobias Winkler},
  title        = {Weighted programming: a programming paradigm for specifying mathematical
                  models},
  journal      = {Proc. {ACM} Program. Lang.},
  volume       = {6},
  number       = {{OOPSLA1}},
  pages        = {1--30},
  year         = {2022},
  url          = {https://doi.org/10.1145/3527310},
  doi          = {10.1145/3527310},
  bibsource    = {dblp computer science bibliography, https://dblp.org}
}

@inproceedings{BarnettCDJL05,
  author       = {Michael Barnett and
                  Bor{-}Yuh Evan Chang and
                  Robert DeLine and
                  Bart Jacobs and
                  K. Rustan M. Leino},
  editor       = {Frank S. de Boer and
                  Marcello M. Bonsangue and
                  Susanne Graf and
                  Willem P. de Roever},
  title        = {Boogie: {A} Modular Reusable Verifier for Object-Oriented Programs},
  booktitle    = {Formal Methods for Components and Objects, 4th International Symposium,
                  {FMCO} 2005, Amsterdam, The Netherlands, November 1-4, 2005, Revised
                  Lectures},
  series       = {Lecture Notes in Computer Science},
  volume       = {4111},
  pages        = {364--387},
  publisher    = {Springer},
  address      = {Berlin, Heidelberg},
  year         = {2005},
  url          = {https://doi.org/10.1007/11804192\_17},
  doi          = {10.1007/11804192\_17},
  bibsource    = {dblp computer science bibliography, https://dblp.org}
}

@inproceedings{FilliatreP13,
  author       = {Jean{-}Christophe Filli{\^{a}}tre and
                  Andrei Paskevich},
  editor       = {Matthias Felleisen and
                  Philippa Gardner},
  title        = {Why3 - Where Programs Meet Provers},
  booktitle    = {Programming Languages and Systems - 22nd European Symposium on Programming,
                  {ESOP} 2013, Held as Part of the European Joint Conferences on Theory
                  and Practice of Software, {ETAPS} 2013, Rome, Italy, March 16-24,
                  2013. Proceedings},
  series       = {Lecture Notes in Computer Science},
  volume       = {7792},
  pages        = {125--128},
  publisher    = {Springer},
  address      = {Berlin, Heidelberg},
  year         = {2013},
  url          = {https://doi.org/10.1007/978-3-642-37036-6\_8},
  doi          = {10.1007/978-3-642-37036-6\_8},
  bibsource    = {dblp computer science bibliography, https://dblp.org}
}

@article{weightedNetKAT26,
  author       = {Emmanuel Su{\'a}rez Acevedo and
                  Tiago Ferreira and
                  Kevin Batz and
                  Oliver B{\o}ving and
                  Nate Foster and
                  Alexandra Silva},
  title        = {Weighted {NetKAT}: A Programming Language For Quantitative Network Verification},
  journal      = {Proc. {ACM} Program. Lang.},
  volume       = {10},
  number       = {{PLDI}},
  articleno    = {240},
  numpages     = {73},
  year         = {2026},
  url          = {https://doi.org/10.1145/3808318},
  doi          = {10.1145/3808318}
}

@article{Dijkstra75,
  author       = {Edsger W. Dijkstra},
  title        = {Guarded Commands, Nondeterminacy and Formal Derivation of Programs},
  journal      = {Commun. {ACM}},
  volume       = {18},
  number       = {8},
  pages        = {453--457},
  year         = {1975},
  url          = {https://doi.org/10.1145/360933.360975},
  doi          = {10.1145/360933.360975},
  bibsource    = {dblp computer science bibliography, https://dblp.org}
}

@book{art_of_multiprocessor_programming,
  title={The Art of Multiprocessor Programming},
  author={Herlihy, M. and Shavit, N. and Luchangco, V. and Spear, M.},
  isbn={9780123914064},
  url={https://books.google.de/books?id=7MqcBAAAQBAJ},
  year={2020},
  edition={2},
  publisher={Morgan Kaufmann},
  address={Cambridge, MA},
  doi={10.1016/C2011-0-06993-4}
}

@inproceedings{nature_of_progress,
  author = {Herlihy, Maurice and Shavit, Nir},
  editor = {Fern{\'a}ndez Anta, Antonio and Lipari, Giuseppe and Roy, Matthieu},
  title = {On the Nature of Progress},
  booktitle = {Principles of Distributed Systems},
  series = {Lecture Notes in Computer Science},
  volume = {7109},
  pages = {313--328},
  publisher = {Springer},
  address = {Berlin, Heidelberg},
  year = {2011},
  isbn = {978-3-642-25872-5},
  doi = {10.1007/978-3-642-25873-2_22},
  url = {https://doi.org/10.1007/978-3-642-25873-2_22}
}

@article{specificationstatement,
  author       = {Carroll Morgan},
  title        = {The Specification Statement},
  journal      = {{ACM} Trans. Program. Lang. Syst.},
  volume       = {10},
  number       = {3},
  pages        = {403--419},
  year         = {1988},
  url          = {https://doi.org/10.1145/44501.44503},
  doi          = {10.1145/44501.44503},
  bibsource    = {dblp computer science bibliography, https://dblp.org}
}

@InProceedings{k_induction,
  author       = {Kevin Batz and
                  Mingshuai Chen and
                  Benjamin Lucien Kaminski and
                  Joost{-}Pieter Katoen and
                  Christoph Matheja and
                  Philipp Schr{\"{o}}er},
  editor       = {Alexandra Silva and
                  K. Rustan M. Leino},
  title        = {Latticed k-Induction with an Application to Probabilistic Programs},
  booktitle    = {Computer Aided Verification - 33rd International Conference, {CAV}
                  2021, Virtual Event, July 20-23, 2021, Proceedings, Part {II}},
  series       = {Lecture Notes in Computer Science},
  volume       = {12760},
  pages        = {524--549},
  publisher    = {Springer},
  address      = {Cham},
  year         = {2021},
  url          = {https://doi.org/10.1007/978-3-030-81688-9\_25},
  doi          = {10.1007/978-3-030-81688-9\_25},
  bibsource    = {dblp computer science bibliography, https://dblp.org}
}

@inproceedings{whysemiring,
  author    = {Katrin M. Dannert and Erich Gr{\"a}del and Matthias Naaf and Val Tannen},
  editor    = {Christel Baier and Jean Goubault{-}Larrecq},
  title     = {Semiring Provenance for Fixed-Point Logic},
  booktitle = {29th EACSL Annual Conference on Computer Science Logic (CSL 2021)},
  series    = {Leibniz International Proceedings in Informatics (LIPIcs)},
  volume    = {183},
  pages     = {17:1--17:22},
  publisher = {Schloss Dagstuhl -- Leibniz-Zentrum f{\"u}r Informatik},
  address   = {Dagstuhl, Germany},
  year      = {2021},
  isbn      = {978-3-95977-175-7},
  doi       = {10.4230/LIPIcs.CSL.2021.17},
  url       = {https://doi.org/10.4230/LIPIcs.CSL.2021.17}
}

@article{complexityofwhyprovenance,
author = {Calautti, Marco and Livshits, Ester and Pieris, Andreas and Schneider, Markus},
title = {The Complexity of Why-Provenance for Datalog Queries},
year = {2024},
issue_date = {May 2024},
publisher = {Association for Computing Machinery},
address = {New York, NY, USA},
volume = {2},
number = {2},
url = {https://doi.org/10.1145/3651146},
doi = {10.1145/3651146},
journal = {Proc. ACM Manag. Data},
month = may,
articleno = {83},
numpages = {16}
}

@incollection{gradel2025provenance,
  author    = {Erich Gr{\"a}del and Val Tannen},
  title     = {Provenance Analysis and Semiring Semantics for First-Order Logic},
  booktitle = {Model Theory, Computer Science, and Graph Polynomials: Festschrift in Honor of Johann A. Makowsky},
  editor    = {Klaus Meer and Alexander Rabinovich and Elena Ravve and Andr{\'e}s Villaveces},
  series    = {Trends in Mathematics},
  pages     = {351--401},
  publisher = {Springer Nature Switzerland},
  address   = {Cham},
  year      = {2025},
  doi       = {10.1007/978-3-031-86319-6\_21},
  url       = {https://doi.org/10.1007/978-3-031-86319-6\_21}
}

@inproceedings{provenance_foundational_paper,
author = {Green, Todd J. and Karvounarakis, Grigoris and Tannen, Val},
title = {Provenance semirings},
year = {2007},
isbn = {9781595936851},
publisher = {Association for Computing Machinery},
address = {New York, NY, USA},
url = {https://doi.org/10.1145/1265530.1265535},
doi = {10.1145/1265530.1265535},
booktitle = {Proceedings of the Twenty-Sixth ACM SIGMOD-SIGACT-SIGART Symposium on Principles of Database Systems},
pages = {31--40},
numpages = {10},
location = {Beijing, China},
series = {PODS '07}
}

@INPROCEEDINGS{attack_graph_model_checking,
  author={Sheyner, O. and Haines, J. and Jha, S. and Lippmann, R. and Wing, J.M.},
  booktitle={Proceedings 2002 IEEE Symposium on Security and Privacy},
  title={Automated generation and analysis of attack graphs},
  year={2002},
  pages={273--284},
  publisher={IEEE Computer Society},
  address={Los Alamitos, CA},
  doi={10.1109/SECPRI.2002.1004377},
  url={https://doi.org/10.1109/SECPRI.2002.1004377}
}

@incollection{leino2007a,
  author = {Leino, K. Rustan M. and Schulte, Wolfram},
  editor = {Broy, Manfred and Gr{\"u}nbauer, Johannes and Hoare, Tony},
  title = {A Verifying Compiler for a Multi-Threaded Object-Oriented Language},
  booktitle = {Software System Reliability and Security},
  series = {NATO Security through Science Series D: Information and Communication Security},
  volume = {9},
  pages = {351--416},
  publisher = {IOS Press},
  address = {Amsterdam},
  year = {2007},
  isbn = {978-1-58603-731-4},
  url = {https://www.microsoft.com/en-us/research/publication/a-verifying-compiler-for-a-multi-threaded-object-oriented-language/}
}

@misc{baaz2024goedellogicsprenexfragments,
      title={Goedel logics: Prenex fragments},
      author={Matthias Baaz and Mariami Gamsakhurdia},
      year={2024},
      eprint={2407.16683},
      archivePrefix={arXiv},
      primaryClass={cs.LO},
      url={https://arxiv.org/abs/2407.16683},
}

@inproceedings{parthasarathy2021formallyvalidatingpracticalverification,
  author    = {Gaurav Parthasarathy and Peter M{\"{u}}ller and Alexander J. Summers},
  title     = {Formally Validating a Practical Verification Condition Generator},
  booktitle = {Computer Aided Verification},
  series    = {Lecture Notes in Computer Science},
  pages     = {704--727},
  publisher = {Springer},
  address   = {Cham},
  year      = {2021},
  doi       = {10.1007/978-3-030-81688-9\_33},
  url       = {https://doi.org/10.1007/978-3-030-81688-9\_33},
}

@book{StoneSpaces,
  title={Stone Spaces},
  author={Johnstone, P.T.},
  isbn={9780521337793},
  lccn={lc82004506},
  series={Cambridge Studies in Advanced Mathematics},
  url={https://books.google.de/books?id=CiWwoLNbpykC},
  year={1982},
  publisher={Cambridge University Press},
  address={Cambridge}
}

@book{Kleene02,
  author = {Stephen Cole Kleene},
  title = {Mathematical Logic},
  publisher = {Dover Publications},
  address = {Mineola, NY},
  edition = {Dover ed.},
  year = {2002},
  isbn = {978-0-486-42533-7},
  note = {Originally published by Wiley in 1967}
}

@inproceedings{DBLP:conf/fossacs/BatzKO25,
  author       = {Kevin Batz and
                  Joost{-}Pieter Katoen and
                  Nora Orhan},
  editor       = {Parosh Aziz Abdulla and
                  Delia Kesner},
  title        = {Quantifier Elimination and Craig Interpolation: The Quantitative Way},
  booktitle    = {Foundations of Software Science and Computation Structures - 28th
                  International Conference, FoSSaCS 2025, Held as Part of the International
                  Joint Conferences on Theory and Practice of Software, {ETAPS} 2025,
                  Hamilton, ON, Canada, May 3-8, 2025, Proceedings},
  series       = {Lecture Notes in Computer Science},
  volume       = {15691},
  pages        = {176--197},
  publisher    = {Springer},
  address      = {Cham},
  year         = {2025},
  url          = {https://doi.org/10.1007/978-3-031-90897-2\_9},
  doi          = {10.1007/978-3-031-90897-2\_9},
  bibsource    = {dblp computer science bibliography, https://dblp.org}
}

@article{Cook74,
  author       = {Stephen A. Cook},
  title        = {An Observation on Time-Storage Trade Off},
  journal      = {J. Comput. Syst. Sci.},
  volume       = {9},
  number       = {3},
  pages        = {308--316},
  year         = {1974},
  url          = {https://doi.org/10.1016/S0022-0000(74)80046-2},
  doi          = {10.1016/S0022-0000(74)80046-2},
  bibsource    = {dblp computer science bibliography, https://dblp.org}
}

@inproceedings{cmp-lg-9705006,
  author       = {Stefan Riezler},
  title        = {Quantitative Constraint Logic Programming for Weighted Grammar Applications},
  editor       = {Christian Retor{\'e}},
  booktitle    = {Logical Aspects of Computational Linguistics},
  series       = {Lecture Notes in Computer Science},
  volume       = {1328},
  pages        = {346--365},
  publisher    = {Springer},
  address      = {Berlin, Heidelberg},
  year         = {1997},
  doi          = {10.1007/BFb0052166},
  url          = {https://doi.org/10.1007/BFb0052166},
  bibsource    = {dblp computer science bibliography, https://dblp.org}
}

@article{GuZ21,
  author       = {Tao Gu and
                  Fabio Zanasi},
  title        = {Coalgebraic Semantics for Probabilistic Logic Programming},
  journal      = {Log. Methods Comput. Sci.},
  volume       = {17},
  number       = {2},
  pages        = {2:1--2:35},
  year         = {2021},
  url          = {https://lmcs.episciences.org/7365},
  doi          = {10.23638/LMCS-17(2:2)2021},
  bibsource    = {dblp computer science bibliography, https://dblp.org}
}

@article{Zilberstein25,
  author       = {Noam Zilberstein},
  title        = {Outcome Logic: {A} Unified Approach to the Metatheory of Program Logics
                  with Branching Effects},
  journal      = {{ACM} Trans. Program. Lang. Syst.},
  volume       = {47},
  number       = {3},
  pages        = {14:1--14:71},
  year         = {2025},
  url          = {https://doi.org/10.1145/3743131},
  doi          = {10.1145/3743131},
  bibsource    = {dblp computer science bibliography, https://dblp.org}
}

@inproceedings{KoeveringR025,
  author       = {Spencer Van Koevering and
                  Wojciech Rozowski and
                  Alexandra Silva},
  editor       = {Keren Censor{-}Hillel and
                  Fabrizio Grandoni and
                  Jo{\"{e}}l Ouaknine and
                  Gabriele Puppis},
  title        = {Weighted {GKAT:} Completeness and Complexity},
  booktitle    = {52nd International Colloquium on Automata, Languages, and Programming,
                  {ICALP} 2025, Aarhus, Denmark, July 8-11, 2025},
  series       = {LIPIcs},
  volume       = {334},
  pages        = {172:1--172:18},
  publisher    = {Schloss Dagstuhl - Leibniz-Zentrum f{\"{u}}r Informatik},
  address      = {Dagstuhl, Germany},
  year         = {2025},
  url          = {https://doi.org/10.4230/LIPIcs.ICALP.2025.172},
  doi          = {10.4230/LIPICS.ICALP.2025.172},
  bibsource    = {dblp computer science bibliography, https://dblp.org}
}

@article{BelleR20,
  author       = {Vaishak Belle and
                  Luc De Raedt},
  title        = {Semiring programming: {A} semantic framework for generalized sum product
                  problems},
  journal      = {Int. J. Approx. Reason.},
  volume       = {126},
  pages        = {181--201},
  year         = {2020},
  url          = {https://doi.org/10.1016/j.ijar.2020.08.001},
  doi          = {10.1016/J.IJAR.2020.08.001},
  bibsource    = {dblp computer science bibliography, https://dblp.org}
}

@article{HansenLMP18,
  author       = {Mikkel Hansen and
                  Kim Guldstrand Larsen and
                  Radu Mardare and
                  Mathias Ruggaard Pedersen},
  title        = {Reasoning About Bounds in Weighted Transition Systems},
  journal      = {Log. Methods Comput. Sci.},
  volume       = {14},
  number       = {4},
  pages        = {19:1--19:32},
  year         = {2018},
  url          = {https://doi.org/10.23638/LMCS-14(4:19)2018},
  doi          = {10.23638/LMCS-14(4:19)2018},
  bibsource    = {dblp computer science bibliography, https://dblp.org}
}

@inproceedings{BistarelliMS08,
  author       = {Stefano Bistarelli and
                  Fabio Martinelli and
                  Francesco Santini},
  editor       = {Chunming Rong and
                  Martin Gilje Jaatun and
                  Frode Eika Sandnes and
                  Laurence Tianruo Yang and
                  Jianhua Ma},
  title        = {A Semantic Foundation for Trust Management Languages with Weights:
                  An Application to the RTFamily},
  booktitle    = {Autonomic and Trusted Computing, 5th International Conference, {ATC}
                  2008, Oslo, Norway, June 23-25, 2008, Proceedings},
  series       = {Lecture Notes in Computer Science},
  volume       = {5060},
  pages        = {481--495},
  publisher    = {Springer},
  address      = {Berlin, Heidelberg},
  year         = {2008},
  url          = {https://doi.org/10.1007/978-3-540-69295-9\_38},
  doi          = {10.1007/978-3-540-69295-9\_38},
  bibsource    = {dblp computer science bibliography, https://dblp.org}
}

@incollection{DrosteK21,
  author       = {Manfred Droste and
                  Dietrich Kuske},
  editor       = {Jean{-}{\'{E}}ric Pin},
  title        = {Weighted automata},
  booktitle    = {Handbook of Automata Theory},
  pages        = {113--150},
  publisher    = {European Mathematical Society Publishing House},
  address      = {Z{\"{u}}rich, Switzerland},
  year         = {2021},
  url          = {https://doi.org/10.4171/Automata-1/4},
  doi          = {10.4171/AUTOMATA-1/4},
  bibsource    = {dblp computer science bibliography, https://dblp.org}
}

@inproceedings{BolligG09,
  author       = {Benedikt Bollig and
                  Paul Gastin},
  editor       = {Volker Diekert and
                  Dirk Nowotka},
  title        = {Weighted versus Probabilistic Logics},
  booktitle    = {Developments in Language Theory, 13th International Conference, {DLT}
                  2009, Stuttgart, Germany, June 30 - July 3, 2009. Proceedings},
  series       = {Lecture Notes in Computer Science},
  volume       = {5583},
  pages        = {18--38},
  publisher    = {Springer},
  address      = {Berlin, Heidelberg},
  year         = {2009},
  url          = {https://doi.org/10.1007/978-3-642-02737-6\_2},
  doi          = {10.1007/978-3-642-02737-6\_2},
  bibsource    = {dblp computer science bibliography, https://dblp.org}
}

@inproceedings{AhrensKGK25,
  author       = {Emma Ahrens and
                  Jan{-}Christoph Kassing and
                  J{\"{u}}rgen Giesl and
                  Joost{-}Pieter Katoen},
  editor       = {Maribel Fern{\'{a}}ndez},
  title        = {Weighted Rewriting: Semiring Semantics for Abstract Reduction Systems},
  booktitle    = {10th International Conference on Formal Structures for Computation
                  and Deduction, {FSCD} 2025, Birmingham, UK, July 14-20, 2025},
  series       = {LIPIcs},
  volume       = {337},
  pages        = {6:1--6:21},
  publisher    = {Schloss Dagstuhl - Leibniz-Zentrum f{\"{u}}r Informatik},
  address      = {Dagstuhl, Germany},
  year         = {2025},
  url          = {https://doi.org/10.4230/LIPIcs.FSCD.2025.6},
  doi          = {10.4230/LIPICS.FSCD.2025.6},
  bibsource    = {dblp computer science bibliography, https://dblp.org}
}

@inproceedings{AvanziniY25,
  author       = {Martin Avanzini and
                  Akihisa Yamada},
  editor       = {Ren{\'{e}} Thiemann and
                  Christoph Weidenbach},
  title        = {Weighted Rewriting},
  booktitle    = {Frontiers of Combining Systems - 15th International Symposium, FroCoS
                  2025, Reykjavik, Iceland, September 29 - October 1, 2025, Proceedings},
  series       = {Lecture Notes in Computer Science},
  volume       = {15979},
  pages        = {191--208},
  publisher    = {Springer},
  address      = {Cham},
  year         = {2025},
  url          = {https://doi.org/10.1007/978-3-032-04167-8\_11},
  doi          = {10.1007/978-3-032-04167-8\_11},
  bibsource    = {dblp computer science bibliography, https://dblp.org}
}

@incollection{Kuich97,
  author       = {Werner Kuich},
  editor       = {Grzegorz Rozenberg and
                  Arto Salomaa},
  title        = {Semirings and Formal Power Series: Their Relevance to Formal Languages
                  and Automata},
  booktitle    = {Handbook of Formal Languages, Volume 1: Word, Language, Grammar},
  pages        = {609--677},
  publisher    = {Springer},
  address      = {Berlin, Heidelberg},
  year         = {1997},
  url          = {https://doi.org/10.1007/978-3-642-59136-5\_9},
  doi          = {10.1007/978-3-642-59136-5\_9},
  bibsource    = {dblp computer science bibliography, https://dblp.org}
}

\clearpage
\newwrite\mainMatterLastPageFile
\immediate\openout\mainMatterLastPageFile=\jobname-main-last-page.tex
\immediate\write\mainMatterLastPageFile{\string\def\string\submissionMainLastPage{\number\numexpr\value{page}-1\relax}}
\immediate\closeout\mainMatterLastPageFile
\appendix
\allowdisplaybreaks[4]

\section{\titleFullref{Additional Theory of }{sec:weighted_programming}}
\label{app:additional-theory}

\begin{definition}[Semiring]
    A \emph{semiring} $\Semiringlong$ is the set $S$ together with two operations and corresponding neutral elements such that
    \begin{itemize}
        \item the structure $(S, \oplus, \snull)$ is a commutative monoid,

        \item the structure $(S, \odot, \sone)$ is a monoid,

        \item the multiplication is distributive over the addition, hence for all $a,b,c \in S$
        $$ a \odot (b \oplus c) = (a \odot b) \oplus (a \odot c) \quad \text{and} \quad (a \oplus b) \odot c = (a \odot c) \oplus (b \odot c), $$

        \item and multiplication by zero annihilates, hence for all $a \in S$
        $$ a \odot \snull = \snull \odot a = \snull. $$
    \end{itemize}
\end{definition}

\begin{lemma} \label{lem:semiring_special_module}
    A semiring $\Semiringlong$ is a special instance of a module $\moduledef$ over the monoid $\monoiddef$, where
    \begin{itemize}
        \item the monoid $\monoiddef$ satisfies $W = S$, $\odot_{\Monoid} = \odot_{\Semiring}$, and $\sone_{\Monoid} = \sone_{\Semiring}$, and

        \item the module $\moduledef$ satisfies $M = S$, $\oplus_{\Module} = \oplus_{\Semiring}$, $\snull_{\Module} = \snull_{\Semiring}$, and $\otimes_{\Module} = \odot$.
    \end{itemize}
\end{lemma}

\begin{proof}
    Due to the definition of semirings, we know that $\moduledef$ is a monoid and $\monoiddef$ is a commutative monoid.
    The scalar multiplication $\otimes: W \times M \to M$ is associative,  distributive, and has the neutral element $\sone$, since the multiplicative operation $\odot_{\Semiring}$ is associative, distributive, and has $\sone$ as neutral element. Furthermore, $\snull$ annihilates, since it also annihilates in the semiring.
\end{proof}

\begin{lemma}[Monotonicity of Scalar Multiplication]
    \label[lemma]{lem:scalar-multiplication-monotone}
    Let $\moduledef$ be a monoid-module over a monoid $\Monoid$.
    For all $a\in\monoid$ and $p,q\in\module$,
    \[
        p\sleq q \qimplies a\stimes p\sleq a\stimes q.
    \]
\end{lemma}

\begin{proof}
    \begin{align*}
        p\sleq q
        &\qimplies \exists m\in\module.\ p\splus m=q
          \inlineExplain{definition of the natural order} \\
        &\qimplies a \stimes (p\splus m) = a \stimes q
          \inlineExplain{scalar multiplication} \\
        &\qimplies (a\stimes p)\splus(a\stimes m)=a\stimes q
          \inlineExplain{distributivity of $\stimes$ over $\splus$} \\
        &\qimplies a\stimes p\sleq a\stimes q
          \inlineExplain{definition of the natural order}.
    \end{align*}
\end{proof}

\begin{lemma}[Monotonicity of Weighted Addition]
    \label[lemma]{lem:weighted-addition-monotone}
    Let $\moduledef$ be a monoid-module over a monoid $\Monoid$.
    For all $c,p,q\in\module$,
    \[
        p\sleq q \qimplies p\splus c\sleq q\splus c.
    \]
\end{lemma}

\begin{proof}
    Assume $p\sleq q$.
    By definition of the natural order, there exists $m\in\module$ such that $p\splus m=q$.
    Hence
    \[
        (p\splus c)\splus m
        =
        (p\splus m)\splus c
        =
        q\splus c,
    \]
    using associativity and commutativity of $\splus$.
    Therefore $p\splus c\sleq q\splus c$.
\end{proof}

\subsection{Operational Semantics}
\label{app:operational-semantics}

We use program states to construct configurations and define the semantics of weighted programs via a weighted relation on the set of configurations.

\begin{definition}[Configuration]
    A configuration is a tuple $(C, \sigma) \in (\wGCL \cup \{\downarrow\}) \times \Sigma =: Q$, %
    where $\downarrow$ indicates termination.
\end{definition}

For a specific configuration $(C, \sigma)$, the weighted program $C$ contains the remaining commands and $\sigma$ is the current program state. %
Now, we define the semantics of a weighted program as a weighted relation between configurations.

\begin{definition}[Operational Semantics of $\wGCL$] \label{def:wGCL_semantics_states}
    The computation relation $\vdash\; \subset Q \times \Monoid \times Q$ is the smallest set satisfying the rules in \Cref{fig:semantics}.
\end{definition}

\begin{figure}
    \begin{minipage}{0.3\textwidth}
        \begin{align*}
        \renewcommand{\arraystretch}{1.3}%
            \begin{array}{l!{\color{Platinum}\vrule width 1.5pt}l}
                \toprule
                \aexpr & \sem{\aexpr}(\sigma) \in \Nats \\ \midrule
                \cellcolor{basicBackground!2}{n \in \Nats} & \cellcolor{basicBackground!2}{n} \\
                \cellcolor{basicBackground!2}{x \in \Vars} & \cellcolor{basicBackground!2}{\sigma(x)} \\
                \cellcolor{basicBackground!2}{\aexpr + \aexpr'} & \cellcolor{basicBackground!2}{\sem{\aexpr}(\sigma) + \sem{\aexpr'}(\sigma)} \\
                \cellcolor{basicBackground!2}{\aexpr \cdot \aexpr'} & \cellcolor{basicBackground!2}{\sem{\aexpr}(\sigma) \cdot \sem{\aexpr'}(\sigma)} \\
                \cellcolor{basicBackground!2}{\aexpr \monus \aexpr'} & \cellcolor{basicBackground!2}{\sem{\aexpr}(\sigma) \monus \sem{\aexpr'}(\sigma)} \\ \bottomrule
            \end{array}
        \end{align*}
    \end{minipage}
    \begin{minipage}{0.5\textwidth}
        \begin{align*}
        \renewcommand{\arraystretch}{1.2}%
            \begin{array}{l!{\color{Platinum}\vrule width 1.5pt}l}
                \toprule
                \bexpr & \sem{\bexpr}(\sigma) \in \{0,1\} \\ \midrule
                \cellcolor{basicBackground!2}{\aexpr < \aexpr'} & \cellcolor{basicBackground!2}{\ifThenElse{\sem{\aexpr}(\sigma) < \sem{\aexpr'}(\sigma)}{1}{0}} \\[1em]
                \cellcolor{basicBackground!2}{\neg \bexpr} & \cellcolor{basicBackground!2}{\ifThenElse{\sem{\bexpr}(\sigma) = 0}{1}{0}} \\[1em]
                \cellcolor{basicBackground!2}{\bexpr \wedge \bexpr'} & \cellcolor{basicBackground!2}{\ifThenElse{\sem{\bexpr}(\sigma) + \sem{\bexpr'}(\sigma) \geq 1}{1}{0}} \\ \bottomrule
            \end{array}
        \end{align*}

    \end{minipage}

    \caption{Semantics of arithmetic and Boolean expressions \ArithExp and \BExp, where monus $\monus$ is subtraction truncated at 0.}
    \label{fig:semantics_arith_bool}
\end{figure}

\begin{figure}
    { \footnotesize
    \begin{minipage}{0.49\textwidth}
        \centering
        \begin{prooftree}
            \hypo{ \sigma' = \sigma [ x \mapsto \sigma(\bexpr)] }
            \infer1[(assign)]{\config{
                \ASSIGN{x}{\bexpr}, \sigma
            } \vdash_{\sonetiny} \config{
                \downarrow, \sigma'
            }}
        \end{prooftree}\\[5pt]

        \begin{prooftree}
            \hypo{\config{
                C_1, \sigma
            } \vdash_u \config{
                \downarrow, \sigma'
            }}
            \infer1[(seq. 1)]{\config{
                C_1 \semcol C_2, \sigma
            } \vdash_u \config{
                C_2, \sigma
            }}
        \end{prooftree}\\[5pt]

        \begin{prooftree}
            \hypo{ \sigma \models \bexpr }
            \infer1[(if)]{\config{
                \ITE{\bexpr}{C_1}{C_2}, \sigma
            } \vdash_{\sonetiny} \config{
                C_1, \sigma
            }}
        \end{prooftree}\\[10pt]

        \begin{prooftree}
            \hypo{}
            \infer1[(left branch)]{\config{
                \BRANCH{C_1}{C_2}, \sigma
            } \vdash_{\sonetiny} \config{
                C_1, \sigma
            }}
        \end{prooftree}\\[5pt]

        \begin{prooftree}
            \hypo{ \sigma \models \bexpr}
            \infer1[(while)]{\config{
                \WHILEDO{\bexpr}{C}, \sigma
            } \vdash_{\sonetiny} \config{
                C \semcol \WHILEDO{\bexpr}{C}, \sigma
            }}
        \end{prooftree}
    \end{minipage}
    \begin{minipage}{0.49\textwidth}
        \centering
        \begin{prooftree}
            \hypo{u = \sigma(\mexpr)}
            \infer1[(weight)]{\config{
                \WEIGH{\mexpr}, \sigma
            } \vdash_u \config{
                \downarrow, \sigma
            }}
        \end{prooftree}\\[5pt]
        \begin{prooftree}
            \hypo{\config{
                C_1, \sigma
            } \vdash_u \config{
                C_1', \sigma'
            } \text{ and } C_1' \neq \downarrow}
            \infer1[(seq. 2)]{\config{
                C_1 \semcol C_2, \sigma
            } \vdash_u \config{
                C_1' \semcol C_2, \sigma
            }}
        \end{prooftree}\\[5pt]
        \begin{prooftree}
            \hypo{ \sigma \models \neg\bexpr}
            \infer1[(else)]{\config{
                \ITE{\bexpr}{C_1}{C_2}, \sigma
            } \vdash_{\sonetiny} \config{
                C_2, \sigma
            }}
        \end{prooftree}\\[10pt]
        \begin{prooftree}
            \hypo{}
            \infer1[(right branch)]{\config{
                \BRANCH{C_1}{C_2}, \sigma
            } \vdash_{\sonetiny} \config{
                C_2, \sigma
            }}
        \end{prooftree}\\[5pt]
        \begin{prooftree}
            \hypo{ \sigma \models \neg\bexpr}
            \infer1[(break)]{\config{
                \WHILEDO{\bexpr}{C}, \sigma
            } \vdash_{\sonetiny} \config{
                \downarrow, \sigma
            }}
        \end{prooftree}
    \end{minipage}
    }

    \caption{The structural operational semantics for $\wGCL$-programs from \cite{BatzGKKW22} without the computation counter and the branching history and extended by weighted expressions.}
    \label{fig:semantics}
    \Description{The structural operational semantics defined via proof rules.}
\end{figure}

\begin{definition}[Weakest Liberal Preweightings] \label{def:wlp}
    For an $\omega$-cocontinuous module $\Module$ over the monoid $\Monoid$, we define the weakest \emph{liberal} preweightings iteratively as the weighting transformer $\symWlp: \wGCL \to (\Weighting \to \Weighting)$
    via the rules in \Cref{fig:weakest_pre}.
\end{definition}

The weakest liberal preweighting of a loop is defined via the greatest fixed point of the characteristic function $\wlpPhi: \Weighting \to \Weighting, X \mapsto [\neg\bexpr] \wla \oplus [\bexpr] \wlp{C}(X)$.

\begin{theorem} \label{thm:weakest_liberal_pre}
    The weighting transformer $\wlp{C}(\cdot)$ is a well-defined, $\omega$-cocontinuous function. In particular, the greatest fixed point of the function $\wlpPhi$ exists for all $w\in \Weighting$ and $C \in \wGCL$.
\end{theorem}

The proof of this theorem can be found in \cite{BatzGKKW22}.

\section{\titleFullref{Theory and Proofs of }{sec:heyting}}

\label{app:heyting}

\begin{definition}[Bounded Lattice] \label{def:bounded_lattice}
    A partially ordered set $(S, \hleq)$ is a \emph{bounded lattice}, if
    \begin{itemize}
        \item it is a \emph{lattice}, hence there exist greatest lower bounds $a \hand b \in S$ and least upper bounds $a \hor b \in S$ for all $a, b\in S$ with regard to the order $\hleq$, and

        \item there exists a greatest element $\top \in S$ and least element $\bot \in S$ such that $a \hand \top = a$ and $a \hor \bot = a$ for all $a \in S$.
    \end{itemize}
\end{definition}

Both the meet $\hand$ and join $\hor$ are commutative ($a \hand b = b \hand a$), associative ($a \hand (b \hand c) = (a \hand b) \hand c$), and absorbing ($a \hand (a \hor b) = a$ for $a, b, c \in S$).

\begin{lemma}
    Meet and join are distributive in Heyting algebras.
\end{lemma}
\begin{proof}
    See \cite{StoneSpaces}.
\end{proof}

\begin{lemma}
    Let $\Module$ be an $\omega$-bicontinuous monoid-module. If the natural order
    $\sleq$ on $\module$ admits binary meets and joins, then $(\module,\sleq)$ is a
    bounded lattice.
\end{lemma}

\begin{proof}
    By $\omega$-continuity, $\sleq$ is a partial order with a least element, which
    we denote by $\bot$. By $\omega$-cocontinuity, the reversed order $\sgeq$ has a
    least element. Equivalently, $\sleq$ has a greatest element, which we denote by
    $\top$.

    By assumption, all binary meets and joins with respect to $\sleq$ exist. Hence
    $(\module,\sleq)$ is a lattice with least element $\bot$ and greatest element
    $\top$, and therefore a bounded lattice.
\end{proof}

\thmDeductionTheorem*
\begin{proof}
\leavevmode\par\smallskip
    \begin{enumerate}
        \item For the first statement, instantiate the Heyting adjunction with $x=\top$:
        \[
            a \hand \top \hleq b \quad\text{iff}\quad \top \hleq a \impl b .
        \]
        Since $a \hand \top = a$ and $\top \hleq y$ holds iff $y=\top$, this yields
        $a \hleq b$ iff $a \impl b = \top$.
        \item Dually, instantiate the co-Heyting adjunction with $x=\bot$:
        \[
            b \hleq a \hor \bot \quad\text{iff}\quad a \coimpl b \hleq \bot .
        \]
        Since $a \hor \bot = a$ and $y \hleq \bot$ holds iff $y=\bot$, this yields
        $a \hgeq b$ iff $a \coimpl b = \bot$.
    \end{enumerate}
\end{proof}

\lemBoolEmbeddingImplCoimpl*
\begin{proof}
    By the adjunction laws of bi-Heyting algebras (\abCref{def:heyting_algebra}), for every $x \in \hdom$,
    \[
        \top \impl x = x,\qquad
        \bot \impl x = \top,\qquad
        \top \coimpl x = \bot,\qquad
        \bot \coimpl x = x .
    \]
    This follows immediately, e.g. $a \hleq \top \impl x$ iff $a=\top\hand a\hleq x$, and
    $(\bot \coimpl x)\hleq a$ iff $x\hleq \bot\hor a=a$. The other two
    identities are exactly given in \Cref{thm:deduction_theorem}.
\end{proof}

Recall \Cref{lem:impl_as_residual} from \cpageref{lem:impl_as_residual}:
\lemImplAsResidual*
\begin{proof}
\leavevmode
\noindent($\Rightarrow$)
\begin{quote}
    For $a, b \in \hdom$, we show $a \impl b$ is the greatest element of $S := \Set{x\in\hdom \mid a\hand x\hleq b}$.
    \begin{enumerate}
        \item Since $a \impl b \hleq a \impl b$, we have $a \hand (a \impl b) \hleq b$ by \fullref{def:heyting_algebra}, hence $a \impl b \in S$.
        \item Let $x \in S$ be arbitrary, i.e. $ a \hand x \hleq b$. By \Cref{def:heyting_algebra}, this gives $x \hleq a \impl b$.
    \end{enumerate}
    Since $a \impl b \in S$ and every $x \in S$ satisfies $x \hleq a \impl b$, we conclude that $a \impl b$ is the greatest element of $S$.
    Since a greatest element of a subset of a poset is always its supremum, $\bigsqcup S$ exists and $\bigsqcup S = a\impl b$, as claimed.
\end{quote}

\noindent($\Leftarrow$)
\begin{quote}
    For $a,b\in\hdom$, define
    \(
        a\impl b := \max\{x\in\hdom \mid a\hand x \hleq b\},
    \)
    which is well-defined because the assumption guarantees the max exists, and it is unique since $\hleq$ is
    antisymmetric. We check that $\impl$ satisfies the defining adjunction: for all
    $a,b,x\in\hdom$,
    \[
        a\hand x \hleq b \iff x\hleq a\impl b.
    \]
    \begin{description}[style=unboxed,leftmargin=2.5em,labelsep=0.4em,font=\normalfont]
        \item[($\Rightarrow$)] If $a\hand x\hleq b$, then $x\in\{x\mid a\hand x\hleq b\}$, so $x\hleq a\impl b$, as $a\impl b$ is an upper bound of this set.
        \item[($\Leftarrow$)] If $x\hleq a\impl b$, then since $a\impl b$ itself lies in $\{x\mid a\hand x\hleq b\}$, we have $a\hand(a\impl b)\hleq b$. As $\hand$ is monotone, $x\hleq a\impl b$ gives $a\hand x \hleq a\hand(a\impl b) \hleq b$.
    \end{description}

    Hence $\impl$ satisfies the adjunction property, so $(\hdom,\hleq,\hand,\hor,\impl)$ is a Heyting algebra
\end{quote}
\end{proof}

\begin{lemma}[Join Below Weighted Addition]
    \label[lemma]{lem:join-below-weighted-addition}
    Let $\moduledef$ be a monoid-module over a monoid $\Monoid$ whose natural order admits the binary join $p\hor q$.
    Then
    \[
        p\hor q \sleq p\splus q.
    \]
\end{lemma}

\begin{proof}
    We have $p\sleq p\splus q$ and $q\sleq p\splus q$ by the definition of the natural order.
    Hence $p\splus q$ is an upper bound of $p$ and $q$, so the \emph{least} upper bound must satisfy $p\hor q \sleq p\splus q$.
\end{proof}

Recall \Cref{thm:totimpl_totcoimpl_bounded_lattice} from \cpageref{thm:totimpl_totcoimpl_bounded_lattice}:

\thmTotimplTotcoimplBoundedLattice*

\begin{proof}
    We show that $\hsym = (\hdom, \hleq, \hand, \hor, \totimpl, \totcoimpl)$ is a bi-Heyting algebra by showing that $(\hdom, \hleq, \hand, \hor, \totimpl)$ and $(\hdom, \hleq, \hand, \hor, \totcoimpl)$ are (co-)Heyting algebras via the statement in \Cref{lem:impl_as_residual}.

    For all $a, b \in \hdom$ with $a \hleq b$, we observe that $x := (a \totimpl b) = \top$. This is the greatest element in $\hdom$ and satisfies $a \hand x = a \hand \top = a \hleq b$. For $a \not\hleq b$, we obtain $x := (a \totimpl b) = b$ and the inequality $a \hand x = a \hand b = b \hleq b$ is satisfied. There exists no greater element $b' \in \hdom$ such that $a \hand b' \hleq b$, since then $b' \not\hleq b$ and $b \not\hleq b'$ hold, which is a contradiction. Thus $(\hdom, \hleq, \hand, \hor, \totimpl)$ is a Heyting algebra.

    Now we consider $a, b \in \hdom$ with $b \hleq a$ and observe that $x := (a \totcoimpl b) = \bot$. This is the least element in $\hdom$ and satisfies $b \hleq a = a \hor \bot = a \hor x$. For $b \not\hleq a$, we obtain $x := (a \totcoimpl b) = b$, which satisfies $b \hleq b = a \hor b = a \hor x$. Assume that there exists a smaller element $b' \in \hdom$ that satisfies the inequality $b \hleq b' = a \hor b'$. This is a contradiction with $b \not\hleq b'$ and hence $a \totcoimpl b$ always yields the least element that satisfies $b \hleq a \hor x$ and $(\hdom, \hleq, \hand, \hor, \totcoimpl)$ is a co-Heyting algebra.
\end{proof}

\thmMatimplMatcoimplBoundedLattice*

\begin{proof}
    We verify the Heyting and co-Heyting adjunctions directly.

    \begin{description}[style=unboxed,leftmargin=0pt,labelsep=0.5em,font=\normalfont]
        \item[Implication.]
        Define $a \matimpl b := \neg a \hor b$.
        We show that, for all $a,b,x\in\hdom$,
        \[
            a\hand x \hleq b
            \quad\Longleftrightarrow\quad
            x \hleq \neg a \hor b .
        \]

        First assume $a\hand x \hleq b$. Since $a\hor\neg a=\top$, distributivity gives
        \[
            x
            = x\hand\top
            = x\hand(a\hor\neg a)
            = (x\hand a)\hor(x\hand\neg a)
            \hleq b\hor\neg a
            = \neg a\hor b .
        \]
        Conversely, assume $x\hleq \neg a\hor b$. Then
        \[
            a\hand x
            \hleq a\hand(\neg a\hor b)
            = (a\hand\neg a)\hor(a\hand b)
            = \bot\hor(a\hand b)
            = a\hand b
            \hleq b .
        \]
        Hence $a\matimpl b=\neg a\hor b$ satisfies the Heyting adjunction.

        \item[Co-implication.]
        Define $a \matcoimpl b := \neg a \hand b$.
        We show that, for all $a,b,x\in\hdom$,
        \[
            b \hleq a\hor x
            \quad\Longleftrightarrow\quad
            \neg a\hand b \hleq x .
        \]

        First assume $b\hleq a\hor x$. Then
        \[
            \neg a\hand b
            \hleq \neg a\hand(a\hor x)
            = (\neg a\hand a)\hor(\neg a\hand x)
            = \bot\hor(\neg a\hand x)
            = \neg a\hand x
            \hleq x .
        \]
        Conversely, assume $\neg a\hand b\hleq x$. Since $a\hor\neg a=\top$, distributivity gives
        \[
            b
            = b\hand\top
            = b\hand(a\hor\neg a)
            = (b\hand a)\hor(b\hand\neg a)
            \hleq a\hor x .
        \]
        Hence $a\matcoimpl b=\neg a\hand b$ satisfies the co-Heyting adjunction.
    \end{description}

    Therefore $\hsym=(\hdom,\hleq,\hand,\hor,\matimpl,\matcoimpl)$ is a bi-Heyting algebra.
\end{proof}

\thmWhySemiringBiHeyting*

\begin{proof}
    Since $\pDNF(A)$ is defined modulo logical equivalence over the finite atom set $A$, it is finite.
    The propositional operations $\wedge$ and $\vee$ respect logical equivalence, so they are well-defined on $\pDNF(A)$.
    Ordered by entailment, $\pDNF(A)$ is a finite distributive lattice with meet $\wedge$, join $\vee$, bottom $0$, and top $1$.

    In particular, $X = \bigvee \Set{R \in \pDNF(A) \mid P \wedge R \hleq Q}$ exists for all $P, Q \in \pDNF(A)$, since $\pDNF(A)$ is finite.
    What remains to be shown is that $X$ is indeed the greatest element satisfying the condition, i.e. $P \wedge X \hleq Q$. By finite distributivity,
    \[
        P \wedge X
        =
        \bigvee \Set{P \wedge R \mid R \in \pDNF(A),\; P \wedge R \hleq Q}
        \hleq Q.
    \]
    Thus $X$ is the greatest element $Z$ satisfying $P \wedge Z \hleq Q$. This proofs that $(\pDNF(A), \hleq, \wedge, \vee, \impl)$ is a Heyting algebra by \Cref{lem:impl_as_residual}.

    To show the dual statement, define
    \(
        Y = \bigwedge \Set{R \in \pDNF(A) \mid Q \hleq P \vee R}.
    \)
    The same argument as above shows that $Y$ is the least element satisfying $Q \hleq P \vee Y$, so $(\pDNF(A), \hleq, \wedge, \vee, \coimpl)$ is a co-Heyting algebra.
    Combining both statements, we conclude that $(\pDNF(A), \hleq, \wedge, \vee, \impl, \coimpl)$ is a bi-Heyting algebra.
\end{proof}

\begin{lemma} \label[lemma]{lem:countable_sup_inf}
    Let the $\omega$-bicontinuous module $\moduledef$ over the monoid $\monoiddef$ be also a bi-Heyting algebra. Then for all countable $A \subseteq \Module$, the supremum $\hbigor A \in \Module$ and infimum $\hbigand A \in \Module$ exists.
\end{lemma}

\begin{proof}
    If $A=\emptyset$, then $\hbigor A=\bot$ and $\hbigand A=\top$ exist because the bi-Heyting algebra is bounded.
    Hence assume $A$ is nonempty and choose an enumeration $(a_i)_{i \in \Nats}$ of $A$.
    We construct a sequence $(b_i)_{i \in \Nats} \subseteq \Module$ such that
    $$ \hbigor A = \hbigor_{i \in \Nats} a_i = \hbigor_{i \in \Nats} b_i. $$

    Define $b_0 := a_0$ and $b_{i + 1} := b_i \hor a_{i + 1} \in \Module$, which is well-defined since the supremum of two elements exists in a Heyting algebra. We observe that $(b_i)_{i \in \Nats}$ is an ascending $\omega$-chain and hence $L = \hbigor_{i \in \Nats} b_i \in \Module$ exists. This implies that $\hbigor A \in \Module$ exists.

    We further show that $L$ is the supremum of $A$. First, $L$ is an upper bound of all $a_i$ since $a_i \hleq a_0 \hor \ldots \hor a_i = b_i \hleq L$ for all $i \in \Nats$. Second, $L$ is the least such upper bound. Let $L'$ be an upper bound of $A$, so $a_i \hleq L'$ for all $i \in \Nats$. We show by induction that $b_n \hleq L'$ for all $n \in \Nats$. As a base case, we have $b_0 = a_0 \hleq L'$. For the induction step, we assume that $b_n \hleq L'$ and further note that by definition $a_{n + 1} \hleq L'$. Hence $L'$ is an upper bound of both $b_n$ and $a_{n + 1}$, so $b_{n + 1} = b_n \hor a_{n + 1} \hleq L'$. Thus, $L'$ is an upper bound of the chain $(b_i)_{i \in \Nats}$. Since $L$ is the least upper bound of this chain, we conclude that $L \hleq L'$. So $L$ is the least upper bound of $A$. Consequently, $\hbigor_{i \in \Nats} a_i = L$.
    The proof for infima is analogous.
\end{proof}

\thmWHeyLoWellDefined*

\begin{proof}
    To show well-definedness, we need to prove that the formulae $\hbigand_{x}{\hhla}$ (and $\hbigor_{x}{\hhla}$) evaluate to elements in the module $\Module$ for all formulae $\hhla \in \wHeyLo$. Since the module is $\omega$-bicontinuous, we only know that suprema (and infima) of ascending (and descending) $\omega$-chains are well-defined.

    However, \Cref{lem:countable_sup_inf} states that suprema (and infima) of countable sets exist in $\omega$-bicontinuous monoid-module that are bi-Heyting algebras. Hence we need to show that the set
    $$
        \left\{ \sem{\hhla}(\sigma[x \mapsto n]) \mid n \in \Nats \right\}
    $$
    is countable for all $\hhla\in \wHeyLo$. This holds directly, since the natural numbers $\Nats$ are countable.
\end{proof}

Heyting algebras can also be defined in an axiomatic way.
\begin{theorem}[Heyting Algebra Axioms]
    \label[theorem]{thm:heyting2}
    A bounded lattice $(H, \sqsubseteq ,\sqcap, \sqcup)$ with a binary implication $\to: H \times H \to H$ is a Heyting algebra $\hdef$ if and only if the following conditions hold for all $a, b, c \in H$:
    \begin{enumerate}
        \item $ (a \to a) = \top$,
        \item $ a \sqcap (a \to b) = a \sqcap b$,
        \item $ b \sqcap (a \to b) = b$,
        \item $ a \to (b \sqcap c) = (a \to b) \sqcap (a \to c)$.
    \end{enumerate}
\end{theorem}
\begin{proof}
    See \cite{StoneSpaces}.
\end{proof}

\begin{restatable}[Properties of Bi-Heyting Algebras]{proposition}{propBiHeytingAlgebraProperties}
    \label[prop]{prop:bi-heyting:properties}
    Let $\hbidef$ be a bi-Heyting algebra. Then the following properties hold for all $a, b, c \in \hdom$:
    \begin{enumerate}
        \item $a \hleq b \iff a \to b = \top$,
        \itemExplain{$\impl$ internalizes $\hleq$}
        \label[statement]{prop:bi-heyting:1}
        \item $a \hleq b \iff b \coimpl a = \bot$, \itemExplain{$\coimpl$ internalizes $\hleq$}
        \label[statement]{prop:bi-heyting:1b}
        \item $b \hleq c \implies a \to b \hleq a \to c$, \itemExplain{$a\to(-)$ is monotone} \label[implication]{prop:bi-heyting:2}
        \item $b \hleq c \implies a \coimpl b \hleq a \coimpl c$, \itemExplain{$a\coimpl(-)$ is monotone}
        \label[implication]{prop:bi-heyting:2b}
        \item $b \hleq c \implies c \impl a \hleq b \impl a$, \itemExplain{$(-)\impl a$ is antitone}
        \label[implication]{prop:bi-heyting:3}
        \item $b \hleq c \implies c \coimpl a \hleq b \coimpl a$, \itemExplain{$(-)\coimpl a$ is antitone}
        \label[implication]{prop:bi-heyting:3b}
    \end{enumerate}
\end{restatable}
\begin{proof}
\leavevmode
\begin{description}[style=unboxed,leftmargin=0pt,labelsep=0.4em,font=\normalfont]
    \item[\Cref{prop:bi-heyting:1}.] \emph{Proof that \( a \hleq b \qiff a \rightarrow b = \top \):}
    \begin{description}[style=unboxed,leftmargin=0pt,labelsep=0.4em,font=\normalfont]
        \item[] ($\Rightarrow$) Assume \( a \hleq b \). Then \( a \sqcap \top = a \hleq b \).
        Using \Cref{lem:impl_as_residual}
        \[
        a \rightarrow b = \sup \big\{ x \in H \mid a \sqcap x \hleq b \big\} = \top
        \]
        holds, since $\top$ is in the set and is the greatest element in $H$.

        \item[] ($\Leftarrow$) Assume \( a \rightarrow b = \top \). By the definition of the implication
        \[
        a \sqcap \top \hleq b \implies a \hleq b.
        \]
    \end{description}
    The proofs for all the co-Heyting properties follow dually. As an example, we show the proof for \Cref{prop:bi-heyting:1b}, but omit the rest.
    \item[\Cref{prop:bi-heyting:1b}.] \emph{Proof that \( a \hleq b  \qiff  b \leftarrow a = \bot \):}
    \begin{description}[style=unboxed,leftmargin=0pt,labelsep=0.4em,font=\normalfont]
        \item[] ($\Rightarrow$) Assume \( a \hleq b \). Then \( a \hleq \bot \hor b \).
        Using the dual of \Cref{lem:impl_as_residual}:
        \[
        b \leftarrow a = \inf \big\{ x \in H \mid a \hleq x \sqcup b \big\} = \bot
        \]
        holds, since $\bot$ is in the set and is the least element in $H$.

        \item[] ($\Leftarrow$) Assume \( b \leftarrow a = \bot \). By definition of the co-implication:
        \[
        b \leftarrow a \hleq \bot \implies a \hleq b \sqcup \bot \implies a \hleq b.
        \]
    \end{description}
    \item[\Cref{prop:bi-heyting:2,prop:bi-heyting:2b}.] \emph{Proof that \( b \hleq c \implies a \to b \hleq a \to c \):}
    \begin{description}[style=unboxed,leftmargin=0pt,labelsep=0.4em,font=\normalfont]
        \item[] ($\Rightarrow$) Assume \( b \hleq c \). Then
        \(
            a \hand (a \impl b) = a \hand b \hleq b \hleq c.
        \)
        Together with \Cref{lem:impl_as_residual}, we have:
        \[
            a \impl c = \sup \big\{ x \in H \mid a \hand x \hleq c \big\} \hgeq a \impl b.
        \]
    \end{description}
    \item[\Cref{prop:bi-heyting:3,prop:bi-heyting:3b}.] \emph{Proof that \( b \hleq c \implies c \impl a \hleq b \impl a \):}
    \begin{description}[style=unboxed,leftmargin=0pt,labelsep=0.4em,font=\normalfont]
         \item[] ($\Rightarrow$) Assume \( b \hleq c \). Then
        \(
            b \hand (c \impl a) \hleq c \hand (c \impl a) \myrel{\abCref{thm:heyting2}}{=} c \impl a \hleq a.
        \)
        Again, we can use \Cref{lem:impl_as_residual}:
        \[
            b \impl a = \sup \big\{ x \in H \mid b \hand x \hleq a \big\} \hgeq c \impl a.
        \]
    \end{description}
\end{description}
\end{proof}

\begin{lemma}[Order Characterization of Existing Suprema and Infima]
    \label[lemma]{lem:existing-extrema-order-characterization}
    Let $(\hdom,\hleq)$ be a partially ordered set, let $I$ be an index set, and let $\{a_i\}_{i \in I}$ be a family of elements in $\hdom$.
    If the supremum $\hbigor_{i \in I} a_i$ exists, then for every $c \in \hdom$,
    \[
        \hbigor_{i \in I} a_i \hleq c \qqiff \forall i \in I.\ a_i \hleq c.
    \]
    If the infimum $\hbigand_{i \in I} a_i$ exists, then for every $c \in \hdom$,
    \[
        c \hleq \hbigand_{i \in I} a_i \qqiff \forall i \in I.\ c \hleq a_i.
    \]
\end{lemma}

\begin{proof}
    We prove the supremum case; the infimum case is dual.
    Let $s := \hbigor_{i \in I} a_i$ and let $c \in \hdom$.
    If $s \hleq c$, then $a_i \hleq s \hleq c$ for every $i \in I$, because $s$ is an upper bound of the family.
    Conversely, if $a_i \hleq c$ for every $i \in I$, then $c$ is an upper bound of the family.
    Since $s$ is the least upper bound, $s \hleq c$.
\end{proof}

\begin{lemma}[Lattice Preservation of Existing Extrema]
    \label[lemma]{lem:lattice-preservation-existing-extrema}
    Let $\hbidef$ be a bi-Heyting algebra.
    Let $I$ be an index set, let $\{a_i\}_{i \in I}$ be a family of elements in $\hdom$, and let $c \in \hdom$.
    If the joins $\hbigor_{i \in I} a_i$ and $\hbigor_{i \in I}(c \hand a_i)$ exist, then
    \[
        c \hand \hbigor_{i \in I} a_i = \hbigor_{i \in I}(c \hand a_i).
    \]
    If the meets $\hbigand_{i \in I} a_i$ and $\hbigand_{i \in I}(c \hor a_i)$ exist, then
    \[
        c \hor \hbigand_{i \in I} a_i = \hbigand_{i \in I}(c \hor a_i).
    \]
\end{lemma}

\begin{proof}
    We first prove preservation of existing joins by meet.
    For every $d \in \hdom$,
    \begin{align*}
        c \hand \hbigor_{i \in I} a_i \hleq d
        &\qiff \hbigor_{i \in I} a_i \hleq c \impl d
            \inlineExplain{adjunction for $\impl$} \\
        &\qiff \forall i \in I.\ a_i \hleq c \impl d
            \inlineExplain{\Cref{lem:existing-extrema-order-characterization}} \\
        &\qiff \forall i \in I.\ c \hand a_i \hleq d
            \inlineExplain{adjunction for $\impl$} \\
        &\qiff \hbigor_{i \in I}(c \hand a_i) \hleq d
            \inlineExplain{\Cref{lem:existing-extrema-order-characterization}}.
    \end{align*}
    We now prove preservation of existing meets by join.
    For every $d \in \hdom$,
    \begin{align*}
        d \hleq c \hor \hbigand_{i \in I} a_i
        &\qiff c \coimpl d \hleq \hbigand_{i \in I} a_i
            \inlineExplain{adjunction for $\coimpl$} \\
        &\qiff \forall i \in I.\ c \coimpl d \hleq a_i
            \inlineExplain{\Cref{lem:existing-extrema-order-characterization}} \\
        &\qiff \forall i \in I.\ d \hleq c \hor a_i
            \inlineExplain{adjunction for $\coimpl$} \\
        &\qiff d \hleq \hbigand_{i \in I}(c \hor a_i)
            \inlineExplain{\Cref{lem:existing-extrema-order-characterization}}.
    \end{align*}
    Since both statements hold for every $d \in \hdom$, reflexivity and antisymmetry yields the claim.
\end{proof}

\begin{restatable}[Meet Preservation of Implication]{lemma}{lemImplPreservesMeet}
    \label[lemma]{lem:impl_preserves_meet}
    Let $\hdef$ be a Heyting algebra. Let $I$ be an arbitrary index set and let $\{b_i\}_{i \in I}$ be a family of elements in $\hdom$ indexed by $I$. For any element $a \in \hdom$, if the meet $\hbigand_{i \in I} b_i$ and $\hbigand_{i \in I} (a \impl b_i)$ exist, then
    \[
        a \impl \hbigand_{i \in I} b_i = \hbigand_{i \in I} (a \impl b_i).
    \]
\end{restatable}
\begin{proof}
    For every $c \in \hdom$,
    \begin{align*}
        c \hleq a \impl \hbigand_{i \in I} b_i
        &\qiff a \hand c \hleq \hbigand_{i \in I} b_i
            \inlineExplain{adjunction for $\impl$} \\
        &\qiff \forall i \in I.\ a \hand c \hleq b_i
            \inlineExplain{\Cref{lem:existing-extrema-order-characterization}} \\
        &\qiff \forall i \in I.\ c \hleq a \impl b_i
            \inlineExplain{adjunction for $\impl$} \\
        &\qiff c \hleq \hbigand_{i \in I} (a \impl b_i)
            \inlineExplain{\Cref{lem:existing-extrema-order-characterization}}.
    \end{align*}
    Since this holds for every $c \in \hdom$, reflexivity and antisymmetry yields the claim.
\end{proof}

\begin{lemma}[Join Preservation of Coimplication]
    \label[lemma]{lem:coimpl_preserves_join}
    Let $\hbidef$ be a bi-Heyting algebra.
    Let $I$ be an index set, let $\{b_i\}_{i \in I}$ be a family of elements in $\hdom$, and let $a \in \hdom$.
    If the joins $\hbigor_{i \in I} b_i$ and $\hbigor_{i \in I} (a \coimpl b_i)$ exist, then
    \[
        a \coimpl \hbigor_{i \in I} b_i = \hbigor_{i \in I} (a \coimpl b_i).
    \]
\end{lemma}

\begin{proof}
    For every $c \in \hdom$,
    \begin{align*}
        \hbigor_{i \in I} (a \coimpl b_i) \hleq c
        &\qiff \forall i \in I.\ a \coimpl b_i \hleq c
            \inlineExplain{\Cref{lem:existing-extrema-order-characterization}} \\
        &\qiff \forall i \in I.\ b_i \hleq a \hor c
            \inlineExplain{adjunction for $\coimpl$} \\
        &\qiff \hbigor_{i \in I} b_i \hleq a \hor c
            \inlineExplain{\Cref{lem:existing-extrema-order-characterization}} \\
        &\qiff a \coimpl \hbigor_{i \in I} b_i \hleq c
            \inlineExplain{adjunction for $\coimpl$}.
    \end{align*}
    Since this holds for every $c \in \hdom$, reflexivity and antisymmetry yields the claim.
\end{proof}

\begin{lemma}[Implication Turns Joins into Meets]
    \label[lemma]{lem:quant-implication-distributes-over-inf}
    Let $\hdef$ be a Heyting algebra. Let $I$ be an arbitrary index set, let $\{a_i\}_{i \in I}$ be a family of elements in $\hdom$, and let $b \in \hdom$.
    If the join $\hbigor_{i \in I} a_i$ and the meet $\hbigand_{i \in I}(a_i \impl b)$ exist, then
    \[
        \left(\hbigor_{i \in I} a_i\right) \impl b = \hbigand_{i \in I} (a_i \impl b).
    \]
\end{lemma}

\begin{proof}
    For every $c \in \hdom$,
    \begin{align*}
        c \hleq \left(\hbigor_{i \in I} a_i\right) \impl b
        &\qiff \left(\hbigor_{i \in I} a_i\right) \hand c \hleq b
            \inlineExplain{adjunction for $\impl$} \\
        &\qiff c \hand \left(\hbigor_{i \in I} a_i\right) \hleq b
            \inlineExplain{commutativity of $\hand$} \\
        &\qiff \hbigor_{i \in I} a_i \hleq c \impl b
            \inlineExplain{adjunction for $\impl$} \\
        &\qiff \forall i \in I.\ a_i \hleq c \impl b
            \inlineExplain{\Cref{lem:existing-extrema-order-characterization}} \\
        &\qiff \forall i \in I.\ c \hand a_i \hleq b
            \inlineExplain{adjunction for $\impl$} \\
        &\qiff \forall i \in I.\ a_i \hand c \hleq b
            \inlineExplain{commutativity of $\hand$} \\
        &\qiff \forall i \in I.\ c \hleq a_i \impl b
            \inlineExplain{adjunction for $\impl$} \\
        &\qiff c \hleq \hbigand_{i \in I}(a_i \impl b)
            \inlineExplain{\Cref{lem:existing-extrema-order-characterization}}.
    \end{align*}
    Since $c \in \hdom$ was arbitrary, reflexivity and antisymmetry yields the claim.
\end{proof}

\begin{lemma}[Coimplication Turns Meets into Joins]
    \label[lemma]{lem:coimpl_turns_meets_into_joins}
    Let $\hbidef$ be a bi-Heyting algebra.
    Let $I$ be an index set, let $\{a_i\}_{i \in I}$ be a family of elements in $\hdom$, and let $b \in \hdom$.
    If the meet $\hbigand_{i \in I} a_i$ and the join $\hbigor_{i \in I} (a_i \coimpl b)$ exist, then
    \[
        \left(\hbigand_{i \in I} a_i\right) \coimpl b = \hbigor_{i \in I} (a_i \coimpl b).
    \]
\end{lemma}

\begin{proof}
    For every $c \in \hdom$,
    \begin{align*}
        \hbigor_{i \in I} (a_i \coimpl b) \hleq c
        &\qiff \forall i \in I.\ a_i \coimpl b \hleq c
            \inlineExplain{\Cref{lem:existing-extrema-order-characterization}} \\
        &\qiff \forall i \in I.\ b \hleq a_i \hor c
            \inlineExplain{adjunction for $\coimpl$} \\
        &\qiff b \hleq \hbigand_{i \in I} (a_i \hor c)
            \inlineExplain{\Cref{lem:existing-extrema-order-characterization}} \\
        &\qiff b \hleq \left(\hbigand_{i \in I} a_i\right) \hor c
            \inlineExplain{\Cref{lem:lattice-preservation-existing-extrema}} \\
        &\qiff \left(\hbigand_{i \in I} a_i\right) \coimpl b \hleq c
            \inlineExplain{adjunction for $\coimpl$}.
    \end{align*}
    Since this holds for every $c \in \hdom$, reflexivity and antisymmetry yield the claim.
\end{proof}

\subsection{An Exemplary Lightweight Encoding of Weighted NetKAT}
\label{sec:appendix-weighted-netkat-lightweight}
\providecommand{\wnlight}[1]{\mathcal{T}(#1)}

Weighted NetKAT~\cite{weightedNetKAT26} is parametric in an $\omega$-continuous semiring.
By \Cref{lem:semiring_special_module}, every such semiring induces an instance of \wGCL{}.
We represent a Weighted NetKAT packet by a \wGCL{} program state: each packet field becomes one program variable. Since Weighted NetKAT assumes finite field domains, symbolic values such as switches, hosts, and ports can be encoded by natural numbers.
The following lightweight translation is not supposed to reproduce the Weighted NetKAT automata construction.
Instead, it aims to show that the scalar verification problems expressible by their policies can be phrased as weakest-preweighting bounds of \wGCL{} programs, which we can verify in our framework.

For every Weighted NetKAT test $t$, let $B_t$ be the corresponding Boolean expression over the packet variables.
We write $\wnlight{p}$ for the \wGCL{} command obtained by translating the Weighted NetKAT policy $p$.
The structural translation is:
\[
\begin{array}{@{}l>{\raggedright\arraybackslash}p{0.28\linewidth}l@{}}
    \toprule
    \text{wNetKAT policy} & meaning & \wGCL{} \text{ encoding} \\
    \midrule
    t
        & filter packets by test $t$
        & \ITE{B_t}{\WEIGH{\sone}}{\WEIGH{\snull}}
        \\
    f \leftarrow n
        & update field $f$
        & \ASSIGN{f}{n}
        \\
    r \odot p
        & run $p$ with weight $r$
        & \WEIGH{r}\fatsemi \wnlight{p}
        \\
    p \oplus q
        & choose between policies
        & \BRANCH{\wnlight{p}}{\wnlight{q}}
        \\
    p;q
        & run $p$, then $q$
        & \wnlight{p}\fatsemi \wnlight{q}
        \\
    p^*
        & finitely many repetitions
        & \text{see below}
        \\
    \mathsf{dup}
        & log packet history
        & \WEIGH{\sone}
        \\
    \bottomrule
\end{array}
\]

The only nontrivial clause next to $\mathsf{dup}$ is Kleene star $p^*$.
For every occurrence of $p^*$, choose a fresh Boolean control variable $b$ not used in $\wnlight{p}$ and define
\[
    \wnlight{p^*}
    \definedAs
    \BRANCH{\ASSIGN{b}{0}}{\ASSIGN{b}{1}}
    \fatsemi
    \WHILEDO{b=1}{
        \wnlight{p}
        \fatsemi
        \BRANCH{\ASSIGN{b}{0}}{\ASSIGN{b}{1}}
    } .
\]

The initial branch chooses whether to execute zero iterations.
After every executed copy of $\wnlight{p}$, the program again chooses whether to stop or continue.
Thus the loop encoding reproduces the semantic role of Weighted NetKAT star: it aggregates all finite repetitions of $p$.
By the Kleene-iteration semantics of \wGCL{} loops and the definition of
countable sums in $\omega$-continuous modules
\cite[Thm.~4.7, Def.~A.15 and Eq.~(6)]{BatzGKKW22}, the resulting weakest
preweighting is
\[
    \wp{\wnlight{p^*}}(w)
    =
    \sbigadd_{n\in\Nats}\wp{\wnlight{p}^n}(w),
\]
for postweightings $w$ independent of the fresh variable $b$, where $\wnlight{p}^0$ is the empty command and $\wnlight{p}^{n+1}=\wnlight{p}\fatsemi\wnlight{p}^n$.

The table maps $\mathsf{dup}$ to the neutral weight because this is enough for scalar safety and reachability properties.
The role of $\mathsf{dup}$ in Weighted NetKAT is to record complete packet histories.
To also track the packet history in our setting, we could add a ghost history variable to the program state, or abstractly use a monoid-module of $S$-weighted languages over guarded strings, i.e. functions from guarded strings to the semiring carrier $S$ that assign a weight to each packet-history trace.
However, we do not need this heavier construction for scalar verification problems.

\paragraph{Example.}
Consider the first switch in Figure~2 of Weighted NetKAT \cite{weightedNetKAT26}.
The policy $p_1$ is the forwarding-table policy: it chooses an outgoing port from the destination field.
The policy $\ell_1$ is the link policy: it follows the link attached to the chosen port.
Latencies are therefore attached to $\ell_1$, not to $p_1$.
For packets destined to $\mathsf{H}_2$, the relevant fragment is $p_1;\ell_1^{\mathit{lat}}$, where
\[
\begin{aligned}
    p_1
    \definedAs{}&
    \mathsf{if}\ \mathit{dst}=\mathsf{H}_2\
        \mathsf{then}\left(
            \mathit{pt}\leftarrow 3 \oplus \mathit{pt}\leftarrow 4
        \right)
    \\
    &\mathsf{else\ if}\ \mathit{dst}=\mathsf{H}_1\
        \mathsf{then}\ \mathit{pt}\leftarrow 1
    \\
    &\mathsf{else}\ \mathsf{drop},
    \\
    \ell_1^{\mathit{lat}}
    \definedAs{}&
    \mathsf{if}\ \mathit{pt}=3\
        \mathsf{then}\ 2\,\mathrm{ms}\odot
            (\mathit{sw}\leftarrow \mathsf{S}_3;\mathit{pt}\leftarrow 1)
    \\
    &\mathsf{else\ if}\ \mathit{pt}=4\
        \mathsf{then}\ 4\,\mathrm{ms}\odot
            (\mathit{sw}\leftarrow \mathsf{S}_4;\mathit{pt}\leftarrow 1)
    \\
    &\mathsf{else}\ \mathit{pt}=1 .
\end{aligned}
\]
Writing weights in milliseconds, the translation is
\[
\begin{array}{@{}r@{\ }c@{\ }l@{}}
    \wnlight{p_1;\ell_1^{\mathit{lat}}}
    &=
    &\wnlight{p_1}\fatsemi\wnlight{\ell_1^{\mathit{lat}}}
    \\[0.4em]
    \wnlight{p_1}
    &\definedAs
    &\begin{array}[t]{@{}l@{}}
        \mathsf{if}\ \mathit{dst}=\mathsf{H}_2\ \mathsf{then} \\
        \quad
        \BRANCH{\ASSIGN{\mathit{pt}}{3}}{\ASSIGN{\mathit{pt}}{4}}
        \\
        \mathsf{else\ if}\ \mathit{dst}=\mathsf{H}_1\ \mathsf{then} \\
        \quad
        \ASSIGN{\mathit{pt}}{1}
        \\
        \mathsf{else} \\
        \quad
        \WEIGH{\snull}
    \end{array}
    \\[1.0em]
    \wnlight{\ell_1^{\mathit{lat}}}
    &\definedAs
    &\begin{array}[t]{@{}l@{}}
        \mathsf{if}\ \mathit{pt}=3\ \mathsf{then} \\
        \quad
        \WEIGH{2}\fatsemi
        \ASSIGN{\mathit{sw}}{\mathsf{S}_3}\fatsemi
        \ASSIGN{\mathit{pt}}{1}
        \\
        \mathsf{else\ if}\ \mathit{pt}=4\ \mathsf{then} \\
        \quad
        \WEIGH{4}\fatsemi
        \ASSIGN{\mathit{sw}}{\mathsf{S}_4}\fatsemi
        \ASSIGN{\mathit{pt}}{1}
        \\
        \mathsf{else} \\
        \quad
        \ITE{\mathit{pt}=1}{\WEIGH{\sone}}{\WEIGH{\snull}}
    \end{array}
\end{array}
\]
If a packet at switch $\mathsf{S}_1$ is destined for host $\mathsf{H}_2$, then $\wnlight{p_1}$ chooses between ports $3$ and $4$, and $\wnlight{\ell_1^{\mathit{lat}}}$ contributes the corresponding link weight.
Thus the translated program exposes exactly the branch weights $2$ and $4$.
With the tropical instantiation this yields the best-case latency
\[
    2 \splus 4 = \min\{2,4\}=2,
\]
whereas with the arctic instantiation it yields the worst-case latency $ 2 \splus 4 = \max\{2,4\}=4. $

This example illustrates that scalar Weighted NetKAT analyses can be expressed in \wGCL{}.
Weighted NetKAT obtains decision procedures by compiling finite packet-processing policies to finite automata~\cite[Sec.~3.1, Sec.~5, Thm.~1]{weightedNetKAT26}.
Our deductive infrastructure has a different scope: using invariants, it can also reason about parameterized or infinite-state weighted programs, such as the unbounded network-security case study.

\begingroup
\setlength{\parindent}{0pt}
\setlength{\parskip}{0.35\baselineskip plus 0.1\baselineskip minus 0.05\baselineskip}

\section{\titleFullref{Proofs of }{sec:proof_rules}}
\label{sec:appendix-proof-rules}

\thmWpEncodingExactness*

\begin{proof}
    We proceed by structural induction on \(\wwcom\).
    \proofheading{Base cases}
    \begin{proofcases}
        \proofcase{$\wwcom = x := E$}
            \begin{align*}
                &\interpretsimple{\vc{\procenc{\wwcom}}(\hhla)} &&\\
                &= \interpretsimple{\vc{\stmtRasgn{x}{E}}(\hhla)} \inlineExplain{Def. of $\procenc{\cdot}$} &&\\
                &= \interpretsimple{\hhla\substBy{x}{E}} \inlineExplain{Def. of $\symRasgn$} &&\\
                &= \lambda\sigma: \, \interpretsimple{\hhla}(\sigma[x \mapsto \sigma(E)]) &&\\
                &= \wp{x := E}(\interpretsimple{\hhla})
            \end{align*}
        \proofcase{$\wwcom = \odot\,a$}
            \begin{align*}
                &\interpretsimple{\vc{\procenc{\wwcom}}(\hhla)} &&\\
                &= \interpretsimple{\vc{\stmtWeigh{a}}(\hhla)} \inlineExplain{Def. of $\procenc{\cdot}$} &&\\
                &= \interpretsimple{a \cdot \hhla} \inlineExplain{Def. of $\symWeigh$} &&\\
                &= a \odot \interpretsimple{\hhla} &&\\
                &= \wp{\odot\,a}(\interpretsimple{\hhla})
            \end{align*}
    \end{proofcases}
    \proofheading{Inductive cases}
    Assume the induction hypothesis for arbitrary but fixed \(\ccs{1}, \ccs{2} \in \wGCL\).
    \begin{proofcases}
        \proofcase{$\wwcom = \ccs{1} ; \ccs{2}$}
            \begin{align*}
                &\interpretsimple{\vc{\procenc{\wwcom}}(\hhla)} &&\\
                &= \interpretsimple{\vc{\stmtSeq{\procenc{\ccs{1}}}{\procenc{\ccs{2}}}}(\hhla)} \inlineExplain{Def. of $\procenc{\cdot}$} &&\\
                &= \interpretsimple{\vc{\procenc{\ccs{1}}}(\vc{\procenc{\ccs{2}}}(\hhla))} &&\\
                &= \wp{\ccs{1}}(\interpretsimple{\vc{\procenc{\ccs{2}}}(\hhla)}) \inlineExplain{I.H.} &&\\
                &= \wp{\ccs{1}}(\wp{\ccs{2}}(\interpretsimple{\hhla})) &&\\
                &= \wp{\wwcom}(\interpretsimple{\hhla})
            \end{align*}
        \proofcase{$\wwcom = \{\,\ccs{1}\,\}\oplus\{\,\ccs{2}\,\}$}
            \begin{align*}
                &\interpretsimple{\vc{\procenc{\wwcom}}(\hhla)} &&\\
                &= \interpretsimple{\vc{\stmtWeighted{\procenc{\ccs{1}}}{\procenc{\ccs{2}}}}(\hhla)} &&\\
                &= \interpretsimple{\vc{\procenc{\ccs{1}}}(\hhla) + \vc{\procenc{\ccs{2}}}(\hhla)} \inlineExplain{Def. $\symWeightedBranching$} &&\\
                &= \interpretsimple{\vc{\procenc{\ccs{1}}}(\hhla)} \oplus \interpretsimple{\vc{\wpenc{\ccs{2}}}(\hhla)} &&\\
                &= \wp{\ccs{1}}(\interpretsimple{\hhla}) \oplus \wp{\ccs{2}}(\interpretsimple{\hhla}) \inlineExplain{I.H.} &&\\
                &= \wp{\wwcom}(\interpretsimple{\hhla})
            \end{align*}
        \proofcase{$\wwcom = \symIf\;(\bexpr)\; \{\,\ccs{1}\,\}\; \symElse\; \{\,\ccs{2}\,\}$. For all $\State \in \States$}
            \begin{align*}
                &\interpretsimple{\vc{\wpenc{\wwcom}}(\hhla)}(\State) && \\
                &= \llbracket \symVc\llbracket\symDemonic \{\stmtSeq{\stmtAssume{\embed{\bexpr}}}{\wpenc{\ccs{1}}} \} \\
                &\qquad \text{else } \{\stmtSeq{\stmtAssume{\embed{\neg\bexpr}}}{\wpenc{\ccs{2}}}\}\rrbracket(\varphi)\rrbracket(\sigma) && \\
                &=  \llbracket\vc{
                        \stmtSeq{\stmtAssume{\embed{\bexpr}}}{\wpenc{\ccs{1}}}
                    } (\hhla) \\
                &\qquad \dotsqcap
                    \vc{
                        \stmtSeq{\stmtAssume{\embed{\neg\bexpr}}}{\wpenc{\ccs{2}}}
                    } (\hhla)
                    \rrbracket (\State)
                    &&\inlineExplain{Def. \symDemonic}\\
                &= \interpretsimpleState{\vc{
                        \stmtSeq{\stmtAssume{\embed{\bexpr}}}{\wpenc{\ccs{1}}}
                    } (\hhla)} \\
                &\qquad \sqcap
                    \interpretsimpleState{\vc{
                        \stmtSeq{\stmtAssume{\embed{\neg\bexpr}}}{\wpenc{\ccs{2}}}
                    } (\hhla)} &&\inlineExplain{Def. $\dotsqcap$} \\
                &= \interpretsimpleState{
                    \vc{
                        \stmtAssume{\embed{\bexpr}}
                    } (\vc{\wpenc{\ccs{1}}}(\hhla)) } \\
                &\qquad \sqcap \interpretsimpleState{
                    \vc{
                        \stmtAssume{\embed{\neg\bexpr}}
                    } (\vc{\wpenc{\ccs{2}}}(\hhla))} &&\inlineExplain{Def. \seq}\\
                &= \interpretsimpleState{\embed{\bexpr} \impl \vc{\wpenc{\ccs{1}}}(\hhla)} \\
                &\qquad \sqcap \interpretsimpleState{\embed{\neg\bexpr} \impl \vc{\wpenc{\ccs{2}}}(\hhla)}
                &&\inlineExplain{Def. \symAssume} \\
                &= (\interpretsimpleState{\embed{\bexpr}} \impl \interpretsimpleState{\vc{\wpenc{\ccs{1}}}(\hhla)}) \\
                &\qquad \sqcap (\interpretsimpleState{\embed{\neg\bexpr}} \impl \interpretsimpleState{\vc{\wpenc{\ccs{2}}}(\hhla)}) && \\
                &= (\interpretsimpleState{\embed{\bexpr}} \impl \wp{\ccs{1}}(\interpretsimple{\hhla})(\State))\\
                &\qquad \sqcap (\interpretsimpleState{\embed{\neg\bexpr}} \impl \wp{\ccs{2}}(\interpretsimple{\hhla})(\State)) \displaybreak[3] &&\inlineExplain{I.H.} \\
                &= \ifThenElse{\interpretsimpleState{b} = \true}{\top \impl \wp{\ccs{1}}(\interpretsimple{\hhla})(\State) \hand \bot \impl \wp{\ccs{2}}(\interpretsimple{\hhla})(\State)}{\bot \impl \wp{\ccs{1}}(\interpretsimple{\hhla})(\State) \hand \top \impl \wp{\ccs{2}}(\interpretsimple{\hhla})(\State)} && \\
                &= \ifThenElse{\interpretsimpleState{b} = \true}{(\top \impl \wp{\ccs{1}}(\interpretsimple{\hhla})(\State)) \hand \top}{\top \hand (\top \impl \wp{\ccs{2}}(\interpretsimple{\hhla})(\State))} &&\inlineExplain{\Cref{def:heyting_algebra}} \\
                &= \ifThenElse{\interpretsimpleState{b} = \true}{\top \impl \wp{\ccs{1}}(\interpretsimple{\hhla})(\State)}{\top \impl \wp{\ccs{2}}(\interpretsimple{\hhla})(\State)} &&\inlineExplain{\Cref{def:bounded_lattice}} \\
                &= \ifThenElse{\interpretsimpleState{b} = \true}{\wp{\ccs{1}}(\interpretsimple{\hhla})(\State)}{\wp{\ccs{2}}(\interpretsimple{\hhla})(\State)} &&\inlineExplain{\Cref{def:heyting_algebra}} \\
                &= (\iverson{b} \wp{\ccs{1}}(\interpretsimple{\hhla}))(\State) \oplus (\iverson{\neg b} \wp{\ccs{2}}(\interpretsimple{\hhla}))(\State) && \\
                &= \wp{\wwcom}(\interpretsimple{\hhla})(\State) \qedhere
            \end{align*}
    \end{proofcases}
\end{proof}

\subsection{\titleFullref{Proofs of }{sec:spec_statement}}
\label[section]{sec:appendix-specification-statements}

\lemHeyvlSpecificationEncoding*

\begin{proof}
    By the semantics in \Cref{tab:heyvl-semantics-modern} and sequential composition,
    \begin{align*}
        &\vc{\procenc{\stmtSpec{\beta}{\hlc}{\hlb}}}(\hhla)(\State) \\
        &= \vc{\stmtAssert{\hlc}\symSemi\stmtHavoc{\beta}\symSemi\stmtValidate\symSemi\stmtAssume{\hlb}}(\hhla)(\State) \\
        &= \vc{\stmtAssert{\hlc}}(\vc{\stmtHavoc{\beta}\symSemi\stmtValidate\symSemi\stmtAssume{\hlb}}(\hhla))(\State) \\
        &= \evalState{\hlc} \sqcap \vc{\stmtHavoc{\beta}}(\vc{\stmtValidate\symSemi\stmtAssume{\hlb}}(\hhla))(\State).
    \end{align*}
    The suffix $\stmtValidate\symSemi\stmtAssume{\hvcpostweight}$ checks whether $\hvcpostweight \hleq \hhla$ holds and turns this qualitative test into a boolean-valued weighting. Since $\stmtHavoc{\beta}$ ranges over exactly the states in $\varsetreach{\State}$, the resulting value is $\top$ iff
    \[
        \forall \State' \in \varsetreach{\State}.\ \eval{\hvcpostweight}(\State') \hleq \eval{\hhla}(\State'),
    \]
    and $\bot$ otherwise. Meeting this result with $\evalState{\hvcpreweight}$ yields the claim.
\end{proof}

\thmHeyvlSpecificationSoundness*

\begin{proof}
    Let $\State \in \States$ be arbitrary.

    First assume that
    \[
        \forall \State' \in \varsetreach{\State}.\ \eval{\hvcpostweight}(\State') \hleq \eval{\hhla}(\State').
    \]
    By \Cref{lem:heyvl-specification-encoding},
    \[
        \vc{\procenc{\stmtSpec{\beta}{\hvcpreweight}{\hvcpostweight}}}(\hhla)(\State)=\evalState{\hvcpreweight}.
    \]
    Since $\sstmt$ adheres to the specification,
    \[
        \evalState{\hvcpreweight} \hleq \vc{\sstmt}(\hvcpostweight)(\State).
    \]
    Moreover, because $\sstmt$ may modify only variables in $\beta$, every state that is relevant for evaluating $\vc{\sstmt}(\cdot)(\State)$ is contained in $\varsetreach{\State}$. Hence the assumption above implies that replacing $\hvcpostweight$ by $\hhla$ can only increase the value of the verification preweighting, i.e.
    \[
        \vc{\sstmt}(\hvcpostweight)(\State) \hleq \vc{\sstmt}(\hhla)(\State).
    \]
    Therefore,
    \[
        \vc{\procenc{\stmtSpec{\beta}{\hvcpreweight}{\hvcpostweight}}}(\hhla)(\State)
        \hleq
        \vc{\sstmt}(\hhla)(\State).
    \]

    Otherwise, \Cref{lem:heyvl-specification-encoding} yields
    \[
        \vc{\procenc{\stmtSpec{\beta}{\hvcpreweight}{\hvcpostweight}}}(\hhla)(\State)=\bot \hleq \vc{\sstmt}(\hhla)(\State). \qedhere
    \]
\end{proof}

\subsection{\titleFullref{Proofs of }{sec:loop_unrolling}}

\lemFiniteLoopUnrolling*

\begin{proof}
    Let $\wpchar$ be the characteristic function of $\wp{\wwcom}(\nosem{\wla})$.
    By the loop semantics and Kleene's fixed-point theorem,
    \[
        \wp{\wwcom}(\nosem{\wla})
        =
        \lfp{\wX}{\wpchar(\wX)}
        =
        \hbigor_{m \in \Nats} \wpchar^m(\bot).
    \]
    For every $n \in \Nats$, the approximant $\wpchar^n(\bot)$ is below this supremum.
    Hence $\nosem{\wlc} \hleq \wpchar^n(\bot)$ implies
    $\nosem{\wlc} \hleq \wp{\wwcom}(\nosem{\wla})$ by transitivity.
\end{proof}

\lemHeyvlCharacteristicEncodingSemantics*

\begin{proof}
    \begin{align*}
        &\vc{\charEnc{\wwcom}{\hinv}}(\hhla)(\State) \\
        &= \vc{\symDemonic \{\stmtSeq{\stmtAssume{\embed{\bexpr}}}{\procenc{\cctick}}\symSemi \Assert{\hinv}\symSemi \Assume{\embed{\false}} \} \text{else } \{\stmtAssume{\embed{\neg\bexpr}}\}}(\hhla)(\State) \\
        &= (\interpretsimpleState{\embed{\bexpr}} \impl \vc{\procenc{\cctick}\symSemi \Assert{\hinv}\symSemi \Assume{\embed{\false}}}(\hhla)(\State)) \sqcap (\interpretsimpleState{\embed{\neg\bexpr}} \impl \interpretsimpleState{\hhla}) \\
        &= (\interpretsimpleState{\embed{\bexpr}} \impl \vc{\procenc{\cctick}}(\vc{\Assert{\hinv}\symSemi \Assume{\embed{\false}}}(\hhla))(\State)) \\
            &\qquad \sqcap (\interpretsimpleState{\embed{\neg\bexpr}} \impl \interpretsimpleState{\hhla}) \\
        &= (\interpretsimpleState{\embed{\bexpr}} \impl \vc{\procenc{\cctick}}(\hinv)(\State)) \sqcap (\interpretsimpleState{\embed{\neg\bexpr}} \impl \interpretsimpleState{\hhla}) \\
        &= \ifThenElse{\interpretsimpleState{\bexpr} = \true}{\vc{\procenc{\cctick}}(\hinv)(\State)}{\interpretsimpleState{\hhla}} \qedhere
    \end{align*}
\end{proof}

\lemCharacteristicEncodingSoundness*

\begin{proof}
    By \Cref{lem:heyvl-characteristic-encoding-semantics} and the soundness assumption on $\procenc{\wwcoma}$, each characteristic encoding under-approximates the semantic characteristic function:
    \[
        \nosem{\vc{\charEnc{\wwcom, \hinv}}(\hhla)}
        \hhleq
        \wpcharhhla(\nosem{\hinv}).
    \]
    Therefore, by induction on $n$,
    \[
        \nosem{\vc{\charEnc{\wwcom, \bot}^n}(\hhla)}
        \hhleq
        \wpcharhhla^n(\bot).
    \]
    The claim follows from \Cref{lem:loop_unrolling}.
\end{proof}

\subsection{\titleFullref{Proofs of }{sec:local_park_induction}}

\lemClosedness*

\begin{proof}
    We prove by structural induction on $C$ that for all global weightings
    $Y,Z \in \weightings$ with $Y(\State)=Z(\State)$ for all
    $\State\in\statesclosed$, it holds that
    \(\wlp{C}(Y)(\State)=\wlp{C}(Z)(\State)\) for all
    $\State \in \statesclosed$. The corresponding statement for $\wp{C}$
    is analogous, but we concentrate on $\wlp{C}$ here.

    \proofheading{Base cases}
    \begin{proofcases}
        \proofcase{$C = \ASSIGN{x}{E}$}
            Since $x\in\beta$ and $\statesclosed$ is closed under
            $\beta$-reachability, every successor state
            $\State[x\mapsto\eval{E}(\State)]$ is also in
            $\statesclosed$. Hence, for all $\State \in \statesclosed$,
            \begin{align*}
                \wlp{\ASSIGN{x}{E}}(Y)(\State)
                &= Y(\State[x\mapsto\eval{E}(\State)]) \\
                &= Z(\State[x\mapsto\eval{E}(\State)]) \\
                &= \wlp{\ASSIGN{x}{E}}(Z)(\State).
            \end{align*}

        \proofcase{$C = \WEIGH{a}$}
            Immediate from the definition of the weight multiplication:
            $\mexpr \otimes Y$ depends only on $Y(\State)$, and therefore on
            $Z(\State)$ when $Y$ and $Z$ agree on $\statesclosed$.
    \end{proofcases}

    \proofheading{Inductive cases}
    Assume the statement holds for $C_1$ and $C_2$.
    \begin{proofcases}
        \proofcase{$C = C_1 ; C_2$}
            For all $\State\in\statesclosed$, the induction hypothesis gives
            $\wlp{C_2}(Y)(\State)=\wlp{C_2}(Z)(\State)$.
            Applying the induction hypothesis to $C_1$ then yields the
            desired equality for $C$.

        \proofcase{$C = \BRANCH{C_1}{C_2}$}
            The semantics of branching is pointwise addition, so the claim
            follows directly from the induction hypothesis on both branches.

        \proofcase{$C = \ITE{b}{C_1}{C_2}$}
            For $\State\in\statesclosed$, either $\nosem{b}(\State)$ or
            $\nosem{\neg b}(\State)$ holds. In the former case, the result is
            determined by $C_1$; in the latter by $C_2$. The induction
            hypothesis therefore implies the desired equality for $C$.

        \proofcase{$C = \WHILEDO{b}{C_1}$}
            Let
            \(\Psi_Y(X)= [\neg b]\cdot Y \oplus [b]\cdot \wlp{C_1}(X)\)
            and
            \(\Psi_Z(X)= [\neg b]\cdot Z \oplus [b]\cdot \wlp{C_1}(X)\).
            By the induction hypothesis on $C_1$, any two weightings that
            agree on $\statesclosed$ are mapped by $\wlp{C_1}$ to weightings
            that also agree on $\statesclosed$. Hence $\Psi_Y$ and
            $\Psi_Z$ agree on $\statesclosed$ whenever $Y$ and $Z$ do.

            Define the approximants
            $X_0 = \top$ and $X_{n+1} = \Psi_Y(X_n)$, and analogously
            $X'_0 = \top$ and $X'_{n+1} = \Psi_Z(X'_n)$. By induction on
            $n$, each pair $X_n,X'_n$ agrees on $\statesclosed$. Therefore
            the greatest fixed points $\gfp\,\Psi_Y$ and
            $\gfp\,\Psi_Z$ also agree on $\statesclosed$. This proves
            the claim for $\wlp{C}$.
    \end{proofcases}

    Finally, closure under $\statesclosed$-restricted weightings follows
    from the fact that for any restricted weighting $\hla_{|\statesclosed}$
    and any two global extensions $Y,Z$ of it, the values of
    $\wlp{C}(Y)$ on $\statesclosed$ agree with those of
    $\wlp{C}(Z)$. Thus we may choose $\hla' = \wlp{C}(Y)$ as a global
    extension of the result. The analogous argument for $\wp{C}$ is
    straightforward.
\end{proof}

\thmLocalParkInduction*

\begin{proof}
    We assume that the inequality $ \nosem{I_{|\statesclosed}} \bhleq \betawlpchar(\nosem{I_{|\statesclosed}})$ holds for an invariant $I \in \weightings$.
    Due to \Cref{lem:closedness_under_reachability}, the transformed weighting $\betawlpchar(\nosem{I_{|\statesclosed}}) : \statesclosed \to \Module$ is again $\statesclosed$-restricted. Hence, we can use monotonicity of $\betawlpchar$ and obtain via natural induction on $n \in \Nats$ that
    $$ \nosem{I_{|\statesclosed}} \bhleq \betawlpchar^n(\nosem{I_{|\statesclosed}}). $$

    Furthermore, we derive via Park induction that $\nosem{I_{|\statesclosed}} \bhleq \gfp \betawlpchar = \wlp{C}(\hla)$, since the set of $\statesclosed$-restricted weightings is an $\omega$-bicomplete partially ordered set and the weighting transformer $\betawlpchar$ is $\omega$-bicontinuous.
\end{proof}

\begin{restatable}[Semantics of Encoding]{lemma}{lemSemanticsPark}
    \label[lemma]{lem:semantics-local-park-induction-encoding}
    For a loop $\wwcom = \WHILEDO{b}{\wwcoma} \in \wGCL$ with local variables $\varset \subseteq \Vars$, weightings $\hhla, \hinv \in \wHeyLo$, and an encoding $\procenc{\wwcoma}$ with $\vc{\procenc{\wwcoma}}(\hhlb) \hleq \wp{\wwcoma}(\hhlb)$ for all $\hhlb \in \wHeyLo$, the encoding of Park induction in \Cref{tab:encoding_spec_char_park} satisfies
    $$
        \nosem{\vc{\parkEnc{\wwcom, \hinv}}(\hhla)}(\State) \hleq
         \ifThenElseDot{\hinv \hleq_{\varsetreach{\sigma}} \betawlpchar(\hinv)}{\nosem{\hinv}(\sigma)}{\bot}
    $$
\end{restatable}

\begin{proof}
    Since the encoding of local Park induction is the sequential composition of the specification statement encoding and the characteristic function encoding, we calculate the semantics using \Cref{lem:heyvl-specification-encoding} and \Cref{lem:heyvl-characteristic-encoding-semantics} by
    \begin{align*}
        & ~ \vc{\parkEnc{C, I}}(\hla)(\sigma) \\
        = & ~ \vc{\specEnc{\stmtSpec{\varset}{I}{I}}}(\vc{\charEnc{C, I}(\hla)})(\sigma) \\
        = & ~ \begin{cases}
            I(\sigma), & \text{if } I \hleq_{\varsetreach{\sigma}} \vc{\charEnc{C, I}}(\hla) \\
            \bot(\sigma), & \text{otherwise }
        \end{cases} \\
        \hleq & ~ \begin{cases}
            I(\sigma), & \text{if } I(\sigma') \hleq_{\varsetreach{\sigma}} \Psi_{\hla}(I)(\sigma') \\
            \bot(\sigma), & \text{otherwise.}
        \end{cases}
    \end{align*}
\end{proof}

\thmHeyvlSoundnessLocalParkInductionEncoding*

\begin{proof}
    For an arbitrary state $\sigma \in \Sigma$, the set $\statesclosed = \varsetreach{\sigma} \subseteq \Sigma$ is the smallest superset of states that is closed under $\varset$-reachability. We distinguish two cases.
    \begin{itemize}
        \item Assume that $I \bhleq \betawlpchar(I)$ holds. Then $\vc{\parkEnc{C, \beta, I}}(\varphi)(\sigma) \hleq I(\sigma)$ due to \Cref{lem:semantics-local-park-induction-encoding}. Due to local Park induction and \Cref{thm:local-park-induction}, we know that $I \bhleq \wlp{C}(\hla)$ and hence also $I(\sigma) \hleq \wlp{C}(\hla)(\sigma)$. Then inequality follows by
        $$ \vc{\parkEnc{C, \beta, I}}(\hla)(\sigma) \hleq I(\sigma) \hleq \wlp{C}(\hla)(\sigma). $$

        \item Assume that $I \not\bhleq \betawlpchar(I)$ holds. Then $\vc{\parkEnc{C, \beta, I}}(\varphi)(\sigma) = \bot$ due to \Cref{lem:semantics-local-park-induction-encoding}. We cannot apply local Park induction and \Cref{thm:local-park-induction}. However, the inequality follows by
        $$ \vc{\parkEnc{C, \beta, I}}(\hla)(\sigma) = \bot \hleq \wlp{C}(\hla)(\sigma). $$
    \end{itemize}
    This concludes the proof.
\end{proof}

\subsection{\titleFullref{Proofs of }{sec:k-induction}}

\begin{restatable}[Closedness under Reachability]{lemma}{lemClosednessKInduction}
    \label[lemma]{lem:closedness_under_reachability_k-induction}
    For a program $\wwcom \in \wGCL$ with local variables $\varset \subseteq \Vars$ and a set $\statesclosed \subseteq \Sigma$ of states that is closed under $\varset$-reachability, the $\kk$-induction operator $\kindsym$ is closed under $\statesclosed$-restricted weightings, hence for all $\wla \in \weightings$ exists $\wlb \in \weightings$ such that
    $$ \kindopk(\wla_{|\statesclosed}) = \wlb_{|\statesclosed}. $$
\end{restatable}

\begin{proof}
    We know that $\betawlpchar$ is closed under $\varset$-reachability. Furthermore, we can show by structural induction that for all $\wla, \wlb \in \weightings$ there exists $\wlc \in \weightings$ such that $\wla_{|\statesclosed} \hor \wlb_{|\statesclosed} = \wlc_{|\statesclosed}$. This implies the statement by induction on $\kk$.
\end{proof}

\thmKappaInduction*

\begin{proof}
    We show that $\kindopk(\wlb) \bhleq \kindop{\kk+1}(\wlb)$ for any $\wlb \in \weightings$ by natural induction on $\kk \in \Nats$: It is $\wlb \bhleq I \hor \betawlpchar(I) = \kindsym(I)$ and $\kindopk(\wlb) \bhleq \kindop{\kk+1}(\wlb) \iff I \hor \betawlpchar(\kindop{\kk-1}) \bhleq I \hor \betawlpchar(\kindopk(I))$, if we assume that $\kindop{\kk-1}(\wlb) \bhleq \kindopk(I)$ for an arbitrary but fixed $\kk > 0$ as the induction hypothesis.

    The previous statement together with the premise $\nosem{I} \bhleq \betawlpchar(\kindopk(\nosem{I}))$ show that $\kindopk(I)$ is a local subinvariant by
    \begin{align*}
        \kindopk(I) \bhleq \kindop{\kk+1}(I) = I \hor \betawlpchar(\kindopk(I)) \bhleq \betawlpchar(\kindopk(I)).
    \end{align*}
    We apply local Park induction (\Cref{thm:local-park-induction}) to obtain that $\nosem{I} \bhleq \wlp{C}(\wla)$.
\end{proof}

\begin{restatable}[Semantics of Encoding]{lemma}{lemSemanticsKappaEncoding}
    \label[lemma]{lem:kappa_induction_semantics}
    For a loop $\wwcom = \WHILEDO{b}{\wwcoma} \in \wGCL$ with local variables $\varset \subseteq \Vars$, weightings $\hhla, \hinv \in \wHeyLo$, an encoding $\procenc{\wwcoma}$ with $\vc{\procenc{\wwcoma}}(\hhlb) \hleq \wp{\wwcoma}(\hhlb)$ for all $\hhlb \in \wHeyVL$, and $\sigma \in \Sigma$, the encoding of $\kappa$-induction in \Cref{fig:encoding_kappa_induction} satisfies
    $$
        \nosem{\vc{\kappaEnc{\wwcom, \varset, \hinv}}(\hhla)}(\sigma)
        \hleq \ifThenElseDot{\nosem{\hinv} \bhleq \betawlpchar(\kindopk(\nosem{\hinv}))}{\nosem{\hinv}(\sigma)}{\bot}
    $$
\end{restatable}

\begin{proof}
    We show by induction that for a loop $C \in \wHeyVL$ with local variables $\varset \subseteq \Sigma$, weightings $\hla, I \in \wHeyLo$, and state $\sigma\in\Sigma$, the semantics of $\kappaUnrollEncShort$ is
    $$\nosem{\vc{\kappaUnrollEnc{C, \varset, I}}(\hla)}(\sigma) \hleq \betawlpchar(\kindopk(I))(\sigma).$$

    The base case is $\kk = 0$ with $\unrollEnc{0}{C, \varset, I} = \charEnc{C, \varset, I}$ and hence \newline $\nosem{\vc{\unrollEnc{0}{C, \varset, I}}}(\hla)(\sigma) \hleq \betawlpchar(I)(\sigma)$.

    We assume that the statement holds for an arbitrary, but fixed $\kk \in \Nats$ and obtain
    \begin{align*}
        & ~ \vc{\unrollEnc{($\kk$+1)}{C, \varset, I}}(\hla) \\
        = & ~ \vc{\stmtAssume{\embed{b}}\symSemi \procenc{C'}\symSemi
            \coAssert{I}\symSemi \kappaUnrollEnc{C, \varset, I}}(\hla) \hand \vc{\stmtAssume{\embed{\neg b}}}(\hla) \\
        = & ~ (\embed{b} \impl \vc{\procenc{C'}}(\vc{\coAssert{I}\symSemi \kappaUnrollEnc{C, \varset, I}}(\hla))) \hand \vc{\stmtAssume{\embed{\neg b}}}(\hla) \\
        = & ~ (\embed{b} \impl \vc{\procenc{C'}}(I \hor \kappaUnrollEnc{C, \varset, I}(\hla)))
            \hand (\embed{\neg b} \impl \hla) \\
        \substack{\text{(IH)} \\ \hleq} & ~ (\embed{b} \impl \vc{\procenc{C'}}(I \hor \betawlpchar(\kindopk(I)))) \hand (\embed{\neg b} \impl \hla) \\
        \hleq & ~ (\embed{b} \impl \vc{\procenc{C'}}(\kindop{\kk + 1}(I))) \hand (\embed{\neg b} \impl \hla) \\
        \hleq & ~ \betawlpchar(\kindop{\kk + 1}(I)).
    \end{align*}
    This shows the statement by induction.

    We know that $\kappaEnc{C, \varset, I} = \specEnc{C, \varset, I}\symSemi \kappaUnrollEnc{C, \varset, I}$ and hence
    \begin{align*}
        & ~ \nosem{\vc{\kappaEnc{C, \varset, I}}(\hla)}(\sigma) \\
        = & ~ \nosem{\vc{\specEnc{C, \varset, I}\symSemi \kappaUnrollEnc{C, \varset, I}}(\hla)}(\sigma) \\
        = & ~ \nosem{\vc{\specEnc{C, \varset, I}}(\vc{\kappaUnrollEnc{C, \varset, I}}(\hla))}(\sigma) \\
        = & ~ \begin{cases}
            \nosem{I}(\sigma), & \text{if } \nosem{I} \bhleq \nosem{\vc{\kappaUnrollEnc{C, \varset, I}}(\hla)} \\
            \bot, & \text{otherwise,}
        \end{cases} \\
        = & ~ \begin{cases}
            \nosem{I}(\sigma), & \text{if } \nosem{I} \bhleq \betawlpchar(\kindopk(\nosem{I})) \\
            \bot, & \text{otherwise.}
        \end{cases}
    \end{align*}
\end{proof}

\begin{restatable}[Soundness]{theorem}{thmSoundnessKappaEncoding}
    \label[theorem]{thm:kappa_induction_soundness}
    For a loop $\wwcom = \WHILEDO{b}{\wwcoma} \in \wGCL$ with local variables $\varset \subseteq \Vars$, weightings $\hhla, \hinv \in \wHeyLo$, and $\kk \in \Nats$, the encoding $\kappaEncShort$ is a sound under-approximation, thus
    $$
        \nosem{\vc{\kappaEnc{\wwcom, \varset, \hinv}}(\hhla)} \hhleq \wlp{\wwcom}(\nosem{\hhla}).
    $$
\end{restatable}

\begin{proof}
    For an arbitrary state $\sigma \in \Sigma$, the set $\statesclosed = \varsetreach{\beta} \subseteq \Sigma$ is the smallest superset of states that is closed under $\varset$-reachability. We distinguish two cases.
    \begin{itemize}
        \item Assume that $I \bhleq \betawlpchar(\kindopk(I))$ holds. Then $\vc{\kappaEnc{C, \beta, I}}(\varphi)(\sigma) \hleq I(\sigma)$ due to \Cref{lem:kappa_induction_semantics}. Due to local $\kk$-induction and \Cref{thm:kappa_induction}, we know that $I \bhleq \wlp{C}(\hla)$ and hence also $I(\sigma) \hleq \wlp{C}(\hla)(\sigma)$. Then inequality follows by
        $$ \vc{\kappaEnc{C, \beta, I}}(\hla)(\sigma) \hleq I(\sigma) \hleq \wlp{C}(\hla)(\sigma). $$

        \item Assume that $I \not\bhleq \betawlpchar(\kindopk(I))$ holds. Then $\vc{\kappaEnc{C, \beta, I}}(\varphi)(\sigma) = \bot$ due to \Cref{lem:kappa_induction_semantics}. We cannot apply local $\kk$-induction and \Cref{thm:kappa_induction}. However, the inequality follows by
        $$ \vc{\kappaEnc{C, \beta, I}}(\hla)(\sigma) = \bot \hleq \wlp{C}(\hla)(\sigma). $$
    \end{itemize}
    This concludes the proof.
\end{proof}

\subsection{\titleFullref{Proofs of }{sec:w-invariants}}

\thmOmegaInvariants*

\begin{proof}
    We show via natural induction that $I_{n} \hleq \wpchar^n(\bot)$ holds for all $n \in \Nats$.
    \begin{itemize}
        \item Suppose that $n = 0$. Then $I(0) = \bot$ and $\bot \hleq \wpchar^0(\bot) = \bot$.
        \item Assume that $I_n \hleq \wpchar^n(\bot)$ holds for an arbitrary, but fixed number $n \in \Nats$.
            Then $I_{n+1} \hleq \wpchar(I_n) \hleq \wpchar(\wpchar^n(\bot))$ follows from the induction hypothesis, from the conditions on subinvariants, and from the $\omega$-continuity of the characteristic function $\wpchar$.
    \end{itemize}
    This concludes the statement via natural induction.

    Since $I_n \hleq \wpchar^n(\bot)$ holds for all $n \in \Nats$, we derive that
    $ \sup_{n \in \Nats} I_n \hleq \hbigor_{n \in \Nats} \wpchar^n(\bot) = \wp{C}(\hla). $
\end{proof}

\lemCheckInvariant*

\begin{proof}
    By applying the rules in \Cref{tab:heyvl-semantics-modern}, we obtain the semantics of $base$ and $step$ by
    \begin{align*}
        & I(0) \hleq \vc{base}(\embed{\false}) \tag{definition of \texttt{proc}} \\
        \text{iff}~ & I(0) \hleq \embed{\false} = \bot.
    \end{align*}
    and, again $I(n+1) \hleq \vc{step}(\hla)$, and
    \begin{align*}
        &~ \vc{step}(\hla) \\
        = &~ \vc{\stmtDeclInit{x_1}{\typevar_1}{\overline{x}_1}\symSemi \ldots\symSemi
            \stmtDemonic{\stmtAssume{\embed{b}}\symSemi \ldots}{\stmtAssume{\embed{\neg b}}}}(\hla) \\
        = &~ \symVc \llbracket
            \stmtDeclInit{x_1}{\typevar_1}{\overline{x}_1}\symSemi \ldots\symSemi \stmtDeclInit{x_k}{\typevar_k}{\overline{x}_k} \rrbracket \\
        & \quad (\vc{\stmtAssume{\embed{b}}\symSemi \procenc{C'}\symSemi
            \stmtAssert{I(n)}\symSemi \stmtAssume{\embed{\false}}}(\hla) \hand \vc{\stmtAssume{\embed{\neg b}}}(\hla) ) \\
        = &~ \symVc \llbracket
            \stmtDeclInit{x_1}{\typevar_1}{\overline{x}_1}\symSemi \ldots\symSemi \stmtDeclInit{x_k}{\typevar_k}{\overline{x}_k} \rrbracket \\
        & \quad (\vc{\stmtAssume{\embed{b}}\symSemi \procenc{C'}\symSemi
            \stmtAssert{I(n)}}(\embed{\false} \impl \hla) \hand (\embed{\neg b} \impl \hla)) \\
        = &~ \symVc \llbracket
            \stmtDeclInit{x_1}{\typevar_1}{\overline{x}_1}\symSemi \ldots\symSemi \stmtDeclInit{x_k}{\typevar_k}{\overline{x}_k} \rrbracket \\
        & \quad (\vc{\stmtAssume{\embed{b}}\symSemi \procenc{C'}\symSemi
            \stmtAssert{I(n)}}(\bot) \hand (\embed{\neg b} \impl \hla)) \\
        = &~ \symVc \llbracket
            \stmtDeclInit{x_1}{\typevar_1}{\overline{x}_1}\symSemi \ldots\symSemi \stmtDeclInit{x_k}{\typevar_k}{\overline{x}_k} \rrbracket \\
        & \quad (\vc{\stmtAssume{\embed{b}}\symSemi \procenc{C'}}(I(n)) \hand (\embed{\neg b} \impl \hla)) \\
        = &~ \symVc \llbracket
            \stmtDeclInit{x_1}{\typevar_1}{\overline{x}_1}\symSemi \ldots\symSemi \stmtDeclInit{x_k}{\typevar_k}{\overline{x}_k} \rrbracket \\
        & \quad (\vc{\stmtAssume{\embed{b}}}(\vc{\procenc{C'}}(I(n))) \hand (\embed{\neg b} \impl \hla)) \\
        = &~ \symVc \llbracket
            \stmtDeclInit{x_1}{\typevar_1}{\overline{x}_1}\symSemi \ldots\symSemi \stmtDeclInit{x_k}{\typevar_k}{\overline{x}_k} \rrbracket \\
        & \quad (\embed{b} \impl \vc{\procenc{C'}}(I(n)) \hand (\embed{\neg b} \impl \hla)) \\
        = &~ \embed{b} \impl \vc{\procenc{C'}}(I(n)) \hand (\embed{\neg b} \impl \hla) \\
        \hleq &~ \wpchar(I(n)),
    \end{align*}
    if $\vc{\procenc{C'}}(\hlb) \hleq \wp{\procenc{C'}}(\hlb)$ for all $\hlb \in \wHeyLo$. Hence, $I(n+1) \hleq \vc{\procenc{C'}}(\hla) \hleq \wp{\procenc{C'}}(\hla)$ follows for all $n \in \Nats$.
\end{proof}

\begin{restatable}[Semantics of Encoding]{lemma}{lemSemanticsOmega}
    \label[lemma]{lem:semantics-omega-encoding}
    For a loop $\wwcom = \WHILEDO{b}{\wwcoma} \in \wGCL$, a postweighting $\hhla \in \wHeyLo$, and a weighting $\hinv(n) \in \wHeyLo$ with free variable $n \in \Vars$, the encoding of the $\omega$-rule in \Cref{fig:encoding_omega} satisfies
    $$
        \nosem{\vc{\omegaEnc{\wwcom, (\hinv(n))_{n \in \Nats}}}(\hhla)}(\State)
        = \hbigor_{n \in \Nats} \hinv(n).
    $$
\end{restatable}

\begin{proof}
    By applying the rules in \Cref{tab:heyvl-semantics-modern}, we obtain the semantics of $\omegaEnc{C, I(n)}$ by
    \begin{align*}
        &~ \vc{\omegaEnc{C, I(n)}}(\hla) \\
        = &~ \vc{\coHavoc{n}\symSemi \stmtAssert{I(n)}}(\vc{\stmtAssume{\embed{\false}}}(\hla)) \\
        = &~ \vc{\coHavoc{n}\symSemi \stmtAssert{I(n)}}(\embed{\false} \impl \hla) \\
        = &~ \vc{\coHavoc{n}\symSemi \stmtAssert{I(n)}}(\top) \\
        = &~ \vc{\coHavoc{n}}(\vc{\stmtAssert{I(n)}}(\top)) \\
        = &~ \vc{\coHavoc{n}}(I(n) \hand \top) \\
        = &~ \vc{\coHavoc{n}}(I(n)) \\
        = &~ \hbigor_{n \in \Nats} I(n).
    \end{align*}
\end{proof}

\begin{restatable}[Soundness]{theorem}{thmOmegaEncodingSoundness}
    \label[theorem]{thm:omega-encoding-soundness}
    For a loop $\wwcom = \WHILEDO{b}{\wwcoma} \in \wGCL$, a postweighting $\hhla \in \wHeyLo$, and a weighting $\hinv(n) \in \wHeyLo$, the encoding $\omegaEncShort$ is a sound under-approximation, if the auxiliary procedures $base$ and $step$ in \Cref{fig:encoding_omega} are valid for an $n \in \Nats$. Then
    $$
        \nosem{\vc{\omegaEnc{\wwcom, \hinv}}(\hhla)} \hhleq \wp{\wwcom}(\nosem{\hhla}).
    $$
\end{restatable}

\begin{proof}
    Suppose that the programs $base$ and $step$ are valid, then the sequence $(I(n))_{n \in \Nats}$ is an $\omega$-subinvariant with $I(0) = \bot$ and $I_{n+1} \hleq \wpchar(I_n)$ for all $n \in \Nats$ due to \Cref{lem:semantics-check-invariant}. Then, we apply \Cref{thm:omega-invariants} and obtain $\vc{\omegaEnc{C,I(n)}}(\hla) = \hbigor_{n \in \Nats} I(n) \hleq \wp{C}(\hla)$.
\end{proof}

\endgroup

\section{\titleFullref{Proofs of }{sec:quantifier-elimination}}
\label{sec:appendix-automated-reasoning}

\thmQelimStrip*

\begin{proof}
  Let $\State \in \States$.
  If $\hhla$ is valid, then $\evalStateSubstBy{\hhla}{x}{v}=\top$ for all $v \in \Vals$, and hence
  \[
    \evalState{\linfquant{x}{\hhla}}
    = \inf \Set{\evalStateSubstBy{\hhla}{x}{v} \mid v \in \Vals}
    = \inf \Set{\top \mid v \in \Vals}
    = \top.
  \]
  Thus $\linfquant{x}{\hhla}$ is valid.

  Conversely, assume $\linfquant{x}{\hhla}$ is valid. Then
  \[
    \top
    = \evalState{\linfquant{x}{\hhla}}
    = \inf \Set{\evalStateSubstBy{\hhla}{x}{v} \mid v \in \Vals}
    \hleq \evalStateSubstBy{\hhla}{x}{\State(x)}
    = \evalState{\hhla}.
  \]
  Since $\top$ is greatest, $\evalState{\hhla}=\top$. As $\State$ was arbitrary, $\hhla$ is valid.

  If $\hhla$ is covalid, then $\evalStateSubstBy{\hhla}{x}{v}=\bot$ for all $v \in \Vals$, and hence
  \[
    \evalState{\lsupquant{x}{\hhla}}
    = \sup \Set{\evalStateSubstBy{\hhla}{x}{v} \mid v \in \Vals}
    = \sup \Set{\bot \mid v \in \Vals}
    = \bot.
  \]
  Thus $\lsupquant{x}{\hhla}$ is covalid.

  Conversely, assume $\lsupquant{x}{\hhla}$ is covalid. Then
  \[
    \evalStateSubstBy{\hhla}{x}{\State(x)}
    \hleq \evalState{\lsupquant{x}{\hhla}}
    = \bot.
  \]
  Since $\bot$ is least, $\evalState{\hhla}=\bot$. As $\State$ was arbitrary, $\hhla$ is covalid.
\end{proof}

\subsection{Prenexing Rules}

We first collect the formula-level prenexing equivalences used by the quantifier-elimination procedure.

\subsubsection{Logical Prenexing}

\begin{theorem}[Positive Prenexing]
  \label[theorem]{thm:quantifier-elimination-positive-prenexing}
  Let $\hhla,\hhlb \in \wHeyLo$ and $x,y \in \Vars$ where $x \notin \freevars{\hhlb}$ and $y \notin \freevars{\hhla}$. Then the following equivalences hold:

  \begin{halfboxl}
    \vspace{-\baselineskip}
    \begin{align*}
      \bigsqcap_x \hhla \sqcap \hhlb \midequiv \bigsqcap_x (\hhla \sqcap \hhlb) \\
      \hhla \sqcap \bigsqcap_y \hhlb \midequiv \bigsqcap_y (\hhla \sqcap \hhlb) \\[1em]
      \bigsqcap_x \hhla \sqcup \hhlb \midequiv \bigsqcap_x (\hhla \sqcup \hhlb) \\
      \hhla \sqcup \bigsqcap_y \hhlb \midequiv \bigsqcap_y (\hhla \sqcup \hhlb) \\[1em]
      \hhla \impl \bigsqcap_y \hhlb \midequivL \bigsqcap_y (\hhla \impl \hhlb) \\
      \heylovalidate{\bigsqcap_y \hhlb} \midequiv \bigsqcap_y \heylovalidate{\hhlb}
    \end{align*}
  \end{halfboxl}%
  \begin{halfboxr}
    \vspace{-\baselineskip}
    \begin{align*}
      \bigsqcup_x \hhla \sqcup \hhlb \midequiv \bigsqcup_x (\hhla \sqcup \hhlb) \\
      \hhla \sqcup \bigsqcup_y \hhlb \midequiv \bigsqcup_y (\hhla \sqcup \hhlb) \\[1em]
      \bigsqcup_x \hhla \sqcap \hhlb \midequiv \bigsqcup_x (\hhla \sqcap \hhlb) \\
      \hhla \sqcap \bigsqcup_y \hhlb \midequiv \bigsqcup_y (\hhla \sqcap \hhlb) \\[1em]
      \hhla \coimpl \bigsqcup_y \hhlb \midequivL \bigsqcup_y (\hhla \coimpl \hhlb) \\
      \heylocovalidate{\bigsqcup_y \hhlb} \midequiv \bigsqcup_y \heylocovalidate{\hhlb}
    \end{align*}
  \end{halfboxr}
\end{theorem}

\begin{proof}
  Let $\hhla,\hhlb \in \wHeyLo$ and $x,y \in \Vars$ where $x \notin \freevars{\hhlb}$ and $y \notin \freevars{\hhla}$, and $\State \in \States$.
  Recall that $\hhla \equiv \hhlb$ iff $\evalState{\hhla} = \evalState{\hhlb}$ for all $\State \in \States$.
  \begin{proofcases}
    \proofcase{$\bigsqcap_x \hhla \sqcap \hhlb$}
    \begin{align*}
      &\evalState{\bigsqcap_x \hhla \sqcap \hhlb} \\
      &= \inf \Set{ \evalState{\bigsqcap_x \hhla},~ \evalState{\hhlb} } \inlineExplain{definition of $\sqcap$} \\
      &= \inf \Set{ \inf \Set{ \evalStateSubstBy{\hhla}{x}{v} \mid v \in \Vals },~ \evalState{\hhlb} } \inlineExplain{definition of $\bigsqcap_x$} \\
      &= \inf \Set{ \evalStateSubstBy{\hhla}{x}{v},~ \evalState{\hhlb} \mid v \in \Vals } \inlineExplain{merge $\inf$} \\
      &= \inf \Set{ \evalStateSubstBy{\hhla}{x}{v},~ \evalStateSubstBy{\hhlb}{x}{v} \mid v \in \Vals } \inlineExplain{$x \notin \freevars{\hhlb}$} \\
      &= \inf \Set{ \inf \Set{\evalStateSubstBy{\hhla}{x}{v},~ \evalStateSubstBy{\hhlb}{x}{v}} \mid v \in \Vals } \inlineExplain{introduce $\inf$} \\
      &= \inf \Set{ \evalStateSubstBy{\hhla \sqcap \hhlb}{x}{v} \mid v \in \Vals } \inlineExplain{definition of $\sqcap$} \\
      &= \evalState{\bigsqcap_x (\hhla \sqcap \hhlb)} \inlineExplain{definition of $\bigsqcap_x$}
    \end{align*}
    The case $\hhla \sqcap \bigsqcap_y \hhlb \equiv \bigsqcap_y (\hhla \sqcap \hhlb)$ is analogous.

    \proofcase{$\bigsqcap_x \hhla \sqcup \hhlb$}
    \begin{align*}
      &\evalState{\bigsqcap_x \hhla \sqcup \hhlb} \\
      &= \sup \Set{ \evalState{\bigsqcap_x \hhla},~ \evalState{\hhlb} } \inlineExplain{definition of $\sqcup$} \\
      &= \sup \Set{ \inf \Set{ \evalStateSubstBy{\hhla}{x}{v} \mid v \in \Vals },~ \evalState{\hhlb} } \inlineExplain{definition of $\bigsqcap_x$} \\
      &= \inf \Set{ \sup \Set{ \evalStateSubstBy{\hhla}{x}{v},~ \evalState{\hhlb} } \mid v \in \Vals } \inlineExplain{\Cref{lem:lattice-preservation-existing-extrema}} \\
      &= \inf \Set{ \sup \Set{ \evalStateSubstBy{\hhla}{x}{v},~ \evalStateSubstBy{\hhlb}{x}{v} } \mid v \in \Vals } \inlineExplain{$x \notin \freevars{\hhlb}$} \\
      &= \inf \Set{ \evalStateSubstBy{\hhla \sqcup \hhlb}{x}{v} \mid v \in \Vals } \inlineExplain{definition of $\sqcup$} \\
      &= \evalState{\bigsqcap_x (\hhla \sqcup \hhlb)} \inlineExplain{definition of $\bigsqcap_x$}
    \end{align*}
    The case $\hhla \sqcup \bigsqcap_y \hhlb \equiv \bigsqcap_y (\hhla \sqcup \hhlb)$ is analogous.

    \proofcase{$\hhla \impl \bigsqcap_y \hhlb$}
    \begin{align*}
      &\evalState{\hhla \impl \bigsqcap_y \hhlb} \\
      &= \evalState{\hhla} \impl \evalState{\linfquant{y}{\hhlb}} \inlineExplain{definition of $\impl$} \\
      &= \evalState{\hhla} \impl \inf \Set{\evalStateSubstBy{\hhlb}{y}{v} \mid v \in \Vals} \inlineExplain{definition of $\bigsqcap_y$} \\
      &= \inf \Set{\evalState{\hhla} \impl \evalStateSubstBy{\hhlb}{y}{v} \mid v \in \Vals} \inlineExplain{\abCref{lem:impl_preserves_meet}} \\
      &= \inf \Set{\evalStateSubstBy{\hhla}{y}{v} \impl \evalStateSubstBy{\hhlb}{y}{v} \mid v \in \Vals} \inlineExplain{$y \not\in \freevars{\hhla}$} \\
      &= \evalState{\bigsqcap_y (\hhla \impl \hhlb)} \inlineExplain{definitions of $\impl$ and $\bigsqcap_y$}
    \end{align*}

    \proofcase{$\heylovalidate{\bigsqcap_y \hhlb}$}
    \begin{align*}
      &\evalState{\heylovalidate{\bigsqcap_y \hhlb}} \\
      &= \ifThenElse{\evalState{\bigsqcap_y \hhlb} = \top}{\top}{\bot}
      \inlineExplain{definition of $\heylovalidate{-}$} \\
      &= \ifThenElse{
        \evalStateSubstBy{\hhlb}{y}{v} = \top \text{ for all } v \in \Vals
      }{\top}{\bot} \inlineExplain{definition of $\bigsqcap$} \\
      &= \inf \Set{
        \ifThenElse{\evalStateSubstBy{\hhlb}{y}{v} = \top}{\top}{\bot}
        \mid v \in \Vals
      } \inlineExplain{infimum of $\{\top, \bot\}$} \\
      &= \evalState{\bigsqcap_y \heylovalidate{\hhlb}} \inlineExplain{definition of $\heylovalidate{-}$ and $\bigsqcap_y$}.
    \end{align*}
    The covalidation case $\heylocovalidate{\bigsqcup_y \hhlb} \equiv \bigsqcup_y \heylocovalidate{\hhlb}$ is dual.

    \proofcase{$\bigsqcup_x \hhla \sqcup \hhlb$}
    \begin{align*}
      &\evalState{\bigsqcup_x \hhla \sqcup \hhlb} \\
      &= \sup \Set{ \evalState{\bigsqcup_x \hhla},~ \evalState{\hhlb} } \inlineExplain{definition of $\sqcup$} \\
      &= \sup \Set{ \sup \Set{ \evalStateSubstBy{\hhla}{x}{v} \mid v \in \Vals },~ \evalState{\hhlb} } \inlineExplain{definition of $\bigsqcup_x$} \\
      &= \sup \Set{ \evalStateSubstBy{\hhla}{x}{v},~ \evalState{\hhlb} \mid v \in \Vals } \inlineExplain{flatten suprema} \\
      &= \sup \Set{ \evalStateSubstBy{\hhla}{x}{v},~ \evalStateSubstBy{\hhlb}{x}{v} \mid v \in \Vals } \inlineExplain{$x \notin \freevars{\hhlb}$} \\
      &= \sup \Set{ \sup \Set{\evalStateSubstBy{\hhla}{x}{v},~ \evalStateSubstBy{\hhlb}{x}{v}} \mid v \in \Vals } \inlineExplain{introduce $\sup$} \\
      &= \sup \Set{ \evalStateSubstBy{\hhla \sqcup \hhlb}{x}{v} \mid v \in \Vals } \inlineExplain{definition of $\sqcup$} \\
      &= \evalState{\bigsqcup_x (\hhla \sqcup \hhlb)} \inlineExplain{definition of $\bigsqcup_x$}
    \end{align*}
    The case $\hhla \sqcup \bigsqcup_y \hhlb \equiv \bigsqcup_y(\hhla \sqcup \hhlb)$ is analogous.

    \proofcase{$\bigsqcup_x \hhla \sqcap \hhlb$}
    \begin{align*}
      &\evalState{\bigsqcup_x \hhla \sqcap \hhlb} \\
      &= \inf \Set{ \evalState{\bigsqcup_x \hhla},~ \evalState{\hhlb} } \inlineExplain{definition of $\sqcap$} \\
      &= \inf \Set{ \sup \Set{ \evalStateSubstBy{\hhla}{x}{v} \mid v \in \Vals },~ \evalState{\hhlb} } \inlineExplain{definition of $\bigsqcup_x$} \\
      &= \sup \Set{ \inf \Set{ \evalStateSubstBy{\hhla}{x}{v},~ \evalState{\hhlb} } \mid v \in \Vals } \inlineExplain{\Cref{lem:lattice-preservation-existing-extrema}} \\
      &= \sup \Set{ \inf \Set{ \evalStateSubstBy{\hhla}{x}{v},~ \evalStateSubstBy{\hhlb}{x}{v} } \mid v \in \Vals } \inlineExplain{$x \notin \freevars{\hhlb}$} \\
      &= \sup \Set{ \evalStateSubstBy{\hhla \sqcap \hhlb}{x}{v} \mid v \in \Vals } \inlineExplain{definition of $\sqcap$} \\
      &= \evalState{\bigsqcup_x (\hhla \sqcap \hhlb)} \inlineExplain{definition of $\bigsqcup_x$}
    \end{align*}
    The case $\hhla \sqcap \bigsqcup_y \hhlb \equiv \bigsqcup_y (\hhla \sqcap \hhlb)$ is analogous.

    \proofcase{$\hhla \coimpl \bigsqcup_y \hhlb$}
    \begin{align*}
      &\evalState{\hhla \coimpl \bigsqcup_y \hhlb} \\
      &= \evalState{\hhla} \coimpl \evalState{\lsupquant{y}{\hhlb}} \inlineExplain{definition of $\coimpl$} \\
      &= \evalState{\hhla} \coimpl \sup \Set{ \evalStateSubstBy{\hhlb}{y}{v} \mid v \in \Vals } \inlineExplain{definition of $\bigsqcup_y$} \\
      &= \sup \Set{ \evalState{\hhla} \coimpl \evalStateSubstBy{\hhlb}{y}{v} \mid v \in \Vals } \inlineExplain{\Cref{lem:coimpl_preserves_join}} \\
      &= \sup \Set{ \evalStateSubstBy{\hhla}{y}{v} \coimpl \evalStateSubstBy{\hhlb}{y}{v} \mid v \in \Vals } \inlineExplain{$y \notin \freevars{\hhla}$} \\
      &= \evalState{\bigsqcup_y (\hhla \coimpl \hhlb)} \inlineExplain{definitions of $\coimpl$ and $\bigsqcup_y$}
    \end{align*}

  \end{proofcases}
\end{proof}

\begin{theorem}[Negative Prenexing]
  \label[theorem]{thm:quantifier-elimination-negative-prenexing}
  Let $\hhla,\hhlb \in \wHeyLo$ and $x,y \in \Vars$ where $x \notin \freevars{\hhlb}$ and $y \notin \freevars{\hhla}$. Then the following equivalences hold:

  \begin{halfboxl}
    \vspace{-\baselineskip}
    \begin{align*}
      (\bigsqcup_x \hhla) \impl \hhlb \midequiv \bigsqcap_x (\hhla \impl \hhlb) \\[1em]
      \neg(\bigsqcup_x \hhla) \midequiv \bigsqcap_x (\neg\hhla)
    \end{align*}
  \end{halfboxl}%
  \begin{halfboxr}
    \vspace{-\baselineskip}
    \begin{align*}
      (\bigsqcap_x \hhla) \coimpl \hhlb \midequiv \bigsqcup_x (\hhla \coimpl \hhlb) \\[1em]
      \coneg(\bigsqcap_x \hhla) \midequiv \bigsqcup_x (\coneg\hhla)
    \end{align*}
  \end{halfboxr}
\end{theorem}

\begin{proof}
    Let $\hhla,\hhlb \in \wHeyLo$ and $x,y \in \Vars$ where $x \notin \freevars{\hhlb}$ and $y \notin \freevars{\hhla}$.
  \begin{proofcases}
    \proofcase{$(\bigsqcup_x \hhla) \impl \hhlb$}
    \begin{align*}
      &\evalState{(\lsupquant{x}{\hhla}) \impl \hhlb} \\
      &= \evalState{\lsupquant{x}{\hhla}} \impl \evalState{\hhlb} \inlineExplain{definition of $\impl$} \\
      &= \sup \Set{\evalStateSubstBy{\hhla}{x}{v} \mid v \in \Vals} \impl \evalState{\hhlb} \inlineExplain{definition of $\bigsqcup_x$} \\
      &= \inf \Set{\evalStateSubstBy{\hhla}{x}{v} \impl \evalState{\hhlb} \mid v \in \Vals} \inlineExplain{\Cref{lem:quant-implication-distributes-over-inf}} \\
      &= \inf \Set{\evalStateSubstBy{\hhla}{x}{v} \impl \evalStateSubstBy{\hhlb}{x}{v} \mid v \in \Vals} \inlineExplain{$x \not\in \freevars{\hhlb}$} \\
      &= \evalState{\bigsqcap_x (\hhla \impl \hhlb)}
    \end{align*}

    \proofcase{$(\bigsqcap_x \hhla) \coimpl \hhlb$}
    \begin{align*}
      &\evalState{(\linfquant{x}{\hhla}) \coimpl \hhlb} \\
      &= \evalState{\linfquant{x}{\hhla}} \coimpl \evalState{\hhlb} \inlineExplain{definition of $\coimpl$} \\
      &= \inf \Set{\evalStateSubstBy{\hhla}{x}{v} \mid v \in \Vals} \coimpl \evalState{\hhlb} \inlineExplain{definition of $\bigsqcap_x$} \\
      &= \sup \Set{\evalStateSubstBy{\hhla}{x}{v} \coimpl \evalState{\hhlb} \mid v \in \Vals} \inlineExplain{\Cref{lem:coimpl_turns_meets_into_joins}} \\
      &= \sup \Set{\evalStateSubstBy{\hhla}{x}{v} \coimpl \evalStateSubstBy{\hhlb}{x}{v} \mid v \in \Vals} \inlineExplain{$x \notin \freevars{\hhlb}$} \\
      &= \evalState{\bigsqcup_x (\hhla \coimpl \hhlb)} \inlineExplain{definitions of $\coimpl$ and $\bigsqcup_x$}
    \end{align*}

    \proofcase{$\neg(\bigsqcup_x \hhla)$}
    Using $\neg\hhla \equiv \hhla \impl \bot$, we compute:
    \begin{align*}
      \evalState{\neg\bigsqcup_x \hhla}
      &= \evalState{(\bigsqcup_x \hhla) \impl \bot} \inlineExplain{definition of $\neg$} \\
      &= \evalState{\bigsqcap_x (\hhla \impl \bot)} \inlineExplain{previous proofcase with $\hhlb = \bot$} \\
      &= \evalState{\bigsqcap_x \neg\hhla} \inlineExplain{definition of $\neg$}.
    \end{align*}

    \proofcase{$\coneg(\bigsqcap_x \hhla)$}
    Using $\coneg\hhla \equiv \hhla \coimpl \top$, we compute:
    \begin{align*}
      \evalState{\coneg\bigsqcap_x \hhla}
      &= \evalState{(\bigsqcap_x \hhla) \coimpl \top} \inlineExplain{definition of $\coneg$} \\
      &= \evalState{\bigsqcup_x (\hhla \coimpl \top)} \inlineExplain{previous proofcase with $\hhlb = \top$} \\
      &= \evalState{\bigsqcup_x \coneg\hhla} \inlineExplain{definition of $\coneg$}.
    \end{align*}
  \end{proofcases}
\end{proof}

The remaining prenexing directions require the relevant quantified extremum to be realized by an actual program value.

\begin{definition}[Well-Foundedness]
  \label[definition]{def:max-min-quantifier-domain}
  $\wHeyLo$ is \textit{well-founded} if, for all $\hhla \in \wHeyLo$, $x \in \Vars$, and $\State \in \States$,
  \[
    \exists v^- \in \Vals \,.\,
    \evalStateSubstBy{\hhla}{x}{v^-} = \evalState{\linfquant{x}{\hhla}}.
  \]
  $\wHeyLo$ is \textit{upwards well-founded} if, for all $\hhla \in \wHeyLo$, $x \in \Vars$, and $\State \in \States$,
  \[
    \exists v^+ \in \Vals \,.\,
    \evalStateSubstBy{\hhla}{x}{v^+} = \evalState{\lsupquant{x}{\hhla}}.
  \]
\end{definition}

Finite totally ordered domains, such as the Boolean and clearance semirings, satisfy both clauses. Domains ordered like $\Nats$, such as the counting semiring, satisfy well-foundedness but not upwards well-foundedness. Reversed natural-number orders, such as the tropical semiring over $\NatsX$, satisfy upwards well-foundedness but not well-foundedness. Dense orders such as Viterbi or bottleneck, and Boolean-algebra domains such as formal languages, satisfy neither clause in general.

\begin{lemma}[Well-Founded Quantifiers under Monotone Maps]
  \label[lemma]{lem:well-founded-quantifiers-monotone}
  Let $\hhla \in \wHeyLo$, $x \in \Vars$, and $\State \in \States$.

  If $\wHeyLo$ is well-founded, then for every monotone map
  $F : \hdom \to \hdom$,
  \[
    F(\evalState{\bigsqcap_x \hhla})
    =
    \inf \Set{
      F(\evalStateSubstBy{\hhla}{x}{v})
      \mid v \in \Vals
    },
  \]
  and for every antitone map $F : \hdom \to \hdom$,
  \[
    F(\evalState{\bigsqcap_x \hhla})
    =
    \sup \Set{
      F(\evalStateSubstBy{\hhla}{x}{v})
      \mid v \in \Vals
    }.
  \]

  Dually, if $\wHeyLo$ is upwards well-founded, then for every monotone map
  $F : \hdom \to \hdom$,
  \[
    F(\evalState{\bigsqcup_x \hhla})
    =
    \sup \Set{
      F(\evalStateSubstBy{\hhla}{x}{v})
      \mid v \in \Vals
    },
  \]
  and for every antitone map $F : \hdom \to \hdom$,
  \[
    F(\evalState{\bigsqcup_x \hhla})
    =
    \inf \Set{
      F(\evalStateSubstBy{\hhla}{x}{v})
      \mid v \in \Vals
    }.
  \]
\end{lemma}

\begin{proof}
  We prove the well-founded clauses; the upwards well-founded clauses are dual.

  By well-foundedness, choose $v_0 \in \Vals$ such that
  \[
    \evalStateSubstBy{\hhla}{x}{v_0} =\evalState{\bigsqcap_x \hhla}=
    \inf \Set{
      \evalStateSubstBy{\hhla}{x}{v}
      \mid v \in \Vals
    }.
  \]
  Hence
  \(
    \evalStateSubstBy{\hhla}{x}{v_0} \hleq \evalStateSubstBy{\hhla}{x}{v}
  \)
  for all $v \in \Vals$.
  If $F$ is monotone, then
  \(
    F(\evalStateSubstBy{\hhla}{x}{v_0}) \hleq F(\evalStateSubstBy{\hhla}{x}{v})
  \)
  for all $v \in \Vals$. Thus
  $F(\evalStateSubstBy{\hhla}{x}{v_0})$ is a lower bound of
  \(
    \Set{
      F(\evalStateSubstBy{\hhla}{x}{v})
      \mid v \in \Vals
    }.
  \)
  Since it is also an element of this set, it is its least element and hence
  its infimum. Therefore
  \[
    F(\evalState{\bigsqcap_x \hhla}) =
    \inf \Set{
      F(\evalStateSubstBy{\hhla}{x}{v})
      \mid v \in \Vals
    }.
  \]

  If $F$ is antitone, then
  \(
    F(\evalStateSubstBy{\hhla}{x}{v}) \hleq F(\evalStateSubstBy{\hhla}{x}{v_0})
  \)
  for all $v \in \Vals$. Thus
  $F(\evalStateSubstBy{\hhla}{x}{v_0})$ is an upper bound of
  \(
    \Set{
      F(\evalStateSubstBy{\hhla}{x}{v})
      \mid v \in \Vals
    }.
  \)
  Since it is also an element of this set, it is its greatest element and hence
  its supremum. Therefore
  \[
    F(\evalState{\bigsqcap_x \hhla})
    =
    \sup \Set{
      F(\evalStateSubstBy{\hhla}{x}{v})
      \mid v \in \Vals
    }.
  \]
\end{proof}

\begin{theorem}[Additional Logical Prenexing from Above]
  \label[theorem]{thm:quantifier-elimination-attained-prenexing-above}
  Assume $\wHeyLo$ is upwards well-founded in the sense of \Cref{def:max-min-quantifier-domain}. Let $\hhla,\hhlb \in \wHeyLo$ and $x,y \in \Vars$ where $x \notin \freevars{\hhlb}$ and $y \notin \freevars{\hhla}$. Then the following additional equivalences hold:

  \begin{halfboxl}
    \vspace{-\baselineskip}
    \begin{align*}
      \hhla \impl \bigsqcup_y \hhlb \midequiv \bigsqcup_y (\hhla \impl \hhlb) \\[1em]
      \heylovalidate{\bigsqcup_y \hhlb} \midequiv \bigsqcup_y \heylovalidate{\hhlb}
    \end{align*}
  \end{halfboxl}%
  \begin{halfboxr}
    \vspace{-\baselineskip}
    \begin{align*}
      (\bigsqcup_x \hhla) \coimpl \hhlb \midequiv \bigsqcap_x (\hhla \coimpl \hhlb) \\[1em]
      \coneg(\bigsqcup_x \hhla) \midequiv \bigsqcap_x (\coneg\hhla)
      \end{align*}
  \end{halfboxr}
\end{theorem}

\begin{proof}
  Let $\hhla,\hhlb \in \wHeyLo$ and $x,y \in \Vars$ where $x \notin \freevars{\hhlb}$ and $y \notin \freevars{\hhla}$.
  \begin{proofcases}
    \proofcase{$\hhla \impl \bigsqcup_y \hhlb$}
    \begin{align*}
      &\evalState{\hhla \impl \bigsqcup_y \hhlb} \\
      &= \evalState{\hhla} \impl \evalState{\bigsqcup_y \hhlb} \inlineExplain{definition of $\impl$} \\
      &= \sup \Set{\evalState{\hhla} \impl \evalStateSubstBy{\hhlb}{y}{v} \mid v \in \Vals} \inlineExplain{\Cref{lem:well-founded-quantifiers-monotone} with $F(z)=\evalState{\hhla} \impl z$} \\
      &= \sup \Set{\evalStateSubstBy{\hhla}{y}{v} \impl \evalStateSubstBy{\hhlb}{y}{v} \mid v \in \Vals} \inlineExplain{$y \not\in \freevars{\hhla}$} \\
      &= \evalState{\bigsqcup_y (\hhla \impl \hhlb)} \inlineExplain{definition of $\bigsqcup_y$}.
    \end{align*}
    The validation clause is obtained by the same argument with $F(z)=\heylovalidate{z}$.

    \proofcase{$(\bigsqcup_x \hhla) \coimpl \hhlb$}
    \begin{align*}
      &\evalState{(\bigsqcup_x \hhla) \coimpl \hhlb} \\
      &= \evalState{\bigsqcup_x \hhla} \coimpl \evalState{\hhlb} \inlineExplain{definition of $\coimpl$} \\
      &= \inf \Set{\evalStateSubstBy{\hhla}{x}{v} \coimpl \evalState{\hhlb} \mid v \in \Vals} \inlineExplain{\Cref{lem:well-founded-quantifiers-monotone} with $F(z)=z \coimpl \evalState{\hhlb}$} \\
      &= \inf \Set{\evalStateSubstBy{\hhla}{x}{v} \coimpl \evalStateSubstBy{\hhlb}{x}{v} \mid v \in \Vals} \inlineExplain{$x \not\in \freevars{\hhlb}$} \\
      &= \evalState{\bigsqcap_x (\hhla \coimpl \hhlb)} \inlineExplain{definition of $\bigsqcap_x$}.
    \end{align*}
    The co-negation clause is the special case $\coneg\hhla \equiv \hhla \coimpl \top$.

  \end{proofcases}
\end{proof}

\begin{theorem}[Additional Logical Prenexing from Below]
  \label[theorem]{thm:quantifier-elimination-attained-prenexing-below}
  Assume $\wHeyLo$ is well-founded in the sense of \Cref{def:max-min-quantifier-domain}. Let $\hhla,\hhlb \in \wHeyLo$ and $x,y \in \Vars$ where $x \notin \freevars{\hhlb}$ and $y \notin \freevars{\hhla}$. Then the following additional equivalences hold:

  \begin{halfboxl}
    \vspace{-\baselineskip}
    \begin{align*}
      (\bigsqcap_x \hhla) \impl \hhlb \midequiv \bigsqcup_x (\hhla \impl \hhlb) \\[1em]
      \neg(\bigsqcap_x \hhla) \midequiv \bigsqcup_x (\neg\hhla)
    \end{align*}
  \end{halfboxl}%
  \begin{halfboxr}
    \vspace{-\baselineskip}
    \begin{align*}
      \hhla \coimpl \bigsqcap_y \hhlb \midequiv \bigsqcap_y (\hhla \coimpl \hhlb) \\[1em]
      \heylocovalidate{\bigsqcap_y \hhlb} \midequiv \bigsqcap_y \heylocovalidate{\hhlb}
      \end{align*}
  \end{halfboxr}
\end{theorem}

\begin{proof}
  Let $\hhla,\hhlb \in \wHeyLo$ and $x,y \in \Vars$ where $x \notin \freevars{\hhlb}$ and $y \notin \freevars{\hhla}$.
  \begin{proofcases}
    \proofcase{$(\bigsqcap_x \hhla) \impl \hhlb$}
    \begin{align*}
      &\evalState{(\bigsqcap_x \hhla) \impl \hhlb} \\
      &= \evalState{\bigsqcap_x \hhla} \impl \evalState{\hhlb} \inlineExplain{definition of $\impl$} \\
      &= \sup \Set{\evalStateSubstBy{\hhla}{x}{v} \impl \evalState{\hhlb} \mid v \in \Vals} \inlineExplain{\Cref{lem:well-founded-quantifiers-monotone} with $F(z)=z \impl \evalState{\hhlb}$} \\
      &= \sup \Set{\evalStateSubstBy{\hhla}{x}{v} \impl \evalStateSubstBy{\hhlb}{x}{v} \mid v \in \Vals} \inlineExplain{$x \not\in \freevars{\hhlb}$} \\
      &= \evalState{\bigsqcup_x (\hhla \impl \hhlb)} \inlineExplain{definition of $\bigsqcup_x$}.
    \end{align*}
    The negation clause is the special case $\neg \hhla \equiv \hhla \impl \bot$.

    \proofcase{$\hhla \coimpl \bigsqcap_y \hhlb$}
    \begin{align*}
      &\evalState{\hhla \coimpl \bigsqcap_y \hhlb} \\
      &= \evalState{\hhla} \coimpl \evalState{\bigsqcap_y \hhlb} \inlineExplain{definition of $\coimpl$} \\
      &= \inf \Set{\evalState{\hhla} \coimpl \evalStateSubstBy{\hhlb}{y}{v} \mid v \in \Vals} \inlineExplain{\Cref{lem:well-founded-quantifiers-monotone} with $F(z)=\evalState{\hhla} \coimpl z$} \\
      &= \inf \Set{\evalStateSubstBy{\hhla}{y}{v} \coimpl \evalStateSubstBy{\hhlb}{y}{v} \mid v \in \Vals} \inlineExplain{$y \not\in \freevars{\hhla}$} \\
      &= \evalState{\bigsqcap_y (\hhla \coimpl \hhlb)} \inlineExplain{definition of $\bigsqcap_y$}.
    \end{align*}
    The covalidation clause is obtained by the same argument with $F(z)=\heylocovalidate{z}$.
  \end{proofcases}
\end{proof}

\subsubsection{Weighted Prenexing}

The weighted prenexing laws below are not consequences of the bare monoid-module setup alone, let alone the bi-Heyting algebra structure.
We first isolate the laws that remain valid in the general $\omega$-bicontinuous monoid-module setting, and then state the additional preservation assumptions under which they become formula equivalences.

We first want to identify which prenexing rules hold in the general framework of $\omega$-bicontinuous bi-Heyting algebras.
The crucial observation is that the monoid-module operations $\splus$ and $\stimes$ are inherently tied to the bottom element $\bot$ through definitions like $p \splus \bot = p$ and $a \stimes \bot = \bot$.

\begin{lemma}[Bottom Facts for Weighted Addition]
  \label[lemma]{lem:weighted-addition-bottom-facts}
  Let $\moduledef$ be a monoid-module over a monoid $\Monoid$.
  For all $p,q \in \module$,
  \begin{equation}
    p \splus q = \bot \qiff p=\bot \text{ and } q=\bot.
    \label{eq:binary-splus-bottom-fact}
  \end{equation}
  Moreover, for every family $(r_i)_{i\in I}$ whose supremum exists,
  \begin{equation}
    \hbigor_{i\in I} r_i = \bot \qiff \forall i\in I.\ r_i=\bot.
    \label{eq:supremum-bottom-fact}
  \end{equation}
  Finally, for every nonempty family $(s_i)_{i\in I}$ and every $c \in \module$ for which the displayed extrema exist,
  \begin{align}
    \left(\hbigor_{i\in I}s_i\right)\splus c=\bot
    &\qiff
    \hbigor_{i\in I}(s_i\splus c)=\bot,
    \label{eq:supremum-addition-bottom-fact}
    \\
    \left(\hbigand_{i\in I}s_i\right)\splus c=\bot
    &\qiff
    \hbigand_{i\in I}(s_i\splus c)=\bot.
    \label{eq:infimum-addition-bottom-fact}
  \end{align}
\end{lemma}

\begin{proof}
  For the \Cref{eq:binary-splus-bottom-fact},
  \begin{align*}
    p\splus q=\bot
    &\qimplies p=\bot \text{ and } q=\bot
      \inlineExplain{$p,q\hleq p\splus q=\bot$} \\
    &\qimplies p\splus q=\bot
      \inlineExplain{$\bot=\snull$ is neutral for $\splus$}.
  \end{align*}
  The converse follows, cause $\bot=\snull$ is neutral for $\splus$.
  For \Cref{eq:supremum-bottom-fact},
  \begin{align*}
    \hbigor_{i\in I}r_i=\bot
    &\qimplies \forall i\in I.\ r_i\hleq\bot
      \inlineExplain{$r_i\hleq\hbigor_{i\in I}r_i$ and \Cref{lem:existing-extrema-order-characterization}} \\
    &\qimplies \forall i\in I.\ r_i=\bot,
  \end{align*}
  and the converse follows because $\bot$ is the smallest upper bound if all $r_i$ are $\bot$.

  The \Cref{eq:supremum-addition-bottom-fact} equivalence now follows immediately:
  \begin{align*}
    \left(\hbigor_{i\in I}s_i\right)\splus c=\bot
    &\qiff \hbigor_{i\in I}s_i=\bot \text{ and } c=\bot
      \inlineExplain{\Cref{eq:binary-splus-bottom-fact}} \\
    &\qiff \forall i\in I.\ s_i=\bot \text{ and } c=\bot
      \inlineExplain{\Cref{eq:supremum-bottom-fact}} \\
    &\qiff \forall i\in I.\ s_i\splus c=\bot
      \inlineExplain{\Cref{eq:binary-splus-bottom-fact}} \\
    &\qiff \hbigor_{i\in I}(s_i\splus c)=\bot
      \inlineExplain{\Cref{eq:supremum-bottom-fact}}.
  \end{align*}

  For infima, both sides are equivalent to $\hbigand_{i\in I}s_i=\bot$ and $c=\bot$:
  \begin{alignat*}{3}
    1. \quad & \left(\hbigand_{i \in I} s_i\right) \splus c = \bot
    && \qiff \hbigand_{i \in I} s_i = \bot \text{ and } c = \bot
    && \inlineExplain{\Cref{eq:binary-splus-bottom-fact}}, \\
    2. \quad & \hbigand_{i \in I} (s_i \splus c) = \bot
    && \qimplies c \hleq \hbigand_{i \in I} (s_i \splus c)
    && \inlineExplain{$c \hleq s_i \splus c$ for all $i$} \\
    & && \qimplies c = \bot \qimplies \hbigand_{i \in I} s_i = \hbigand_{i \in I} (s_i \splus c) = \bot
    && \inlineExplain{$\bot = \snull$}.
  \end{alignat*}
  Conversely, $\hbigand_{i\in I}s_i=\bot$ and $c=\bot$ implies
  $\hbigand_{i\in I}(s_i\splus c)=\hbigand_{i\in I}s_i=\bot$.
\end{proof}

For the remainder, we (still) assume the standing \wHeyLo setting of an $\omega$-bicontinuous monoid-module over a monoid whose natural-order lattice is a bi-Heyting algebra. The following scalar lemma and prenexing laws hold in this general setting, without any additional assumptions on the monoid-module structure.

\begin{lemma}[Bottom Facts for Scalar Multiplication over Suprema]
  \label[lemma]{lem:scalar-multiplication-supremum-bottom}
  Let $a\in\monoid$.
  For every family $(p_v)_{v\in\Vals}$ in $\module$,
  \[
    a\stimes\hbigor_{v\in\Vals}p_v=\bot
    \qiff
    \hbigor_{v\in\Vals}(a\stimes p_v)=\bot.
  \]
\end{lemma}

\begin{proof}
  \begin{align*}
    a\stimes\hbigor_{v\in\Vals}p_v=\bot
    &\qimplies \forall v\in\Vals.\ a\stimes p_v\hleq a\stimes\hbigor_{v\in\Vals}p_v = \bot
      \inlineExplain{$p_v\hleq\hbigor_{v\in\Vals}p_v$ and \Cref{lem:scalar-multiplication-monotone}} \\
    &\qimplies \forall v\in\Vals.\ a\stimes p_v=\bot
      \inlineExplain{def. of $\bot$} \\
    &\qimplies \hbigor_{v\in\Vals}(a\stimes p_v)=\bot
      \inlineExplain{\Cref{eq:supremum-bottom-fact}}.
  \end{align*}

  Conversely, assume $\hbigor_{v\in\Vals}(a\stimes p_v)=\bot$.
  Then $a\stimes p_v=\bot$ for all $v\in\Vals$ by \Cref{eq:supremum-bottom-fact}.
  Enumerate $\Vals$ as $(v_n)_{n\in\Nats}$ and define
  \[
    u_0 \definedAs p_{v_0},
    \qquad
    u_{n+1}\definedAs u_n\hor p_{v_{n+1}}.
  \]
  Then $\hbigor_{n\in\Nats}u_n=\hbigor_{v\in\Vals}p_v$ by the construction in \Cref{lem:countable_sup_inf}, and $a\stimes u_n=\bot$ for all $n\in\Nats$ by induction:
  \begin{alignat*}{2}
    \text{Base:} \quad & a\stimes u_0
    &&= a\stimes p_{v_0}
     = \bot, \\
    \text{Step:} \quad & a\stimes u_{n+1}
    &&= a\stimes(u_n\hor p_{v_{n+1}}) \\
    &&&\hleq a\stimes(u_n\splus p_{v_{n+1}})
      \inlineExplain{\Cref{lem:join-below-weighted-addition} and \Cref{lem:scalar-multiplication-monotone}} \\
    &&&= (a\stimes u_n)\splus(a\stimes p_{v_{n+1}})
      \inlineExplain{distributivity of $\stimes$ over $\splus$} \\
    &&&= \bot\splus\bot
      \inlineExplain{induction hypothesis and assumption} \\
    &&&= \bot.
  \end{alignat*}
  Hence
  \begin{align*}
    a\stimes\hbigor_{v\in\Vals}p_v
    &= a\stimes\hbigor_{n\in\Nats}u_n
      \inlineExplain{join construction} \\
    &= \hbigor_{n\in\Nats}(a\stimes u_n)
      \inlineExplain{$\omega$-continuity of $\stimes$} \\
    &= \bot
      \inlineExplain{\Cref{eq:supremum-bottom-fact}}.
  \end{align*}
\end{proof}

\begin{theorem}[General Weighted Prenexing for Covalidity]
  \label[theorem]{thm:general-weighted-prenexing-covalidity}
  Let $a \in \monoid$ be a constant scalar, and let $\hhla,\hhlb \in \wHeyLo$ with $x \notin \freevars{\hhlb}$ and $y \notin \freevars{\hhla}$.
  Then the following covalidity equivalences hold:
  \begin{align*}
    \isCovalid{(\lsupquant{x}{\hhla}) \splus \hhlb}
    &\qiff
    \isCovalid{\lsupquant{x}{(\hhla \splus \hhlb)}},
    &
    \isCovalid{\hhla \splus (\lsupquant{y}{\hhlb})}
    &\qiff
    \isCovalid{\lsupquant{y}{(\hhla \splus \hhlb)}},
    \\
    \isCovalid{(\linfquant{x}{\hhla}) \splus \hhlb}
    &\qiff
    \isCovalid{\linfquant{x}{(\hhla \splus \hhlb)}},
    &
    \isCovalid{\hhla \splus (\linfquant{y}{\hhlb})}
    &\qiff
    \isCovalid{\linfquant{y}{(\hhla \splus \hhlb)}},
    \\
    \isCovalid{a \stimes (\lsupquant{x}{\hhla})}
    &\qiff
    \isCovalid{\lsupquant{x}{(a \stimes \hhla)}}.
  \end{align*}
\end{theorem}

\begin{proof}
  For an arbitrary state $\State \in \States$: %
  \begin{proofcases}
    \proofcase{Addition over supremum}
    \begin{align*}
      \evalState{(\lsupquant{x}{\hhla}) \splus \hhlb}=\bot
      &\qiff \left(\hbigor_{v\in\Vals}\evalStateSubstBy{\hhla}{x}{v}\right)\splus \evalState{\hhlb}=\bot
        \inlineExplain{definition of $\lsupquant{x}{}$} \\
      &\qiff \hbigor_{v\in\Vals}(\evalStateSubstBy{\hhla}{x}{v}\splus \evalState{\hhlb})=\bot
        \inlineExplain{\Cref{lem:weighted-addition-bottom-facts}} \\
      &\qiff \hbigor_{v\in\Vals}(\evalStateSubstBy{\hhla}{x}{v}\splus \evalStateSubstBy{\hhlb}{x}{v})=\bot
        \inlineExplain{$x\notin\freevars{\hhlb}$} \\
      &\qiff \evalState{\lsupquant{x}{(\hhla \splus \hhlb)}}=\bot
        \inlineExplain{definition of $\lsupquant{x}{}$}.
    \end{align*}
    \begin{align*}
      \evalState{\hhla \splus (\lsupquant{y}{\hhlb})}=\bot
      &\qiff \evalState{\hhla}\splus\left(\hbigor_{v\in\Vals}\evalStateSubstBy{\hhlb}{y}{v}\right)=\bot
        \inlineExplain{definition of $\lsupquant{y}{}$} \\
      &\qiff \hbigor_{v\in\Vals}(\evalState{\hhla}\splus \evalStateSubstBy{\hhlb}{y}{v})=\bot
        \inlineExplain{\Cref{lem:weighted-addition-bottom-facts}} \\
      &\qiff \hbigor_{v\in\Vals}(\evalStateSubstBy{\hhla}{y}{v}\splus \evalStateSubstBy{\hhlb}{y}{v})=\bot
        \inlineExplain{$y\notin\freevars{\hhla}$} \\
      &\qiff \evalState{\lsupquant{y}{(\hhla \splus \hhlb)}}=\bot
        \inlineExplain{definition of $\lsupquant{y}{}$}.
    \end{align*}

    \proofcase{Addition over infimum}
    \begin{align*}
      \evalState{(\linfquant{x}{\hhla}) \splus \hhlb}=\bot
      &\qiff \left(\hbigand_{v\in\Vals}\evalStateSubstBy{\hhla}{x}{v}\right)\splus \evalState{\hhlb}=\bot
        \inlineExplain{definition of $\linfquant{x}{}$} \\
      &\qiff \hbigand_{v\in\Vals}(\evalStateSubstBy{\hhla}{x}{v}\splus \evalState{\hhlb})=\bot
        \inlineExplain{\Cref{lem:weighted-addition-bottom-facts}} \\
      &\qiff \hbigand_{v\in\Vals}(\evalStateSubstBy{\hhla}{x}{v}\splus \evalStateSubstBy{\hhlb}{x}{v})=\bot
        \inlineExplain{$x\notin\freevars{\hhlb}$} \\
      &\qiff \evalState{\linfquant{x}{(\hhla \splus \hhlb)}}=\bot
        \inlineExplain{definition of $\linfquant{x}{}$}.
    \end{align*}
    \begin{align*}
      \evalState{\hhla \splus (\linfquant{y}{\hhlb})}=\bot
      &\qiff \evalState{\hhla}\splus\left(\hbigand_{v\in\Vals}\evalStateSubstBy{\hhlb}{y}{v}\right)=\bot
        \inlineExplain{definition of $\linfquant{y}{}$} \\
      &\qiff \hbigand_{v\in\Vals}(\evalState{\hhla}\splus \evalStateSubstBy{\hhlb}{y}{v})=\bot
        \inlineExplain{\Cref{lem:weighted-addition-bottom-facts}} \\
      &\qiff \hbigand_{v\in\Vals}(\evalStateSubstBy{\hhla}{y}{v}\splus \evalStateSubstBy{\hhlb}{y}{v})=\bot
        \inlineExplain{$y\notin\freevars{\hhla}$} \\
      &\qiff \evalState{\linfquant{y}{(\hhla \splus \hhlb)}}=\bot
        \inlineExplain{definition of $\linfquant{y}{}$}.
    \end{align*}

    \proofcase{Scalar multiplication over supremum}
    \begin{align*}
      \evalState{a\stimes(\lsupquant{x}{\hhla})}=\bot
      &\qiff a\stimes\hbigor_{v\in\Vals}\evalStateSubstBy{\hhla}{x}{v}=\bot
        \inlineExplain{definition of $\lsupquant{x}{}$} \\
      &\qiff \hbigor_{v\in\Vals}(a\stimes \evalStateSubstBy{\hhla}{x}{v})=\bot
        \inlineExplain{\Cref{lem:scalar-multiplication-supremum-bottom}} \\
      &\qiff \evalState{\lsupquant{x}{(a\stimes\hhla)}}=\bot
        \inlineExplain{definitions of $\stimes$ and $\lsupquant{x}{}$}.
    \end{align*}
  \end{proofcases}
\end{proof}

The analogous scalar rule for infimum quantifiers is not included here: in the general monoid-module setting, $a\stimes\bigsqcap_x\hhla$ does not need to preserve covalidity in both directions without additional assumptions.

The preceding theorem gives judgment-preservation facts rather than formula equivalences.
For the goal-directed algorithm, these general covalidity-preservation facts would in principle be enough for the weighted rules in the upper branch.
However, the lower branch also moves weighted addition across infimum quantifiers and there we need validity preservation.
To obtain a soundness statement for both branches, and to understand when the stronger formula equivalences are available, we now demand that the monoid-module fulfills additional assumptions.

Conceptually, the approach is based on the idea that similarly to \Cref{lem:scalar-multiplication-supremum-bottom,lem:countable_sup_inf} we would like to turn the family we quantify over into a countable chain to utilize the $\omega$-bicontinuity assumptions. The current construction lets us turn a countable family $(p_v)_{v\in\Vals}$ into an ascending or descending chain $u_n$ with the same supremum or infimum, e.g. $\hbigor_n u_n = \hbigor_{v\in\Vals} p_v$.
However, by $\omega$-continuity of scalar multiplication we get $a\stimes\hbigor_{v\in\Vals} p_v=a\stimes\hbigor_n u_n=\hbigor_n(a\stimes u_n)$, but the prenexed expression would be $\hbigor_{v\in\Vals}(a\stimes p_v)$. These are not the same unless scalar multiplication also preserves the finite joins used to build $u_n$, i.e. unless we know $a\stimes(p\hor q) = (a\stimes p)\hor(a\stimes q)$.

\begin{lemma}[Countable Extrema Preservation for Weighted Maps]
  \label[lemma]{lem:weighted-map-countable-extrema-preservation}
  Let $\moduledef$ be an $\omega$-bicontinuous monoid-module over a monoid $\Monoid$ whose natural-order lattice is a bi-Heyting algebra, and let $(p_v)_{v\in\Vals}$ be a family in $\module$.
  Then:
  \begin{enumerate}[label=(\roman*)]
    \item If a fixed $c\in\module$ satisfies
    \[
      (p\hor q)\splus c = (p\splus c)\hor(q\splus c)
      \qquad
      \text{for all } p,q\in\module,
    \]
    then
    \[
      \left(\hbigor_{v\in\Vals}p_v\right)\splus c
      =
      \hbigor_{v\in\Vals}(p_v\splus c).
    \]
    \item If a fixed $c\in\module$ satisfies
    \[
      (p\hand q)\splus c = (p\splus c)\hand(q\splus c)
      \qquad
      \text{for all } p,q\in\module,
    \]
    then
    \[
      \left(\hbigand_{v\in\Vals}p_v\right)\splus c
      =
      \hbigand_{v\in\Vals}(p_v\splus c).
    \]
    \item If a fixed $a\in\monoid$ satisfies
    \[
      a\stimes(p\hor q) = (a\stimes p)\hor(a\stimes q)
      \qquad
      \text{for all } p,q\in\module,
    \]
    then
    \[
      a\stimes\hbigor_{v\in\Vals}p_v
      =
      \hbigor_{v\in\Vals}(a\stimes p_v).
    \]
    \item If a fixed $a\in\monoid$ satisfies
    \[
      a\stimes(p\hand q) = (a\stimes p)\hand(a\stimes q)
      \qquad
      \text{for all } p,q\in\module,
    \]
    then
    \[
      a\stimes\hbigand_{v\in\Vals}p_v
      =
      \hbigand_{v\in\Vals}(a\stimes p_v).
    \]
  \end{enumerate}
\end{lemma}

\begin{proof}
  We prove the first and fourth case; the second and third is identical, replacing ascending chains by descending chains or addition by scalar multiplication.
  For (i), enumerate $\Vals$ as $(v_n)_{n\in\Nats}$ and define
  \[
    u_0\definedAs p_{v_0},
    \qquad
    u_{n+1}\definedAs u_n\hor p_{v_{n+1}}.
  \]
  Then $(u_n)_{n\in\Nats}$ is ascending and
  $\hbigor_{n\in\Nats}u_n=\hbigor_{v\in\Vals}p_v$ according to the proof of \Cref{lem:countable_sup_inf}.
  To see that $(u_n\splus c)_{n\in\Nats}$ is the corresponding chain for $(p_v\splus c)_{v\in\Vals}$, define
  \[
    w_0\definedAs p_{v_0}\splus c,
    \qquad
    w_{n+1}\definedAs w_n\hor(p_{v_{n+1}}\splus c).
  \]
  We show by induction that $w_n=u_n\splus c$ for all $n$.
  The base case is immediate.
  For the induction step, using the binary join-preservation assumption,
  \[
    w_{n+1}
    =
    (u_n\splus c)\hor(p_{v_{n+1}}\splus c)
    =
    (u_n\hor p_{v_{n+1}})\splus c
    =
    u_{n+1}\splus c.
  \]
  Therefore, again by the construction in \Cref{lem:countable_sup_inf},
  \[
    \hbigor_{n\in\Nats}(u_n\splus c)
    =
    \hbigor_{n\in\Nats}w_n
    =
    \hbigor_{v\in\Vals}(p_v\splus c).
  \]
  Hence
  \[
    \left(\hbigor_{v\in\Vals}p_v\right)\splus c
    =
    \left(\hbigor_{n\in\Nats}u_n\right)\splus c
    =
    \hbigor_{n\in\Nats}(u_n\splus c)
    =
    \hbigor_{v\in\Vals}(p_v\splus c),
  \]
  where the middle equality uses $\omega$-continuity of weighted addition.

  For (iv), enumerate $\Vals$ as $(v_n)_{n\in\Nats}$ and define
  \[
    d_0\definedAs p_{v_0},
    \qquad
    d_{n+1}\definedAs d_n\hand p_{v_{n+1}}.
  \]
  Then $(d_n)_{n\in\Nats}$ is descending and
  $\hbigand_{n\in\Nats}d_n=\hbigand_{v\in\Vals}p_v$ according to the proof of \Cref{lem:countable_sup_inf}.
  To identify $(a\stimes d_n)_{n\in\Nats}$ with the chain generated by $(a\stimes p_v)_{v\in\Vals}$, define
  \[
    e_0\definedAs a\stimes p_{v_0},
    \qquad
    e_{n+1}\definedAs e_n\hand(a\stimes p_{v_{n+1}}).
  \]
  Again $e_n=a\stimes d_n$ follows by induction: the base case is immediate, and
  \[
    e_{n+1}
    =
    (a\stimes d_n)\hand(a\stimes p_{v_{n+1}})
    =
    a\stimes(d_n\hand p_{v_{n+1}})
    =
    a\stimes d_{n+1}
  \]
  by the binary meet-preservation assumption.
  Thus
  \[
    \hbigand_{n\in\Nats}(a\stimes d_n)
    =
    \hbigand_{v\in\Vals}(a\stimes p_v).
  \]
  Thus
  \[
    a\stimes\hbigand_{v\in\Vals}p_v
    =
    a\stimes\hbigand_{n\in\Nats}d_n
    =
    \hbigand_{n\in\Nats}(a\stimes d_n)
    =
    \hbigand_{v\in\Vals}(a\stimes p_v),
  \]
  where the middle equality uses $\omega$-cocontinuity of scalar multiplication.
\end{proof}

The next theorem brings these preservation properties into the form of the prenexing laws used by the prenexing algorithm.
\begin{theorem}[Weighted Prenexing]
  \label[theorem]{thm:arithmetic-prenexing}
  Assume that for every $c\in\module$ and every $a\in\monoid$,
  \begin{align}
    (p\hor q)\splus c &= (p\splus c)\hor(q\splus c),
    &
    (p\hand q)\splus c &= (p\splus c)\hand(q\splus c),
    \label{eq:weighted-addition-binary-extrema-preservation}
    \\
    a\stimes(p\hor q) &= (a\stimes p)\hor(a\stimes q)
    \label{eq:scalar-binary-join-preservation}
  \end{align}
  for all $p,q\in\module$.
  Let $a \in \monoid$ be a constant scalar, and let $\hhla,\hhlb \in \wHeyLo$ with $x \notin \freevars{\hhlb}$ and $y \notin \freevars{\hhla}$.
  Then:
  \begin{align*}
    (\lsupquant{x}{\hhla}) \splus \hhlb
    &\equiv
    \lsupquant{x}{(\hhla \splus \hhlb)},
    &
    \hhla \splus (\lsupquant{y}{\hhlb})
    &\equiv
    \lsupquant{y}{(\hhla \splus \hhlb)},
    \\
    (\linfquant{x}{\hhla}) \splus \hhlb
    &\equiv
    \linfquant{x}{(\hhla \splus \hhlb)},
    &
    \hhla \splus (\linfquant{y}{\hhlb})
    &\equiv
    \linfquant{y}{(\hhla \splus \hhlb)},
    \\
    a \stimes (\lsupquant{x}{\hhla})
    &\equiv
    \lsupquant{x}{(a \stimes \hhla)}.
  \end{align*}
  If, additionally, this scalar $a$ preserves binary meets, i.e.
  \[
    a\stimes(p\hand q)=(a\stimes p)\hand(a\stimes q)
    \qquad
    \text{for all } p,q\in\module,
  \]
  then also
  \[
    a\stimes(\linfquant{x}{\hhla})
    \equiv
    \linfquant{x}{(a\stimes\hhla)}.
  \]
\end{theorem}

\begin{proof}
  We show equality of denotations in an arbitrary state $\State\in\States$.
  For the first addition-over-supremum clause:
  \begin{align*}
    \evalState{(\lsupquant{x}{\hhla})\splus\hhlb}
    &=
    \left(\hbigor_{v\in\Vals}\evalStateSubstBy{\hhla}{x}{v}\right)\splus \evalState{\hhlb} \\
    &=
    \hbigor_{v\in\Vals}(\evalStateSubstBy{\hhla}{x}{v}\splus \evalState{\hhlb})
    \inlineExplain{\Cref{lem:weighted-map-countable-extrema-preservation}} \\
    &=
    \hbigor_{v\in\Vals}(\evalStateSubstBy{\hhla}{x}{v}\splus\evalStateSubstBy{\hhlb}{x}{v})
      \inlineExplain{$x\notin\freevars{\hhlb}$} \\
    &=
    \evalState{\lsupquant{x}{(\hhla\splus\hhlb)}}.
  \end{align*}
  The right-hand addition-over-supremum clause is symmetric.

  The two addition-over-infimum clauses are proved in the same way, using the meet-preservation part of \Cref{lem:weighted-map-countable-extrema-preservation}.

  For scalar multiplication over supremum:
  \begin{align*}
    \evalState{a\stimes(\lsupquant{x}{\hhla})}
    &=
    a\stimes\hbigor_{v\in\Vals}\evalStateSubstBy{\hhla}{x}{v} \\
    &=
    \hbigor_{v\in\Vals}(a\stimes \evalStateSubstBy{\hhla}{x}{v})
      \inlineExplain{\Cref{lem:weighted-map-countable-extrema-preservation}} \\
    &=
    \evalState{\lsupquant{x}{(a\stimes\hhla)}}.
  \end{align*}
  The scalar-over-infimum clause follows analogously from the additional binary meet-preservation assumption for $a$.
\end{proof}

We now connect the preservation assumptions in \Cref{thm:arithmetic-prenexing} to the main classes of domains used in this paper, as summarized in \Cref{fig:modules}.
The point is that totality, Booleanity of the natural-order lattice, and idempotence of weighted addition are distinct assumptions: Booleanity gives the logical connectives, whereas idempotence identifies weighted addition with lattice join.

\begin{corollary}[Totally Ordered Weighted Prenexing]
  \label[corollary]{cor:totally-ordered-weighted-prenexing}
  If the natural order of the monoid-module is total, then all weighted prenexing equivalences from \Cref{thm:arithmetic-prenexing}, including scalar multiplication over infimum quantifiers, hold.
  This covers the totally ordered weighted domains considered in \Cref{fig:modules}.
\end{corollary}

\begin{proof}
  Fix $c\in\module$, $a\in\monoid$, and $p,q\in\module$.
  Since the natural order is total, assume without loss of generality that $p\sleq q$.
  Then $p\hor q=q$ and $p\hand q=p$.
  By \Cref{lem:weighted-addition-monotone,lem:scalar-multiplication-monotone}, weighted addition and scalar multiplication are monotone, so
  $p\splus c\sleq q\splus c$ and $a\stimes p\sleq a\stimes q$.
  Hence
  \[
    (p\hor q)\splus c=(p\splus c)\hor(q\splus c),
    \qquad
    (p\hand q)\splus c=(p\splus c)\hand(q\splus c),
  \]
  and likewise
  \[
    a\stimes(p\hor q)=(a\stimes p)\hor(a\stimes q),
    \qquad
    a\stimes(p\hand q)=(a\stimes p)\hand(a\stimes q).
  \]
  Thus the addition assumptions, the scalar join assumption, and the additional scalar meet-preservation assumption of \Cref{thm:arithmetic-prenexing} all hold.
  Hence \Cref{thm:arithmetic-prenexing} applies.
\end{proof}

\begin{lemma}[Idempotent Addition Is Lattice Join]
  \label[lemma]{lem:idempotent-addition-is-join}
  Let $\moduledef$ be a monoid-module whose natural order is a partial order and admits binary joins.
  Then
  \[
    \forall p\in\module.\ p\splus p=p
    \qquad\qiff\qquad
    \forall p,q\in\module.\ p\splus q=p\hor q.
  \]
\end{lemma}

\begin{proof}
  The $\Leftarrow$ direction is immediate because $p\hor p=p$.
  For the $\Rightarrow$ direction, assume that $\splus$ is idempotent.
  The element $p\splus q$ is an upper bound of $p$ and $q$ by the definition of the natural order.
  Let $r$ be any upper bound of $p$ and $q$.
  Then $p\sleq r$ and $q\sleq r$, so choose $m,n\in\module$ with $p\splus m=r$ and $q\splus n=r$.
  By idempotence,
  \[
    p\splus r=p\splus p\splus m=r,
    \qquad
    q\splus r=q\splus q\splus n=r.
  \]
  Hence $(p\splus q)\splus r=p\splus(q\splus r)=p\splus r=r$, and therefore $p\splus q\sleq r$.
  Thus $p\splus q$ is the least upper bound of $p$ and $q$.
\end{proof}

\begin{corollary}[Idempotent Weighted Prenexing]
  \label[corollary]{cor:idempotent-weighted-prenexing}
  If weighted addition is idempotent, then the addition rules and the scalar-over-supremum rule from \Cref{thm:arithmetic-prenexing} hold.
  This applies in particular when $\splus$ is disjunction or union, as in the Boolean semiring and the formal-language module.
  If, additionally, each scalar map $p\mapsto a\stimes p$ preserves binary meets, then scalar multiplication also prenexes over infimum quantifiers.
\end{corollary}

\begin{proof}
  By \Cref{lem:idempotent-addition-is-join}, weighted addition coincides with the lattice join.
  Thus $p\mapsto p\splus c$ is the map $p\mapsto p\hor c$, which preserves binary joins by associativity of $\hor$ and binary meets by distributivity of the bi-Heyting lattice:
  \[
    (p\hand q)\hor c = (p\hor c)\hand(q\hor c).
  \]
  Moreover,
  \[
    a\stimes(p\hor q)
    =
    a\stimes(p\splus q)
    =
    (a\stimes p)\splus(a\stimes q)
    =
    (a\stimes p)\hor(a\stimes q),
  \]
  so every scalar map preserves binary joins.
  The claims follow from \Cref{thm:arithmetic-prenexing}.
\end{proof}

\paragraph{A Note on Boolean natural-order lattices.}
A Boolean natural-order lattice gives the Boolean implication and coimplication, but it does not by itself demand the required behavior of $\splus$ or $\stimes$.
The concrete Boolean examples in \Cref{fig:modules} satisfy the stronger idempotent-addition assumption above: in the Boolean semiring addition is disjunction, and in the formal-language module addition is union.
In these examples scalar meet preservation also holds, because Boolean multiplication and fixed-word prefixing preserve intersections.

\begin{example}[Counterexample: Scalar Infimum Prenexing]
  \label{example:counterexample-arithmetic-prenexing}
  The previous formal-language counterexample used a language such as $\{c,cc\}$ as a scalar, but in the formal-language monoid-module $\modlang$ scalars are words $g \in \Gamma^*$, not languages.
  Thus that construction was not a counterexample for our monoid-module framework.

  Consider the expectation monoid-module $\modprob$ with its componentwise natural order.
  Let the scalar be $a=(1,1)$, and define
  \[
    \hhla \definedAs (\embed{x<1}\hand (0,100)) \hor (\embed{\neg(x<1)}\hand (1,0)).
  \]
  Then $\interpretsimpleStateSubstBy{\hhla}{x}{0}=(0,100)$ and $\interpretsimpleStateSubstBy{\hhla}{x}{n}=(1,0)$ for all $n\ge 1$, hence
  \[
    a\otimes\interpretsimpleState{\bigsqcap_x\hhla}
    = (1,1)\otimes(0,0)
    = (0,0),
  \]
  \begin{align*}
    \interpretsimpleState{\bigsqcap_x(a\otimes\hhla)}
    &= ((1,1)\otimes(0,100)) \sqcap ((1,1)\otimes(1,0)) \\
    &= (0,100) \hand (1, 1) \\
    &= (0,1).
  \end{align*}

  Therefore $a\stimes\bigsqcap_x\hhla \not\equiv \bigsqcap_x(a\stimes\hhla)$.
  Thus scalar multiplication does not prenex over infimum quantifiers in the general monoid-module setting.
\end{example}

\begin{lemma}[Conditional Weighted Prenexing]
  \label[lemma]{lem:conditional-arithmetic-prenexing}
  Let $a \in \monoid$ and $\hhla \in \wHeyLo$.
  If $\hhla$ is boolean-valued (i.e., $\Ima \interpretsimple{\hhla} \subseteq \{\bot, \top\}$), then scalar multiplication distributes over the infimum:
  \[
    a \stimes \left(\bigsqcap_x \hhla\right) \equiv \bigsqcap_x (a \stimes \hhla).
  \]
\end{lemma}

\begin{proof}
  Let $\State \in \States$. Let $S = \{ \interpretsimpleStateSubstBy{\hhla}{x}{v} \mid v \in \Vals \}$ be the set of values over which the infimum is taken. Since $\hhla$ is boolean-valued, $S \subseteq \{\bot, \top\}$.
  We distinguish two cases based on the value of the quantifier:

  \begin{proofcases}
    \proofcase{1: $\interpretsimpleState{\bigsqcap_x \hhla} = \bot$}
    By the definition of infimum, this implies that $\bot \in S$, because $\Ima \interpretsimple{\hhla} \subseteq \{\bot, \top\}$.
    \begin{itemize}
      \item LHS: $a \stimes \bot = \bot$ by the annihilation axiom of monoid-modules.
      \item RHS: Since $\bot \in S$, the set of distributed values $\{ a \stimes s \mid s \in S \}$ contains the element $a \stimes \bot$. Since $a \stimes \bot = \bot$ and $\bot$ is the least element, the infimum of this set is $\bot$.
    \end{itemize}
    Thus, LHS = RHS = $\bot$.

    \proofcase{2: $\interpretsimpleState{\bigsqcap_x \hhla} = \top$}
    This implies that for all $v$, $\interpretsimpleStateSubstBy{\hhla}{x}{v} = \top$. Thus, $S = \{\top\}$.
    \begin{itemize}
      \item LHS: $a \stimes \top$.
      \item RHS: The set of distributed values is $\{ a \stimes \top \}$. The infimum is trivially $a \stimes \top$.
    \end{itemize}
    Thus, LHS = RHS = $a \stimes \top$.
  \end{proofcases}
\end{proof}

\paragraph{Summary of weighted prenexing assumptions.}
In the table below, $\checkmark$ means that the corresponding formula equivalence is guaranteed, $\mathrm{cov}$ means that only the covalidity-preservation statement from \Cref{thm:general-weighted-prenexing-covalidity} is guaranteed, and -- means that the property is not forced by the assumption.
The hook arrows indicate strengthening implications between the assumptions, i.e. the second row implies the first.
\begin{center}
  \small
  \renewcommand{\arraystretch}{1.14}
  \begin{tabularx}{\linewidth}{@{}>{\raggedright\arraybackslash}Xcccc@{}}
    \toprule
    Assumption
      & \makecell{$\splus$\\over $\bigsqcup$}
      & \makecell{$\splus$\\over $\bigsqcap$}
      & \makecell{$\stimes$\\over $\bigsqcup$}
      & \makecell{$\stimes$\\over $\bigsqcap$} \\
    \midrule
    General $\omega$-bicontinuous module
      & $\mathrm{cov}$ & $\mathrm{cov}$ & $\mathrm{cov}$ & -- \\
    Total natural order
      & $\checkmark$ & $\checkmark$ & $\checkmark$ & $\checkmark$ \\
    Boolean natural-order lattice
      & -- & -- & -- & -- \\
    \(\hookrightarrow\) Idempotent weighted addition
      & $\checkmark$ & $\checkmark$ & $\checkmark$ & -- \\
    \(\hookrightarrow\;\hookrightarrow\) Idempotent addition and scalar meet preservation
      & $\checkmark$ & $\checkmark$ & $\checkmark$ & $\checkmark$ \\
    \bottomrule
  \end{tabularx}
\end{center}
The total-order row covers the totally ordered examples from \Cref{fig:modules}, such as \(\modtrop\), \(\modnats\), and \(\modprob\).
The last row covers the concrete Boolean examples there, such as \(\modbool\) and \(\modlang\).

\subsection{Goal-Directed \texorpdfstring{$\mathsf{qelim}$}{qelim}}

In this subsection, we fix the standing assumptions from \Cref{sec:weighted_programming,sec:heyting} together with the \wHeyLo semantics from \Cref{fig:wheylo_semantics}: the underlying domain is an $\omega$-bicontinuous monoid-module over a monoid, together with an implication and coimplication such that the induced weighting domain forms a bi-Heyting algebra. Thus, for every $\hhla \in \wHeyLo$,
\[
  \isValid{\hhla} \qiff \forall \State \in \States.~\evalState{\hhla} = \top,
  \qquad
  \isCovalid{\hhla} \qiff \forall \State \in \States.~\evalState{\hhla} = \bot.
\]

The proof of \Cref{thm:qelim-soundness} proceeds in three steps. First, auxiliary procedures $\downPrenex{\cdot}$ and $\upPrenex{\cdot}$ move quantifiers to the front while preserving formula equivalence. Second, auxiliary procedures $\downStrip{\cdot}{\cdot}$ and $\upStrip{\cdot}{\cdot}$ remove the resulting leading quantifier block. Third, we show that the displayed procedure in \Cref{fig:qelim-algorithm} is exactly $\mathsf{strip}\circ\mathsf{prenex}$. Completeness is then treated separately on a syntactic fragment where the prenexed matrix is quantifier-free.

\noindent
\begin{definition}[Quantifier-Free Matrices and Prenex Outputs]
  \label[definition]{def:qelim-prenex-fragments}
  We write $\psi \in \qfFrag$ iff $\psi$ contains no occurrences of $\bigsqcap_x$ or $\bigsqcup_x$. We write $\epsilon$ for the empty variable list, $x\vec x$ for prefixing a variable to a list, and $\vec x\vec y$ for concatenation.

  For a variable list $\vec x$, define iterated quantification by
  \[
    \bigsqcap_{\epsilon} \psi \definedAs \psi,
    \qquad
    \bigsqcap_{x\vec x} \psi \definedAs \bigsqcap_x \bigsqcap_{\vec x} \psi,
  \]
  and dually
  \[
    \bigsqcup_{\epsilon} \psi \definedAs \psi,
    \qquad
    \bigsqcup_{x\vec x} \psi \definedAs \bigsqcup_x \bigsqcup_{\vec x} \psi.
  \]

  We identify lower prenex outputs with formulas by
  \[
    \langle \vec x,\psi\rangle \in \downPFrag \quad\widehat{=}\quad \bigsqcap_{\vec x}\psi,
  \]
  and dually upper prenex outputs by
  \[
    \langle \vec x,\psi\rangle \in \upPFrag \quad\widehat{=}\quad \bigsqcup_{\vec x}\psi.
  \]
  Accordingly, we write
  \[
    \downPFrag \definedAs \{ \langle \vec x, \psi \rangle \mid \psi \in \wHeyLo \},
    \qquad
    \upPFrag \definedAs \{ \langle \vec x, \psi \rangle \mid \psi \in \wHeyLo \}.
  \]
\end{definition}

\noindent
\begin{definition}[Algorithm \texorpdfstring{$\mathsf{prenex}$}{prenex}]
  \label[definition]{def:qelim-prenex}
  \Cref{fig:prenex-algorithm} defines functions $\downPrenex{\cdot} : \wHeyLo \to \downPFrag$ and $\upPrenex{\cdot} : \wHeyLo \to \upPFrag$. If $\downPrenex{\hhla}=\langle \vec x,\psi\rangle$, then $\hhla \equiv \bigsqcap_{\vec x}\psi$; if $\upPrenex{\hhla}=\langle \vec x,\psi\rangle$, then $\hhla \equiv \bigsqcup_{\vec x}\psi$. Moreover, by \Cref{thm:prenex-complete-fragment},
  \begin{align*}
    \hhla \in \downQFrag &\Longrightarrow
      \downPrenex{\hhla}=\langle \vec x,\psi\rangle \text{ with } \psi \in \qfFrag,\\
    \hhla \in \upQFrag &\Longrightarrow
      \upPrenex{\hhla}=\langle \vec x,\psi\rangle \text{ with } \psi \in \qfFrag.
  \end{align*}
\end{definition}

\begin{figure}
  \centering
  \renewcommand{\arraystretch}{1.24}
  \begin{minipage}[t]{0.49\linewidth}
    \vspace{0pt}
    \centering
    \textbf{Lower Branch}\\[-0.1em]
    $\downPrenex{\cdot}$ (prenexing)
    \[
      \begin{array}{@{}r@{\,\mapsto\,}l@{}}
        \alpha & \langle \epsilon,\alpha \rangle \quad \text{for atomic $\alpha$, $\top$, $\bot$}
        \\[0.35em]
        \heylovalidate{\Psi} & \langle \vec x,\heylovalidate{\psi}\rangle
        \\[0.05em]
        \multicolumn{2}{@{}l@{}}{\qquad \text{where $\downPrenex{\Psi}=\langle \vec x,\psi\rangle$}}
        \\[0.35em]
        \hhla_1 \mathbin{\star} \hhla_2 & \langle \vec x \vec y,\psi \mathbin{\star} \chi\rangle
        \\[0.05em]
        \multicolumn{2}{@{}l@{}}{\qquad \text{where $\downPrenex{\hhla_1}=\langle \vec x,\psi\rangle$}}
        \\[0.05em]
        \multicolumn{2}{@{}l@{}}{\qquad \text{and $\downPrenex{\hhla_2}=\langle \vec y,\chi\rangle$}}
        \\[0.05em]
        \multicolumn{2}{@{}l@{}}{\qquad \text{for $\star \in \{\sqcap,\sqcup,\splus\}$}}
        \\[0.35em]
        a \stimes \Psi & \langle \vec x,a \stimes \psi\rangle
        \\[0.05em]
        \multicolumn{2}{@{}l@{}}{\qquad \text{where $\downPrenex{\Psi}=\langle \vec x,\psi\rangle$}}
        \\[0.35em]
        \hhla_1 \impl \hhla_2 & \langle \vec x \vec y,\psi \impl \chi\rangle
        \\[0.05em]
        \multicolumn{2}{@{}l@{}}{\qquad \text{where $\upPrenex{\hhla_1}=\langle \vec x,\psi\rangle$}}
        \\[0.05em]
        \multicolumn{2}{@{}l@{}}{\qquad \text{and $\downPrenex{\hhla_2}=\langle \vec y,\chi\rangle$}}
        \\[0.35em]
        \bigsqcap_x \Psi & \langle x\vec x,\psi\rangle
        \\[0.05em]
        \multicolumn{2}{@{}l@{}}{\qquad \text{where $\downPrenex{\Psi}=\langle \vec x,\psi\rangle$}}
        \\[0.35em]
        \hhla & \langle \epsilon,\hhla\rangle \quad \text{otherwise}
      \end{array}
    \]
  \end{minipage}
  \hfill
  \begin{minipage}[t]{0.49\linewidth}
    \vspace{0pt}
    \centering
    \textbf{Upper Branch}\\[-0.1em]
    $\upPrenex{\cdot}$ (prenexing)
    \[
      \begin{array}{@{}r@{\,\mapsto\,}l@{}}
        \alpha & \langle \epsilon,\alpha \rangle \quad \text{for atomic $\alpha$, $\top$, $\bot$}
        \\[0.35em]
        \heylocovalidate{\Psi} & \langle \vec x,\heylocovalidate{\psi}\rangle
        \\[0.05em]
        \multicolumn{2}{@{}l@{}}{\qquad \text{where $\upPrenex{\Psi}=\langle \vec x,\psi\rangle$}}
        \\[0.35em]
        \hhla_1 \mathbin{\star} \hhla_2 & \langle \vec x \vec y,\psi \mathbin{\star} \chi\rangle
        \\[0.05em]
        \multicolumn{2}{@{}l@{}}{\qquad \text{where $\upPrenex{\hhla_1}=\langle \vec x,\psi\rangle$}}
        \\[0.05em]
        \multicolumn{2}{@{}l@{}}{\qquad \text{and $\upPrenex{\hhla_2}=\langle \vec y,\chi\rangle$}}
        \\[0.05em]
        \multicolumn{2}{@{}l@{}}{\qquad \text{for $\star \in \{\sqcap,\sqcup,\splus\}$}}
        \\[0.35em]
        a \stimes \Psi & \langle \vec x,a \stimes \psi\rangle
        \\[0.05em]
        \multicolumn{2}{@{}l@{}}{\qquad \text{where $\upPrenex{\Psi}=\langle \vec x,\psi\rangle$}}
        \\[0.35em]
        \hhla_1 \coimpl \hhla_2 & \langle \vec x \vec y,\psi \coimpl \chi\rangle
        \\[0.05em]
        \multicolumn{2}{@{}l@{}}{\qquad \text{where $\downPrenex{\hhla_1}=\langle \vec x,\psi\rangle$}}
        \\[0.05em]
        \multicolumn{2}{@{}l@{}}{\qquad \text{and $\upPrenex{\hhla_2}=\langle \vec y,\chi\rangle$}}
        \\[0.35em]
        \bigsqcup_x \Psi & \langle x\vec x,\psi\rangle
        \\[0.05em]
        \multicolumn{2}{@{}l@{}}{\qquad \text{where $\upPrenex{\Psi}=\langle \vec x,\psi\rangle$}}
        \\[0.35em]
        \hhla & \langle \epsilon,\hhla\rangle \quad \text{otherwise}
      \end{array}
    \]
  \end{minipage}
  \caption{Algorithm $\mathsf{prenex}$. The lower branch produces a prefix of infimum quantifiers, and the upper branch a prefix of supremum quantifiers. In binary cases, prefixes are concatenated in left-to-right order.}
  \label{fig:prenex-algorithm}
\end{figure}

\noindent
\begin{lemma}[Validation and Covalidation Reflection]
  \label[lemma]{lem:qelim-validate-covalidate}
  For every $\hhla \in \wHeyLo$:
  \[
    \isValid{\heylovalidate{\hhla}}
    \qiff
    \isValid{\hhla},
    \qquad
    \isCovalid{\heylocovalidate{\hhla}}
    \qiff
    \isCovalid{\hhla}.
  \]
\end{lemma}

\begin{proof}
  Let $\State \in \States$.
  By definition, $\evalState{\heylovalidate{\hhla}} = \top$ iff $\evalState{\hhla} = \top$. Hence $\isValid{\heylovalidate{\hhla}} \qiff \isValid{\hhla}$. Dually, $\evalState{\heylocovalidate{\hhla}} = \bot$ iff $\evalState{\hhla} = \bot$, so $\isCovalid{\heylocovalidate{\hhla}} \qiff \isCovalid{\hhla}$.
\end{proof}

\noindent
\begin{theorem}[Soundness of $\mathsf{prenex}$]
  \label[theorem]{thm:prenex-soundness}
  Assume the hypotheses of \Cref{thm:arithmetic-prenexing}. Then:
  \begin{enumerate}[label=(\roman*)]
    \item If $\downPrenex{\hhla}=\langle \vec x, \psi \rangle$, then $\hhla \equiv \bigsqcap_{\vec x} \psi$.
    \item If $\upPrenex{\hhla}=\langle \vec x, \psi \rangle$, then $\hhla \equiv \bigsqcup_{\vec x} \psi$.
  \end{enumerate}
\end{theorem}

\begin{proof}
  This is a direct mutual induction: each recursive clause of \Cref{fig:prenex-algorithm} is justified by the corresponding prenexing theorem stated above.
  We prove both clauses simultaneously by induction on the recursive evaluation of \Cref{fig:prenex-algorithm}.

  \paragraph{Lower branch.}
  \begin{proofcases}
    \proofcase{$\alpha$, $\top$, or $\bot$}
    Immediate.

    \proofcase{$\heylovalidate{\Psi}$}
    Let $\downPrenex{\Psi}=\langle \vec x, \psi \rangle$. IH: $\Psi \equiv \bigsqcap_{\vec x} \psi$.
    By repeated application of the validation clause of \Cref{thm:quantifier-elimination-positive-prenexing},
    \[
      \heylovalidate{\Psi}
      \equiv
      \heylovalidate{\bigsqcap_{\vec x} \psi}
      \equiv
      \bigsqcap_{\vec x} \heylovalidate{\psi}.
    \]

    \proofcase{$\hhla_1 \star \hhla_2$, $\star \in \{\sqcap,\sqcup\}$}
    Let $\downPrenex{\hhla_i}=\langle \vec x_i, \psi_i \rangle$. IH: $\hhla_i \equiv \bigsqcap_{\vec x_i} \psi_i$.
    By repeated application of the $\star$-clauses of \Cref{thm:quantifier-elimination-positive-prenexing},
    \[
      \hhla_1 \star \hhla_2
      \equiv
      (\bigsqcap_{\vec x_1} \psi_1) \star (\bigsqcap_{\vec x_2} \psi_2)
      \equiv
      \bigsqcap_{\vec x_1 \vec x_2} (\psi_1 \star \psi_2).
    \]

    \proofcase{$\hhla_1 \splus \hhla_2$}
    Let $\downPrenex{\hhla_i}=\langle \vec x_i, \psi_i \rangle$. IH: $\hhla_i \equiv \bigsqcap_{\vec x_i} \psi_i$.
    By repeated application of \Cref{thm:arithmetic-prenexing},
    \[
      \hhla_1 \splus \hhla_2
      \equiv
      (\bigsqcap_{\vec x_1} \psi_1) \splus (\bigsqcap_{\vec x_2} \psi_2)
      \equiv
      \bigsqcap_{\vec x_1 \vec x_2} (\psi_1 \splus \psi_2).
    \]

    \proofcase{$a \stimes \Psi$}
    Let $\downPrenex{\Psi}=\langle \vec x, \psi \rangle$. IH: $\Psi \equiv \bigsqcap_{\vec x} \psi$.
    By repeated application of the scalar-infimum clause of \Cref{thm:arithmetic-prenexing},
    \[
      a \stimes \Psi
      \equiv
      a \stimes \bigsqcap_{\vec x} \psi
      \equiv
      \bigsqcap_{\vec x} (a \stimes \psi).
    \]

    \proofcase{$\hhla_1 \impl \hhla_2$}
    Let $\upPrenex{\hhla_1}=\langle \vec x, \psi \rangle$ and $\downPrenex{\hhla_2}=\langle \vec y, \chi \rangle$. IH: $\hhla_1 \equiv \bigsqcup_{\vec x} \psi$ and $\hhla_2 \equiv \bigsqcap_{\vec y} \chi$.
    By repeated application of \Cref{thm:quantifier-elimination-negative-prenexing,thm:quantifier-elimination-positive-prenexing},
    \[
      \hhla_1 \impl \hhla_2
      \equiv
      (\bigsqcup_{\vec x} \psi) \impl (\bigsqcap_{\vec y} \chi)
      \equiv
      \bigsqcap_{\vec x \vec y} (\psi \impl \chi).
    \]

    \proofcase{$\bigsqcap_x \Psi$}
    Let $\downPrenex{\Psi}=\langle \vec x, \psi \rangle$. IH: $\Psi \equiv \bigsqcap_{\vec x} \psi$. Since $\bigsqcap_x \Psi \equiv \bigsqcap_{x \vec x} \psi$, immediate.

    \proofcase{otherwise}
    Immediate from $\downPrenex{\hhla}=\langle \epsilon, \hhla \rangle$ and $\bigsqcap_{\epsilon} \hhla = \hhla$.
  \end{proofcases}

  \paragraph{Upper branch.}
  \begin{proofcases}
    \proofcase{$\alpha$, $\top$, or $\bot$}
    Immediate.

    \proofcase{$\heylocovalidate{\Psi}$}
    Let $\upPrenex{\Psi}=\langle \vec x, \psi \rangle$. IH: $\Psi \equiv \bigsqcup_{\vec x} \psi$.
    By repeated application of the covalidation clause of \Cref{thm:quantifier-elimination-positive-prenexing},
    \[
      \heylocovalidate{\Psi}
      \equiv
      \heylocovalidate{\bigsqcup_{\vec x} \psi}
      \equiv
      \bigsqcup_{\vec x} \heylocovalidate{\psi}.
    \]

    \proofcase{$\hhla_1 \star \hhla_2$, $\star \in \{\sqcap,\sqcup\}$}
    Let $\upPrenex{\hhla_i}=\langle \vec x_i, \psi_i \rangle$. IH: $\hhla_i \equiv \bigsqcup_{\vec x_i} \psi_i$.
    By repeated application of the upper clauses of \Cref{thm:quantifier-elimination-positive-prenexing},
    \[
      \hhla_1 \star \hhla_2
      \equiv
      (\bigsqcup_{\vec x_1} \psi_1) \star (\bigsqcup_{\vec x_2} \psi_2)
      \equiv
      \bigsqcup_{\vec x_1 \vec x_2} (\psi_1 \star \psi_2).
    \]

    \proofcase{$\hhla_1 \splus \hhla_2$}
    Let $\upPrenex{\hhla_i}=\langle \vec x_i, \psi_i \rangle$. IH: $\hhla_i \equiv \bigsqcup_{\vec x_i} \psi_i$.
    By repeated application of \Cref{thm:arithmetic-prenexing},
    \[
      \hhla_1 \splus \hhla_2
      \equiv
      (\bigsqcup_{\vec x_1} \psi_1) \splus (\bigsqcup_{\vec x_2} \psi_2)
      \equiv
      \bigsqcup_{\vec x_1 \vec x_2} (\psi_1 \splus \psi_2).
    \]

    \proofcase{$a \stimes \Psi$}
    Let $\upPrenex{\Psi}=\langle \vec x, \psi \rangle$. IH: $\Psi \equiv \bigsqcup_{\vec x} \psi$.
    By repeated application of the scalar clause of \Cref{thm:arithmetic-prenexing},
    \[
      a \stimes \Psi
      \equiv
      a \stimes \bigsqcup_{\vec x} \psi
      \equiv
      \bigsqcup_{\vec x} (a \stimes \psi).
    \]

    \proofcase{$\hhla_1 \coimpl \hhla_2$}
    Let $\downPrenex{\hhla_1}=\langle \vec x, \psi \rangle$ and $\upPrenex{\hhla_2}=\langle \vec y, \chi \rangle$. IH: $\hhla_1 \equiv \bigsqcap_{\vec x} \psi$ and $\hhla_2 \equiv \bigsqcup_{\vec y} \chi$.
    By repeated application of \Cref{thm:quantifier-elimination-negative-prenexing,thm:quantifier-elimination-positive-prenexing},
    \[
      \hhla_1 \coimpl \hhla_2
      \equiv
      (\bigsqcap_{\vec x} \psi) \coimpl (\bigsqcup_{\vec y} \chi)
      \equiv
      \bigsqcup_{\vec x \vec y} (\psi \coimpl \chi).
    \]

    \proofcase{$\bigsqcup_x \Psi$}
    Let $\upPrenex{\Psi}=\langle \vec x, \psi \rangle$. IH: $\Psi \equiv \bigsqcup_{\vec x} \psi$. Since $\bigsqcup_x \Psi \equiv \bigsqcup_{x \vec x} \psi$, immediate.

    \proofcase{otherwise}
    Immediate from $\upPrenex{\hhla}=\langle \epsilon, \hhla \rangle$ and $\bigsqcup_{\epsilon} \hhla = \hhla$.
  \end{proofcases}
\end{proof}

\noindent
\begin{theorem}[Prenex Completeness on the \texorpdfstring{$\mathsf{qelim}$}{qelim} Fragment]
  \label[theorem]{thm:prenex-complete-fragment}
  Assume the hypotheses of \Cref{thm:arithmetic-prenexing}. Then:
  \begin{enumerate}[label=(\roman*)]
    \item If $\hhla \in \downQFrag$, then $\downPrenex{\hhla}=\langle \vec x,\psi\rangle$ for some $\vec x,\psi$, and every such result satisfies $\psi \in \qfFrag$.
    \item If $\hhla \in \upQFrag$, then $\upPrenex{\hhla}=\langle \vec x,\psi\rangle$ for some $\vec x,\psi$, and every such result satisfies $\psi \in \qfFrag$.
  \end{enumerate}
\end{theorem}

\begin{proof}
  This is again a straightforward mutual induction over the fragment grammar: every constructor in the fragment maps quantifier-free recursive results to a quantifier-free matrix.
  We prove both clauses simultaneously by induction on the derivation of $\hhla \in \downQFrag$ and $\hhla \in \upQFrag$.

  \begin{proofcases}
    \proofcase{$\hhla \in \qfFrag$}
    Immediate, since $\downPrenex{\hhla}=\langle \epsilon,\hhla \rangle$ and $\upPrenex{\hhla}=\langle \epsilon,\hhla \rangle$.

    \proofcase{$\heylovalidate{\Psi} \in \downQFrag$}
    Let $\downPrenex{\Psi}=\langle \vec x,\psi \rangle$. By IH, $\psi \in \qfFrag$, hence $\heylovalidate{\psi} \in \qfFrag$.

    \proofcase{$\heylocovalidate{\Psi} \in \upQFrag$}
    Dual.

    \proofcase{$\bigsqcap_x \Psi \in \downQFrag$}
    Let $\downPrenex{\Psi}=\langle \vec x,\psi \rangle$. By IH, $\psi \in \qfFrag$.

    \proofcase{$\bigsqcup_x \Psi \in \upQFrag$}
    Dual.

    \proofcase{$\hhla_1 \star \hhla_2 \in \downQFrag$, $\star \in \{\sqcap,\sqcup,\splus\}$}
    Let $\downPrenex{\hhla_i}=\langle \vec x_i,\psi_i \rangle$. By IH, $\psi_1,\psi_2 \in \qfFrag$, hence $\psi_1 \star \psi_2 \in \qfFrag$.

    \proofcase{$\hhla_1 \star \hhla_2 \in \upQFrag$, $\star \in \{\sqcap,\sqcup,\splus\}$}
    Dual.

    \proofcase{$\hhla_1 \impl \hhla_2 \in \downQFrag$}
    Let $\upPrenex{\hhla_1}=\langle \vec x,\psi \rangle$ and $\downPrenex{\hhla_2}=\langle \vec y,\chi \rangle$. By IH, $\psi,\chi \in \qfFrag$, hence $\psi \impl \chi \in \qfFrag$.

    \proofcase{$a \stimes \Psi \in \downQFrag$}
    Let $\downPrenex{\Psi}=\langle \vec x,\psi \rangle$. By IH, $\psi \in \qfFrag$, hence $a \stimes \psi \in \qfFrag$.

    \proofcase{$\hhla_1 \coimpl \hhla_2 \in \upQFrag$}
    Let $\downPrenex{\hhla_1}=\langle \vec x,\psi \rangle$ and $\upPrenex{\hhla_2}=\langle \vec y,\chi \rangle$. By IH, $\psi,\chi \in \qfFrag$, hence $\psi \coimpl \chi \in \qfFrag$.

    \proofcase{$a \stimes \Psi \in \upQFrag$}
    Let $\upPrenex{\Psi}=\langle \vec x,\psi \rangle$. By IH, $\psi \in \qfFrag$, hence $a \stimes \psi \in \qfFrag$.
  \end{proofcases}
\end{proof}

\noindent
\begin{definition}[Algorithm \texorpdfstring{$\mathsf{strip}$}{strip}]
  \label[definition]{def:qelim-strip}
  The functions $\downStrip{\cdot}{\cdot} : \downPFrag \to \wHeyLo$ and $\upStrip{\cdot}{\cdot} : \upPFrag \to \wHeyLo$ are defined by the following recursive equations:
  \begin{align*}
    \downStrip{x \vec x}{\psi}
    &\definedAs \downStrip{\vec x}{\mathsf{fresh}_x(\psi)},
    &
    \upStrip{x \vec x}{\psi}
    &\definedAs \upStrip{\vec x}{\mathsf{fresh}_x(\psi)},
    \\
    \downStrip{\epsilon}{\psi}
    &\definedAs \psi,
    &
    \upStrip{\epsilon}{\psi}
    &\definedAs \psi.
  \end{align*}
  By \Cref{lem:qelim-strip-soundness},
  \[
    \isValid{\bigsqcap_{\vec x}\psi} \qiff \isValid{\downStrip{\vec x}{\psi}},
    \qquad
    \isCovalid{\bigsqcup_{\vec x}\psi} \qiff \isCovalid{\upStrip{\vec x}{\psi}};
  \]
  if $\psi \in \qfFrag$, then $\downStrip{\vec x}{\psi}, \upStrip{\vec x}{\psi} \in \qfFrag$ whenever the corresponding branch is defined.
\end{definition}

\noindent
\begin{lemma}[Strip Equations]
  \label[lemma]{lem:qelim-strip-helpers}
  For all variable lists $\vec x,\vec y$ and formulas $\phi,\psi$:
  \begin{align*}
    \downStrip{\vec x}{\heylovalidate{\psi}}
    &= \heylovalidate{\downStrip{\vec x}{\psi}},
    \\
    \downStrip{\vec x \vec y}{\phi \star \psi}
    &= \downStrip{\vec x}{\phi} \star \downStrip{\vec y}{\psi}
    && \text{for } \star \in \{\sqcap,\sqcup,\splus\},
    \\
    \downStrip{\vec x}{a \stimes \psi}
    &= a \stimes \downStrip{\vec x}{\psi},
    \\
    \downStrip{\vec x \vec y}{\phi \impl \psi}
    &= \upStrip{\vec x}{\phi} \impl \downStrip{\vec y}{\psi},
    \\
    \upStrip{\vec x}{\heylocovalidate{\psi}}
    &= \heylocovalidate{\upStrip{\vec x}{\psi}},
    \\
    \upStrip{\vec x \vec y}{\phi \star \psi}
    &= \upStrip{\vec x}{\phi} \star \upStrip{\vec y}{\psi}
    && \text{for } \star \in \{\sqcap,\sqcup,\splus\},
    \\
    \upStrip{\vec x}{a \stimes \psi}
    &= a \stimes \upStrip{\vec x}{\psi},
    \\
    \upStrip{\vec x \vec y}{\phi \coimpl \psi}
    &= \downStrip{\vec x}{\phi} \coimpl \upStrip{\vec y}{\psi}.
  \end{align*}
\end{lemma}

\begin{proof}
  All clauses follow by induction on the leftmost variable list, using \Cref{def:qelim-strip} and the freshness convention built into \Cref{fig:prenex-algorithm}.
\end{proof}

\noindent
\begin{lemma}[Stripping a Prenex Output]
  \label[lemma]{lem:qelim-strip-soundness}
  For every $\langle \vec x, \psi \rangle \in \downPFrag$,
  \[
    \isValid{\bigsqcap_{\vec x} \psi}
    \qiff
    \isValid{\downStrip{\vec x}{\psi}}.
  \]
  For every $\langle \vec x, \psi \rangle \in \upPFrag$,
  \[
    \isCovalid{\bigsqcup_{\vec x} \psi}
    \qiff
    \isCovalid{\upStrip{\vec x}{\psi}}.
  \]
  If additionally $\psi \in \qfFrag$, then
  \[
    \downStrip{\vec x}{\psi} \in \qfFrag,
    \qquad
    \upStrip{\vec x}{\psi} \in \qfFrag.
  \]
\end{lemma}

\begin{proof}
  We prove both clauses simultaneously by induction on the recursive evaluation of \Cref{def:qelim-strip}.

  \begin{proofcases}
    \proofcase{$\downStrip{x \vec x}{\psi}$}
    By \Cref{thm:qelim-strip},
    \[
      \isValid{\bigsqcap_{x \vec x} \psi}
      \qiff
      \isValid{\bigsqcap_{\vec x} \mathsf{fresh}_x(\psi)}.
    \]
    Then apply IH to $(\vec x,\mathsf{fresh}_x(\psi))$.

    \proofcase{$\upStrip{x \vec x}{\psi}$}
    Dual, using \Cref{thm:qelim-strip}.

    \proofcase{$\downStrip{\epsilon}{\psi}$ or $\upStrip{\epsilon}{\psi}$}
    Immediate.
  \end{proofcases}
  The quantifier-freeness clause follows by the same induction: the recursive equations of \Cref{def:qelim-strip} only remove leading quantifiers, and the base case returns the unchanged matrix.
\end{proof}

\noindent
\begin{lemma}[\texorpdfstring{$\mathsf{qelim}$}{qelim} Factors Through $\mathsf{strip}\circ\mathsf{prenex}$]
  \label[lemma]{lem:qelim-is-strip-prenex}
  If $\downPrenex{\hhla}=\langle \vec x, \psi \rangle$, then
  \[
    \downQelim{\hhla}=\downStrip{\vec x}{\psi}.
  \]
  Dually, if $\upPrenex{\hhla}=\langle \vec x, \psi \rangle$, then
  \[
    \upQelim{\hhla}=\upStrip{\vec x}{\psi}.
  \]
\end{lemma}

\begin{proof}
  We prove both clauses simultaneously by induction on the recursive evaluation of \Cref{fig:qelim-algorithm}.

  \paragraph{Lower branch.}
  \begin{proofcases}
    \proofcase{$\alpha$, $\top$, or $\bot$}
    Immediate.

    \proofcase{$\heylovalidate{\Psi}$}
    Let $\downPrenex{\Psi}=\langle \vec x, \psi \rangle$.
    By IH and \Cref{lem:qelim-strip-helpers},
    \[
      \downQelim{\heylovalidate{\Psi}}
      =
      \heylovalidate{\downQelim{\Psi}}
      =
      \heylovalidate{\downStrip{\vec x}{\psi}}
      =
      \downStrip{\vec x}{\heylovalidate{\psi}}.
    \]

    \proofcase{$\hhla_1 \star \hhla_2$, $\star \in \{\sqcap,\sqcup\}$}
    Let $\downPrenex{\hhla_i}=\langle \vec x_i, \psi_i \rangle$.
    By IH and \Cref{lem:qelim-strip-helpers},
    \[
      \downQelim{\hhla_1 \star \hhla_2}
      =
      \downQelim{\hhla_1} \star \downQelim{\hhla_2}
      =
      \downStrip{\vec x_1}{\psi_1} \star \downStrip{\vec x_2}{\psi_2}
      =
      \downStrip{\vec x_1 \vec x_2}{\psi_1 \star \psi_2}.
    \]

    \proofcase{$\hhla_1 \splus \hhla_2$}
    Let $\downPrenex{\hhla_i}=\langle \vec x_i, \psi_i \rangle$.
    By IH and \Cref{lem:qelim-strip-helpers},
    \[
      \downQelim{\hhla_1 \splus \hhla_2}
      =
      \downQelim{\hhla_1} \splus \downQelim{\hhla_2}
      =
      \downStrip{\vec x_1}{\psi_1} \splus \downStrip{\vec x_2}{\psi_2}
      =
      \downStrip{\vec x_1 \vec x_2}{\psi_1 \splus \psi_2}.
    \]

    \proofcase{$a \stimes \Psi$}
    Let $\downPrenex{\Psi}=\langle \vec x, \psi \rangle$.
    By IH and \Cref{lem:qelim-strip-helpers},
    \[
      \downQelim{a \stimes \Psi}
      =
      a \stimes \downQelim{\Psi}
      =
      a \stimes \downStrip{\vec x}{\psi}
      =
      \downStrip{\vec x}{a \stimes \psi}.
    \]

    \proofcase{$\hhla_1 \impl \hhla_2$}
    Let $\upPrenex{\hhla_1}=\langle \vec x, \psi \rangle$ and $\downPrenex{\hhla_2}=\langle \vec y, \chi \rangle$.
    By IH and \Cref{lem:qelim-strip-helpers},
    \[
      \downQelim{\hhla_1 \impl \hhla_2}
      =
      \upQelim{\hhla_1} \impl \downQelim{\hhla_2}
      =
      \upStrip{\vec x}{\psi} \impl \downStrip{\vec y}{\chi}
      =
      \downStrip{\vec x \vec y}{\psi \impl \chi}.
    \]

    \proofcase{$\bigsqcap_x \Psi$}
    Let $\downPrenex{\Psi}=\langle \vec x, \psi \rangle$.
    By IH and \Cref{def:qelim-strip},
    \[
      \downQelim{\bigsqcap_x \Psi}
      =
      \mathsf{fresh}_x(\downQelim{\Psi})
      =
      \mathsf{fresh}_x(\downStrip{\vec x}{\psi})
      =
      \downStrip{x \vec x}{\psi}.
    \]

    \proofcase{otherwise}
    Immediate, since $\downPrenex{\hhla}=\langle \epsilon,\hhla \rangle$ and $\downStrip{\epsilon}{\hhla}=\hhla=\downQelim{\hhla}$.
  \end{proofcases}

  \paragraph{Upper branch.}
  \begin{proofcases}
    \proofcase{$\alpha$, $\top$, or $\bot$}
    Immediate.

    \proofcase{$\heylocovalidate{\Psi}$}
    Let $\upPrenex{\Psi}=\langle \vec x, \psi \rangle$.
    By IH and \Cref{lem:qelim-strip-helpers},
    \[
      \upQelim{\heylocovalidate{\Psi}}
      =
      \heylocovalidate{\upQelim{\Psi}}
      =
      \heylocovalidate{\upStrip{\vec x}{\psi}}
      =
      \upStrip{\vec x}{\heylocovalidate{\psi}}.
    \]

    \proofcase{$\hhla_1 \star \hhla_2$, $\star \in \{\sqcap,\sqcup\}$}
    Let $\upPrenex{\hhla_i}=\langle \vec x_i, \psi_i \rangle$.
    By IH and \Cref{lem:qelim-strip-helpers},
    \[
      \upQelim{\hhla_1 \star \hhla_2}
      =
      \upQelim{\hhla_1} \star \upQelim{\hhla_2}
      =
      \upStrip{\vec x_1}{\psi_1} \star \upStrip{\vec x_2}{\psi_2}
      =
      \upStrip{\vec x_1 \vec x_2}{\psi_1 \star \psi_2}.
    \]

    \proofcase{$\hhla_1 \splus \hhla_2$}
    Let $\upPrenex{\hhla_i}=\langle \vec x_i, \psi_i \rangle$.
    By IH and \Cref{lem:qelim-strip-helpers},
    \[
      \upQelim{\hhla_1 \splus \hhla_2}
      =
      \upQelim{\hhla_1} \splus \upQelim{\hhla_2}
      =
      \upStrip{\vec x_1}{\psi_1} \splus \upStrip{\vec x_2}{\psi_2}
      =
      \upStrip{\vec x_1 \vec x_2}{\psi_1 \splus \psi_2}.
    \]

    \proofcase{$a \stimes \Psi$}
    Let $\upPrenex{\Psi}=\langle \vec x, \psi \rangle$.
    By IH and \Cref{lem:qelim-strip-helpers},
    \[
      \upQelim{a \stimes \Psi}
      =
      a \stimes \upQelim{\Psi}
      =
      a \stimes \upStrip{\vec x}{\psi}
      =
      \upStrip{\vec x}{a \stimes \psi}.
    \]

    \proofcase{$\hhla_1 \coimpl \hhla_2$}
    Let $\downPrenex{\hhla_1}=\langle \vec x, \psi \rangle$ and $\upPrenex{\hhla_2}=\langle \vec y, \chi \rangle$.
    By IH and \Cref{lem:qelim-strip-helpers},
    \[
      \upQelim{\hhla_1 \coimpl \hhla_2}
      =
      \downQelim{\hhla_1} \coimpl \upQelim{\hhla_2}
      =
      \downStrip{\vec x}{\psi} \coimpl \upStrip{\vec y}{\chi}
      =
      \upStrip{\vec x \vec y}{\psi \coimpl \chi}.
    \]

    \proofcase{$\bigsqcup_x \Psi$}
    Let $\upPrenex{\Psi}=\langle \vec x, \psi \rangle$.
    By IH and \Cref{def:qelim-strip},
    \[
      \upQelim{\bigsqcup_x \Psi}
      =
      \mathsf{fresh}_x(\upQelim{\Psi})
      =
      \mathsf{fresh}_x(\upStrip{\vec x}{\psi})
      =
      \upStrip{x \vec x}{\psi}.
    \]

    \proofcase{otherwise}
    Immediate, since $\upPrenex{\hhla}=\langle \epsilon,\hhla \rangle$ and $\upStrip{\epsilon}{\hhla}=\hhla=\upQelim{\hhla}$.
  \end{proofcases}
\end{proof}

\thmQelimSoundness*

\begin{proof}
  \begin{proofcases}
    \proofcase{$\downQelim{\hhla}=\hhlb$}
    Let $\downPrenex{\hhla}=\langle \vec x, \psi \rangle$.
    By \Cref{lem:qelim-is-strip-prenex}, $\hhlb=\downStrip{\vec x}{\psi}$.
    By \Cref{thm:prenex-soundness}, $\hhla \equiv \bigsqcap_{\vec x} \psi$.
    Hence
    \begin{align*}
      \isValid{\hhla}
      &\qiff \isValid{\bigsqcap_{\vec x} \psi}
      \inlineExplain{formula equivalence}
      \\
      &\qiff \isValid{\downStrip{\vec x}{\psi}}
      \inlineExplain{\Cref{lem:qelim-strip-soundness}}
      \\
      &\qiff \isValid{\hhlb}.
    \end{align*}

    \proofcase{$\upQelim{\hhla}=\hhlb$}
    Let $\upPrenex{\hhla}=\langle \vec x, \psi \rangle$.
    By \Cref{lem:qelim-is-strip-prenex}, $\hhlb=\upStrip{\vec x}{\psi}$.
    By \Cref{thm:prenex-soundness}, $\hhla \equiv \bigsqcup_{\vec x} \psi$.
    Hence
    \begin{align*}
      \isCovalid{\hhla}
      &\qiff \isCovalid{\bigsqcup_{\vec x} \psi}
      \inlineExplain{formula equivalence}
      \\
      &\qiff \isCovalid{\upStrip{\vec x}{\psi}}
      \inlineExplain{\Cref{lem:qelim-strip-soundness}}
      \\
      &\qiff \isCovalid{\hhlb}.
    \end{align*}
    \qedhere
  \end{proofcases}
\end{proof}

\thmQelimCompleteFragment*

\begin{proof}
  \begin{proofcases}
    \proofcase{$\hhla \in \downQFrag$}
    Let $\downPrenex{\hhla}=\langle \vec x,\psi \rangle$. By \Cref{thm:prenex-complete-fragment}, $\psi \in \qfFrag$. By \Cref{lem:qelim-is-strip-prenex}, $\downQelim{\hhla}=\downStrip{\vec x}{\psi}$. Hence $\downQelim{\hhla} \in \qfFrag$ by \Cref{lem:qelim-strip-soundness}.

    \proofcase{$\hhla \in \upQFrag$}
    Let $\upPrenex{\hhla}=\langle \vec x,\psi \rangle$. By \Cref{thm:prenex-complete-fragment}, $\psi \in \qfFrag$. By \Cref{lem:qelim-is-strip-prenex}, $\upQelim{\hhla}=\upStrip{\vec x}{\psi}$. Hence $\upQelim{\hhla} \in \qfFrag$ by \Cref{lem:qelim-strip-soundness}. \qedhere
  \end{proofcases}
\end{proof}

\clearpage
\ifarxivPreprint
\begin{landscape}
\fi
    \begingroup
    \footnotesize
    \captionsetup{font=footnotesize,position=bottom,skip=8pt,width=0.84\linewidth,justification=justified,singlelinecheck=false}
    \setlength{\tabcolsep}{5pt}
    \setlength{\extrarowheight}{1.1pt}
    \renewcommand{\arraystretch}{1.2}
    \colorlet{sourceStmtBg}{black!6}
    \newcommand{\asIs}{\textit{as is}}
    \newcommand{\noTrans}{\textemdash}
    \newcommand{\lax}[1]{\ensuremath{#1^{\dagger}}}
    \newcommand{\choiceplus}[2]{\ensuremath{\{#1\}\splus\{#2\}}}
    \newcommand{\choicemeet}[2]{\ensuremath{\{#1\}\hand\{#2\}}}
    \newcommand{\choicejoin}[2]{\ensuremath{\{#1\}\hor\{#2\}}}
    \newcommand{\seqtrans}{\ensuremath{T(S_1)\fatsemi T(S_2)}}
    \newcommand{\dualembed}[1]{\ensuremath{\widehat{#1}}}
    \newcommand{\mathcell}[1]{\ensuremath{\begin{array}[t]{@{}l@{}}#1\end{array}}}
    \newcommand{\compactmathcell}[1]{\begingroup\scriptsize\renewcommand{\arraystretch}{0.96}\ensuremath{\begin{array}[t]{@{}l@{}}#1\end{array}}\endgroup}
    \newcommand{\Tlow}{\ensuremath{T_\downarrow}}
    \newcommand{\Tupr}{\ensuremath{T_\uparrow}}
    \newcommand{\TlowLax}{\ensuremath{T_\downarrow^\dagger}}
    \newcommand{\TuprLax}{\ensuremath{T_\uparrow^\dagger}}
    \newlength{\sourceStmtColWidth}
    \setlength{\sourceStmtColWidth}{0.073\linewidth}
    \newsavebox{\loweringOverviewTable}
    \newcolumntype{S}{>{\raggedright\arraybackslash}p{\sourceStmtColWidth}}
    \newcolumntype{Y}{>{\raggedright\arraybackslash}X}
    \ifarxivPreprint
        \newcommand{\overviewPageCenter}{\hspace*{\dimexpr(\paperheight-\linewidth)/2 - 1in - \topmargin - \headheight - \headsep\relax}}
    \else
        \newcommand{\overviewPageCenter}{}
    \fi
    \null
    \vfill
    \noindent\overviewPageCenter
    \begin{minipage}{\linewidth}
    \centering
    \sbox{\loweringOverviewTable}{%
    \begin{tabularx}{\linewidth}{@{}
        S
        Y
        Y
        Y
        Y
        Y
        @{}}
        \toprule
        \makecell[l]{Source\\statement}
        & Tropical
        & \makecell[l]{Access\\control}
        & \makecell[l]{Why/\\provenance}
        & \makecell[l]{Formal\\languages:\\lower monitor}
        & \makecell[l]{Formal\\languages:\\upper monitor} \\
        \midrule
        $\stmtRasgn{x}{e}$
        & \asIs
        & \asIs
        & \asIs
        & \asIs
        & \asIs \\

        $S_1 \fatsemi S_2$
        & \seqtrans
        & \seqtrans
        & \seqtrans
        & $\Tlow(S_1)\fatsemi \Tlow(S_2)$
        & $\Tupr(S_1)\fatsemi \Tupr(S_2)$ \\

        \choiceplus{S_1}{S_2}
        & \choicemeet{T(S_1)}{T(S_2)}
        & \choicemeet{T(S_1)}{T(S_2)}
        & \choicejoin{T(S_1)}{T(S_2)}
        & \choiceplus{\TlowLax(S_1)}{\TlowLax(S_2)}
        & \choicemeet{\Tupr(S_1)}{\Tupr(S_2)} \\

        \choicemeet{S_1}{S_2}
        & \choicejoin{T(S_1)}{T(S_2)}
        & \choicejoin{T(S_1)}{T(S_2)}
        & \choicemeet{T(S_1)}{T(S_2)}
        & \choicemeet{\Tlow(S_1)}{\Tlow(S_2)}
        & \choicejoin{\TuprLax(S_1)}{\TuprLax(S_2)} \\

        \choicejoin{S_1}{S_2}
        & \choicemeet{T(S_1)}{T(S_2)}
        & \choicemeet{T(S_1)}{T(S_2)}
        & \choicejoin{T(S_1)}{T(S_2)}
        & \choicejoin{\TlowLax(S_1)}{\TlowLax(S_2)}
        & \choicemeet{\Tupr(S_1)}{\Tupr(S_2)} \\
        \addlinespace[0.22em]

        $\stmtWeigh{a}$
        & $\stmtTick{\dualembed{a}}$
        & $\coAssert{\dualembed{a}}$
        & $\stmtAssert{\WhyBool{a}}$
        & \compactmathcell{
          \stmtAssert{\FLMonVar \subseteq a\cdot\FLAllWords}\fatsemi{}\\
          \stmtRasgn{\FLMonVar}{a^{-1}\cdot\FLMonVar}}
        & $\stmtRasgn{\FLMonVar}{a^{-1}\cdot\FLMonVar}$ \\
        \addlinespace[0.22em]

        $\stmtAssert{\hlb}$
        & $\coAssert{\dualembed{\hlb}}$
        & $\coAssert{\dualembed{\hlb}}$
        & $\stmtAssert{\WhyBool{\hlb}}$
        & $\stmtAssert{\FLMonVar \subseteq \hlb}$
        & $\stmtRasgn{\FLMonVar}{\FLMonVar\cup(\FLAllWords\setminus\hlb)}$ \\

        $\stmtAssume{\hlb}$
        & $\coAssume{\dualembed{\hlb}}$
        & $\coAssume{\dualembed{\hlb}}$
        & $\stmtAssume{\WhyBool{\hlb}}$
        & $\stmtRasgn{\FLMonVar}{\FLMonVar\cap\hlb}$
        & $\stmtAssert{(\FLAllWords\setminus\hlb)\subseteq\FLMonVar}$ \\

        \mathcell{\stmtValidate\fatsemi{}\\\stmtAssume{\hlb}}
        & $\coValidate\fatsemi\coAssume{\dualembed{\hlb}}$
        & $\coValidate\fatsemi\coAssume{\dualembed{\hlb}}$
        & \compactmathcell{
          \stmtHavoc{p_{a_1},\ldots,p_{a_n}}\fatsemi{}\\
          \stmtValidate\fatsemi\stmtAssume{\WhyBool{\hlb}}}
        & $\Tlow(\stmtValidate)\fatsemi\Tlow(\stmtAssume{\hlb})$
        & $\TuprLax(\stmtValidate)\fatsemi\Tupr(\stmtAssume{\hlb})$ \\

        $\coAssert{\hlb}$
        & $\stmtAssert{\dualembed{\hlb}}$
        & $\stmtAssert{\dualembed{\hlb}}$
        & $\coAssert{\WhyBool{\hlb}}$
        & $\stmtRasgn{\FLMonVar}{\FLMonVar\setminus\hlb}$
        & $\stmtAssert{\hlb\subseteq\FLMonVar}$ \\

        $\coAssume{\hlb}$
        & $\stmtAssume{\dualembed{\hlb}}$
        & $\stmtAssume{\dualembed{\hlb}}$
        & \noTrans
        & $\stmtAssert{\FLMonVar\subseteq\FLAllWords\setminus\hlb}$
        & $\stmtRasgn{\FLMonVar}{\FLMonVar\cup\hlb}$ \\
        \addlinespace[0.22em]

        $\stmtHavoc{x}$
        & $\coHavoc{x}$
        & $\coHavoc{x}$
        & \asIs
        & \asIs
        & \lax{\coHavoc{x}} \\

        $\coHavoc{x}$
        & $\stmtHavoc{x}$
        & $\stmtHavoc{x}$
        & \asIs
        & \lax{\coHavoc{x}}
        & $\stmtHavoc{x}$ \\
        \addlinespace[0.22em]

        $\stmtValidate$
        & $\coValidate$
        & $\coValidate$
        & \noTrans
        & \compactmathcell{
          \stmtAssume{\FLMonVar\neq\emptyset}\fatsemi{}\\
          \stmtRasgn{\FLMonVar}{\FLAllWords}}
        & \lax{\stmtAssert{\FLAllWords\subseteq\FLMonVar}} \\

        $\coValidate$
        & $\stmtValidate$
        & $\stmtValidate$
        & \noTrans
        & \lax{\stmtAssert{\FLMonVar\subseteq\emptyset}}
        & \compactmathcell{
          \coAssert{\FLAllWords\subseteq\FLMonVar}\fatsemi{}\\
          \stmtRasgn{\FLMonVar}{\emptyset}\fatsemi{}\\
          \coValidate} \\
        \bottomrule
    \end{tabularx}
    }%
    \begin{adjustbox}{max width=0.99\linewidth,max totalheight=0.83\textheight,center}
    \makebox[0pt][l]{%
        \raisebox{-\dp\loweringOverviewTable}{%
            \textcolor{sourceStmtBg}{%
                \rule{\dimexpr\sourceStmtColWidth+\tabcolsep\relax}{\dimexpr\ht\loweringOverviewTable+\dp\loweringOverviewTable\relax}%
            }%
        }%
    }%
    \usebox{\loweringOverviewTable}
    \end{adjustbox}
    \par\addvspace{0.35\baselineskip}
    \captionof{table}{%
        Statement-level translations proved in \Cref{sec:appendix-weighted-to-ordinary-heyvl}.
        Unmarked entries are exact equivalences from \Cref{thm:tropical-translation,thm:clearance-translation,thm:why-translation,thm:residual-language-monitor,thm:residual-language-monitor-upper}.
        In the tropical and clearance columns, $\widehat{\psi}$ denotes the translated source formula in ordinary $\eureal$.
        Tropical costs are embedded by inclusion, while clearance levels use $\widehat{\mathsf{P}}=0$, $\widehat{\mathsf{C}}=1$, $\widehat{\mathsf{S}}=2$, and $\widehat{\mathsf{TS}}=\infty$.
        The superscript $\dagger$ marks one-sided lax rules from \Cref{lem:residual-language-monitor-lax-joins,lem:residual-language-monitor-upper-lax-meets}.
        A dash means that this appendix gives no compositional exact lowering for that statement form.
        In the Why/provenance column, $\WhyBool{\psi}$ is defined only for crisp $\psi$, i.e. $\psi(\sigma)\in\{\bot,\top\}$ for every state~$\sigma$; ordinary state transformers must not write the fresh provenance variables, except that the displayed $\stmtValidate\fatsemi\stmtAssume{\hlb}$ block explicitly havocs them to quantify over all provenance valuations.
    }
    \label{tab:wheyvl-lowering-overview}
    \end{minipage}
    \vfill
    \ifarxivPreprint
    \else
        \null
    \fi
    \endgroup
\ifarxivPreprint
\end{landscape}
\fi
\clearpage

\section{Lowerings of \texorpdfstring{\wHeyVL{}}{wHeyVL}}
\label{sec:appendix-weighted-to-ordinary-heyvl}

\Cref{tab:wheyvl-lowering-overview} summarizes the statement-level translations proved in this appendix.

\subsection{Residual-Language Monitor Translation}
\label{sec:appendix-formal-language-monitors}

We can translate (fragments of) \wHeyVLLang{} over formal languages into \wHeyVLBool{} over predicates on residual languages.
We present two translations, one for lower bounds and one for upper bounds.
They are based on \emph{residual languages}, which we explain first.

\subsubsection{Residual Languages}

The translation maintains a language variable~$\FLMonVar$ that stores the \emph{residual language} after the current execution prefix.
That is, $\FLMonVar$ contains exactly those continuations that would still complete the current prefix to a word in the original language.

Formally, for a language $L \subseteq \FLAllWords$ and a word $u \in \Gamma^*$, the residual language $u^{-1} \cdot L$ is defined by
$u^{-1} \cdot L = \{\, w \mid uw \in L \,\}$.
The upper- and lower-bound translations rely on \Cref{eq:residual-language-law-upper,eq:residual-language-law-lower}.
\begin{subequations}\label{eq:residual-language-laws}
    \begin{align}
        a \cdot \hla \subseteq \FLMonVar
         & \qiff \hla \subseteq a^{-1} \cdot \FLMonVar
        \label{eq:residual-language-law-upper}
        \\
        \FLMonVar \subseteq a \cdot \hla
         & \qiff (a^{-1} \cdot \FLMonVar \subseteq \hla)
        \land
        (\FLMonVar \subseteq a \cdot \FLAllWords)
        \label{eq:residual-language-law-lower}
\end{align}
\end{subequations}
\Cref{eq:residual-language-law-lower} includes the additional condition $\FLMonVar \subseteq a \cdot \FLAllWords$ because quotienting alone is insufficient.
For example, if $\FLMonVar = \{ b \}$ for some $b \neq a$, then $a^{-1} \cdot \FLMonVar = \emptyset$, and hence $a^{-1} \cdot \FLMonVar \subseteq \hla$ holds for every $\hla$.
Nevertheless, $\FLMonVar \subseteq a \cdot \hla$ is false, since the words in $\FLMonVar$ do not start with~$a$.

\subsubsection{Lower-Bound Translation}

The translation rewrites each statement into a sequence of \wHeyVLBool{} statements that manipulate the residual language $\FLMonVar$.
The translated postcondition is $\stmtAssert{(\FLMonVar \subseteq \hla)}$.

The statement $\stmtWeigh{a}$ now means that the next symbol must be $a$.
To implement this, we first check that every word in the current residual starts with $a$ (i.e., $\FLMonVar \subseteq a \cdot \FLAllWords$), and then update the residual by taking the quotient with $a$ (i.e., $\FLMonVar := a^{-1} \cdot \FLMonVar$).
Therefore, $\stmtWeigh{a}$ is translated to $$\stmtAssert{(\FLMonVar \subseteq a \cdot \FLAllWords)} \fatsemi \stmtRasgn{\FLMonVar}{a^{-1} \cdot \FLMonVar}~.$$

The statement $\stmtAssert{\hlb}$ checks that all possible continuations of the current execution (i.e., $\FLMonVar$) are included in $\hlb$.
Thus, $\stmtAssert{\hlb}$ is translated to $\stmtAssert{(\FLMonVar \subseteq \hlb)}$.

The statement $\stmtAssume{\hlb}$ restricts the current residual to continuations contained in $\hlb$.
Thus, $\stmtAssume{\hlb}$ is translated to $\stmtRasgn{\FLMonVar}{\FLMonVar \cap \hlb}$.

The following definition formalizes the translation for a fragment of \wHeyVLLang{} statements.

\begin{definition}[Residual-language monitor translation]\label{def:residual-language-monitor}
    Let $\FLMonVar$ be fresh and range over languages over $\FLAllWords$.
    For
    \[
        \begin{aligned}
            S ::= {} & \stmtRasgn{x}{e}
            \mid \stmtAssert{\hlb}
            \mid \stmtAssume{\hlb}
            \mid \coAssert{\hlb}
            \mid \coAssume{\hlb}        \\
            \mid {}  & \stmtValidate
            \mid \stmtWeigh{a}
            \mid \stmtHavoc{x}
            \mid \stmtDemonic{S}{S}
            \mid S \fatsemi S
        \end{aligned}
    \]
    define $\FLMonitor{-}$ by
    \begin{align*}
        \FLMonitor{\stmtAssert{\hlb}}
         & = \stmtAssert{(\FLMonVar \subseteq \hlb)}                       \\
        \FLMonitor{\stmtAssume{\hlb}}
         & = \stmtRasgn{\FLMonVar}{\FLMonVar \cap \hlb}                    \\
        \FLMonitor{\coAssert{\hlb}}
         & = \stmtRasgn{\FLMonVar}{\FLMonVar \setminus \hlb}               \\
        \FLMonitor{\coAssume{\hlb}}
         & = \stmtAssert{(\FLMonVar \subseteq \FLAllWords \setminus \hlb)} \\
        \FLMonitor{\stmtValidate}
         & = \stmtAssume{(\FLMonVar \neq \emptyset)}
        \fatsemi
        \stmtRasgn{\FLMonVar}{\FLAllWords}                                 \\
        \FLMonitor{\stmtWeigh{a}}
         & = \stmtAssert{(\FLMonVar \subseteq a \cdot \FLAllWords)}
        \fatsemi
        \stmtRasgn{\FLMonVar}{a^{-1} \cdot \FLMonVar}                       \\
        \FLMonitor{\stmtDemonic{S_1}{S_2}}
         & = \stmtDemonic{\FLMonitor{S_1}}{\FLMonitor{S_2}}                 \\
        \FLMonitor{S_1 \fatsemi S_2}
         & = \FLMonitor{S_1} \fatsemi \FLMonitor{S_2}.
    \end{align*}
    It leaves $\stmtHavoc{x}$ and $\stmtRasgn{x}{e}$ unchanged.
    In the Boolean algebra of languages, $\hlb \coimpl \hla = \hla \setminus \hlb$.
\end{definition}

The next theorem states that the encoding is sound for the lower-bound query $\FLIn{\vc{S}(\hla)}$, i.e., the monitor correctly tracks the residual language of the initial language $\hla$ after executing $S$.
\Cref{def:residual-language-upper-monitor} gives the dual translation for the upper-bound query $\FLOut{\vc{S}(\hla)}$.

\begin{theorem}
    \label{thm:residual-language-monitor}
    Assume that $\FLMonVar$ is fresh, i.e., $\FLMonVar$ does not occur in $S$ or~$\hla$.
    Then, for every statement~$S$ in the fragment above and language~$\hla$,
    \[
        \FLIn{\vc{S}(\hla)}
        \qiff
        \vcB{\FLMonitor{S}}(\FLIn{\hla}).
    \]
\end{theorem}
\begin{proof}
        We prove by structural induction on~$S$.

    \begin{itemize}
        \proofcase{$S = \stmtAssert{\hlb}$}
        \begin{align*}
             & \FLIn{\vc{\stmtAssert{\hlb}}(\hla)}                                                 &  & \\
             & \qiff \FLIn{\hlb \cap \hla} \inlineExplain{Def. \symVc}                             &  & \\
             & \qiff \FLIn{\hlb} \land \FLIn{\hla} \inlineExplain{Defs. subset/intersection}       &  & \\
             & \qiff \vcB{\stmtAssert{\FLIn{\hlb}}}(\FLIn{\hla}) \inlineExplain{Def. \symVcB}      &  & \\
             & \qiff \vcB{\FLMonitor{\stmtAssert{\hlb}}}(\FLIn{\hla}) \inlineExplain{Def. monitor}
        \end{align*}

        \proofcase{$S = \stmtAssume{\hlb}$}
        \begin{align*}
             & \FLIn{\vc{\stmtAssume{\hlb}}(\hla)}                                                              &  & \\
             & \qiff \FLIn{\hlb \impl \hla} \inlineExplain{Def. \symVc}                                         &  & \\
             & \qiff \FLMonVar \cap \hlb \subseteq \hla \inlineExplain{Def. implication}                        &  & \\
             & \qiff \FLIn{\hla}\substBy{\FLMonVar}{\FLMonVar \cap \hlb} \inlineExplain{Def. substBy}           &  & \\
             & \qiff \vcB{\stmtRasgn{\FLMonVar}{\FLMonVar \cap \hlb}}(\FLIn{\hla}) \inlineExplain{Def. \symVcB} &  & \\
             & \qiff \vcB{\FLMonitor{\stmtAssume{\hlb}}}(\FLIn{\hla}) \inlineExplain{Def. monitor}
        \end{align*}

        \proofcase{$S = \coAssert{\hlb}$}
        \begin{align*}
             & \FLIn{\vc{\coAssert{\hlb}}(\hla)}                                                                     &  & \\
             & \qiff \FLIn{\hlb \cup \hla} \inlineExplain{Def. \symVc}                                               &  & \\
             & \qiff \FLMonVar \setminus \hlb \subseteq \hla \inlineExplain{Set difference}                          &  & \\
             & \qiff \FLIn{\hla}\substBy{\FLMonVar}{\FLMonVar \setminus \hlb} \inlineExplain{Def. substBy}           &  & \\
             & \qiff \vcB{\stmtRasgn{\FLMonVar}{\FLMonVar \setminus \hlb}}(\FLIn{\hla}) \inlineExplain{Def. \symVcB} &  & \\
             & \qiff \vcB{\FLMonitor{\coAssert{\hlb}}}(\FLIn{\hla}) \inlineExplain{Def. monitor}
        \end{align*}

        \proofcase{$S = \coAssume{\hlb}$}
        \begin{align*}
                                                       & \FLIn{\vc{\coAssume{\hlb}}(\hla)}                                                                    &   & \\
                                                       & \qiff \FLIn{\hla \setminus \hlb} \inlineExplain{Def. \symVc}                                         &   & \\
                                                       & \qiff \FLIn{\FLAllWords \setminus \hlb}
            \land
            \FLIn{\hla} \inlineExplain{Set difference} &                                                                                                      &     \\
                                                       & \qiff \vcB{\stmtAssert{\FLIn{\FLAllWords \setminus \hlb}}}(\FLIn{\hla}) \inlineExplain{Def. \symVcB} &   & \\
                                                       & \qiff \vcB{\FLMonitor{\coAssume{\hlb}}}(\FLIn{\hla}) \inlineExplain{Def. monitor}
        \end{align*}

        \proofcase{$S = \stmtValidate$}
        \begin{align*}
             & \FLIn{\vc{\stmtValidate}(\hla)}                                                                                                                   &  & \\
             & \qiff \FLIn{\heylovalidate{\hla}} \inlineExplain{Def. \symVc}                                                                                     &  & \\
             & \qiff (\FLMonVar = \emptyset) \lor (\FLAllWords \subseteq \hla) \inlineExplain{Def. \stmtValidate}                                                &  & \\
             & \qiff (\FLMonVar = \emptyset) \lor \FLIn{\hla}\substBy{\FLMonVar}{\FLAllWords} \inlineExplain{Def. substBy}                                       &  & \\
             & \qiff (\FLMonVar = \emptyset) \lor \vcB{\stmtRasgn{\FLMonVar}{\FLAllWords}}(\FLIn{\hla}) \inlineExplain{Def. \symVcB}                             &  & \\
             & \qiff \vcB{\stmtAssume{(\FLMonVar \neq \emptyset)}}\bigl(\vcB{\stmtRasgn{\FLMonVar}{\FLAllWords}}(\FLIn{\hla})\bigr) \inlineExplain{Def. \symVcB} &  & \\
             & \qiff \vcB{\stmtAssume{(\FLMonVar \neq \emptyset)} \fatsemi \stmtRasgn{\FLMonVar}{\FLAllWords}}(\FLIn{\hla}) \inlineExplain{Def. \symVcB}         &  & \\
             & \qiff \vcB{\FLMonitor{\stmtValidate}}(\FLIn{\hla}) \inlineExplain{Def. monitor}
        \end{align*}

        \proofcase{$S = \stmtWeigh{a}$}
        \begin{align*}
                                                                                                          & \FLIn{\vc{\stmtWeigh{a}}(\hla)}                                                                                                                              &   & \\
                                                                                                          & \qiff \FLIn{a \cdot \hla} \inlineExplain{Def. \symVc}                                                                                                        &   & \\
                                                                                                          & \qiff \FLIn{a \cdot \FLAllWords}
            \land
            a^{-1} \cdot \FLMonVar \subseteq \hla \inlineExplain{\Cref{eq:residual-language-law-lower}} &                                                                                                                                                              &     \\
                                                                                                          & \qiff \FLIn{a \cdot \FLAllWords}
            \land
            \FLIn{\hla}\substBy{\FLMonVar}{a^{-1} \cdot \FLMonVar} \inlineExplain{Def. substBy}          &                                                                                                                                                              &     \\
                                                                                                          & \qiff \FLIn{a \cdot \FLAllWords}
            \land
            \vcB{\stmtRasgn{\FLMonVar}{a^{-1} \cdot \FLMonVar}}(\FLIn{\hla}) \inlineExplain{Def. \symVcB} &                                                                                                                                                              &     \\
                                                                                                          & \qiff \vcB{\stmtAssert{\FLIn{a \cdot \FLAllWords}}}\bigl(\vcB{\stmtRasgn{\FLMonVar}{a^{-1} \cdot \FLMonVar}}(\FLIn{\hla})\bigr) \inlineExplain{Def. \symVcB} &   & \\
                                                                                                          & \qiff \vcB{\stmtAssert{\FLIn{a \cdot \FLAllWords}} \fatsemi \stmtRasgn{\FLMonVar}{a^{-1} \cdot \FLMonVar}}(\FLIn{\hla}) \inlineExplain{Def. \symVcB}         &   & \\
                                                                                                          & \qiff \vcB{\FLMonitor{\stmtWeigh{a}}}(\FLIn{\hla}) \inlineExplain{Def. monitor}
        \end{align*}

        \proofcase{$S = \stmtRasgn{x}{e}$}
        \begin{align*}
             & \FLIn{\vc{\stmtRasgn{x}{e}}(\hla)}                                                 &  & \\
             & \qiff \FLIn{\hla\substBy{x}{e}} \inlineExplain{Def. \symVc}                        &  & \\
             & \qiff \FLIn{\hla}\substBy{x}{e} \inlineExplain{Freshness}                          &  & \\
             & \qiff \vcB{\stmtRasgn{x}{e}}(\FLIn{\hla}) \inlineExplain{Def. \symVcB}             &  & \\
             & \qiff \vcB{\FLMonitor{\stmtRasgn{x}{e}}}(\FLIn{\hla}) \inlineExplain{Def. monitor}
        \end{align*}

        \proofcase{$S = S_1 \fatsemi S_2$}
        Since $\FLMonVar$ is fresh and does not occur in $S_2$ or~$\hla$, it also does not occur in $\vc{S_2}(\hla)$.
        \begin{align*}
             & \FLIn{\vc{S_1 \fatsemi S_2}(\hla)}                                                                       &  & \\
             & \qiff \FLIn{\vc{S_1}(\vc{S_2}(\hla))} \inlineExplain{Def. \symVc}                                        &  & \\
             & \qiff \vcB{\FLMonitor{S_1}}(\FLIn{\vc{S_2}(\hla)}) \inlineExplain{I.H. for $S_1$}                        &  & \\
             & \qiff \vcB{\FLMonitor{S_1}}\bigl(\vcB{\FLMonitor{S_2}}(\FLIn{\hla})\bigr) \inlineExplain{I.H. for $S_2$} &  & \\
             & \qiff \vcB{\FLMonitor{S_1} \fatsemi \FLMonitor{S_2}}(\FLIn{\hla}) \inlineExplain{Def. \symVcB}           &  & \\
             & \qiff \vcB{\FLMonitor{S_1 \fatsemi S_2}}(\FLIn{\hla}) \inlineExplain{Def. monitor}
        \end{align*}

        \proofcase{$S = \stmtHavoc{x}$}
        This is immediate from the definition of~$\symVc$, the definition
        of~$\symVcB$, and the fact that $\FLMonitor{\stmtHavoc{x}} = \stmtHavoc{x}$.

        \proofcase{$S = \stmtDemonic{S_1}{S_2}$}
        This is immediate from the induction hypotheses and
        \(
        \FLIn{X \cap Y}
        \iff
        \FLIn{X} \land \FLIn{Y}
        \).
    \end{itemize}
\end{proof}

\begin{lemma}[Lax lower-bound translation]
    \label{lem:residual-language-monitor-lax-joins}
    Let $S_0$ range over the fragment of \Cref{thm:residual-language-monitor}. For
    \[
        S ::= S_0
        \mid \stmtWeighted{S}{S}
        \mid \stmtAngelic{S}{S}
        \mid \coHavoc{x}
        \mid \coValidate
        \mid S \fatsemi S
    \]
    define $\FLMonitorPrime{-}$ by
    \begin{align*}
        \FLMonitorPrime{S_0}
         & = \FLMonitor{S_0}                                   \\
        \FLMonitorPrime{\stmtWeighted{S_1}{S_2}}
         & = \stmtWeighted{\FLMonitorPrime{S_1}}{\FLMonitorPrime{S_2}} \\
        \FLMonitorPrime{\stmtAngelic{S_1}{S_2}}
         & = \stmtAngelic{\FLMonitorPrime{S_1}}{\FLMonitorPrime{S_2}} \\
        \FLMonitorPrime{\coHavoc{x}}
         & = \coHavoc{x}                                              \\
        \FLMonitorPrime{\coValidate}
         & = \stmtAssert{(\FLMonVar \subseteq \emptyset)}              \\
        \FLMonitorPrime{S_1 \fatsemi S_2}
         & = \FLMonitorPrime{S_1} \fatsemi \FLMonitorPrime{S_2}.
    \end{align*}
    Then, for every such statement~$S$ and every language~$\hla$,
    \[
        \vcB{\FLMonitorPrime{S}}(\FLIn{\hla})
        \Rightarrow
        \FLIn{\vc{S}(\hla)}.
    \]
\end{lemma}

\begin{proof}
    The cases from \Cref{thm:residual-language-monitor} are identical. We only
    show the additional lax cases and sequencing.

    \begin{itemize}
        \proofcase{$S = \stmtWeighted{S_1}{S_2}$ or $S = \stmtAngelic{S_1}{S_2}$}
        \begin{align*}
             & \vcB{\FLMonitorPrime{S}}(\FLIn{\hla})                                                                                                  &  & \\
             & \qiff \vcB{\FLMonitorPrime{S_1}}(\FLIn{\hla}) \lor \vcB{\FLMonitorPrime{S_2}}(\FLIn{\hla}) \inlineExplain{Defs. monitor'/\symVcB} &  & \\
             & \qimplies \FLIn{\vc{S_1}(\hla)} \lor \FLIn{\vc{S_2}(\hla)} \inlineExplain{I.H.}                                                 &  & \\
             & \qimplies \FLIn{\vc{S_1}(\hla) \cup \vc{S_2}(\hla)} \inlineExplain{Union}                                                       &  & \\
             & \qiff \FLIn{\vc{S}(\hla)} \inlineExplain{Def. \symVc}
        \end{align*}

        \proofcase{$S = \coHavoc{x}$}
        \begin{align*}
             & \vcB{\FLMonitorPrime{\coHavoc{x}}}(\FLIn{\hla})                            &  & \\
             & \qiff \vcB{\coHavoc{x}}(\FLIn{\hla}) \inlineExplain{Def. monitor'}         &  & \\
             & \qiff \exists v.\, \FLIn{\hla\substBy{x}{v}} \inlineExplain{Def. \symVcB} &  & \\
             & \qimplies \FLIn{\bigcup_v \hla\substBy{x}{v}} \inlineExplain{Witness}   &  & \\
             & \qiff \FLIn{\vc{\coHavoc{x}}(\hla)} \inlineExplain{Def. \symVc}
        \end{align*}

        \proofcase{$S = \coValidate$}
        \begin{align*}
             & \vcB{\FLMonitorPrime{\coValidate}}(\FLIn{\hla})                                      &  & \\
             & \qiff \FLIn{\emptyset} \land \FLIn{\hla} \inlineExplain{Defs. monitor'/\symVcB}      &  & \\
             & \qimplies \FLIn{\emptyset}                                                          &  & \\
             & \qimplies \FLIn{\vc{\coValidate}(\hla)} \inlineExplain{$\emptyset \subseteq X$}
        \end{align*}

        \proofcase{$S = S_1 \fatsemi S_2$}
        Since $\FLMonVar$ is fresh and does not occur in $S_2$ or~$\hla$, it also
        does not occur in $\vc{S_2}(\hla)$.
        \begin{align*}
             & \vcB{\FLMonitorPrime{S_1 \fatsemi S_2}}(\FLIn{\hla})                                                                         &  & \\
             & \qiff \vcB{\FLMonitorPrime{S_1}}\bigl(\vcB{\FLMonitorPrime{S_2}}(\FLIn{\hla})\bigr) \inlineExplain{Defs. \symVcB/monitor'} &  & \\
             & \qimplies \vcB{\FLMonitorPrime{S_1}}(\FLIn{\vc{S_2}(\hla)}) \inlineExplain{I.H. for $S_2$, monotonicity}                     &  & \\
             & \qimplies \FLIn{\vc{S_1}(\vc{S_2}(\hla))} \inlineExplain{I.H. for $S_1$}                                           &  & \\
             & \qiff \FLIn{\vc{S_1 \fatsemi S_2}(\hla)} \inlineExplain{Def. \symVc}
        \end{align*}
    \end{itemize}
\end{proof}

The left-hand side of \Cref{fig:cas-counter-livelock-monitor} shows the modified CAS livelock example $\clf'$ from \Cref{sec:cas-counter}, encoded by Park induction (cf.~\Cref{sec:local_park_induction}) in the formal-language instance of \wHeyVL{}, together with its translated monitor version in the Boolean instance over residual-language predicates according to \Cref{lem:residual-language-monitor-lax-joins}.
Because a procedure precondition is implemented as an initial $\symAssume$ and the monitor translation maps $\symAssume$ to an assignment on~$\FLMonVar$, the figure internalizes the original pre/postconditions as explicit wrapper statements rather than translating them through the procedure header.
The schematic branch $\BIGBRANCH{j=1}{N}{\ASSIGN{i}{j}}$ on the left is rendered on the right by the existential explicit choice $\coHavoc{i} \fatsemi \coAssume{\coEmbed{1 \le i \le N}}$.
Here $\coAssume{\coEmbed{1 \le i \le N}}$ keeps exactly the valid thread indices, since $\coEmbed{b}$ denotes the Boolean complement of the embedding of~$b$ and $\coAssume{\psi}$ excludes the states satisfying~$\psi$.
The lowered procedure verifies.
Thus, by soundness of local Park induction (\Cref{thm:heyvl-soundness-local-park-induction-encoding}) together with \Cref{lem:residual-language-monitor-lax-joins}, this implies $I_k \subseteq \wlp{\clf'}(\szero)$. In \Caesar{}, this example still has to be instantiated to a concrete finite nondeterministic choice.

\begin{figure}[H]
    \begingroup
    \small
    \let\qquad\quad
    \setlength{\tabcolsep}{0pt}
    \centering
    \begin{adjustbox}{max width=\linewidth}
        \begin{tabular}{@{}p{0.49\linewidth}@{\hspace{0.2em}}p{0.49\linewidth}@{}}
            \toprule
            \centering\makecell{Source\\$\wHeyVLLang{}$} &
            \centering\makecell{Lowered\\$\wHeyVLBool{}$}
            \tabularnewline
            \midrule
            \begin{minipage}[t]{\linewidth}
                \vspace{0pt}
                \[
                \begin{aligned}
                    & \weightedProcHead{\rigLang}{casLivelock}{\typeof{k}{\Nats}} \\
                    & \weightedProcOut{\typeof{i,v}{\Nats}, \typeof{l_v}{\listtype}} \\
                    & \qquad \Requires{\top} \\
                    & \qquad \Ensures{\top} \\
                    & \blockStart \\
                    & \qquad \Assume{I_k} \fatsemi \\
                    & \qquad \Assert{I_k} \fatsemi \\
                    & \qquad \Havoc{i, v, l_v} \fatsemi \\
                    & \qquad \Validate \fatsemi \\
                    & \qquad \Assume{I_k} \fatsemi \\
                    & \qquad \IF{\true} \\
                    & \qquad\qquad \BIGBRANCH{j=1}{N}{\ASSIGN{i}{j}} \fatsemi \\
                    & \qquad\qquad \IF{l_v[i] = v} \\
                    & \qquad\qquad\qquad \ASSIGN{v}{v+1} \fatsemi \\
                    & \qquad\qquad\qquad \stmtWeigh{S_i} \fatsemi \\
                    & \qquad\qquad\qquad \ASSIGN{l_v[i]}{v} \\
                    & \qquad\qquad \ELSE \\
                    & \qquad\qquad\qquad \stmtWeigh{F_i} \\
                    & \qquad\qquad \} \fatsemi \\
                    & \qquad\qquad \Assert{I_k} \fatsemi \Assume{\embed{\false}} \\
                    & \qquad \ELSE~\} \fatsemi \\
                    & \qquad \Assert{\emptyset} \\
                    & \blockEnd
                \end{aligned}
                \]
            \end{minipage}
            &
            \begin{minipage}[t]{\linewidth}
                \vspace{0pt}
                \[
                \begin{aligned}
                    & \weightedProcHead{\rigBool}{casLivelock}{\typeof{k}{\Nats}} \\
                    & \weightedProcOut{\typeof{i,v}{\Nats}, \typeof{l_v}{\listtype}, \typeof{\FLMonVar}{\FLType}} \\
                    & \qquad \Requires{\top} \\
                    & \qquad \Ensures{\top} \\
                    & \blockStart \\
                    & \qquad \ASSIGN{\FLMonVar}{\FLMonVar \cap I_k} \fatsemi \\
                    & \qquad \Assert{\FLIn{I_k}} \fatsemi \\
                    & \qquad \Havoc{i, v, l_v} \fatsemi \\
                    & \qquad \Assume{(\FLMonVar \neq \emptyset)} \fatsemi
                        \ASSIGN{\FLMonVar}{\FLAllWords} \fatsemi \\
                    & \qquad \ASSIGN{\FLMonVar}{\FLMonVar \cap I_k} \fatsemi \\
                    & \qquad \IF{\true} \\
                    & \qquad\qquad \coHavoc{i} \fatsemi \coAssume{\coEmbed{1 \le i \le N}} \fatsemi
                        \vphantom{\BIGBRANCH{j=1}{N}{\ASSIGN{i}{j}}} \\
                    & \qquad\qquad \IF{l_v[i] = v} \\
                    & \qquad\qquad\qquad \ASSIGN{v}{v+1} \fatsemi \\
                    & \qquad\qquad\qquad \Assert{\FLIn{S_i \cdot \FLAllWords}} \fatsemi
                        \ASSIGN{\FLMonVar}{S_i^{-1} \cdot \FLMonVar} \fatsemi \\
                    & \qquad\qquad\qquad \ASSIGN{l_v[i]}{v} \\
                    & \qquad\qquad \ELSE \\
                    & \qquad\qquad\qquad \Assert{\FLIn{F_i \cdot \FLAllWords}} \fatsemi
                        \ASSIGN{\FLMonVar}{F_i^{-1} \cdot \FLMonVar} \\
                    & \qquad\qquad \} \fatsemi \\
                    & \qquad\qquad \Assert{\FLIn{I_k}} \fatsemi
                        \ASSIGN{\FLMonVar}{\emptyset} \\
                    & \qquad \ELSE~\} \fatsemi \\
                    & \qquad \Assert{\FLIn{\emptyset}} \\
                    & \blockEnd
                \end{aligned}
                \]
            \end{minipage}
            \tabularnewline
            \bottomrule
        \end{tabular}
    \end{adjustbox}
    \caption{Modified CAS livelock witness from \Cref{sec:cas-counter} for $I_k \definedAs \iverson{l_v[k] \neq v} \cdot \Set{F_k^\omega}$ and its monitor translation into the Boolean instance over residual-language predicates, following \Cref{lem:residual-language-monitor-lax-joins}. Because assumptions lower to assignments, the original pre/post are made explicit as wrapper statements, so both procedures use the trivial outer specification. Here $\FLType$ denotes languages over $\Gamma^\infty$.}
    \label{fig:cas-counter-livelock-monitor}
    \endgroup
\end{figure}

\subsubsection{Upper-Bound Translation}

\begin{definition}[Residual-language upper-bound monitor translation]
    \label{def:residual-language-upper-monitor}
    Let $\FLMonVar$ be fresh and range over languages over $\FLAllWords$.
    For
    \[
        \begin{aligned}
            S ::= {} & \stmtRasgn{x}{e}
            \mid \stmtAssert{\hlb}
            \mid \stmtAssume{\hlb}
            \mid \coAssert{\hlb}
            \mid \coAssume{\hlb}
            \mid \stmtWeigh{a}          \\
            \mid {}  & \coValidate
            \mid \stmtWeighted{S}{S}
            \mid \stmtAngelic{S}{S}
            \mid \coHavoc{x}
            \mid S \fatsemi S
        \end{aligned}
    \]
    define $\FLUpMonitor{-}$ by
    \begin{align*}
        \FLUpMonitor{\stmtWeighted{S_1}{S_2}}
         & = \stmtDemonic{\FLUpMonitor{S_1}}{\FLUpMonitor{S_2}}                 \\
        \FLUpMonitor{\stmtAngelic{S_1}{S_2}}
         & = \stmtDemonic{\FLUpMonitor{S_1}}{\FLUpMonitor{S_2}}                 \\
        \FLUpMonitor{\coHavoc{x}}
         & = \stmtHavoc{x}                                                      \\
        \FLUpMonitor{S_1 \fatsemi S_2}
         & = \FLUpMonitor{S_1} \fatsemi \FLUpMonitor{S_2}                       \\
        \FLUpMonitor{\stmtAssert{\hlb}}
         & = \stmtRasgn{\FLMonVar}{\FLMonVar \cup (\FLAllWords \setminus \hlb)} \\
        \FLUpMonitor{\stmtAssume{\hlb}}
         & = \stmtAssert{((\FLAllWords \setminus \hlb) \subseteq \FLMonVar)}    \\
        \FLUpMonitor{\coAssert{\hlb}}
         & = \stmtAssert{(\hlb \subseteq \FLMonVar)}                            \\
        \FLUpMonitor{\coAssume{\hlb}}
         & = \stmtRasgn{\FLMonVar}{\FLMonVar \cup \hlb}                         \\
        \FLUpMonitor{\stmtWeigh{a}}
         & = \stmtRasgn{\FLMonVar}{a^{-1} \cdot \FLMonVar}                      \\
        \FLUpMonitor{\coValidate}
         & = \coAssert{(\FLAllWords \subseteq \FLMonVar)}
        \fatsemi
        \stmtRasgn{\FLMonVar}{\emptyset}
        \fatsemi
        \coValidate.
    \end{align*}
    It leaves $\stmtRasgn{x}{e}$ unchanged.
\end{definition}

\begin{theorem}
    \label{thm:residual-language-monitor-upper}
    Assume that $\FLMonVar$ is fresh, i.e., it does not occur in $S$ or~$\hla$.
    Then, for every statement~$S$ in the fragment above and language~$\hla$,
    \[
        \FLOut{\vc{S}(\hla)}
        \qiff
        \vcB{\FLUpMonitor{S}}(\FLOut{\hla})~.
    \]
\end{theorem}

\begin{proof}
        The proof is by structural induction.

    \begin{itemize}
        \proofcase{$S = \stmtAssert{\hlb}$}
        \begin{align*}
             & \FLOut{\vc{\stmtAssert{\hlb}}(\hla)}                                                                                      &  & \\
             & \qiff \FLOut{\hlb \cap \hla} \inlineExplain{Def. \symVc}                                                                  &  & \\
             & \qiff \hla \subseteq \FLMonVar \cup (\FLAllWords \setminus \hlb) \inlineExplain{Set difference}                           &  & \\
             & \qiff \FLOut{\hla}\substBy{\FLMonVar}{\FLMonVar \cup (\FLAllWords \setminus \hlb)} \inlineExplain{Def. substBy}           &  & \\
             & \qiff \vcB{\stmtRasgn{\FLMonVar}{\FLMonVar \cup (\FLAllWords \setminus \hlb)}}(\FLOut{\hla}) \inlineExplain{Def. \symVcB} &  & \\
             & \qiff \vcB{\FLUpMonitor{\stmtAssert{\hlb}}}(\FLOut{\hla}) \inlineExplain{Def. monitor}
        \end{align*}

        \proofcase{$S = \stmtAssume{\hlb}$}
        \begin{align*}
                                                        & \FLOut{\vc{\stmtAssume{\hlb}}(\hla)}                                                                                   &   & \\
                                                        & \qiff \FLOut{\hlb \impl \hla} \inlineExplain{Def. \symVc}                                                              &   & \\
                                                        & \qiff (\FLAllWords \setminus \hlb) \subseteq \FLMonVar
            \land
            \FLOut{\hla} \inlineExplain{Set difference} &                                                                                                                        &     \\
                                                        & \qiff \vcB{\stmtAssert{((\FLAllWords \setminus \hlb) \subseteq \FLMonVar)}}(\FLOut{\hla}) \inlineExplain{Def. \symVcB} &   & \\
                                                        & \qiff \vcB{\FLUpMonitor{\stmtAssume{\hlb}}}(\FLOut{\hla}) \inlineExplain{Def. monitor}
        \end{align*}

        \proofcase{$S = \coAssert{\hlb}$}
        \begin{align*}
                                               & \FLOut{\vc{\coAssert{\hlb}}(\hla)}                                                             &   & \\
                                               & \qiff \FLOut{\hlb \cup \hla} \inlineExplain{Def. \symVc}                                       &   & \\
                                               & \qiff \hlb \subseteq \FLMonVar
            \land
            \FLOut{\hla} \inlineExplain{Union} &                                                                                                &     \\
                                               & \qiff \vcB{\stmtAssert{(\hlb \subseteq \FLMonVar)}}(\FLOut{\hla}) \inlineExplain{Def. \symVcB} &   & \\
                                               & \qiff \vcB{\FLUpMonitor{\coAssert{\hlb}}}(\FLOut{\hla}) \inlineExplain{Def. monitor}
        \end{align*}

        \proofcase{$S = \coAssume{\hlb}$}
        \begin{align*}
             & \FLOut{\vc{\coAssume{\hlb}}(\hla)}                                                                &  & \\
             & \qiff \FLOut{\hla \setminus \hlb} \inlineExplain{Def. \symVc}                                     &  & \\
             & \qiff \hla \subseteq \FLMonVar \cup \hlb \inlineExplain{Set difference}                           &  & \\
             & \qiff \FLOut{\hla}\substBy{\FLMonVar}{\FLMonVar \cup \hlb} \inlineExplain{Def. substBy}           &  & \\
             & \qiff \vcB{\stmtRasgn{\FLMonVar}{\FLMonVar \cup \hlb}}(\FLOut{\hla}) \inlineExplain{Def. \symVcB} &  & \\
             & \qiff \vcB{\FLUpMonitor{\coAssume{\hlb}}}(\FLOut{\hla}) \inlineExplain{Def. monitor}
        \end{align*}

        \proofcase{$S = \stmtWeigh{a}$}
        \begin{align*}
             & \FLOut{\vc{\stmtWeigh{a}}(\hla)}                                                                     &  & \\
             & \qiff \FLOut{a \cdot \hla} \inlineExplain{Def. \symVc}                                               &  & \\
             & \qiff \hla \subseteq a^{-1} \cdot \FLMonVar \inlineExplain{\Cref{eq:residual-language-law-upper}} &  & \\
             & \qiff \FLOut{\hla}\substBy{\FLMonVar}{a^{-1} \cdot \FLMonVar} \inlineExplain{Def. substBy}           &  & \\
             & \qiff \vcB{\stmtRasgn{\FLMonVar}{a^{-1} \cdot \FLMonVar}}(\FLOut{\hla}) \inlineExplain{Def. \symVcB} &  & \\
             & \qiff \vcB{\FLUpMonitor{\stmtWeigh{a}}}(\FLOut{\hla}) \inlineExplain{Def. monitor}
        \end{align*}

        \proofcase{$S = \coValidate$}
        \begin{align*}
                                                                                                                       & \FLOut{\vc{\coValidate}(\hla)}                                                                                                                                             &   & \\
                                                                                                                       & \qiff \FLOut{\heylocovalidate{\hla}} \inlineExplain{Def. \symVc}                                                                                                           &   & \\
                                                                                                                       & \qiff (\FLAllWords \subseteq \FLMonVar)
            \lor
            \heylocovalidate{\left(\FLOut{\hla}\substBy{\FLMonVar}{\emptyset}\right)} \inlineExplain{Def. \coValidate} &                                                                                                                                                                            &     \\
                                                                                                                       & \qiff (\FLAllWords \subseteq \FLMonVar)
            \lor
            \vcB{\coValidate}(\FLOut{\hla}\substBy{\FLMonVar}{\emptyset}) \inlineExplain{Def. \symVcB}                 &                                                                                                                                                                            &     \\
                                                                                                                       & \qiff \vcB{\coAssert{(\FLAllWords \subseteq \FLMonVar)}}\bigl(\vcB{\coValidate}(\FLOut{\hla}\substBy{\FLMonVar}{\emptyset})\bigr) \inlineExplain{Def. \symVcB}             &   & \\
                                                                                                                       & \qiff \vcB{\coAssert{(\FLAllWords \subseteq \FLMonVar)}}\bigl(\vcB{\stmtRasgn{\FLMonVar}{\emptyset} \fatsemi \coValidate}(\FLOut{\hla})\bigr) \inlineExplain{Def. \symVcB} &   & \\
                                                                                                                       & \qiff \vcB{\coAssert{(\FLAllWords \subseteq \FLMonVar)} \fatsemi \stmtRasgn{\FLMonVar}{\emptyset} \fatsemi \coValidate}(\FLOut{\hla}) \inlineExplain{Def. \symVcB}         &   & \\
                                                                                                                       & \qiff \vcB{\FLUpMonitor{\coValidate}}(\FLOut{\hla}) \inlineExplain{Def. monitor}
        \end{align*}

        \proofcase{$S = \stmtWeighted{S_1}{S_2}$ or $S = \stmtAngelic{S_1}{S_2}$}
        This is immediate from the induction hypotheses and
        \(
        \FLOut{X \cup Y}
        \iff
        \FLOut{X} \land \FLOut{Y}
        \).

        \proofcase{$S = \coHavoc{x}$}
        This is immediate from the definitions and
        \(
        \FLOut{\bigcup_{v} X_v}
        \iff
        \forall v.\ \FLOut{X_v}
        \).

        \proofcase{$S = \stmtRasgn{x}{e}$}
        This follows exactly as in the proof of \Cref{thm:residual-language-monitor}, since the translation leaves $\stmtRasgn{x}{e}$ unchanged.

        \proofcase{$S = S_1 \fatsemi S_2$}
        This follows exactly as in the proof of \Cref{thm:residual-language-monitor}, using sequential composition and the induction hypotheses.
    \end{itemize}
\end{proof}

\begin{lemma}[Lax upper-bound translation]
    \label{lem:residual-language-monitor-upper-lax-meets}
    Let $S_0$ range over the fragment of \Cref{thm:residual-language-monitor-upper}. For
    \[
        S ::= S_0
        \mid \stmtDemonic{S}{S}
        \mid \stmtHavoc{x}
        \mid \stmtValidate
        \mid S \fatsemi S
    \]
    define $\FLUpMonitorPrime{-}$ by
    \begin{align*}
        \FLUpMonitorPrime{S_0}
         & = \FLUpMonitor{S_0}                                      \\
        \FLUpMonitorPrime{\stmtDemonic{S_1}{S_2}}
         & = \stmtAngelic{\FLUpMonitorPrime{S_1}}{\FLUpMonitorPrime{S_2}} \\
        \FLUpMonitorPrime{\stmtHavoc{x}}
         & = \coHavoc{x}                                            \\
        \FLUpMonitorPrime{\stmtValidate}
         & = \stmtAssert{(\FLAllWords \subseteq \FLMonVar)}          \\
        \FLUpMonitorPrime{S_1 \fatsemi S_2}
         & = \FLUpMonitorPrime{S_1} \fatsemi \FLUpMonitorPrime{S_2}.
    \end{align*}
    Then, for every such statement~$S$ and every language~$\hla$,
    \[
        \vcB{\FLUpMonitorPrime{S}}(\FLOut{\hla})
        \Rightarrow
        \FLOut{\vc{S}(\hla)}.
    \]
\end{lemma}

\begin{proof}
        The cases from \Cref{thm:residual-language-monitor-upper} are identical. We only
    show the additional lax cases and sequencing.

    \begin{itemize}
        \proofcase{$S = \stmtDemonic{S_1}{S_2}$}
        \begin{align*}
             & \vcB{\FLUpMonitorPrime{S}}(\FLOut{\hla})                                                                                                  &  & \\
             & \qiff \vcB{\FLUpMonitorPrime{S_1}}(\FLOut{\hla}) \lor \vcB{\FLUpMonitorPrime{S_2}}(\FLOut{\hla}) \inlineExplain{Defs. $\mathsf{monitor}^{\uparrow\prime}$/\symVcB} &  & \\
             & \qimplies \FLOut{\vc{S_1}(\hla)} \lor \FLOut{\vc{S_2}(\hla)} \inlineExplain{I.H.}                                                 &  & \\
             & \qimplies \FLOut{\vc{S_1}(\hla) \cap \vc{S_2}(\hla)} \inlineExplain{Intersection}                                                       &  & \\
             & \qiff \FLOut{\vc{S}(\hla)} \inlineExplain{Def. \symVc}
        \end{align*}

        \proofcase{$S = \stmtHavoc{x}$}
        \begin{align*}
             & \vcB{\FLUpMonitorPrime{\stmtHavoc{x}}}(\FLOut{\hla})                          &  & \\
             & \qiff \vcB{\coHavoc{x}}(\FLOut{\hla}) \inlineExplain{Def. $\mathsf{monitor}^{\uparrow\prime}$} &  & \\
             & \qiff \exists v.\, \FLOut{\hla\substBy{x}{v}} \inlineExplain{Def. \symVcB} &  & \\
             & \qimplies \FLOut{\bigcap_v \hla\substBy{x}{v}} \inlineExplain{Witness} &  & \\
             & \qiff \FLOut{\vc{\stmtHavoc{x}}(\hla)} \inlineExplain{Def. \symVc}
        \end{align*}

        \proofcase{$S = \stmtValidate$}
        \begin{align*}
             & \vcB{\FLUpMonitorPrime{\stmtValidate}}(\FLOut{\hla})                                      &  & \\
             & \qiff (\FLAllWords \subseteq \FLMonVar) \land \FLOut{\hla} \inlineExplain{Defs. $\mathsf{monitor}^{\uparrow\prime}$/\symVcB} &  & \\
             & \qimplies \FLAllWords \subseteq \FLMonVar                                                  &  & \\
             & \qimplies \FLOut{\vc{\stmtValidate}(\hla)} \inlineExplain{$\vc{\stmtValidate}(\hla) \subseteq \FLAllWords$}
        \end{align*}

        \proofcase{$S = S_1 \fatsemi S_2$}
        Since $\FLMonVar$ is fresh and does not occur in $S_2$ or~$\hla$, it also
        does not occur in $\vc{S_2}(\hla)$.
        \begin{align*}
             & \vcB{\FLUpMonitorPrime{S_1 \fatsemi S_2}}(\FLOut{\hla})                                                                         &  & \\
             & \qiff \vcB{\FLUpMonitorPrime{S_1}}\bigl(\vcB{\FLUpMonitorPrime{S_2}}(\FLOut{\hla})\bigr) \inlineExplain{Defs. \symVcB/$\mathsf{monitor}^{\uparrow\prime}$} &  & \\
             & \qimplies \vcB{\FLUpMonitorPrime{S_1}}(\FLOut{\vc{S_2}(\hla)}) \inlineExplain{I.H. for $S_2$, monotonicity}                     &  & \\
             & \qimplies \FLOut{\vc{S_1}(\vc{S_2}(\hla))} \inlineExplain{I.H. for $S_1$}                                           &  & \\
             & \qiff \FLOut{\vc{S_1 \fatsemi S_2}(\hla)} \inlineExplain{Def. \symVc}
        \end{align*}
    \end{itemize}
\end{proof}

\subsection{Tropical Translation}
\label{sec:appendix-tropical-lowering}

The tropical monoid-module $\modtrop$ has weights in $\NatsX$, with $\oplus=\min$ and scalar multiplication $+$.
We embed weights into $\eureal$ by inclusion, so $\stmtWeighted{S_1}{S_2}$ maps to demonic choice and $\stmtWeigh{c}$ maps to $\stmtTick{c}$.
A source lower-bound check maps to a target upper-bound check, hence source procedures map to target coprocedures.
The verification statements map as in the definition below: assert, assume, havoc, and validate map to their \texttt{co} variants, and conversely.

\begin{definition}[Tropical translation]
    \label{def:tropical-translation}
    Define a translation
    \[
        \TropicalTrans{-} : \wHeyVLTrop{} \to \wHeyVLEUReal{}
    \]
    recursively by
    \begin{align*}
        \TropicalTrans{\stmtRasgn{x}{e}}
         & = \stmtRasgn{x}{e}                                                     \\
        \TropicalTrans{\stmtWeigh{c}}
         & = \stmtTick{c}                                                         \\
        \TropicalTrans{\stmtWeighted{S_1}{S_2}}
         & = \stmtDemonic{\TropicalTrans{S_1}}{\TropicalTrans{S_2}}               \\
        \TropicalTrans{S_1 \fatsemi S_2}
         & = \TropicalTrans{S_1} \fatsemi \TropicalTrans{S_2}                     \\
        \TropicalTrans{\stmtDemonic{S_1}{S_2}}
         & = \stmtAngelic{\TropicalTrans{S_1}}{\TropicalTrans{S_2}}               \\
        \TropicalTrans{\stmtAngelic{S_1}{S_2}}
         & = \stmtDemonic{\TropicalTrans{S_1}}{\TropicalTrans{S_2}}               \\
        \TropicalTrans{\stmtAssert{\hlb}}
         & = \coAssert{\widehat{\hlb}}                                            \\
        \TropicalTrans{\stmtAssume{\hlb}}
         & = \coAssume{\widehat{\hlb}}                                            \\
        \TropicalTrans{\coAssert{\hlb}}
         & = \stmtAssert{\widehat{\hlb}}                                          \\
        \TropicalTrans{\coAssume{\hlb}}
         & = \stmtAssume{\widehat{\hlb}}                                          \\
        \TropicalTrans{\stmtHavoc{x}}
         & = \coHavoc{x}                                                          \\
        \TropicalTrans{\coHavoc{x}}
         & = \stmtHavoc{x}                                                        \\
        \TropicalTrans{\stmtValidate}
         & = \coValidate                                                          \\
        \TropicalTrans{\coValidate}
         & = \stmtValidate.
    \end{align*}
    On the right-hand side, each weight $c \in \NatsX$ is read as an element of $\eureal$.
    For formulae, $\widehat{\hlb}$ embeds constants pointwise, swaps $\sqcap$ with $\sqcup$, swaps $\impl$ with $\coimpl$, and swaps quantified infima with suprema.
    In particular, Boolean embeddings are dualized as $\widehat{\embed{b}} = \embed{\neg b}$.
    Any derived statement is translated after expanding it into this language.
\end{definition}

\begin{theorem}
    \label{thm:tropical-translation}
    For every statement $S$ of \wHeyVLTrop{}, every postweighting $r : \States \to \NatsX$, and every state $\State \in \States$,
    \[
        \vc{\TropicalTrans{S}}(\widehat{r})(\State)
        =
        \widehat{\vcTropical{S}(r)}(\State).
    \]
\end{theorem}

\begin{proof}
    We prove the claim by structural induction on $S$.
    \proofheading{Semiring-specific cases}
    \begin{align*}
        & \widehat{\vcTropical{\stmtWeighted{S_1}{S_2}}(r)}(\State) \\
        &= \min\!\left(\widehat{\vcTropical{S_1}(r)}(\State), \widehat{\vcTropical{S_2}(r)}(\State)\right)
            \inlineExplain{Def. \symVcTropical} \\
        &= \vc{\TropicalTrans{S_1}}(\widehat{r})(\State) \hand \vc{\TropicalTrans{S_2}}(\widehat{r})(\State)
            \inlineExplain{I.H., target meet is $\min$} \\
        &= \vc{\stmtDemonic{\TropicalTrans{S_1}}{\TropicalTrans{S_2}}}(\widehat{r})(\State)
            \inlineExplain{Def. \symVc} \\
        &= \vc{\TropicalTrans{\stmtWeighted{S_1}{S_2}}}(\widehat{r})(\State)
            \inlineExplain{Def. translation}
    \end{align*}
    For $S = \stmtWeigh{c}$, we calculate
    \begin{align*}
        & \widehat{\vcTropical{\stmtWeigh{c}}(r)}(\State) \\
        &= c + r(\State)
            \inlineExplain{Def. \symVcTropical} \\
        &= \vc{\stmtTick{c}}(\widehat{r})(\State)
            \inlineExplain{Def. \symVc} \\
        &= \vc{\TropicalTrans{\stmtWeigh{c}}}(\widehat{r})(\State)
            \inlineExplain{Def. translation.}
    \end{align*}
    \proofheading{Remaining cases}
    Sequential composition and assignment follow from the induction hypotheses and substitution commuting with $\widehat{-}$.
    The remaining cases use the same equations for $\sqcap/\sqcup$, $\impl/\coimpl$, quantified infima/suprema, and validation/covalidation.
\end{proof}

\Cref{fig:tropical-lowering-example} shows the tropical ski-rental example and its translation.

\begin{figure}[t]
    \centering
    \begin{adjustbox}{max width=\linewidth}
        \begin{tabular}{@{}p{0.49\linewidth}@{\hspace{0.2em}}p{0.49\linewidth}@{}}
            \toprule
            \centering\makecell{Source\\$\wHeyVLTrop{}$} &
            \centering\makecell{Lowered\\$\wHeyVLEUReal{}$}
            \tabularnewline
            \midrule
            \begin{minipage}[t]{\linewidth}
                \vspace{0pt}
                \[
                \begin{aligned}
                    & \weightedProcHead{\rigTrop}{skiRental}{\typeof{n_0}{\Nats}, \typeof{c_B}{\Nats}} \\
                    & \weightedProcOut{\typeof{n}{\Nats}} \\
                    & \qquad \Requires{I(n_0)} \\
                    & \qquad \Ensures{0} \\
                    & \blockStart \\
                    & \qquad \ASSIGN{n}{n_0} \fatsemi \Assert{I(n)} \fatsemi \\
                    & \qquad \Havoc{n} \fatsemi \Validate \fatsemi \Assume{I(n)} \fatsemi \\
                    & \qquad \IF{n > 0} \\
                    & \qquad\qquad \ASSIGN{n}{n - 1} \fatsemi \\
                    & \qquad\qquad \symWeightedBranching~\{ \\
                    & \qquad\qquad\qquad \WEIGH{1} \\
                    & \qquad\qquad \}~\stmtElseStart \\
                    & \qquad\qquad\qquad \WEIGH{c_B} \fatsemi \ASSIGN{n}{0} \\
                    & \qquad\qquad \} \fatsemi \\
                    & \qquad\qquad \Assert{I(n)} \fatsemi \Assume{\embed{\false}} \\
                    & \qquad \ELSE~\} \\
                    & \blockEnd
                \end{aligned}
                \]
            \end{minipage}
            &
            \begin{minipage}[t]{\linewidth}
                \vspace{0pt}
                \[
                \begin{aligned}
                    & \weightedCoprocHead{\rigEUReal}{skiRental}{\typeof{n_0}{\Nats}, \typeof{c_B}{\Nats}} \\
                    & \weightedProcOut{\typeof{n}{\Nats}} \\
                    & \qquad \Requires{\widehat{I}(n_0)} \\
                    & \qquad \Ensures{0} \\
                    & \blockStart \\
                    & \qquad \ASSIGN{n}{n_0} \fatsemi \coAssert{\widehat{I}(n)} \fatsemi \\
                    & \qquad \coHavoc{n} \fatsemi \coValidate \fatsemi \coAssume{\widehat{I}(n)} \fatsemi \\
                    & \qquad \IF{n > 0} \\
                    & \qquad\qquad \ASSIGN{n}{n - 1} \fatsemi \\
                    & \qquad\qquad \stmtDemonicStart \\
                    & \qquad\qquad\qquad \stmtTick{1} \\
                    & \qquad\qquad \}~\stmtElseStart \\
                    & \qquad\qquad\qquad \stmtTick{c_B} \fatsemi \ASSIGN{n}{0} \\
                    & \qquad\qquad \} \fatsemi \\
                    & \qquad\qquad \coAssert{\widehat{I}(n)} \fatsemi \coAssume{\embed{\true}} \\
                    & \qquad \ELSE~\} \\
                    & \blockEnd
                \end{aligned}
                \]
            \end{minipage}
            \tabularnewline
            \bottomrule
        \end{tabular}
    \end{adjustbox}
    \caption{%
        Tropical ski-rental lower-bound proof and its translation to the $\eureal$ instance according to \Cref{def:tropical-translation}.
        The source procedure is translated to a target coprocedure, so \Caesar{} checks the source lower bound as a target upper bound.
        Here $I(n) \definedAs \mathit{offlineCost}(n, c_B)$ and $\widehat{I}$ is its embedded target formula.
    }
    \label{fig:tropical-lowering-example}
\end{figure}

\subsection{Clearance Translation}
\label{sec:appendix-clearance-lowering}

The clearance semiring $\Semiring_{\clear} = (\clearUniverse, \min, \max, \mathsf{TS}, \mathsf{P})$ uses levels $\mathsf{P}<\mathsf{C}<\mathsf{S}<\mathsf{TS}$.
We embed these levels into ordinary $\eureal$ by
\[
    \widehat{\mathsf{P}}=0, \qquad \widehat{\mathsf{C}}=1, \qquad \widehat{\mathsf{S}}=2, \qquad \widehat{\mathsf{TS}}=\infty.
\]
Thus the top clearance level $\mathsf{TS}$ maps to the top element $\infty$ of the $\eureal$ instance.
Under this embedding, $\stmtWeighted{S_1}{S_2}$ maps to demonic choice and $\stmtWeigh{c}$ maps to $\coAssert{\widehat{c}}$:
\[
    \vc{\coAssert{\widehat{c}}}(\widehat{r})(\State)
    =
    \widehat{c} \hor \widehat{r}(\State)
    =
    \widehat{\max(c, r(\State))}.
\]
A source lower-bound check maps to a target upper-bound check, hence source procedures map to target coprocedures.
Dually, a source upper-bound check maps to a target lower-bound check, hence source coprocedures map to target procedures.
The verification statements map as in the definition below: assert, assume, havoc, and validate map to their \texttt{co} variants, and conversely.
\begin{definition}[Clearance translation]
    \label{def:clearance-translation}
    Define a translation
    \[
        \AccessControlTrans{-} : \wHeyVLClearance{} \to \wHeyVLEUReal{}
    \]
    recursively by
    \begin{align*}
        \AccessControlTrans{\stmtRasgn{x}{e}}
         & = \stmtRasgn{x}{e}                                               \\
        \AccessControlTrans{\stmtWeighted{S_1}{S_2}}
         & = \stmtDemonic{\AccessControlTrans{S_1}}{\AccessControlTrans{S_2}} \\
        \AccessControlTrans{\stmtWeigh{c}}
         & = \coAssert{\widehat{c}}                                          \\
        \AccessControlTrans{S_1 \fatsemi S_2}
         & = \AccessControlTrans{S_1} \fatsemi \AccessControlTrans{S_2}      \\
        \AccessControlTrans{\stmtDemonic{S_1}{S_2}}
         & = \stmtAngelic{\AccessControlTrans{S_1}}{\AccessControlTrans{S_2}} \\
        \AccessControlTrans{\stmtAngelic{S_1}{S_2}}
         & = \stmtDemonic{\AccessControlTrans{S_1}}{\AccessControlTrans{S_2}} \\
        \AccessControlTrans{\stmtAssert{\hlb}}
         & = \coAssert{\widehat{\hlb}}                                      \\
        \AccessControlTrans{\stmtAssume{\hlb}}
         & = \coAssume{\widehat{\hlb}}                                      \\
        \AccessControlTrans{\coAssert{\hlb}}
         & = \stmtAssert{\widehat{\hlb}}                                    \\
        \AccessControlTrans{\coAssume{\hlb}}
         & = \stmtAssume{\widehat{\hlb}}                                    \\
        \AccessControlTrans{\stmtHavoc{x}}
         & = \coHavoc{x}                                                    \\
        \AccessControlTrans{\coHavoc{x}}
         & = \stmtHavoc{x}                                                  \\
        \AccessControlTrans{\stmtValidate}
         & = \coValidate                                                    \\
        \AccessControlTrans{\coValidate}
         & = \stmtValidate.
    \end{align*}
    On the right-hand side, clearance constants and formulae are read via the embedding above into~$\eureal$.
    For formulae, $\widehat{\hlb}$ embeds constants pointwise, swaps $\sqcap$ with $\sqcup$, swaps $\impl$ with $\coimpl$, and swaps quantified infima with suprema.
    In particular, Boolean embeddings are dualized as $\widehat{\embed{b}} = \embed{\neg b}$.
    Any derived statement is translated after expanding it into this language.
\end{definition}

\begin{theorem}
    \label{thm:clearance-translation}
    For every statement $S$ of \wHeyVLClearance{}, every clearance postweighting $r : \States \to \clearUniverse$, and every state $\State \in \States$,
    \[
        \vc{\AccessControlTrans{S}}(\widehat{r})(\State)
        =
        \widehat{\vcAccess{S}(r)}(\State).
    \]
\end{theorem}

\begin{proof}
    We prove the claim by structural induction on $S$.
    \proofheading{Semiring-specific cases}
    For $S = \stmtWeighted{S_1}{S_2}$, we calculate
    \begin{align*}
        & \widehat{\vcAccess{\stmtWeighted{S_1}{S_2}}(r)}(\State) \\
        &= \min\!\left(\widehat{\vcAccess{S_1}(r)}(\State), \widehat{\vcAccess{S_2}(r)}(\State)\right)
            \inlineExplain{Def. \symVcAccess} \\
        &= \vc{\AccessControlTrans{S_1}}(\widehat{r})(\State) \hand \vc{\AccessControlTrans{S_2}}(\widehat{r})(\State)
            \inlineExplain{I.H., target meet is $\min$} \\
        &= \vc{\stmtDemonic{\AccessControlTrans{S_1}}{\AccessControlTrans{S_2}}}(\widehat{r})(\State)
            \inlineExplain{Def. \symVc} \\
        &= \vc{\AccessControlTrans{\stmtWeighted{S_1}{S_2}}}(\widehat{r})(\State)
            \inlineExplain{Def. translation.}
    \end{align*}

    For $S = \stmtWeigh{c}$, we calculate
    \begin{align*}
        & \widehat{\vcAccess{\stmtWeigh{c}}(r)}(\State) \\
        &= \widehat{\max(c, r(\State))}
            \inlineExplain{Def. \symVcAccess} \\
        &= \widehat{c} \hor \widehat{r}(\State)
            \inlineExplain{definition of $\widehat{-}$} \\
        &= \vc{\coAssert{\widehat{c}}}(\widehat{r})(\State)
            \inlineExplain{Def. \symVc} \\
        &= \vc{\AccessControlTrans{\stmtWeigh{c}}}(\widehat{r})(\State)
            \inlineExplain{Def. translation.}
    \end{align*}

    \proofheading{Remaining cases}
    Sequential composition and assignment follow from the induction hypotheses and substitution commuting with $\widehat{-}$.
    The remaining cases use the same equations for $\sqcap/\sqcup$, $\impl/\coimpl$, quantified infima/suprema, and validation/covalidation.
\end{proof}

\Cref{fig:clearance-lowering-example} shows the network-security example and its translation.

\begin{figure}[t]
    \centering
    \begin{adjustbox}{max width=\linewidth}
        \begin{tabular}{@{}p{0.47\linewidth}@{\hspace{1.8em}}p{0.47\linewidth}@{}}
            \toprule
            \centering\makecell{Source\\$\wHeyVLClearance{}$} &
            \centering\makecell{Lowered\\$\wHeyVLEUReal{}$}
            \tabularnewline
            \midrule
            \begin{minipage}[t]{\linewidth}
                \vspace{0pt}
                \[
                \begin{aligned}
                    & \weightedCoprocHead{\rigClear}{networkSecurity}{\typeof{n_0}{\Nats}} \\
                    & \weightedProcOut{\typeof{n}{\Nats}} \\
                    & \qquad \Requires{I(n_0)} \\
                    & \qquad \Ensures{P} \\
                    & \blockStart \\
                    & \qquad \ASSIGN{n}{n_0} \fatsemi \coAssert{I(n)} \fatsemi \\
                    & \qquad \coHavoc{n} \fatsemi \coValidate \fatsemi \coAssume{I(n)} \fatsemi \\
                    & \qquad \IF{n > 0} \\
                    & \qquad\qquad \ASSIGN{n}{n - 1} \fatsemi \\
                    & \qquad\qquad \symWeightedBranching~\{ \\
                    & \qquad\qquad\qquad \WEIGH{S} \\
                    & \qquad\qquad \}~\stmtElseStart \\
                    & \qquad\qquad\qquad \IF{n = 0} \\
                    & \qquad\qquad\qquad\qquad \WEIGH{C} \\
                    & \qquad\qquad\qquad \ELSE \\
                    & \qquad\qquad\qquad\qquad \WEIGH{TS} \\
                    & \qquad\qquad\qquad \} \\
                    & \qquad\qquad \} \fatsemi \\
                    & \qquad\qquad \coAssert{I(n)} \fatsemi \coAssume{\embed{\true}} \\
                    & \qquad \ELSE~\} \\
                    & \blockEnd
                \end{aligned}
                \]
            \end{minipage}
            &
            \begin{minipage}[t]{\linewidth}
                \vspace{0pt}
                \[
                \begin{aligned}
                    & \weightedProcHead{\rigEUReal}{networkSecurity}{\typeof{n_0}{\Nats}} \\
                    & \weightedProcOut{\typeof{n}{\Nats}} \\
                    & \qquad \Requires{\widehat{I}(n_0)} \\
                    & \qquad \Ensures{0} \\
                    & \blockStart \\
                    & \qquad \ASSIGN{n}{n_0} \fatsemi \Assert{\widehat{I}(n)} \fatsemi \\
                    & \qquad \Havoc{n} \fatsemi \Validate \fatsemi \Assume{\widehat{I}(n)} \fatsemi \\
                    & \qquad \IF{n > 0} \\
                    & \qquad\qquad \ASSIGN{n}{n - 1} \fatsemi \\
                    & \qquad\qquad \stmtDemonicStart \\
                    & \qquad\qquad\qquad \coAssert{2} \\
                    & \qquad\qquad \}~\stmtElseStart \\
                    & \qquad\qquad\qquad \IF{n = 0} \\
                    & \qquad\qquad\qquad\qquad \coAssert{1} \\
                    & \qquad\qquad\qquad \ELSE \\
                    & \qquad\qquad\qquad\qquad \coAssert{\infty} \\
                    & \qquad\qquad\qquad \} \\
                    & \qquad\qquad \} \fatsemi \\
                    & \qquad\qquad \Assert{\widehat{I}(n)} \fatsemi \Assume{\embed{\false}} \\
                    & \qquad \ELSE~\} \\
                    & \blockEnd
                \end{aligned}
                \]
            \end{minipage}
            \tabularnewline
            \bottomrule
        \end{tabular}
    \end{adjustbox}
    \caption{%
        Network-security example from \Cref{sec:network-security-clearance} and its translation to the $\eureal$ instance according to \Cref{def:clearance-translation}.
        The source coprocedure is translated to a target procedure, so \Caesar{} checks the source upper bound as a target lower bound under the order-reversing embedding.
        On the right, we use the embedding $\widehat{\mathsf{P}}=0$, $\widehat{\mathsf{C}}=1$, $\widehat{\mathsf{S}}=2$, and $\widehat{\mathsf{TS}}=\infty$.
        Here $I(n) \definedAs \iverson{n = 0} \stimes P \ssplus \iverson{n = 1} \stimes C \ssplus \iverson{n \ge 2} \stimes S$.
    }
    \label{fig:clearance-lowering-example}
\end{figure}

\subsection{Why-Semiring Translation}
\label{sec:appendix-why-lowering}

Let $A$ be a finite set of provenance atoms.
The carrier of the Why semiring is $\pDNF(A)$.
Source assertions therefore have type $H : \Sigma \to \pDNF(A)$.

To lower into Boolean \wHeyVL{}, we add one fresh Boolean variable $p_a$ for each provenance atom $a \in A$, and write $\widehat{\Sigma} \definedAs \Sigma \times (A \to \{\false,\true\})$ for the extended target state space.
An extended state is written $\widehat{\sigma} = (\sigma,\rho)$, where $\rho(a)$ is the value of~$p_a$.

For a Why formula $W \in \pDNF(A)$, define $\WhyBool{W}$ to be the Boolean assertion on~$\widehat{\Sigma}$ obtained by replacing each atom $a$ with~$p_a$, i.e., $\WhyBool{W}(\sigma,\rho) \definedAs \WhyEval{\rho}{W}$.
For assertions $H : \Sigma \to \pDNF(A)$, define $\WhyBool{H}(\sigma,\rho) \definedAs \WhyEval{\rho}{H(\sigma)}$.
The target verification calculus is still the ordinary Boolean calculus $\symVcBool$, now interpreted over assertions on~$\widehat{\Sigma}$.
For source \wHeyVLWhy{} programs in this fragment, $\stmtWeigh{W}$ with $W \in \pDNF(A)$ maps the postweight $H$ to $\sigma \mapsto W \land H(\sigma)$, and weighted choice combines branches by disjunction.
The pointwise Booleanization below preserves finite conjunctions and disjunctions of provenance formulas.
It does not, in general, preserve Heyting implication, coimplication, validity, or covalidity.
Thus, arbitrary $\stmtAssume{-}$, $\coAssume{-}$, $\stmtValidate$, and $\coValidate$ are not part of the structural Why lowering theorem.
The special fragment $\stmtValidate \fatsemi \stmtAssume{P}$, which is used in the proof rules, is still exact.
For Why assertions $P,H : \Sigma \to \pDNF(A)$, the original Why-\wHeyVL source validates the implication $P(\sigma) \impl_{\pDNF} H(\sigma)$, and therefore only checks whether this value is $\top$.
In a Heyting algebra this is equivalent to an order check, and for $\pDNF(A)$ the order is entailment of the represented predicates:
\begin{align*}
    P(\sigma) \impl_{\pDNF} H(\sigma) = \top
    &\Longleftrightarrow
    P(\sigma) \hleq H(\sigma) \\
    &\Longleftrightarrow
    \forall \rho : A \to \Bools.\; \WhyEval{\rho}{P(\sigma)} \Rightarrow \WhyEval{\rho}{H(\sigma)}.
\end{align*}
Thus, Boolean implication is used only for this entailment check, not to represent the Why implication itself.

\begin{definition}[Why-semiring translation]
    \label{def:why-translation}
    Fix an enumeration $A = \{a_1,\ldots,a_n\}$.
    On the supported fragment, the statement translation $\symWhyTrans : \wHeyVLWhy{} \to \wHeyVLBool{}$ over~$\widehat{\Sigma}$ is defined by
    \[
        \begin{aligned}
            \WhyTrans{\stmtWeigh{W}}
             & = \stmtAssert{\WhyBool{W}}, \\
            \WhyTrans{\stmtWeighted{S_1}{S_2}}
             & = \stmtAngelic{\WhyTrans{S_1}}{\WhyTrans{S_2}}, \\
            \WhyTrans{S_1 \fatsemi S_2}
             & = \WhyTrans{S_1} \fatsemi \WhyTrans{S_2}, \\
            \WhyTrans{\stmtAssert{H}}
             & = \stmtAssert{\WhyBool{H}}, \\
            \WhyTrans{\coAssert{H}}
             & = \coAssert{\WhyBool{H}}, \\
            \WhyTrans{\stmtAssume{G}}
             & = \stmtAssume{\WhyBool{G}}
             \qquad\text{if } G(\sigma) \in \{\bot,\top\} \text{ for all } \sigma, \\
            \WhyTrans{\stmtValidate \fatsemi \stmtAssume{P}}
             & =
             \stmtHavoc{p_{a_1}} \fatsemi \cdots \fatsemi
             \stmtHavoc{p_{a_n}} \fatsemi
             \stmtValidate \fatsemi
             \stmtAssume{\WhyBool{P}}, \\
            \WhyTrans{\stmtDemonic{S_1}{S_2}}
             & = \stmtDemonic{\WhyTrans{S_1}}{\WhyTrans{S_2}}, \\
            \WhyTrans{\stmtAngelic{S_1}{S_2}}
             & = \stmtAngelic{\WhyTrans{S_1}}{\WhyTrans{S_2}},
        \end{aligned}
    \]
    for weights $W \in \pDNF(A)$ and assertions $P,H,G : \Sigma \to \pDNF(A)$.
    State transformers that leave the fresh Boolean variables $p_a$ unchanged are translated unchanged.
\end{definition}

\begin{theorem}
    \label{thm:why-translation}
    For every statement $S$ in the supported fragment of \Cref{def:why-translation} and every Why assertion $H : \Sigma \to \pDNF(A)$,
    \[
        \vcB{\WhyTrans{S}}(\WhyBool{H})
        =
        \WhyBool{\vcWhy{S}(H)}
    \]
    as Boolean assertions on~$\widehat{\Sigma}$.
\end{theorem}

\begin{proof}
        We prove the claim pointwise by induction on the syntax of the supported fragment.
    Fix an extended state $\widehat{\sigma} = (\sigma,\rho) \in \widehat{\Sigma}$.
    It suffices to show
    \[
        \vcB{\WhyTrans{S}}(\WhyBool{H})(\widehat{\sigma})
        \qiff
        \WhyEval{\rho}{\vcWhy{S}(H)(\sigma)}.
    \]

    \begin{itemize}
        \proofcase{$S = \stmtWeigh{W}$}
        \begin{align*}
            & \vcB{\WhyTrans{\stmtWeigh{W}}}(\WhyBool{H})(\widehat{\sigma}) \\
                &\qiff \vcB{\stmtAssert{\WhyBool{W}}}(\WhyBool{H})(\widehat{\sigma})
                    && \inlineExplain{Def. translation} \\
                &\qiff \WhyBool{W}(\widehat{\sigma}) \land \WhyBool{H}(\widehat{\sigma})
                    && \inlineExplain{Def. \symVcBool} \\
                &\qiff \WhyEval{\rho}{W} \land \WhyEval{\rho}{H(\sigma)}
                    && \inlineExplain{Def. Booleanization} \\
                &\qiff \WhyEval{\rho}{W \land H(\sigma)}
                    && \inlineExplain{Evaluation of positive DNF} \\
                &\qiff \WhyEval{\rho}{\vcWhy{\stmtWeigh{W}}(H)(\sigma)}
                    && \inlineExplain{Def. \symVcWhy}
        \end{align*}

        \proofcase{$S = \stmtWeighted{S_1}{S_2}$}
        \begin{align*}
            & \vcB{\WhyTrans{\stmtWeighted{S_1}{S_2}}}(\WhyBool{H})(\widehat{\sigma}) \\
                &\qiff \vcB{\stmtAngelic{\WhyTrans{S_1}}{\WhyTrans{S_2}}}(\WhyBool{H})(\widehat{\sigma})
                    && \inlineExplain{Def. translation} \\
                &\qiff \vcB{\WhyTrans{S_1}}(\WhyBool{H})(\widehat{\sigma}) \lor \vcB{\WhyTrans{S_2}}(\WhyBool{H})(\widehat{\sigma})
                    && \inlineExplain{Def. \symVcBool} \\
                &\qiff \WhyEval{\rho}{\vcWhy{S_1}(H)(\sigma)} \lor \WhyEval{\rho}{\vcWhy{S_2}(H)(\sigma)}
                    && \inlineExplain{I.H.} \\
                &\qiff \WhyEval{\rho}{\vcWhy{S_1}(H)(\sigma) \lor \vcWhy{S_2}(H)(\sigma)}
                    && \inlineExplain{Evaluation of positive DNF} \\
                &\qiff \WhyEval{\rho}{\vcWhy{\stmtWeighted{S_1}{S_2}}(H)(\sigma)}
                    && \inlineExplain{Def. \symVcWhy}
        \end{align*}

        \proofcase{$S = \stmtAssume{G}$ with $G(\sigma) \in \{\bot,\top\}$ for all $\sigma$}
        If $G(\sigma) = \top$, then both sides are equivalent to $\WhyEval{\rho}{H(\sigma)}$:
        \[
            \WhyBool{G}(\sigma,\rho) \Rightarrow \WhyBool{H}(\sigma,\rho)
            \qiff
            \true \Rightarrow \WhyEval{\rho}{H(\sigma)}
            \qiff
            \WhyEval{\rho}{\top \impl_{\pDNF} H(\sigma)}.
        \]
        If $G(\sigma) = \bot$, then both sides are true:
        \[
            \WhyBool{G}(\sigma,\rho) \Rightarrow \WhyBool{H}(\sigma,\rho)
            \qiff
            \false \Rightarrow \WhyEval{\rho}{H(\sigma)}
            \qiff
            \WhyEval{\rho}{\bot \impl_{\pDNF} H(\sigma)}.
        \]

        \proofcase{$S = \stmtValidate \fatsemi \stmtAssume{P}$}
        The added havoc statements range over all valuations of the provenance variables, hence
        \begin{align*}
             & \vcB{\WhyTrans{\stmtValidate \fatsemi \stmtAssume{P}}}(\WhyBool{H})(\sigma,\rho) \\
            &\qiff
                \forall \rho' : A \to \{\false,\true\}.\;
                \WhyBool{P}(\sigma,\rho') \Rightarrow \WhyBool{H}(\sigma,\rho')
                \inlineExplain{Def. translation, \symVcBool} \\
            &\qiff
                \forall \rho' : A \to \{\false,\true\}.\;
                \WhyEval{\rho'}{P(\sigma)} \Rightarrow \WhyEval{\rho'}{H(\sigma)}
                \inlineExplain{Def. Booleanization} \\
            &\qiff
                P(\sigma) \hleq H(\sigma)
                \inlineExplain{Semantic order reflection} \\
            &\qiff
                P(\sigma) \impl_{\pDNF} H(\sigma) = \top
                \inlineExplain{Heyting implication} \\
            &\qiff
                \WhyEval{\rho}{\heylovalidate{P(\sigma) \impl_{\pDNF} H(\sigma)}}
                \inlineExplain{Def. validation} \\
            &\qiff
                \WhyEval{\rho}{\vcWhy{\stmtValidate \fatsemi \stmtAssume{P}}(H)(\sigma)}
                \inlineExplain{Def. \symVcWhy}
        \end{align*}

        The cases $S_1 \fatsemi S_2$, $\stmtDemonic{S_1}{S_2}$, $\stmtAngelic{S_1}{S_2}$, $\symAssert$, and $\symUp\symAssert$ are immediate from the induction hypotheses and the fact that $\WhyEval{\rho}{-}$ preserves conjunction and disjunction of positive DNF formulas.
        The ordinary state-transformer cases $\symAssign$, $\symHavoc$, and $\symUp\symHavoc$ are immediate when they do not update the fresh Boolean variables $p_a$.
    \end{itemize}
\end{proof}

\Cref{fig:why-lowering-example} shows the graph-reachability example and its translation.
The positive commands, crisp guard assumptions, and the Park-induction check are all covered by \Cref{def:why-translation,thm:why-translation}.

\begin{figure}[t]
    \centering
    \begingroup
    \setlength{\tabcolsep}{0pt}
    \begin{adjustbox}{max width=\linewidth}
        \begin{tabular}{@{}p{0.485\linewidth}@{\hspace{0.8em}}p{0.485\linewidth}@{}}
            \toprule
            \centering\makecell{Source\\$\wHeyVLWhy{}$} &
            \centering\makecell{Translated\\$\wHeyVLBool{}$}
            \tabularnewline
            \midrule
            \begin{minipage}[t]{\linewidth}
                \vspace{0pt}
                \[
                \begin{aligned}
                    & \weightedProcHead{\rigWhy}{graphReachability}{} \\
                    & \weightedProcOut{\typeof{node}{\Nats}} \\
                    & \qquad \Requires{X_{AB} \lor X_{AA}} \\
                    & \qquad \Ensures{\top} \\
                    & \blockStart \\
                    & \qquad \ASSIGN{node}{0} \fatsemi \Assert{I} \fatsemi \\
                    & \qquad \Havoc{node} \fatsemi \Validate \fatsemi \Assume{I} \fatsemi \\
                    & \qquad \IF{node \neq 1} \\
                    & \qquad\qquad \symWeightedBranching~\{ \\
                    & \qquad\qquad\qquad \ASSIGN{node}{1} \fatsemi \WEIGH{X_{AB}} \\
                    & \qquad\qquad \}~\stmtElseStart \\
                    & \qquad\qquad\qquad \ASSIGN{node}{0} \fatsemi \WEIGH{X_{AA}} \\
                    & \qquad\qquad \} \fatsemi \\
                    & \qquad\qquad \Assert{I} \fatsemi \Assume{\embed{\false}} \\
                    & \qquad \ELSE~\} \\
                    & \blockEnd
                \end{aligned}
                \]
            \end{minipage}
            &
            \begin{minipage}[t]{\linewidth}
                \vspace{0pt}
                \[
                \begin{aligned}
                    & \weightedProcHead{\rigBool}{graphReachability}{\typeof{p^0_{AB}, p^0_{AA}}{\Bools}} \\
                    & \weightedProcOut{\typeof{node}{\Nats}, \typeof{p_{AB}, p_{AA}}{\Bools}} \\
                    & \qquad \Requires{p^0_{AB} \lor p^0_{AA}} \\
                    & \qquad \Ensures{\top} \\
                    & \blockStart \\
                    & \qquad \ASSIGN{node}{0} \fatsemi \ASSIGN{p_{AB}}{p^0_{AB}} \fatsemi \ASSIGN{p_{AA}}{p^0_{AA}} \fatsemi \\
                    & \qquad \Assert{I_{\mathrm{bool}}} \fatsemi \\
                    & \qquad \Havoc{node} \fatsemi \Havoc{p_{AB}, p_{AA}} \fatsemi \\
                    & \qquad \Validate \fatsemi \Assume{I_{\mathrm{bool}}} \fatsemi \\
                    & \qquad \IF{node \neq 1} \\
                    & \qquad\qquad \stmtAngelicStart \\
                    & \qquad\qquad\qquad \ASSIGN{node}{1} \fatsemi \Assert{p_{AB}} \\
                    & \qquad\qquad \}~\stmtElseStart \\
                    & \qquad\qquad\qquad \ASSIGN{node}{0} \fatsemi \Assert{p_{AA}} \\
                    & \qquad\qquad \} \fatsemi \\
                    & \qquad\qquad \Assert{I_{\mathrm{bool}}} \fatsemi \Assume{\embed{\false}} \\
                    & \qquad \ELSE~\} \\
                    & \blockEnd
                \end{aligned}
                \]
            \end{minipage}
            \tabularnewline
            \bottomrule
        \end{tabular}
    \end{adjustbox}
    \endgroup
    \caption{Graph-reachability program $\graphreach$ from \Cref{fig:case_studies_prob_why} and its lowering to the Boolean instance on the extended state space. The provenance atoms $X_{AB}$ and $X_{AA}$ become mutable Boolean variables $p_{AB}$ and $p_{AA}$, initialized from read-only inputs $p^0_{AB}$ and $p^0_{AA}$. We encode the two nodes as $0$ and $1$, and write $I \definedAs \iverson{node = 1}\top \lor \iverson{node \neq 1}(X_{AB} \lor X_{AA})$ and $I_{\mathrm{bool}} \definedAs \WhyBool{I} = \iverson{node = 1}\top \lor \iverson{node \neq 1}(p_{AB} \lor p_{AA})$.}
    \label{fig:why-lowering-example}
\end{figure}

\begin{figure}[t]
    \centering
    \begingroup
    \setlength{\tabcolsep}{0pt}
    \begin{adjustbox}{max width=\linewidth}
        \begin{tabular}{@{}p{0.485\linewidth}@{\hspace{0.8em}}p{0.485\linewidth}@{}}
            \toprule
            \centering\makecell{Source\\$\wHeyVLWhy{}$} &
            \centering\makecell{Translated\\$\wHeyVLBool{}$}
            \tabularnewline
            \midrule
            \begin{minipage}[t]{\linewidth}
                \vspace{0pt}
                \[
                \begin{aligned}
                    & \weightedProcHead{\rigWhy}{graphReachability}{} \\
                    & \weightedProcOut{\typeof{e}{\Nats}} \\
                    & \qquad \Requires{X_1 \land X_2 \land X_3 \land X_4} \\
                    & \qquad \Ensures{\top} \\
                    & \blockStart \\
                    & \qquad \WEIGH{X_4}\ \fatsemi \ASSIGN{e}{4} \fatsemi \Assert{I} \fatsemi \\
                    & \qquad \Havoc{e} \fatsemi \Validate \fatsemi \Assume{I} \fatsemi \\[0.5em]
                    & \qquad \IF{e \neq 1} \\
                    & \qquad\qquad \IF{e = 2} \\
                    & \qquad\qquad\qquad \symWeightedBranching~\{ \\
                    & \qquad\qquad\qquad\qquad \WEIGH{X_3} \fatsemi \ASSIGN{e}{3} \\
                    & \qquad\qquad\qquad \}~\stmtElseStart \\
                    & \qquad\qquad\qquad\qquad \WEIGH{X_1} \fatsemi \ASSIGN{e}{1} \\
                    & \qquad\qquad\qquad \} \fatsemi \\
                    & \qquad\qquad \}~\stmtElseStart \\
                    & \qquad\qquad\qquad \IF{e = 3} \\
                    & \qquad\qquad\qquad\qquad \symWeightedBranching~\{ \\
                    & \qquad\qquad\qquad\qquad\qquad \WEIGH{X_2} \fatsemi \ASSIGN{e}{2} \\
                    & \qquad\qquad\qquad\qquad \}~\stmtElseStart \\
                    & \qquad\qquad\qquad\qquad\qquad \WEIGH{X_1} \fatsemi \ASSIGN{e}{1} \\
                    & \qquad\qquad\qquad\qquad \} \fatsemi \\
                    & \qquad\qquad\qquad \}~\stmtElseStart \\
                    & \qquad\qquad\qquad\qquad \IF{(e = 4) \vee (e = 5)} \\
                    & \qquad\qquad\qquad\qquad\qquad \WEIGH{X_2} \fatsemi \ASSIGN{e}{2} \\
                    & \qquad\qquad\qquad\qquad \}~\stmtElseStart \\
                    & \qquad\qquad\qquad\qquad\qquad \WEIGH{X_3} \fatsemi \ASSIGN{e}{3} \\
                    & \qquad\qquad\qquad\qquad \} \\
                    & \qquad\qquad\qquad \} \\
                    & \qquad\qquad \} \\[0.5em]
                    & \qquad\qquad \Assert{I} \fatsemi \Assume{\embed{\false}} \\
                    & \qquad \ELSE~\} \\
                    & \blockEnd
                \end{aligned}
                \]
            \end{minipage}
            &
            \begin{minipage}[t]{\linewidth}
                \vspace{0pt}
                \[
                \begin{aligned}
                    & \weightedProcHead{\rigBool}{graphReachability}{\typeof{p^0_1, p^0_2, p^0_3, p^0_4}{\Bools}} \\
                    & \weightedProcOut{\typeof{e}{\Nats}, \typeof{p_1, p_2, p_3, p_4}{\Bools}} \\
                    & \qquad \Requires{p^0_1 \land p^0_2 \land p^0_3 \land p^0_4} \\
                    & \qquad \Ensures{\top} \\
                    & \blockStart \\
                    & \qquad \ASSIGN{p_1}{p^0_1} \fatsemi \ASSIGN{p_2}{p^0_2} \fatsemi \ASSIGN{p_3}{p^0_3} \fatsemi \ASSIGN{p_4}{p^0_4} \fatsemi \\
                    & \qquad \stmtAssert {p_4}\ \fatsemi \ASSIGN{e}{4} \fatsemi \Assert{I_{bool}} \fatsemi \\
                    & \qquad \Havoc{e} \fatsemi \Havoc{p_1, p_2, p_3, p_4} \fatsemi \\
                    & \qquad \Validate \fatsemi \Assume{I_{bool}} \fatsemi \\[0.5em]
                    & \qquad \IF{e \neq 1} \\
                    & \qquad\qquad \IF{e = 2} \\
                    & \qquad\qquad\qquad \symAngelic~\{ \\
                    & \qquad\qquad\qquad\qquad \stmtAssert {p_3} \fatsemi \ASSIGN{e}{3} \\
                    & \qquad\qquad\qquad \}~\stmtElseStart \\
                    & \qquad\qquad\qquad\qquad \stmtAssert {p_1} \fatsemi \ASSIGN{e}{1} \\
                    & \qquad\qquad\qquad \} \fatsemi \\
                    & \qquad\qquad \}~\stmtElseStart \\
                    & \qquad\qquad\qquad \IF{e = 3} \\
                    & \qquad\qquad\qquad\qquad \symAngelic~\{ \\
                    & \qquad\qquad\qquad\qquad\qquad \stmtAssert {p_2} \fatsemi \ASSIGN{e}{2} \\
                    & \qquad\qquad\qquad\qquad \}~\stmtElseStart \\
                    & \qquad\qquad\qquad\qquad\qquad \stmtAssert {p_1} \fatsemi \ASSIGN{e}{1} \\
                    & \qquad\qquad\qquad\qquad \} \fatsemi \\
                    & \qquad\qquad\qquad \}~\stmtElseStart \\
                    & \qquad\qquad\qquad\qquad \IF{(e = 4) \vee (e = 5)} \\
                    & \qquad\qquad\qquad\qquad\qquad \stmtAssert {p_2} \fatsemi \ASSIGN{e}{2} \\
                    & \qquad\qquad\qquad\qquad \}~\stmtElseStart \\
                    & \qquad\qquad\qquad\qquad\qquad \stmtAssert {p_3} \fatsemi \ASSIGN{e}{3} \\
                    & \qquad\qquad\qquad\qquad \} \\
                    & \qquad\qquad\qquad \} \\
                    & \qquad\qquad \} \\[0.5em]
                    & \qquad\qquad \Assert{I_{bool}} \fatsemi \Assume{\embed{\false}} \\
                    & \qquad \ELSE~\} \\
                    & \blockEnd
                \end{aligned}
                \]
            \end{minipage}
            \tabularnewline
            \bottomrule
        \end{tabular}
    \end{adjustbox}
    \endgroup
    \caption{Employee graph-reachability program from \Cref{sec:database-provenance} and its lowering to the Boolean instance on the extended state space. The provenance atoms $X_i$ become mutable Boolean variables $p_i$, initialized from read-only inputs $p^0_i$. We write $I_{\mathrm{bool}} \definedAs \WhyBool{I}$.}
    \label{fig:why-lowering-example-employees}
\end{figure}

\begin{figure}[t]
    \centering
    \begingroup
    \setlength{\tabcolsep}{0pt}
    \begin{adjustbox}{max width=\linewidth}
        \begin{tabular}{@{}p{0.485\linewidth}@{\hspace{0.8em}}p{0.485\linewidth}@{}}
            \toprule
            \centering\makecell{$\wHeyVLEUReal{}$} &
            \centering\makecell{$\wHeyVLWhy{}$}
            \tabularnewline
            \midrule
            \begin{minipage}[t]{\linewidth}
                \vspace{0pt}
                \[
                \begin{aligned}
                    & \weightedCoprocHead{\rigEUReal}{queueing}{\typeof{k_{0}}{\Nats}, \typeof{e_{0}}{\Nats}, \typeof{b_{0}}{\Bools}} \\
                    & \weightedProcOut{\typeof{k}{\Nats}, \typeof{e}{\Nats}, \typeof{b}{\Bools}} \\
                    & \quad \Requires{e_0 + 4 \cdot k_0} \\
                    & \quad \Ensures{e} \\
                    & \blockStart \\
                    & \quad \ASSIGN{k}{k_0} \fatsemi\, \ASSIGN{e}{e_0} \fatsemi\, \ASSIGN{b}{b_0} \fatsemi \\[5pt]
                    & \quad \coAssert{e + 4 \cdot k + 16 \cdot \iverson{b \neq b_0}} \fatsemi \\
                    & \quad \coHavoc{k, e, b} \fatsemi \coValidate \fatsemi \\
                    & \quad \coAssume{e + 4 \cdot k + 16 \cdot \iverson{b \neq b_0}} \fatsemi \\[5pt]
                    & \quad\IF{k > 0} \\
                    & \quad\quad \symWeightedBranching~\{\, \WEIGH{0.8} \} \\
                    & \quad\quad \stmtElseStart \, \WEIGH{0.2} \fatsemi \\
                    & \quad\quad\quad \IF{b} \,\ASSIGN{e}{e + 4}\, \} \\
                    & \quad\quad\quad \stmtElseStart \,\ASSIGN{e}{e + 10}\, \} \fatsemi \\
                    & \quad\quad\quad \ASSIGN{b}{\neg b} \\
                    & \quad\quad \} \fatsemi \, \ASSIGN{k}{k-1} \fatsemi \\[5pt]
                    & \quad\quad \Assert{e + 4 \cdot k + 16 \cdot \iverson{b \neq b_0}} \\[5pt]
                    & \quad\quad \IF{k > 0} \\
                    & \quad\quad\quad \symWeightedBranching~\{\, \WEIGH{0.8} \} \\
                    & \quad\quad\quad \stmtElseStart \, \WEIGH{0.2} \fatsemi \\
                    & \quad\quad\quad\quad \IF{b} \,\ASSIGN{e}{e + 4}\, \} \\
                    & \quad\quad\quad\quad \stmtElseStart \,\ASSIGN{e}{e + 10}\, \} \fatsemi \\
                    & \quad\quad\quad\quad \ASSIGN{b}{\neg b} \\
                    & \quad\quad\quad \} \fatsemi \, \ASSIGN{k}{k-1} \fatsemi \\[5pt]
                    & \quad\quad\quad \coAssert{e + 4 \cdot k + 16 \cdot \iverson{b \neq b_0}} \fatsemi \\
                    & \quad\quad\quad \coAssume{\embed{\true}} \\
                    & \quad\quad \ELSE~\} \\
                    & \quad \ELSE~\} \\
                    & \blockEnd
                \end{aligned}
                \]
            \end{minipage}
            &
            \begin{minipage}[t]{\linewidth}
                \vspace{0pt}
                \[
                \begin{aligned}
                    & \weightedProcHead{\rigWhy}{graphReachability}{} \\
                    & \weightedProcOut{\typeof{node}{\Nats}} \\
                    & \quad \Requires{X_{AB} \lor X_{AA}} \\
                    & \quad \Ensures{\top} \\
                    & \blockStart \\
                    & \quad \ASSIGN{node}{0} \fatsemi \\[5pt]
                    & \quad \Assert{I} \fatsemi\, \Havoc{node} \fatsemi \\
                    & \quad \Validate \fatsemi\, \Assume{I} \fatsemi \\[5pt]
                    & \quad \IF{node \neq 1} \\
                    & \quad\quad \symWeightedBranching~\{ \\
                    & \quad\quad\quad \ASSIGN{node}{1} \fatsemi \WEIGH{X_{AB}} \\
                    & \quad\quad \}~\stmtElseStart \\
                    & \quad\quad\quad \ASSIGN{node}{0} \fatsemi \WEIGH{X_{AA}} \\
                    & \quad\quad \} \fatsemi \\[5pt]
                    & \quad\quad \Assert{I} \fatsemi \Assume{\embed{\false}} \\
                    & \quad \ELSE~\} \\
                    & \blockEnd
                \end{aligned}
                \]
            \end{minipage}
            \tabularnewline
            \bottomrule
        \end{tabular}
    \end{adjustbox}
    \endgroup
    \caption{The job scheduling program $queueing$ uses $k$-induction with $k = 2$ to prove $\wp{C_{Q}}(e) \heyloleq e_0 + 4 * k_0$. The program $graphReachability$ uses Park induction to prove $\wlp{C_{GR}}(\top) \heylogeq X_{AB} \vee X_{AA}$.}
    \label{fig:case_studies_prob_why}
\end{figure}

\begin{figure}[t]
    \centering
    \begingroup
    \setlength{\tabcolsep}{0pt}
    \begin{adjustbox}{max width=\linewidth}
        \begin{tabular}{@{}p{0.485\linewidth}@{\hspace{0.8em}}p{0.485\linewidth}@{}}
            \toprule
            \centering\makecell{$\wHeyVLClearance{}$} &
            \centering\makecell{$\wHeyVLLang{}$}
            \tabularnewline
            \midrule
            \begin{minipage}[t]{\linewidth}
                \vspace{0pt}
                \[
                \begin{aligned}
                    & \weightedCoprocHead{\rigClear}{networkSecurity}{\typeof{n_0}{\Nats}} \\
                    & \weightedProcOut{\typeof{n}{\Nats}} \\
                    & \quad \Requires{\min\{\, \iverson{n_0 = 0}\cdot \mathsf{P}, \iverson{n_0 = 1}\cdot \mathsf{C},} \\
                    & \quad\quad\quad\quad\quad {\color{prepostColor} \iverson{n_0 \geq 2}\cdot \mathsf{S} \,\} \quad\quad = I(n_0)} \\
                    & \quad \Ensures{\mathsf{P}} \\
                    & \blockStart \\
                    & \quad \ASSIGN{n}{n_0} \fatsemi \\[5pt]
                    & \quad \coAssert{I(n)} \fatsemi\, \coHavoc{n} \fatsemi \\
                    & \quad \coValidate \fatsemi \coAssume{I(n)} \fatsemi \\[5pt]
                    & \quad \IF{n > 0} \\
                    & \quad\quad \ASSIGN{n}{n - 1} \fatsemi \\
                    & \quad\quad \symWeightedBranching~\{ \\
                    & \quad\quad\quad \WEIGH{\mathsf{S}} \\
                    & \quad\quad \}~\stmtElseStart \\
                    & \quad\quad\quad \IF{n = 0} \\
                    & \quad\quad\quad\quad \WEIGH{\mathsf{C}} \\
                    & \quad\quad\quad \ELSE \\
                    & \quad\quad\quad\quad \WEIGH{\mathsf{TS}} \\
                    & \quad\quad\quad \} \\
                    & \quad\quad \} \fatsemi \\[5pt]
                    & \quad\quad \coAssert{I(n)} \fatsemi\, \coAssume{\embed{\true}} \\
                    & \quad \ELSE~\} \\
                    & \blockEnd
                \end{aligned}
                \]
            \end{minipage}
            &
            \begin{minipage}[t]{\linewidth}
                \vspace{0pt}
                \[
                \begin{aligned}
                    & \weightedCoprocHead{\rigClear}{base}{\typeof{\ell_0}{\listtype}} \\
                    & \weightedProcOut{\typeof{\ell}{\listtype}} \\
                    & \quad \Requires{I(0)} \\
                    & \quad \Ensures{\bot} \\
                    & \blockStart \, \coAssert{\embed{\true}} \, \blockEnd \\[15pt]
                    & \weightedCoprocHead{\rigClear}{step}{\typeof{\ell_0}{\listtype}} \\
                    & \weightedProcOut{\typeof{\ell}{\listtype}} \\
                    & \quad \Requires{I(n+1)} \\
                    & \quad \Ensures{\bot} \\
                    & \blockStart \\
                    & \quad \ASSIGN{\ell}{\ell_0} \fatsemi \\[5pt]
                    & \quad \IF{\true} \\
                    & \quad\quad \symWeightedBranching~\{ \, \ASSIGN{i}{1} \,\} \\
                    & \quad\quad \stmtElseStart \symWeightedBranching~\{ \, \ASSIGN{i}{2} \,\} \,\ldots \\
                    & \quad\quad\quad\quad \ldots\, \stmtElseStart \, \ASSIGN{i}{N} \,\} \,\} \fatsemi \\[5pt]
                    & \quad\quad \IF{\ell[i] = v} \\
                    & \quad\quad\quad \WEIGH{S_i}\fatsemi\, \ASSIGN{v}{\ell[i] + 1} \fatsemi\, \ASSIGN{\ell[i]}{v} \\
                    & \quad\quad \ELSE \\
                    & \quad\quad\quad \WEIGH{F_i}\fatsemi\, \ASSIGN{\ell[i]}{v} \\
                    & \quad\quad \} \\[5pt]
                    & \quad\quad \coAssert{I(n)} \fatsemi \coAssume{\embed{\true}} \\
                    & \quad \ELSE~\} \\
                    & \blockEnd
                \end{aligned}
                \]
            \end{minipage}
            \tabularnewline
            \bottomrule
        \end{tabular}
    \end{adjustbox}
    \endgroup
    \caption{The program $networkSecurity$ uses Park induction to prove $\wp{C_{AC}}(P) \geq I(n_0)$ with $I(n) = \min\{ \iverson{n = 0}\cdot \mathsf{P}, \iverson{n = 1}\cdot \mathsf{C}, \iverson{n \geq 2}\cdot \mathsf{S}\}$. The programs $base$ and $step$ show $\wlp{C_{CAS}}(\bot) \heyloleq \inf_{n\in\Nats} I(n)$ via $\omega$-superinvariants.}
    \label{fig:case_studies_clear_lang}
\end{figure}

\section{\titleFullref{Theory and Proofs of }{sec:case-studies}}
\label{app:case-studies}

\begin{figure}[t]
    \centering
    \[
    \begin{aligned}
        & C_Q \in \wGCL_{\prob}
        \qquad \leqannotate{e + 4 \cdot k} \\
        & \WHILE{k > 0} \\
        & \qquad \{ \\
        & \qquad\qquad \WEIGH{0.8} \\
        & \qquad \} \, \BranchSymbol \, \{ \\
        & \qquad\qquad \WEIGH{0.2} \fatsemi \\
        & \qquad\qquad \IF{b} \\
        & \qquad\qquad\qquad \ASSIGN{e}{e + 4} \\
        & \qquad\qquad \ELSE \\
        & \qquad\qquad\qquad \ASSIGN{e}{e + 10} \\
        & \qquad\qquad \} \fatsemi \\
        & \qquad\qquad \ASSIGN{b}{\neg b} \\
        & \qquad \} \fatsemi \\
        & \qquad \ASSIGN{k}{k-1} \fatsemi \\
        & \} \\
        & \wpannotate{e}
    \end{aligned}
    \]
    \caption{The job scheduling program $C_Q$ satisfies $\wp{C_{Q}}(e) \heyloleq e_0 + 4 * k_0$.}
    \label{fig:case_studies_wGCL_prob}
\end{figure}

\begin{figure}[t]
    \centering

    \begin{align*}
        & C_{GR} \in \wGCL_{\why} \qquad \geqannotate{X_1 \wedge X_2 \wedge X_3 \wedge X4} \\
        & \WEIGH{X_4} \fatsemi \ASSIGN{e}{4} \fatsemi \\[0.5em]
        & \WHILE{e \neq 1} \\
        & \qquad \IF{e = 2} \\
        & \qquad\qquad \BRANCH{\WEIGH{X_3} \fatsemi \ASSIGN{e}{3}}{\WEIGH{X_1} \fatsemi \ASSIGN{e}{1}} \\
        & \qquad \ELSE \\
        & \qquad\qquad \IF{e = 3} \\
        & \qquad\qquad\qquad \BRANCH{\WEIGH{X_2} \fatsemi \ASSIGN{e}{2}}{\WEIGH{X_1} \fatsemi \ASSIGN{e}{1}} \\
        & \qquad\qquad \ELSE \\
        & \qquad\qquad\qquad \IF{(e = 4) \vee (e = 5)} \\
        & \qquad\qquad\qquad\qquad \WEIGH{X_2} \fatsemi \ASSIGN{e}{2} \\
        & \qquad\qquad\qquad \ELSE \\
        & \qquad\qquad\qquad\qquad \WEIGH{X_3} \fatsemi \ASSIGN{e}{3} \\
        & \qquad\qquad\qquad \} \\
        & \qquad\qquad \} \\
        & \qquad \} \\
        & \} \\
        & \wlpannotate{\top}
    \end{align*}

    \caption{The recursive database provenance program $C_{GR}$ satisfies $\wlp{C_{GR}}(\top) \heylogeq X_{AB} \vee X_{AA}$.}
    \label{fig:case_studies_wGCL_why}
\end{figure}

\begin{figure}[t]
    \centering
    \begin{adjustbox}{max width=\linewidth}
        \(
        \begin{aligned}
            & \makecell{C_{AC}\in\wGCL_{\clear}} \\
            & \geqannotate{I(n)} \\
            & \WHILE{n > 0} \\
            & \qquad \ASSIGN{n}{n - 1} \fatsemi \\
            & \qquad \{ \\
            & \qquad\qquad \WEIGH{S} \\
            & \qquad \} \mathrel{\BranchSymbol} \{ \\
            & \qquad\qquad \IF{n = 0} \\
            & \qquad\qquad\qquad \WEIGH{C} \\
            & \qquad\qquad \ELSE \\
            & \qquad\qquad\qquad \WEIGH{TS} \\
            & \qquad\qquad \} \\
            & \qquad \} \\
            & \} \\
            & \wlpannotate{P}
        \end{aligned}
        \)
    \end{adjustbox}
    \caption{The program $C_{AC}$ describes network security and satisfies $\wlp{C_{AC}}(P) \heylogeq I(n)$.}
    \label{fig:case_studies_wGCL_clear}
\end{figure}

\subsection{Probabilities: Queueing Background Jobs}
\label{app:case-studies-probabilities}

We show that the probability module $\modprob$ over the monoid $\monprob$ is an $\omega$-bicontinuous monoid-module with Heyting (co-)implications.

\begin{lemma}
    The probability module $\modprob = \modproblong$ over the monoid $\monprob = \monproblong$ is a monoid-module (see \Cref{def:monoid-module}).
\end{lemma}

\begin{proof}
    We observe that $\monproblong$ and $(\PosRealsInf, +, 0)$ are monoids, where the latter is also commutative due to the addition $+$. We check the properties of the scalar multiplication $\cdot: \PosReals \times \PosRealsInf \to \PosRealsInf$, $(p, c) \mapsto p \cdot c$.
    It is associative ($(p \cdot q)\cdot c = p\cdot (q \cdot c)$) and distributive ($p \cdot (c + d) = p\cdot c + p\cdot d$) for $p, q \in \PosReals$ and $c \in \PosRealsInf$. Furthermore, the neutral element $1$ is neutral ($1 \cdot c = c$) and the neutral element $0$ is annihilating ($p \cdot 0 = 0$), hence all properties are satisfied.
\end{proof}

\begin{lemma}
    The probability module $\modprob$ over the monoid $\monprob$ is $\omega$-bicontinuous.
\end{lemma}

\begin{proof}
    The natural order is defined as $c \sleq d$ if and only if there exists $e \in \PosRealsInf$ such that $c + e = d$ if and only if $c \leq d$ in the usual order of the non-negative reals. Hence, the natural order is a partial (even total) order with least element $0$, greatest element $\infty$, where suprema and infima of $\omega$-chains are the usual supremum and infimum in the non-negative reals (since the non-negative reals are complete, they always exist).

    We show that the scalar multiplication $\cdot$ is $\omega$-bicontinuous in the second argument and assume that $p \in \PosReals$ and $(c_i)_{i \in \Nats} \subseteq \PosRealsInf$ is an ascending (descending) $\omega$-chain. Then
    $$
        p \cdot \sup_{i \in \Nats} c_i = p \cdot \sup_{i \in \Nats} c_i = \sup_{i \in \Nats} p \cdot c_i \quad \text{and} \quad p \cdot \inf_{i \in \Nats} c_i = p \cdot \inf_{i \in \Nats} c_i = \inf_{i \in \Nats} p \cdot c_i,
    $$
    since the scalar multiplication is just the usual multiplication of non-negative reals, which is continuous. The addition is also $\omega$-bicontinuous using the same argument.
\end{proof}

\begin{lemma}
    The probability module $\modprob$ over the monoid $\monprob$ admits Heyting (co-)implications.
\end{lemma}

\begin{proof}
    Since the module $\modprob$ is totally ordered, \Cref{thm:totimpl_totcoimpl_bounded_lattice} states that the implications
    $$
        c \totimpl d = \left. \begin{cases}
            \infty & \text{if } c \leq d, \\
            d & \text{otherwise}
        \end{cases} \right\}
        \quad \text{and} \quad
        c \totcoimpl d = \left. \begin{cases}
            0 & \text{if } d \leq c, \\
            d & \text{otherwise}
        \end{cases} \right\}.
    $$
    specify a bi-Heyting algebra $(\PosRealsInf, \sleq, \min, \max, \totimpl, \totcoimpl)$ for $c, d \in \PosRealsInf$.
\end{proof}

\subsection{Why: Database Provenance}
\label{app:case-studies-why}

We show that the Why module $\modwhy$ over the monoid $\monwhy$ is an $\omega$-bicontinuous monoid-module with Heyting (co-)implications. The set of positive disjunctive normal form propositional formulae $\pDNF$ is the set of propositional formulae in disjunctive normal form, where all literals are positive and we consider formulae up to logical equivalence.

The Why module $\modwhy$ is a semiring, hence we call it Why semiring in the main text.

\begin{lemma}
    The Why module $\modwhy = \modwhylong$ over the monoid $\monwhy = \monwhylong$ is a monoid-module (see \Cref{def:monoid-module}).
\end{lemma}

\begin{proof}
    We observe that $\monwhylong$ and $(\pDNF, \vee, 0)$ are monoids, where the latter is also commutative due to the disjunction $\vee$. We check the properties of the scalar multiplication $\wedge: \pDNF \times \pDNF \to \pDNF$, $(p, q) \mapsto p \wedge q$.
    It is associative ($(p \wedge q)\wedge r = p\wedge (q \wedge r)$) and distributive ($p \wedge (q \vee r) = (p\wedge q) \vee (p\wedge r)$) for $p, q, r\in \pDNF$. Furthermore, the neutral element $1$ is neutral ($1 \wedge p = p$) and the neutral element $0$ is annihilating ($p \wedge 0 = 0$), hence all properties are satisfied.
\end{proof}

\begin{lemma}
    The Why module $\modwhy$ over the monoid $\monwhy$ is $\omega$-bicontinuous.
\end{lemma}

\begin{proof}
    The natural order is defined as $p \sleq q$ if and only if there exists $r \in \pDNF$ such that $p \vee r = q$ if and only if $p \leq q$ in the usual order of propositional formulae (where $0$ is the least element and $1$ is the greatest element). Hence, the natural order is a partial order with least element $0$, greatest element $1$, where suprema and infima of $\omega$-chains exist, since the set of propositional formulae is finite (up to logical equivalence).

    We show that the scalar multiplication $\wedge$ is $\omega$-bicontinuous in the second argument and assume that $p \in \pDNF$ and $(q_i)_{i \in \Nats} \subseteq \pDNF$ is an ascending (descending) $\omega$-chain. Since the set of propositional formulae is finite, the $\omega$-chain is eventually constant, i.e., there exists $j \in \Nats$ such that $q_i = q_j$ for all $i \geq j$. Hence,
    $$
        p \wedge \sup_{i \leq j} q_i = \sup_{i \leq j} (p \wedge q_i)
        \quad \text{and} \quad
        p \wedge \inf_{i \leq j} q_i = \inf_{i \leq j} (p \wedge q_i).
    $$
    The additive operation $\vee$ is also $\omega$-bicontinuous using the same argument.
\end{proof}

For the Heyting (co-)implications of the Why module $\modwhy$ over the monoid $\monwhy$, look at the remark in \Cref{sec:heyting}.

\subsection{Clearance: Network Security Analysis}
\label{app:case-studies-clearance}

We show that the clearance module $\modclear$ over the monoid $\monclear$ is an $\omega$-bicontinuous monoid-module with Heyting (co-)implications.

\begin{lemma}
    The clearance module $\modclear$ over the monoid $\monclear$ is a monoid-module (see \Cref{def:monoid-module}).
\end{lemma}

\begin{proof}
    Recall that $\modclear = \modclearlong$ and $\monclear = \monclearlong$.
    We observe that $\monclearlong$ and $(\clearUniverse, \min, \mathsf{TS})$ are monoids, where the latter is also commutative due to the minimum $\min$. We check the properties of the scalar multiplication $\max: \clearUniverse \times \clearUniverse \to \clearUniverse$, $(p, q) \mapsto \max\{p, q\}$.
    It is associative and distributive over $\min$, since for all $p,q,r \in \clearUniverse$,
    \[
        (p \max q)\max r = p\max(q \max r),
        \qquad
        \max\{p,\min\{q,r\}\} =
        \min\{\max\{p,q\}, \max\{p,r\}\}.
    \]
    Furthermore, the neutral element $\mathsf{P}$ is neutral ($\max\{\mathsf{P}, p\} = p$) and $\mathsf{TS}$ is annihilating ($\max\{p, \mathsf{TS}\} = \mathsf{TS}$), hence all properties are satisfied.
\end{proof}

\begin{lemma}
    The clearance module $\modclear$ over the monoid $\monclear$ is $\omega$-bicontinuous.
\end{lemma}

\begin{proof}
    The natural order is defined as $p \sleq q$ if and only if there exists $r \in \clearUniverse$ such that $p \min r = q$ if and only if $p \geq q$ holds on the clearance levels. Hence, the natural order is a partial (even total) order with least element $\mathsf{Top Secret}$, greatest element $\mathsf{Public}$, where all ascending (descending) $\omega$-chains become constant after at most 4 steps and hence the supremum (infimum) exists.

    We show that the scalar multiplication $\max$ is $\omega$-bicontinuous in the second argument and assume that $p \in \clearUniverse$ and $(q_i)_{i \in \Nats} \subseteq \clearUniverse$ is an ascending (descending) $\omega$-chain. Since the $\omega$-chain is eventually constant, there exists $j \in \Nats$ such that $q_i = q_j$ for all $i \geq j$. Hence,
    $$
        \max\{ p, \sup_{i \leq j} q_i \} = \sup_{i \leq j} \max\{ p, q_i \}
        \quad \text{and} \quad
        \max\{p, \inf_{i \leq j} q_i \} = \inf_{i \leq j} \max\{ p, q_i \}.
    $$
    The additive operation $\min$ is also $\omega$-bicontinuous using the same argument.
\end{proof}

Note that the natural order $\sleq$ is the reverse of the semantic order $\geq$ on the clearance levels. This is also the reason why we compute the upper bound $\wp{C_{AC}}(\mathsf{TS}) \sleq I$ to obtain the minimal clearance level $I \leq \wp{C_{AC}}(\mathsf{TS})$.

\begin{lemma}
    The clearance module $\modclear$ over the monoid $\monclear$ admits Heyting (co-)implications.
\end{lemma}

\begin{proof}
    Since the module $\modclear$ is totally ordered, \Cref{thm:totimpl_totcoimpl_bounded_lattice} states that the implications
    $$
        p \totimpl q = \left. \begin{cases}
            \mathsf{P} & \text{if } p \geq q, \\
            q & \text{otherwise}
        \end{cases} \right\}
        \quad \text{and} \quad
        p \totcoimpl q = \left. \begin{cases}
            \mathsf{TS} & \text{if } q \geq p, \\
            q & \text{otherwise}
        \end{cases} \right\}.
    $$
    specify a bi-Heyting algebra $(\clearUniverse, \geq, \max, \min, \totimpl, \totcoimpl)$ for $p, q \in \clearUniverse$.
\end{proof}

\end{document}